\documentclass[11pt, a4paper, twoside, titlepage]{report}

\usepackage{graphicx} 
\usepackage{bm}
\usepackage{graphicx}
\usepackage{epstopdf}
\usepackage[cmex10]{amsmath}
\usepackage{amssymb}
\usepackage{enumerate}
\usepackage{amsfonts}
\usepackage{float}
\usepackage{epigraph}
\usepackage{amsmath}
\usepackage{cuted}
\usepackage{pdfpages}
\usepackage{booktabs}
\usepackage{qcircuit}
\usepackage{braket}
\usepackage{placeins}
\usepackage{bm}
\usepackage{graphicx}
\usepackage{epstopdf}
\usepackage[usenames,dvipsnames]{pstricks}
\usepackage{epsfig}
\usepackage{amsmath}
\usepackage{amssymb}
\usepackage{amsfonts}
\usepackage{amsmath}
\newcommand{\foo}{\hspace{-2.3pt}$\bullet$ \hspace{5pt}}
\usepackage{fancybox}
\usepackage{fancyhdr}
\usepackage{float}
\usepackage{lipsum}
\usepackage{mathtools}
\usepackage{cuted}
\usepackage{glossaries}
\usepackage{dsfont}
\usepackage{placeins}
\usepackage{booktabs}
\usepackage{multirow}
\usepackage{multicol}
\usepackage{dblfloatfix}
\usepackage{caption}
\usepackage[margin=1cm]{caption}
\usepackage{subcaption}
\usepackage{chngcntr}
\usepackage[toc,page]{appendix}
\usepackage{amsmath}
\usepackage{dsfont}
\usepackage{amsthm}
\newtheorem{theorem}{Theorem}[section]
\newtheorem{lemma}{Lemma}[section]
\newtheorem{definition}{Definition}[section]
\newtheorem{corollary}{Corollary}[theorem]

\newtheorem{proposition}{Proposition}[section]
\newtheorem{postulate}{Postulate}
\newtheorem*{postulate2*}{Postulate 2$'$}

\usepackage{tesis_ESI_en}
\usepackage{amsmath}
\usepackage{pstricks}

\newcommand{\comment}[1]{}

\begin{document}

\pagenumbering{roman}
\pagestyle{empty}
\pagestyle{empty}
\vspace*{6cm}
\begin{center}
\begin{minipage}[c][0.7\height][t]{\textwidth}
	\begin{center}
	\LARGE{\bf Decoherence and Quantum Error Correction for Quantum Computing and Communications}
	\end{center}
\end{minipage}\end{center}
\cleardoublepage


\begin{center}
\begin{minipage}[c][0.8\height][t]{\textwidth}
	\begin{center}
	\LARGE{\bf{Decoherence and Quantum Error Correction for Quantum Computing and Communications}}
	\end{center}
\end{minipage}\vspace*{1cm}
\begin{minipage}[c][0.2\height][t]{\textwidth}
	\begin{center}
	\Large{\bf{by}}
	\end{center}
\end{minipage}\vspace*{1cm}
\begin{minipage}[c][0.8\height][t]{\textwidth}
	\begin{center}
	\Large{\bf{Josu Etxezarreta Martinez}}
	\end{center}
\end{minipage}
\begin{minipage}[c]{\textwidth}
\begin{center}
\vspace{1cm}
\includegraphics[scale=0.8]{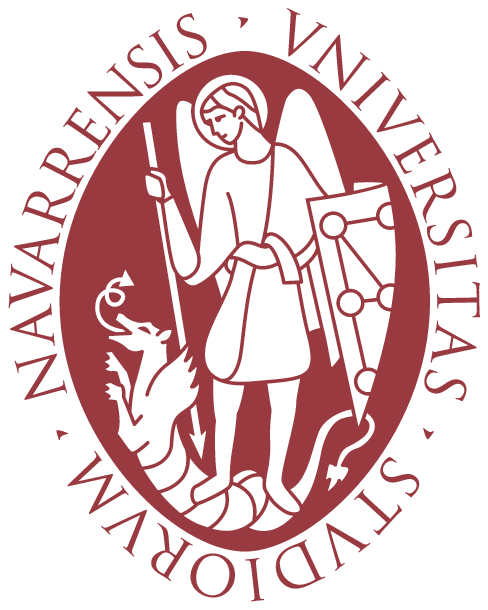}
\vspace{1cm}
\end{center}
\end{minipage}
\vspace{1cm}
\begin{minipage}[t][1.1\height][t]{1.0\textwidth}
	\begin{center}
	\large In Partial Fulfillment of the Requirements \\ for the Degree of \\ Doctor of Philosophy

	\end{center}
	\vspace{0.24cm}
\end{minipage}

\vspace{0cm}
\begin{minipage}[t][1.1\height][t]{\textwidth}
	\begin{center}
	\large Under the supervision of:\\
	\vspace{0.5cm}
	Professor Pedro M. Crespo Bofill
	\end{center}
	\vspace{1.15cm}
\end{minipage}

{\large University of Navarra--TECNUN\\
Faculty of Telecommunication Engineering\\
2021}\end{center}
\clearpage

\normalsize\vspace*{15cm}
\begin{minipage}{13cm}
Doctoral Dissertation by the University of Navarra-TECNUN\\ 

\copyright{Josu Etxezarreta Martinez}, Dec 2021, All Rights Reserved\\[2ex] 

This doctoral dissertation has been funded by the Basque Government predoctoral research
grant.\\

Printed in TECNUN, Donostia-San Sebasti\'an
\end{minipage}

\cleardoublepage
\epigraph{Herodotus of Halicarnassus here presents the results of his study, so that time may not abolish the works of men...}{\textbf{Herodotus}, \textit{Histories}}
\epigraph{Eta azken hitzak ere ahots apalez esaten ahaleginduko naiz, lehenagokoetan ere horixe izan baita ene asmoa, apalki, ia zurrumurrua bezala, ene hitzok ere zarataren munduan gal ez daitezen.}{\textbf{Joseba Sarrionandia},\\ \textit{Ni ez naiz hemengoa}}
\cleardoublepage

\pagestyle{plain}
\pagenumbering{roman}
\setcounter{page}{3}

\pagestyle{fancyplain}

\chapter*{Eskerrak}
Bizitzea ez al da oso arriskutsua? \\

Esaldi honekin hasi eta bukatzen du Joseba Sarrionandiak 2018ko lehenengo bost hilabetetan idatzitako gaukaria. Urte berdineko urtarrilean hasi nuen lau urtez luzatu den abentura hau, Master Amaierako Lana berarekin egin nahi nuela esan nionean Pedrori. Orduztik, galdera erretoriko horrek aipatzen duen arriskua antzeman ahal izan det: dudak, gorabeherak, oinazeak, pandemia mundial bat, angustia ... Hala ere, bizitzak jartzen dizkigun erronka hauei aurre egitea beharrezkoa degu gizakiok. ``Desatseginak diren bizipenak ukatuz plazera lortu daiteke, baina zoriontasuna ez'' idatzi zuen Konrad Lorenz etologoak. Hortaz, bizitzak botatako gorabeherengandik ikasi behar degu momentu oro, hau ez baita behin ere gelditzen Berri Txarrak-ek abesten duen bezela. Ez pentsa, ordea, dena iluna izan denik, urte hauek momentu zoragarriak eman baitizkidate pertsona zoragarriagoez inguraturik. Lorenzekin jarraituz, eta lehenagoko aipua zehaztuz, zoriontasuna onaren eta txarraren arteko orekan datza. Lerro hauen bidez eskerrak eman nahi dizkiet urte hauetan modu batean edo bestean nire alboan egondako pertsonei.  \\

Hasteko, etxekoei. Aitari eta amari, betidanik nire alboan izateagatik eta bide hau hasi nuenetik nigan konfiantza izateagatik. Nire anai Imanoli, nire ikerketei buruzko interesa izateagatik eta nire haserreak jasateagatik, gehienbat pandemia garaiean. Etxezarreta eta Martinez familiei, zuen animo eta hitz gozoengaitik. Zuen babesa ezinbestekoa det nire lana aurrera eraman ahal izateko. \\

Muchas gracias a todo el grupo MAT$\pi$COM de Tecnun, principalmente por el buen trato que he recibido y por hacerme sentir c\'omodo en la universidad. Me parece importante destacar el papel que tuvisteis Pedro, Xabi, Jes\'us y Adam durante mis estudios ya que conseguisteis que sintiera un gran inter\'es por las matem\'aticas y la teor\'ia de la informaci\'on. Gracias en especial a Pedro por darme la oportunidad de poder realizar esta tesis doctoral con total libertad, apoyarme en todo momento y compartir tus conocimientos conmigo. Sigues inspir\'andome todos los d\'ias por tu pasi\'on por la investigaci\'on y la docencia; y espero con ganas seguir estudiando el mundo de la decoherencia y la correcci\'on cu\'antica de errores contigo. Ernesto Sabato escribi\'o que ``es el tema el que lo elige a uno'' y creo que gracias a tu intuci\'on el tema de la computaci\'on cu\'antica me eligi\'o a m\'i. Gracias por todo. \\

Continuando con el grupo, quiero agradecer tambi\'en a mis compa\~neros de doctorado y despacho: Imanol, Patrick, Fernando, I\~nigo, Fran y Ton. Hemos conseguido que estos a\~nos hayan sido tremendamente divertidos entre caf\'es, snatches en el despacho, compras de criptomonedas, gritos extra\~nos y danzas. Imanol, gracias por ayudarme al comienzo de esta aventura y por los momentos vividos en ella. Espero que podamos volver a cruzar nuestros caminos laborales en alg\'un momento. Me llevo un buen amigo. Patrick, muchas gracias por estos m\'as de dos a\~nos y medio en los cuales hemos crecido juntos tanto como investigadores como personas. Creo que tu incorporaci\'on le dio un empuje definitivo a la parte de Quantum del grupo y ha sido un placer investigar las ``eldritch laws of quantum mechanics'' junto a ti, dicho en tu ingl\'es Shakesperiano obviamente. Dicho esto, me llevo una gran amistad que ha crecido m\'as en estos a\~nos de tesis. Ton, moltes gr\`acies per aquests \'ultims mesos de tesi. Espero seguir desenvolupant la nostra relaci\'o tant acad\`emica com personal d'aqu\'i endavant. Enrecorda't, no treballis els caps de setmana. Tengo muchas ganas de seguir investigando los tres juntos, los Quantum Bois. Como solemos decir en el despacho: se vienen cositas. \\

Tambi\'en quiero agradecer a Javier Garcia-Fr\'ias por el apoyo que nos das para a las investigaciones que realizamos en el grupo. Siempre tienes alguna idea interesante para poder estudiar un tema m\'as a fondo y tienes un ojo cl\'inico a la hora de corregir los art\'iculos. Gracias tambi\'en a Rom\'an Or\'us y a Gabriel Molina por confiar en nuestro grupo para pedir algunos proyectos que tienen muy buena pinta y por aceptar estar en el tribunal de esta tesis doctoral. Tengo muchas ganas de poder trabajar juntos en el futuro. Thank you Jonas Bylander for useful discussions regarding decoherence that were critical for our npj Quantum Information paper.

Horretaz gain, milesker Miramoneko jendeari: Mar\'ia, Xabi, Luisvi, Ane, Francesco, Jon ... Momentu ederrak pasa ditugu kafe gehiegi edaten atseden amaiezinetan. Zuen falta sentitu degu Ibaetara jeitsi gintuztenetik, gehienbat paperren argitaratzea tortilla janez ospatzeko. Aitor, Zaba, Nagore eta Iratiri, lizeotik doktoretza tesiraino egindako bideagatik. Zuekin ere kafe ugari hartu ditut, zorte on tesia bukatzeko gelditzen zaizuen denboran. Unairi, ostiral arratsaldeetan kañengatik aldatzen ditugun asteazkenetako kafeengatik. \\

Aurreko parrafoak irakurrita, kafe asko edaten detela antzeman daiteke. Ondorioz, eskerrak eman behar dizkiet Nespressori, Pelican Rougeko makinei, kafetegietako kamareroei eta kafea egin dezakeen edozein makina zein pertsonari ere. Asko zor diet, halaber, Spotify eta Youtuberi. Musikarik gabe ezingo nintzake bizi. \\

Bera Berako taldekideei. Astebururoko borrokengatik. \\

Kourt eta Jaureri. N\"urnberg, Sydney eta Donosti urrun daude, baina gu gertu. Urtero gabonetan egiten degun gure QITCEAI konferentzia, eta bazkaria, egiten jarrai dezagun. Ea COVIDa amaitzean jarduera presentzialera bueltatu gaitezkeen. \\

Kuadrilakoei eta lagun minei. Nire alboan egoteagatik, une onetan eta txarretan. Nire ikerketei buruz galdetzeagatik, barreengatik, negarrengatik. Zertan ari nintzen erabat asmatu ezinda ere nire poz eta tristurak ondo asko ulertu dituzuelako. Bizi ditugunengaitik eta biziko ditugunengaitik. Aristotelesek bazioen lagunik gabe zoriontasuna lortzea ez dela posible. Ez dakit posible den ala ez, baina lagunek zoriontasuna dakarkidazuela argi daukat. \\

Amaitzeko, master amaierako lana hasi nuenetik nire bizitzatik pasa zareten eta noizbait nigan edo nire lanaren inguruan interesa izan dezuenoi. Askorentzako ulergaitza den gai bati buruz galdetzeagatik eta egiazko arreta jartzeagatik. Bestalde, nire bidaia aberasteagatik. \\

Bizitzea ez al da oso arriskutsua? Horrela hasi ditut eskerrak emateko lerro hauek. Erantzuna? Bai, bizitzea oso ariskutsua da. Baina merezi du. Albert Camus-en hitzetan: ``Neguaren erdian neure barnean aurkitu nuen uda garaitezina''.

\cleardoublepage
\chapter*{Laburpena}
Teknologia kuantikoak, konputazionalki konplexuak diren eta egungo teknologien bidez ebatzi ezin daitezkeen arazoei aurre egiteko aukera eskaintzen du. Adibidez, modu eraginkor batean zenbakiak beren osagai lehenetan faktorizatzeko, datu-base desegituratuetan bilaketak burutzeko edota makromolekula konplexuak simulatzeko aukera eskaintzen du. Ondorioz, konputazio kuantikoa gizartearen eta zientziaren aurrerapenean ezinbesteko tresna bilakatu daiteke. Besteak beste, botika diseinuan, finantza-krisien aurreikustean, konputagailu sareen segurtasuna sendotzean edo genomen sekuentziazioan aplikatu daiteke teknologia kuantikoa. Hala ere, egungo teknologiaren bitartez oraindik ezinezkoa da konputazio kuantikoak eskaini ditzaken aukera guztiak burutzeko gai den konputagailu kuantikoa eraikitzea. Informazio kuantikoak erroreak jasateko duen joerak sorturiko fidagarritasun falta da ezgaitasun horren kausa. Errore horiek, sistema kuantikoek euren ingurumenarekin dituzten interakzioen ondorio dira. Prozesu fisiko horien multzoari dekoherentzia deritzaio eta teknologia kuantikoen zeregin guztietan ageri da. Hortaz, informazio kuantikoa errore-zuzentze kodeen bidez babestea beharrezkoa da era zuzenean funtzionatzeko ahalmena duten konputagailu kuantikoak eraiki ahal izateko. Kode horiek modu eraginkor batean sortu ahal izateko, dekoherentzia prozesuak ulertzea eta matematikoki modelatzea funtsezkoa da. Tesi honetan dekoherentziaren modelatze matematikoa eta errore-zuzentze kode kuantikoen optimizazioa ikertu ditugu.

\cleardoublepage

\chapter*{Abstract}

Quantum technologies have shown immeasurable potential to effectively solve several information processing tasks such as prime number factorization, unstructured database search or complex macromolecule simulation. As a result of such capability to solve certain problems that are not classically tractable, quantum machines have the potential revolutionize the modern world via applications such as drug design, process optimization, unbreakable communications or machine learning. However, quantum information is prone to suffer from errors caused by the so-called decoherence, which describes the loss in coherence of quantum states associated to their interactions with the surrounding environment. This decoherence phenomenon is present in every quantum information task, be it transmission, processing or even storage of quantum information. Consequently, the protection of quantum information via quantum error correction codes (QECC) is of paramount importance to construct fully operational quantum computers. Understanding environmental decoherence processes and the way they are modeled is fundamental in order to construct effective error correction methods capable of protecting quantum information. In this thesis, the nature of decoherence is studied and mathematically modelled; and QECCs are designed and optimized so that they exhibit better error correction capabilities.

\cleardoublepage
\chapter*{List of Publications}
\section*{Journal Papers}
\begin{enumerate}
\item \textbf{Josu Etxezarreta Martinez}, Patricio Fuentes, Pedro Crespo and Javier Garcia-Frias, Quantum outage probability for time-varying quantum channels. \emph{Physical Review A} 105, 012432 (2022).
\item \textbf{Josu Etxezarreta Martinez}, Patricio Fuentes, Pedro Crespo and Javier Garcia-Frias, Time-varying quantum channel models for superconducting qubits. \emph{npj Quantum Information} \textbf{7,} 115 (2021).
\item \textbf{Josu Etxezarreta Martinez}, Patricio Fuentes, Pedro Crespo and Javier Garcia-Frias, Approximating Decoherence Processes for the Design and Simulation of Quantum Error Correction Codes on Classical Computers. \emph{IEEE Access} \textbf{8}, 172623-172643 (2020).
\item \textbf{Josu Etxezarreta Martinez}, Pedro Crespo and Javier Garcia-Frias, Depolarizing Channel Mismatch and Estimation Protocols for Quantum Turbo Codes. \emph{Entropy} \textbf{21(12)}, 1133 (2019).
\item \textbf{Josu Etxezarreta Martinez}, Pedro Crespo and Javier Garcia-Frias, On the Performance of Interleavers for Quantum Turbo Codes. \emph{Entropy} \textbf{21(7)}, 633 (2019).
\item Patricio Fuentes, \textbf{Josu Etxezarreta Martinez}, Pedro Crespo and Javier Garcia-Frias, On the logical error rate of sparse quantum codes. \emph{IEEE transactions on Quantum Engineering} (under review).
\item Patricio Fuentes, \textbf{Josu Etxezarreta Martinez}, Pedro Crespo and Javier Garcia-Frias, Degeneracy and its impact on the decoding of sparse quantum codes. \emph{IEEE Access}  \textbf{9,} 89093-89119 (2021).
\item Patricio Fuentes, \textbf{Josu Etxezarreta Martinez}, Pedro Crespo and Javier Garcia-Frias, Design of low-density-generator-matrix–based quantum codes for asymmetric quantum channels. \emph{Physical Review A} \textbf{103,} 2 (2021).
\item Patricio Fuentes, \textbf{Josu Etxezarreta Martinez}, Pedro Crespo and Javier Garcia-Frias, Approach for the construction of non-calderbank-steane-shor low-density-generator-matrix–based quantum codes. \emph{Physical Review A} \textbf{102,} 1 (2020).
\end{enumerate}

\section*{Conference Papers}
\begin{enumerate}
\item \textbf{Josu Etxezarreta Martinez}, Patricio Fuentes, Pedro Crespo and Javier Garcia-Frias, Pauli Channel Online Estimation Protocol for Quantum Turbo Codes. \emph{2020 IEEE International Conference on Quantum Computing and Engineering (QCE)}, 102-108 (2020).
\item Patricio Fuentes, \textbf{Josu Etxezarreta Martinez}, Pedro Crespo and Javier Garcia-Frias, Performance of non-CSS LDGM-based quantum codes over the misidentified depolarizing channel. \emph{2020 IEEE International Conference on Quantum Computing and Engineering (QCE)}, 93-101 (2020).
\end{enumerate}
\cleardoublepage


\chapter*{Glossary}

A list of the most repeated acronyms is provided below.\\

\begin{tabular}{ll}
\textbf{AD}  		&       \textit{Amplitude damping (channel)} \\
\textbf{APD}  		&       \textit{Amplitude and phase damping (channel)} \\
\textbf{AWGN}  		&       \textit{Additive white Gaussian noise} \\
\textbf{CHSH}  		&       \textit{Clauser, Horne, Shimony and Holt} \\
\textbf{CPTP}  		&       \textit{Completely-positive trace-preserving} \\
\textbf{CSI}  		&       \textit{Channel state information} \\
\textbf{CSS}  		&       \textit{Calderbank-Shor-Steane} \\
\textbf{CTA}  		&       \textit{Clifford twirl approximation} \\
\textbf{D}  		&       \textit{Depolarizing (channel)} \\
\textbf{DQLMD}  	&       \textit{Degenerate Quantum Maximum Likelihood decoding} \\
\textbf{EA}  		&       \textit{Entanglement-assisted} \\
\textbf{EPR}  		&       \textit{Einstein-Podolsky-Rosen} \\
\textbf{EXIT}  		&       \textit{EXtrinsic Information Transfer} \\
\textbf{IQR}  		&       \textit{Interquartile range} \\
\textbf{LDGM}  		&       \textit{Low-Density-Generator-Matrix (code)} \\
\textbf{LSD}  		&       \textit{Lloyd-Shor-Devetak} \\
\textbf{MC}  		&       \textit{Medcouple} \\
\textbf{MP}  		&       \textit{Markovian Pauli (channel)} \\
\textbf{MWPM}  		&       \textit{Minimum weight perfect matching} \\
\textbf{NISQ}  		&       \textit{Noisy intermediate-scale quantum} \\
\textbf{NMR}  		&       \textit{Nuclear Magnetic Resonance} \\
\textbf{P}  		&       \textit{Pauli (channel)} \\
\end{tabular}\\

\begin{tabular}{ll}
\textbf{PCM}  		&       \textit{Parity check matrix} \\
\textbf{PD}  		&       \textit{Phase damping or dephasing (channel)} \\
\textbf{PSD}  		&       \textit{Power Spectral Density} \\
\textbf{PTA}  		&       \textit{Pauli twirl approximation} \\
\textbf{QBER}  		&       \textit{QuBit error rate} \\
\textbf{QCC}  		&       \textit{Quantum Convolutional Code} \\
\textbf{QEC}  		&       \textit{Quantum Error Correction} \\
\textbf{QECC}  		&       \textit{Quantum error correction code} \\
\textbf{QHB}  		&       \textit{Quantum Hamming bound} \\
\textbf{QIRCC}  	&       \textit{Quantum IrRegular Convolutional Code} \\
\textbf{QKD}  		&       \textit{Quantum key distribution} \\
\textbf{QLDPC}  	&       \textit{Quantum Low-Density-Parity-Check (code)} \\
\textbf{QMLD}  		&       \textit{Quantum Maximum Likelihood decoding} \\
\textbf{QSC}  		&       \textit{Quantum Stabilizer Code} \\
\textbf{QTC}  		&       \textit{Quantum Turbo Code} \\
\textbf{QURC}  		&       \textit{Quantum Unity Rate Code} \\
\textbf{SISO}  		&       \textit{Soft-input soft-output} \\
\textbf{SLD}  		&       \textit{Symmetric logarithmic derivative} \\
\textbf{SNR}  		&       \textit{Signal-to-noise ratio} \\
\textbf{SQSC}  		&       \textit{Single-qubit single-channel} \\
\textbf{TV}  		&       \textit{Time-varying} \\
\textbf{TVQC}  		&       \textit{Time-varying quantum channel} \\
\textbf{WER}  		&       \textit{Word Error Rate} \\
\textbf{WSS}  		&       \textit{Wide-sense stationary} \\
\textbf{iid}  		&       \textit{independent and identically distributed} \\
\end{tabular}

\cleardoublepage
\chapter*{Notation}

Although all symbols are defined at their first appearance, here we provide a list of the notation for the symbols used throughout the text.  \\

\begin{tabular}{ll}
\(\hbar\) 	&  Reduced Planck's constant. \\
\(\mathcal{H}\) 	&  Complex Hilbert space. \\
\((\mathbb{F}_2)^n\) 	&  Space of binary strings of length $n$. \\
\(A^T\) 	&  Matrix transpose. \\
\(A^*\) 	&  Matrix complex conjugate. \\
\(A^\dagger\) 	&  Hermitian transpose. \\
\(\mathrm{Tr}(\cdot)\) 	&  Trace of a matrix. \\
\(\mathrm{Q}(\cdot)\)		&  Q-function.\\
\(\otimes\) 	&  Tensor product. \\
\(\odot\)	&  Symplectic product.\\
\([A,B]\) 	&  Commutator operator $[A,B] = AB-BA$. \\
\(\{A,B\}\) 	&  Anticommutator operator $\{A,B\} = AB+BA$. \\
\(\delta_{ij}\) 	&  Kronecker delta. \\
\(|\cdot|\) 	&  Cardinality of a set. \\
\(\mathrm{E}[\cdot]\) 	&  Expected value. \\
\(K(\Delta t)\) 	&  Covariance function. \\
\(||\cdot||_\diamond\) 	&  Diamond norm. \\
\(||\cdot||_1\) 	&  Trace norm. \\
\end{tabular}\\

\begin{tabular}{ll}
\(\mathcal{GN}_{[a,b]}(\mu_X,\sigma_X^2)\)		&  Truncated normal random variable $X$.\\
\(\mathcal{CN}(0,1)\)		&  Circularly symmetric complex normal random variable.\\
\([\cdot]\)			&  Equivalence class. \\
\(\mathrm{I,X,Y,Z}\)		&  Pauli matrices. \\
\(\mathcal{P}_n\)		&  set of $n$-fold tensor products of Pauli matrices.\\
\(\mathcal{G}_n\)		&  $n$-fold Pauli group.\\
\([\mathcal{G}_n]\)		&  $n$-fold effective Pauli group.\\
\(\mathcal{C}_1^{\otimes n}\)		&  $n$-fold Clifford group.\\
\(\mathcal{S}_1^{\otimes n}\)		&  $n$-fold Symplectic group.\\
\(\Upsilon\)		&  symplectic-to-Pauli map.\\
\(\hat{\mathrm{H}}\)		&  Hamiltonian.\\
\(\ket{\cdot}\)   	&  State vector (ket). \\
\(\ket{0},\ket{1}\)   	&  Standard basis. \\
\(\ket{+},\ket{-}\)   	&  Hadamard basis. \\
\(\ket{\Phi^+},\ket{\Phi^-},\ket{\Psi^+},\ket{\Psi^-}\)		&  Bell states. \\
\(\bra{\cdot}\)    	&  Vector dual to $\ket{\cdot}$, i.e. $\bra{\cdot}=\ket{\cdot}^\dagger$ (bra). \\
\(\rho\)		&  Density matrix.\\
\(\mathrm{H}\)		&  Hadamard gate.\\
\(\mathrm{R}_\phi\)		&  Phase-shift gate.\\
\(\mathrm{P}\)		&  Phase gate $\mathrm{P}=\mathrm{R}_{\pi/2}$.\\
\(\mathrm{CNOT}\)		&  controlled-NOT gate.\\
\(C(U)\)		&  controlled-Unitary gate.\\
\(\mathrm{SWAP}\)		&  SWAP gate.\\
\(M_m\)		&  Measurement operator.\\
\(P_m\)		&  Projective measurement.\\
\(\gamma\)		&  Damping probability.\\
\(\lambda\)		&  Scattering probability.\\
\(\mathrm{L}_k\)		&  Lindblad or jump operator of decoherence source $k$.\\
\(\Gamma_k\)		&  Interaction rate of decoherence source $k$.\\
\(T_1\)		&  Relaxation time.\\
\(T_2\)		&  Dephasing or Ramsey time.\\
\(T_\phi\)		&  Pure dephasing time.\\
\(f_{01}\)		&  Qubit frequency.\\
\(\Delta \omega_q\)		&  Qubit frequency shift.\\
\(\mathcal{N}\)		&  Quantum channel.\\
\(\mathcal{N}_1\circ\mathcal{N}_2\)		&  Serial channel composition.\\
\(E_k\)	&  Kraus or error operators.\\
\(\mathcal{E}\)	&  Channel error.\\
\(\bar{\mathcal{N}}^{\mathcal{U}}\)		&  twirled channel by set of unitary operators $\mathcal{U}$.\\
\(\alpha\)		& Asymmetry coefficient.\\
\end{tabular}\\

\begin{tabular}{ll}
\(R\)		& Classical coding rate.\\
\(\mathrm{H}_2(\cdot)\)		& Binary entropy of a random variable.\\
\(I(X;Y)\)		& Mutual information of random variables $X$ and $Y$.\\
\(C\)		& Classical channel capacity.\\
\(R_\mathrm{Q}\)		& Quantum coding rate.\\
\(S(\cdot)\)		& von Neumann entropy of quantum state.\\
\(Q_{\mathrm{coh}}\)		& Quantum channel coherent information.\\
\(Q_{\mathrm{reg}}\)		& Regularized quantum channel coherent information.\\
\(C_\mathrm{Q}\)		& Quantum channel capacity.\\
\(C_\mathrm{H}\)		& Hashing bound.\\
\(\mathcal{S}\)		&  Stabilizer set.\\
\(\mathrm{C}(\mathcal{S})\)		& Quantum stabilizer code.\\
\(d\)		& Quantum code distance.\\
\(\mathcal{R}\)		&  Recovery operation.\\
\(\bar{s}\)		&  Error syndrome.\\
\(H\)		&  Parity check matrix.\\
\(\mu_{T_i}\)		&  Mean relaxation/dephasing/pure dephasing time.\\
\(\sigma_{T_i}\)		&  Standard deviation of relaxation/dephasing/pure dephasing time.\\
\(c_\mathrm{v}\)		&  Coefficient of variation.\\
\(\mathrm{Lor}(\omega,t)\)		&  Lorentzian noise process.\\
\(N(\omega,t)\)		&  White noise process.\\
\(\mathrm{BW}\)		&  Bandwidth.\\
\(T_\mathrm{c}\)		&  Stochastic process coherence time.\\
\(t_\mathrm{algo}\)		&  Quantum algorithm processing time.\\
\(t_{\mathrm{1Q}}\)		&  $1$-qubit gate time.\\
\(t_{\mathrm{2Q}}\)		&  $2$-qubit gate time.\\
\(t_{\Delta}\)		&  Gate-to-gate delay time.\\
\(t_{\mathrm{meas}}\)		&  Measurement time.\\
\(\rho_{X,Y}\)		&  Pearson correlation coefficient for random variables $X$ and $Y$.\\
\(\alpha(\omega,t)\)		&  Fading gain random process.\\
\(p_\mathrm{out}(R,\mathrm{SNR})\)		&  Classical outage probability.\\
\(p_\mathrm{out}^\mathrm{Q}(R_\mathrm{Q})\)		&  Quantum outage probability.\\
\(p_\mathrm{out}^\mathrm{H}(R_\mathrm{Q})\)		&  Quantum hashing outage probability.\\
\(\gamma^*(R_\mathrm{Q})\)		&  Noise limit.\\
\(T_1^*(R_\mathrm{Q},t_\mathrm{algo})\)		&  Critical relaxation time.\\
\(\delta_{\mathrm{out}}(@\chi)\)		&  Gap of a code to the $(p_\mathrm{out}^\mathrm{Q}(R_\mathrm{Q})$ in dBs at $\mathrm{WER}=\chi$.\\
\end{tabular}\\

\begin{tabular}{ll}
\(\Pi\)		&  Quantum interleaver.\\
\(\pi\)		&  Scrambling pattern.\\
\(S\)		& Interleaver spread parameter.\\
\(\eta\)		& Interleaver disperssion parameter.\\
\(p\)		& Channel error probability.\\
\(\hat{p}\)		& Estimated channel error probability.\\
\(\sigma\)		& Quantum probe.\\
\(F(p)\)		& Classical Fisher information.\\
\(J(p)\)		& Quantum Fisher information.\\
\(\hat{L}(p)\)		& Symmetric logarithmic derivative.\\
\(\mathrm{P}^a\)		& A priori information.\\
\(\mathrm{P}^o\)		& A posteriori information.\\
\(\mathrm{P}^e\)		& Extrinsic information.\\ 
\end{tabular}

\cleardoublepage

\clearemptydoublepage
\pagestyle{fancyplain}
\def\contentsname{Contents}
\setcounter{tocdepth}{3}
\tableofcontents
\clearemptydoublepage

\pagenumbering{arabic}

\listparindent=10mm
\parskip=10pt

\setcounter{page}{1}
\def\chaptername{CHAPTER}


\chapter{Introduction} \label{cp1_intro}
Since Richard Feynman's original and ground-breaking proposal in \cite{feynman} of constructing computers that follow the laws of quantum mechanics to simulate physical systems that obey said laws, the scientific community has gone to extraordinary lengths in order to build an operational quantum computer. Following the introduction of these novel ideas about constructing quantum machines, research has shown that application of quantum physics theory is not only useful for simulation of complex quantum mechanical systems such as macromolecules for drug discovery \cite{drugs, energies, reactions, chemistry}, but also to efficiently solve tasks which are computationally unmanageable in a reasonable amount of time for classical computers. The most prominent of these tasks are the factorization of prime numbers and the discrete logarithm problem \cite{shor}, Byzantine agreement \cite{byzantine}, or searching an unstructured database or an unordered list \cite{grover1,grover2}. This way, quantum machines are thought to have the potential to revolutionize the modern industry with applications such as the design of medicines optimized to cure specific diseases, the optimization of materials, more accurate weather forecasting or advanced artificial intelligence among others.

Furthermore, quantum computing is a powerful asset for secure communications. For instance, Quantum Key Distribution (QKD) protocols enable two parties to create a shared random secret key only known to them. The key remains secure given that these protocols allow the aforementioned parties to detect if a malicious entity is trying to gain knowledge of it, which enables them to adapt their message exchange so that the eavesdropper extracts no information. The best-known cryptographic QKD protocols are the BB84 protocol proposed by Bennett and Brassard in \cite{BB84} and the E91 protocol by Ekert presented in \cite{E91}. Given the fact that the power of quantum computers can break the state-of-the-art classical cryptographic protocols such as RSA or Diffie-Hellman in the blink of an eye, the development of such QKD protocols or the so-called post-quantum (classical security algorithms that are not compromised by quantum technologies) cryptographic protocols are a necessity for modern security.

In light of the astonishing potential of quantum frameworks, the construction of devices capable of exploiting the benefits offered by the paradigm of quantum computing represents a step forward in the advance of technology, regardless of these machines being in the form of fully operational quantum computers \cite{VonNeumann}, quantum processors as accelerators of classical computers \cite{accel}, or specific devices that perform QKD \cite{QKDhard}. Unfortunately, quantum information is vulnerable to errors that arise and corrupt quantum states while they are being processed, which oftentimes leads to incorrect algorithm outcomes. The frailty of quantum information is caused by the phenomenon known as quantum decoherence, which refers to the destruction of the superposition of quantum states from the interaction that they have with the environment \cite{decoherenceShor}. This quantum noise arises during every task related to the quantum computing paradigm: information storage, processing or communication. Hence, it is necessary to invoke Quantum Error Correction Codes (QECC) in order to have qubits with sufficiently long coherence times\footnote{The coherence time of a qubit is defined as the time during which the superposition that defines said state remains uncorrupted.} for practical applications. Quantum information is so sensitive to decoherence that many think that quantum computation is unfeasible without the aid of quantum error correction tools.

The earliest formulation of QECCs appeared in 1995 \cite{decoherenceShor} when Shor proposed a $9$-qubit code, which was later aptly named after him. This code is capable of correcting errors of weight one and it does not saturate the Quantum Hamming Bound (QHB) \cite{bounds}, which implies that the same task could be achieved with codes of shorter length. Nevertheless, it stands as the first proposal of an error correcting code for the quantum computing paradigm. Since then, research efforts have focused on deriving QECC schemes that approach the quantum capacity limits \cite{quantumcap} at a reasonable complexity cost. Encoding and decoding gate depths play an important role in the paradigm of correcting quantum errors due to the fact that quantum gates introduce additional errors. Additionally, the runtimes of the decoding algorithms should not exceed the coherence times of the qubits that are being processed, else these states would suffer from new decoherence effects while being corrected. A major breakthrough in the development of QECCs came in Gottesman's Ph.D. thesis \cite{QSC}, where he proposed the theory of Quantum Stabilizer Codes (QSC), which is a very useful framework that facilitates the construction of QECC families from classical binary and quaternary error correction codes. This formalism has led to the development of several promising QECC families such as Quantum Reed-Muller codes \cite{QRM}, Quantum Low Density Parity Check (QLDPC) codes \cite{bicycle,qldpc15}, Quantum Low Density Generator Matrix (QLDGM) codes \cite{jgf}, Quantum Convolutional Codes (QCC) \cite{QCC}, Quantum Turbo Codes (QTC) \cite{QTC,EAQTC}, and Quantum Topological Codes \cite{toric,QEClidar}.

The design of QSCs is closely related to the construction of good classical error correction codes \cite{EAQECC,QECCortGeo,CQiso}. The problem of finding good QECCs was reduced to that of constructing classical dual-containing quaternary codes \cite{QECCortGeo}. As a result, the thoroughly studied classical coding theory that was developed since Claude Shannon published his ground-breaking work \emph{A Mathematical Theory of Communication} \cite{shannon} can be integrated in the framework of quantum error correction. Unfortunately, the requirement of needing a self-orthogonal classical parity check matrix (dual-containing code) poses a challenge for importing some of the best classical codes to the quantum realm. The aforementioned restriction is a consequence of the so-called symplectic product criterion, which lays on the core of stabilizer code theory and must be fulfilled for the stabilizer generators to commute. Anyway, coding theorists have been able to import several classical code families to the framework of quantum information in a successful manner \cite{QRM,bicycle,jgf,QCC,QTC} and, thus, this isomorphism between the classical and quantum error correction theories is a priceless asset for the development of the latter.

QECC design requires the assumption of an error model in the form of a quantum channel that accurately represents the decoherence processes that affect quantum information. Based on these error models, appropriate strategies to combat the effects of decoherence can be derived. In order for the designed QECCs to be applicable in realistic quantum devices, the quantum channels should capture the essential characteristics of the physical processes that make qubits lose their coherence. The decoherence effects experienced by the qubits of a quantum processor are generally characterized using the relaxation time ($T_1$) and the dephasing time ($T_2$) \cite{SchlorPhD}. Those measurable parameters of the qubits are then used in order to describe the evolution of the open\footnote{In this context, open means that the quantum system interacts with its surrounding environment. In practice, every quantum system is subjected to open evolution.} quantum system. By doing so, quantum noise models that accurately describe how those qubits decohere can be obtained. However, simulations of those error models show exponential complexity in classical computers \cite{twirl1} and, thus, they cannot be efficiently done with classical resources when the number of qubits of the system grows. As a consequence, approximated decoherence models are needed in order to simulate noisy quantum information in classical computers. The depolarizing model is a widespread quantum error model used to evaluate the error correcting abilities of QECC families \cite{decoherenceShor,bounds,QSC,QRM,bicycle,qldpc15,QCC,QTC,EAQTC,toric,QEClidar}. This decoherence model is especially useful due to the fact that it makes the system fulfill the Gottesman-Knill theorem\footnote{The depolarizing channel can be implemented by applying random Pauli gates to the encoded quantum states. The Gottesman-Knill theorem states that Pauli gates are quantum computations that can be efficiently simulated on classical computers.} and, therefore, it can be efficiently simulated on a classical computer \cite{NielsenChuang,GotKnill}.

At the time of writing, the availability of quantum computers for researchers is limited and the accessible machines operate on a reduced number of qubits. Therefore, classical resources remain an invaluable tool for the design of advanced QECCs that will be used beyond the Noisy Intermediate-Scale Quantum (NISQ) era. The NISQ era \cite{NISQ} is a term coined by Preskill that makes reference to the time when quantum computers will be able to perform tasks that classical computers are incapable of, but will still be too small (in qubit number) to provide fault-tolerant implementations of quantum algorithms. One of the most important milestones in the track of fully operational quantum computers, named quantum advantage\footnote{Proving quantum advantage means running an algorithm in a quantum machine such that said algorithm cannot be executed in a classical machine in a reasonable time.}, has been recently claimed both by Google in 2019 \cite{GoogleSup} and China in 2020 \cite{ChinaSup}. In addition, some companies are already making use of quantum machines with low qubit overheads for applications that are nowadays a reality. For example, the Basque start-up company Multiverse Computing has recently released their Singularity spreadsheet that uses D-Wave quantum computers in order to optimize investment portfolios. Setting aside any controversies regarding the veracity of the claims by Google and China; and taking into account the fact that some practical applications of quantum computing are emerging, it can be said that the begining of the 2020's is accompanied by the begining of the NISQ-era.

\section{Motivation and Objectives}

The motivation and objectives of this thesis are two-fold: the study of decoherence as the source of quantum noise and the optimization of the family of QECCs named Quantum Turbo Codes. Both of those topics are closely related as QECCs are the tools used in order to protect quantum information from the deleterious effects caused by decoherence, therefore, knowledge about how decoherence corrupts quantum information is necessary for the task of designing good protection methods. To that end, we split the thesis into two main parts:
\begin{itemize}
\item \textbf{Part I: }Quantum Information Theory: Decoherence modelling and asymptotical limits
\item \textbf{Part II: }Quantum Error Correction: Optimization of Quantum Turbo Codes
\end{itemize}

In this way, the investigations regarding decoherence modelling and QECC optimization are clearly separated. Next, the motivation and objectives for these two parts of the thesis are described.

\subsection{Quantum Information Theory: Decoherence modelling and asymptotical limits}

Quantum error models that describe the decoherence processes that corrupt quantum information become a necessity when seeking to construct any error correction method. The amplitude damping channel, $\mathcal{N}_{\mathrm{AD}}$, and the combined amplitude and phase damping channel, $\mathcal{N}_{\mathrm{APD}}$, are a pair of widely used quantum channels that provide a mathematical abstraction that describes the decoherence phenomenon. However, such channels cannot be efficiently simulated in a classical computer when the number of qubits exceeds a small amount. For this reason, a quantum information theory technique known as twirling has been used in order to approximate $\mathcal{N}_{\mathrm{AD}}$ and $\mathcal{N}_{\mathrm{APD}}$ to the classically tractable family of Pauli channels, $\mathcal{N}_\mathrm{P}$ \cite{twirl1}. The dynamics of the $\mathcal{N}_{\mathrm{APD}}$ channel depend on the qubit relaxation time, $T_1$, and the dephasing time, $T_2$, while the $\mathcal{N}_{\mathrm{AD}}$ channel depends solely on $T_1$. These dependencies are also displayed by the Pauli channel families obtained by twirling the original channels. All these models consider $T_1$ and $T_2$ to be fixed parameters (i.e., they do not fluctuate over time). This implies that the noise dynamics experienced by the qubits in a quantum device are identical for each quantum information processing task, independently of when the task is performed. However, the assumption that $T_1$ and $T_2$ are time invariant has been disproven in recent experimental studies on quantum processors \cite{decoherenceBenchmarking,klimov,fluctAPS,fluctApp,temperature,fluctGoogle}. The data presented in these studies showed that these parameters experience time variations of up to $50\%$ of their mean value in the sample data, which strongly suggests that the dynamics of the decoherence effects change drastically as a function of time. These fluctuations indicate that qubit-based QECCs implemented in superconducting circuits may not perform, in average, as predicted by considering quantum channels to be static through all the error correction rounds. 

In this thesis, we aim to amalgamate the findings of the aforementioned studies with the existing models for quantum noise. This way the objective is to have a more realistic portrayal of quantum noise so that when designed QECCs are implemented in hardware, their operation will be more reliable. Furthermore, the quantum information limits in error correction for this new class of time-varying quantum channels are also studied.

\subsection{Quantum Error Correction: Optimization of Quantum Turbo Codes}

Quantum turbo codes have shown excellent error correction capabilities {in the setting of quantum communication}, achieving a performance less than 1 dB away from their corresponding hashing bounds. The serial concatenation of the inner and outer Quantum Convolutional Codes (QCCs) used to construct QTCs is realized by means of an {interleaver}, which permutes the symbols so that the error locations are randomized and error correction can be improved. The reason for the use of an interleaver in concatenated coding schemes is that the first stage in the decoding process generates bursts of errors that are more efficiently corrected in the second stage if they are scrambled. The QTCs proposed in the literature \cite{QTC, EAQTC, EAQECC, isomorphism, EXITQTC, QIrCC, EAQIRCC} use the so-called {random interleaver}. However, it is known from classical turbo codes that interleaving design plays a central role in optimizing performance, specially when the error floor region is considered \cite{SrandomConstruction, Lazcano, twoStep, IntDesignforTC, Kovaci, Vafi, turboCoding, SrandomFlexible}.

The QTC decoder, consisting of two serially concatenated soft-input soft-output (SISO) decoders, uses channel information as the input in order to engage in degenerate decoding and to estimate the most probable error coset to correct the corrupted quantum state. Previous works on QTCs \cite{QTC,EAQTC,isomorphism,EXITQTC,QIrCC,EAQIRCC,QURC,
QSBC,Catastrophic} were based on the assumption that such a decoder has {perfect channel knowledge}, that is, the system is able to estimate the depolarizing probability of the quantum channel perfectly. However, this scenario is not realistic since the decoder should work with {estimates} of the depolarizing probability rather than with the exact value. The effect of channel mismatch on QECCs has been studied for quantum low density parity check (QLDPC) codes in \cite{QLDPCmismatch,QLDPCMismatchMethods}, which showed that such codes are pretty insensitive to the errors introduced by the imperfect estimation of the depolarizing probability and proposed methods to improve the performance of such QECCs when the depolarizing probability is estimated. For classical parallel turbo codes, several studies showed their insensitivity to SNR mismatch \cite{ReedMismatch,SummersMismatch,MismatchNoEstimation,MismatchSOVA}. In such works, it was found that the performance loss of parallel turbo codes is small if the estimated SNR is above the actual value of the channel. However, QTCs can only be constructed as serially concatenated convolutional codes, and classical studies about channel information mismatch for such codes \cite{VarianceMismatchSCCC,SCCCmismatchTalwar} showed that such structures are more susceptible to channel identification mismatch, suffering a significant performance loss if the SNR is either overestimated or underestimated. \mbox{The authors} in \cite{VarianceMismatchSCCC} justified the increase in sensitivity by pointing out that in parallel turbo codes, the channel estimates are fed to both decoders, which distributes the errors between the two constituent codes. However, for serial turbo codes, the errors are not as evenly spread as the channel information is just used in the inner decoding stage. This strongly suggests that (serially concatenated) QTCs will be more sensitive to depolarizing probability mismatch.

The main objective of the second part of the thesis is to improve the error rate metric of existing QTCs. First, by making use of the existing classical-quantum isomorphism, we aim to improve the error floor region of QTCs in the error floor region by adapting classical interleaving schemes that have shown to be effective in classical turbo coding theory. Second, by analyzing the performance loss incurred by QTCs due to the channel’s depolarizing probability mismatch, we propose ``blind'' estimation methods that aid the decoder to overcome channel information mismatch so that the performance of the  QTCs can approach to the one shown by QTCs with perfect channel state information.

\section{Outline and Contributions of the Thesis}

This thesis is organized as follows: Chapter 2 gives the background material needed to better understand the Ph.D. dissertation, where concepts regarding Quantum Mechanics, Decoherence of quantum systems and basics of QEC theory are explained in detail; Part I, Quantum Information Theory: Decoherence modelling and asymptotical limits (Chapters 3, 4 and 5), includes the mathematial modelling of quantum channels under time-varying behaviour, the reinterpretation of the QECC asymptotical limits for those channels, and the impact of those quantum channels in the operation and benchmarking of QECCs; Part II (Chapters 6 and 7), Quantum Error Correction: Optimization of Quantum Turbo Codes, covers the optimization of QTCs using interleavers and the operation of QTCs when the decoders are blind to the channel state information; Chapter 8 covers additional research done during the Ph.D. regarding QLDPC design and the study of degeneracy. Finally, Chapter 9 provides the conclusion of the thesis and future work.

Note that the outline of this thesis does not correspond to the actual chronological ordering of the Ph.D. research done during the years. As a matter of fact, the chronological order of the research corresponds to Part II followed by Part I, as it can be seen in the dates of the published papers \cite{Qpout,TVQC,josu3,josuchannels,josu2,josu1}. However, we consider that the outline presented for this dissertation makes it easier to understand as a whole as we go from the most general research done to the most specific one. This way, we intend to make the thesis more readable and easy to understand.

Some passages have been quoted verbatim from the following sources \cite{Qpout,TVQC,josu3,josuchannels,josu2,josu1,logicalRate,degen,patrick3,patrick2,patrick}, which are the articles that have been pusblished throughout the development of this dissertation.

\subsection{Chapter 2: Preliminaries of Quantum Information}

This chapter provides the required background of quantum computing and quantum information theory to better understand the technical work developed in the following chapters. The chapter begins by presenting basic notions of quantum mechanics. We follow by giving a detailed description of decoherence as the source of the quantum noise that corrupts quantum information, and the way that decoherence is integrated into error models known as quantum channels. Furthermore, the twirling method to obtain approximate error models that can be simulated in classical computers is explained. Finally, we provide basic notions of how quantum error correction works and the way it can be simulated in classical computers. The aim of this chapter is to make the dissertation self-contained.

Some of the discussions and descriptions given in this chapter have been published in review article \cite{josuchannels}.

\subsection{Part I: \emph{Quantum Information Theory: Decoherence modelling and asymptotical limits} (Chapters 3, 4 and 5)}
The first part of the dissertation focuses on understanding how decoherence acts on quantum systems and how such set of physical effects can be mathematically modeled. Based on experimental results found in the literature, we propose quantum channel models that vary through time. Thus, the objective is to include the inherent fluctuations of the decoherence parameters, which have been experimentally observed in qubits from the state-of-the-art quantum hardware, into the mathematical error models used to describe the noise processes suffered by quantum information. Next, we study how the asymptotical limits of QEC are changed due to the incorporation of time-variation into the framework of quantum channels. Finally, we study how the performance of QECCs is affected by the proposed time-varying quantum channels and use the new asymptotical limits to benchmark their error correction capabilities.

\subsubsection{Chapter 3: Time-varying quantum channels}
This chapter begins by describing the recent experimental research on the time-varying nature of the decoherence parameters, i.e., the relaxation time $T_1$ and dephasing time $T_2$, of superconducting qubits. First, an extensive analysis of the stochastic processes that describe the time fluctuations is done in order to appropriately include the time-variations in the quantum channel framework. Following that, time-varying quantum channels (TVQCs) are proposed as the decoherence models that include the experimentally observed dynamic nature of $T_1$ and $T_2$. Finally, we discuss the divergence that exists between TVQCs and their static counterparts, which have been considered in previous literature of QEC, by means of a metric known as the diamond norm. In many circumstances this deviation can be significant, which indicates that the time-dependent nature of decoherence must be taken into account, if one wants to construct models that capture the real nature of quantum devices.

The research described in this chapter has been published in the journal article \cite{TVQC}.

\subsubsection{Chapter 4: Quantum outage probability}

Quantum channel capacity establishes the quantum rate limit for which reliable (i.e., with a
vanishing error rate) quantum communication/correction is asymptotically possible. However, the inclusion of the time-varying nature of $T_1$ and $T_2$ in the quantum channel framework, resulting in the TVQCs proposed in Chapter 3, implies that the notions of quantum capacity based on static channels must be reinterpreted. In this chapter, we introduce the concepts of quantum outage probability and quantum hashing outage probability as asymptotically achievable error rates by a QECC with quantum rate $R_Q$ operating over a TVQC. We derive closed-form expressions for the family of time-varying amplitude damping channels (TVAD) and study their behaviour for different scenarios. We quantify the impact of time-variation as a function of the relative variation of $T_1$ around its mean. Furthermore, the behaviour of these limits as a function of the quantum rate $R_Q$ and the coefficient of variation of the qubit relaxation time $c_\mathrm{v}(T_1)$ is also studied.

The research described here has been published in the journal article \cite{Qpout}.

\subsubsection{Chapter 5: Time-varying quantum channels and quantum error correction codes}
In this chapter, we study how the performance of quantum error correction codes is affected when they operate over the TVQCs proposed in Chapter 3. At first glance, it may seem that due to the deviation found in terms of the diamond norm between the TVQCs and their static counterparts, the performance of QECCs when operate over those channels may differ significantly. However, the effect of this norm divergence on error correction is not straightforward. We provide a qualitative analysis regarding the implications of such decoherence model on QECCs by simulating Kitaev toric codes and QTCs. In this way, we obtain conclusions about the impact that the fluctuation of the decoherence parameters produce on the performance of QECCs. It has been observed that in many instances the performance of QECCs is indeed limited by the inherent fluctuations of their decoherence parameters and conclude that parameter stability is crucial to maintain the excellent performance observed under a static quantum channel assumption. We also use the quantum outage probability proposed in Chapter 4 for benchmarking the performance of the simulated QECCs.

The research described here has been published in the journal articles \cite{Qpout} and \cite{TVQC}.

\subsection{Part II: \emph{Quantum Error Correction: Optimization of Quantum Turbo Codes} (Chapters 6 and 7)}
The goal of the second part of the dissertation is to optimize the performance of QTCs for different scenarios. To that end, we first use the classical-quantum isomorphism and use some existing interleaving methods for classical turbo codes in order to lower the error floor of QTCs. Secondly, we study the problem of QTCs when channel state information is not perfectly available at the decoder. We begin by studying the sensitivity shown by the decoder of those QECCs when there exists a mismatch between the true channel information and the actual CSI fed to the decoder. We then propose estimation schemes for the CSI and compare the perfomance of the decoders that use such schemes with the decoders that have perfect CSI available.

\subsubsection{Chapter 6: Optimization of the error floor performance of QTCs via interleaver design}
Entanglement-assisted quantum turbo codes have shown extremely good error correcting ability. EXtrinsic Information Transfer (EXIT) chart techniques have been used to narrow the gap to the Hashing bound, resulting in codes with a performance as close as 0.3 dB to the hashing bounds. However, such optimization of QTCs comes at the expense of increasing their error floor region. Based on classical turbo coding theory, we aim to lower such error floors by using interleavers with some construction rather than the originally proposed random ones. Motivated by such studies, in this chapter we investigate the application of different types of interleavers in QTCs, aiming at reducing the error floors. Simulation results show that the QTCs designed using the proposed interleavers present similar behavior in the turbo-cliff region as the codes with random interleavers, while the performance in the error floor region is improved by up to two orders of magnitude. Simulations also show reduction in memory consumption, while the performance is comparable to or better than that of QTCs with random interleavers.

The research described here has been published in the journal article \cite{josu1}.

\subsubsection{Chapter 7: QTCs without channel state information}
Chapter 7 first analyzes the performance loss incurred by QTCs due to the channel’s depolarizing probability mismatch. Then, different off-line estimation protocols are studied and their impact on the overall QTCs' performance is analyzed. Some heuristic guidelines are provided for selecting the number of quantum probes required for the estimation of the depolarizing channel probability so that the performance degradation of such codes is kept low. Additionally, we propose an on-line estimation procedure by utilizing a modified version of the turbo decoding algorithm that allows, at each decoding iteration, estimating the channel information to be fed to the inner SISO decoder. Finally, we propose an extension of the above online estimation scheme to include the more general case of the Pauli channel with asymmetries.

The research described here has been published in the journal article \cite{josu2} and the conference article \cite{josu3}.

\subsection{Chapter 8: Quantum Low-Density-Generator-Matrix Codes and Degeneracy}
This chapter summarizes the additional work done on quantum error correction throughout the Ph.D. thesis. First, we briefly describe how we used the properties of classical LDGM codes in order to design a class of non-Calderbank-Shor-Steane (non-CSS) QLDPC codes which we have named non-CSS QLDGM codes. We use extensive Monte Carlo simulations in order to compare our design with other QLDPC constructions proposed in the literature, concluding that our non-CSS outerperforms them. Furthermore, we propose CSS QLDGM designs for the general asymmetric Pauli channel and heuristically prove that their performance is better than other state-of-the-art QLDPC families when operating over such error model. We also present the study of QLDGMs when they are blind to the channel state information. As it was done for QTCs in Chapter 7, we studied the sensitivity of QLDGM to channel mismatch and proposed an online estimation methods in order to aid those codes to be succesful in decoding when they ignore the depolarizing probability.

This chapter also describes the research done on the decoding degeneracy property which only appears in the quantum error correction realm. Degeneracy refers to the effect that two distinct quantum errors that share syndrome corrupt encoded quantum information in a similar manner, and so they share the necessary recovery operation. We briefly describe the group theoretical approach we used in order to describe the problem of degeneracy from the point-of-view of sparse quantum codes. To finish Chapter 8, we investigate the problem of obtaining the logical error rate of sparse quantum codes. We observed that literature is obscure regarding the benchmark metrics of QLDPC codes, and consequently, we study existing methods in order to calculate the true error correction ability of such codes and propose another method for doing so from the point-of-view of classical coding theory. By using the proposed method, we heuristically study the percentage of degenerate errors that occur whn decoding QLDGM codes.

The research described here correspond to journal articles \cite{logicalRate,degen,patrick3,patrick} and coference article \cite{patrick2}, for which I am the second author.

\subsection{Chapter 9: Conclusions and Future Work}

We conclude the Ph.D. dissertation in this chapter. We begin by summarizing the main conclusions obtained regarding both parts of the thesis. Once describing the main outcomes of the dissertation, we procede to describe future lines of research that can potentially comence from the research developed here.

\section{Reading this Thesis}

\begin{figure}[h!]
	\begin{center}
		\includegraphics[width=\columnwidth]{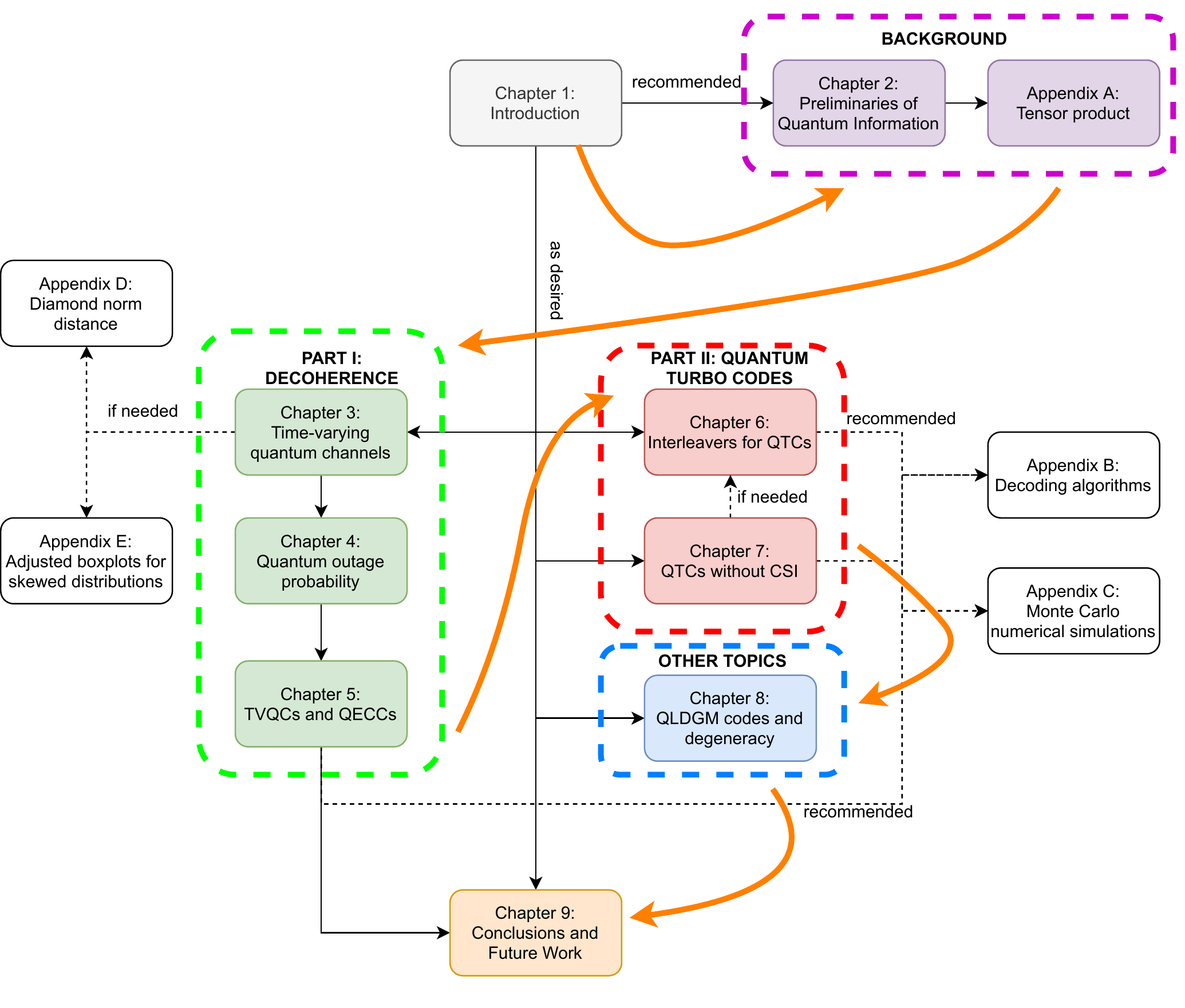}
		\caption{Block diagram detailing the dependencies between chapters.}
		\label{cp1_reading}
	\end{center}
\end{figure}

The reading of this dissertation does not need to be sequential, and each of the parts can be read independently. However, the chapters in Part \ref{part1} should be read in a sequential manner. Figure \ref{cp1_reading} shows the dependencies between the different Chapters of the thesis in an schematic manner. In this way, the reader may use such figure to know which of the parts are needed in order to properly understand some specific part of the thesis. Nevertheless, the dissertation has been written so that the sequential read of the thesis provides the best experience for the reader. Therefore, we encourage the reader to start from Chapter \ref{ch:preliminary} and follow the rest in order (the orange path in Figure \ref{cp1_reading} shows this path that we consider to be optimal).

 \label{chapter1}
\clearemptydoublepage
\chapter{Preliminaries of Quantum Information}\label{ch:preliminary}
This chapter provides the necessary background to understand the basics of quantum information theory. The main objective is to give the reader the elementary tools of quantum mechanics, decoherence and quantum error correction so that the investigations in this Ph.D. dissertation can be understood in the correct manner. It is a self-explanatory introduction so it maybe useful as an introduction to those researchers interested in the field of quantum error correction and information theory.

We start the chapter by presenting the postulates of quantum mechanics and their implications both from the state vector and density matrix perspectives. The no-cloning theorem and the quantum teleportation protocol are also introduced given their importance in quantum error correction. We follow the chapter by providing an entire subsection dedicated to the subject of decoherence and quantum noise. We extensively describe the way quantum information degrades from decoherence and how decoherence can be mathematically modelled via quantum channels. In addition, we cover the techniques often used to obtain simplified quantum channel models so that they can be efficiently simulated in classical computers and explain why these approximated channels are useful when designing quantum error correction codes. The basics on quantum channel capacity are also discussed. Finally, an introduction to stabilizer codes and their simulation via classical resources is given.

\section{Quantum Mechanics}\label{subsec:quantummechanics}
Quantum mechanics is a mathematical framework for the development of physical theories. Quantum mechanics does not tell what laws a physical system must obey, but it does provide a mathematical and conceptual framework for developing such physical laws. The connection between the physical world and the mathematical formalism of quantum mechanics is given by the \textit{postulates of quantum mechanics}.

The postulates of quantum mechanics were derived by the physicists Paul Dirac \cite{Dirac} and John von Neumann \cite{vonNeumann} after a long process of trial and error, which involved a considerable amount of guessing and fumbling. These postulates are the base from where quantum information theory will arise.

This section will also present the formulation of quantum mechanics by the usage of the so-called \textit{density matrix} and some quantum effects such as the \textit{no-cloning theorem} and \textit{quantum teleportation} will be described due to their importance in the paradigm of quantum computing and communications. Finally, we briefly describe the most important technologies used for physically implementing qubits. This section is partly based on chapter 2.2 of \cite{NielsenChuang}.

\subsection{The postulates of quantum mechanics}\label{subsub:postulates}
We begin the description of Quantum Mechanics from the state vector perspective. The first postulate of quantum mechanics deals with the mathematical vector space in which quantum mechanics takes place, the so-called \textit{state space}.

\begin{postulate}[State space]
Associated to any isolated physical system is a complex inner product vector space (that is, a Hilbert space) known as the \textit{state space} of the system. The system is completely described by its \textit{state vector}, which is a unit vector in the system's state space.
\label{post:postulate1}
\end{postulate}

The simplest quantum mechanical system is the \textit{qubit}. A qubit is associated to a two dimensional complex Hilbert space $\mathcal{H}_2$, and an arbitrary state vector in $\mathcal{H}_2$ is denoted by
\begin{equation}\label{eq:qubit}
\ket{\psi}=\alpha\ket{0}+\beta\ket{1},
\end{equation}
where $\alpha,\beta\in\mathbb{C}$ and $\ket{0},\ket{1}\in\mathcal{H}_2$ form an orthonormal basis of such Hilbert space. Later in this section, we will briefly describe the technologies that are being used for the experimental implementation of qubits. The orthonormal basis composed by those two elements is known as the \textit{standard basis} and in two dimensional vector notation are written as
\begin{equation}\nonumber
\ket{0}=\begin{pmatrix}
1 \\
0
\end{pmatrix}; \qquad
\ket{1}=\begin{pmatrix}
0 \\
1
\end{pmatrix}.
\end{equation}

As stated in postulate \ref{post:postulate1}, the vector $\ket{\psi}$ must be a unit vector, so the condition $\braket{\psi|\psi}=1\rightarrow |\alpha|^2+|\beta |^2=1$ must hold. This last condition is known as the \textit{normalization condition}.

The qubit will be the elementary quantum mechanical system for the rest of this document. Intuitively, the two states $\ket{0},\ket{1}$ are analogous to the states $0$ and $1$ of a classical bit. However, qubits differ from a classical bits from the fact that the state $\ket{\psi}$ described by \eqref{qubit} is in a \textit{superposition} of these two states.
Therefore, it cannot be said that a qubit is definitely in one of the two basis states. The coefficients $\alpha$ and $\beta$ are known as \emph{amplitudes} of the qubit. In general, it is said that any linear combination $\sum_k \alpha_k\ket{\psi_k}$ is a superposition of the states $\ket{\psi_k}$ with amplitudes $\alpha_k$.

The second postulate of quantum mechanics refers to how a quantum mechanical system evolves with time.

\begin{postulate}[Evolution]
The evolution of a \textit{closed} quantum system is described by a \textit{unitary transformation}. That is, the state $\ket{\psi}$ of the system at time $t_1$ is related to the state $\ket{\psi'}$ of the system at time $t_2$ by a unitary operator\footnote{Unitary operators fulfill $UU^\dagger=U^\dagger U=I$.} $U$ which depends only on the times $t_1$ and $t_2$,
\begin{equation}
\ket{\psi'}=U\ket{\psi}.
\end{equation}
\label{post:postulate2}
\end{postulate}

It can be seen that postulate \ref{post:postulate2} does not imply which are the unitary operators $U$ that describe real world quantum mechanics, it only assures that the evolution of any closed quantum system should be described in such way. The obvious question that arises is which unitaries are natural to consider when dealing with quantum system evolutions. In the case of single qubits, it turns out that \textit{any} unitary operator can be realized in real systems.

Next, some examples of unitary qubit operators will be presented, the so-called \textit{quantum gates}. Quantum gates are the analogous of logic gates in the classical digital world, and they are the basic elements to construct the \textit{quantum circuits} required to run \textit{quantum algorithms} in quantum processors. Note that the operation of a quantum gate over a qubit is just the evolution of the qubit by the unitary describing such gate.

The first elementary quantum gates are the \textit{Pauli gates}, denoted by $\mathrm{X}$, $\mathrm{Z}$ and $\mathrm{Y}$. Their operation over qubits is described by the Pauli matrices \cite{NielsenChuang}. The $\mathrm{X}$ Pauli gate is known as the quantum \textit{bit flip} gate as its operation over an arbitrary qubit $\ket{\psi}=\alpha\ket{0}+\beta\ket{1}$ is
\begin{equation}\nonumber
\mathrm{X}\ket{\psi} = \mathrm{X}(\alpha\ket{0}+\beta\ket{1}) = \alpha \mathrm{X}\ket{0}+\beta \mathrm{X}\ket{1}=\alpha\ket{1} + \beta\ket{0},
\end{equation}
so it flips the amplitudes respect to the standard basis. The $\mathrm{Z}$ Pauli matrix defines the so called \textit{phase flip} quantum gate as its operation leaves $\ket{0}$ unchanged and transforms $\ket{1}$ to $-\ket{1}$. The operation of a $\mathrm{Z}$ Pauli gate on an arbitrary qubit $\ket{\psi}=\alpha\ket{0}+\beta\ket{1}$ is then
\begin{equation}\nonumber
\mathrm{Z}\ket{\psi}=\mathrm{Z}(\alpha\ket{0}+\beta\ket{1})=\alpha \mathrm{Z}\ket{0}+\beta \mathrm{Z}\ket{1}=\alpha\ket{0}-\beta\ket{1}.
\end{equation}
Finally, the operation of the Pauli matrix $\mathrm{Y}$ on an arbitrary qubit $\ket{\psi}$ is a bit flip and a phase flip times a multiplication by $\pm i$, as $\mathrm{ZX}=i\mathrm{Y}$ and $\mathrm{XZ}=-i\mathrm{Y}$. Consequently, the operation of such gate is
\begin{equation}\nonumber
\mathrm{Y}\ket{\psi}=\mathrm{Y}(\alpha\ket{0}+\beta\ket{1})=\alpha \mathrm{Y}\ket{0}+\beta \mathrm{Y}\ket{1}=i\alpha \mathrm{XZ}\ket{0}+i\beta \mathrm{XZ}\ket{1}=i\alpha\ket{1}-i\beta\ket{0}.
\end{equation}

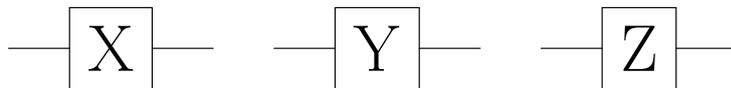
\begin{figure}[h]
\centering
\leavevmode
\Huge
\Qcircuit @C=1em @R=1em {
& \gate{\mathrm{X}}  & \qw &  &\gate{\mathrm{Y}} & \qw  &  & \gate{\mathrm{Z}} & \qw
}
\caption{Schematic representation of Pauli gates for quantum circuits.}
\label{fig:pauligates}
\end{figure}
Other two one-qubit quantum gates of interest for the construction of quantum error correction theory are the \textit{Hadamard} and the \textit{phase shift} quantum gates. The unitary that describes the  operation of the Hadamard gate is given by
\begin{equation}\nonumber
\mathrm{H}=\frac{1}{\sqrt{2}}\begin{pmatrix}
1 & 1 \\
1 & -1
\end{pmatrix},
\label{eq:hadamard}
\end{equation}
and it maps the basis states $\ket{0}$ and $\ket{1}$ into the superposition states $\ket{+}\equiv \frac{1}{\sqrt{2}}(\ket{0}+\ket{1})$ and $\ket{-}\equiv \frac{1}{\sqrt{2}}(\ket{0}-\ket{1})$ respectively. The basis formed by $\ket{+}$ and $\ket{-}$ is named \textit{Hadamard basis}.
The unitary for the general phase shift gate is
\begin{equation}\nonumber
\mathrm{R}_{\phi}=\begin{pmatrix}
1 & 0 \\
0 & e^{i\phi}
\end{pmatrix},
\label{eq:phaseshift}
\end{equation}
and maps $\ket{0}$ and $\ket{1}$ as $\ket{0}$ and $e^{i\phi \ket{1}}$ respectively. In other words, the phase of the quantum state is changed.

Two special cases for the phase shift gate are the $\mathrm{R}_\pi$ gate, as it is equal to the Pauli $\mathrm{Z}$ gate, and the $\mathrm{R}_\frac{\pi}{2}$ gate, as it will take an important role for error correction encoder and decoder design. That's why this gate will have the special notation $\mathrm{P}\equiv \mathrm{R}_\frac{\pi}{2}$ and its unitary matrix is
\begin{equation}\nonumber
\mathrm{P}\equiv \mathrm{R}_\frac{\pi}{2} = \begin{pmatrix}
1 & 0 \\
0 & i
\end{pmatrix}.
\end{equation}

\begin{figure}[h]
\centering
\leavevmode
\Huge
\Qcircuit @C=1em @R=1em {
& \gate{\mathrm{H}}  & \qw &  &\gate{\mathrm{R}_\phi} & \qw  &  & \gate{\mathrm{P}} & \qw
}
\caption{Schematic representation of Hadamard, phase shift and P gates.}
\label{fig:hadpermgates}
\end{figure}
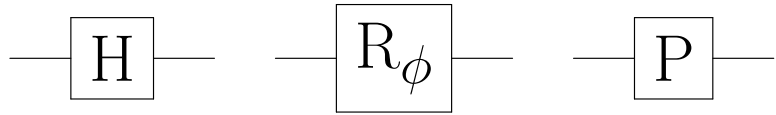

Posutlate \ref{post:postulate2} describes how the evolution of a quantum system at two different times are related. A more refined version of the postulate is given in postulate $2'$ which describes the evolution of a closed quantum system in \textit{continuous time}.
\begin{postulate2*}[Continuous-time evolution]
The time evolution of the state of a closed quantum system is described by the \textit{Sch\"odinger equation},
\begin{equation}\nonumber
i\hbar\frac{d\ket{\psi}}{dt}=\mathrm{\hat{H}}\ket{\psi},
\end{equation}
where $\hbar$ is a physical constant known as the \textit{reduced Planck's constant} and $\mathrm{\hat{H}}$ is a fixed Hermitian operator known as the \textit{Hamiltonian} of the closed quantum system.
\end{postulate2*}

In \cite{NielsenChuang}, postulate \ref{post:postulate2} is derived from this continuous time postulate implying a one-to-one correspondence between the discrete-time dynamics using unitary operators and the continuous-time dynamics via the differential equation that describes the evolution of the Hamiltonians of quantum systems.

The third postulate of quantum mechanics deals with the measurement of quantum systems. This measurement occurs when some external party observes the system to find out what is going on in the system, an interaction that makes the system no longer closed and so not necessarily subject to unitary evolution.

\begin{postulate}[Quantum Measurement]
Quantum measurements are described by a collection $\{M_m\}$ of \textit{measurement operators}. These are operators acting on the state space of the system being measured. The index $m$ refers to the measurement outcomes that may occur in the experiment. If the state of the quantum system is $\ket{\psi}$ immediately before the measurement then the probability that result $m$ occurs is given by
\begin{equation}\nonumber
p(m) = \bra{\psi}M_m^\dagger M_m\ket{\psi},
\end{equation}
and the state of the system after the measurement $\ket{\psi_m'}$ is
\begin{equation}\nonumber
\ket{\psi_m'}=\frac{M_m\ket{\psi}}{\sqrt{\bra{\psi}M_m^\dagger M_m\ket{\psi}}}.
\end{equation}
\label{post:postulate3}
\end{postulate}

The measurement operator set $\{M_m\}$ should verify the so called \textit{completeness equation}
\begin{equation}\nonumber
\sum_m M_m^\dagger M_m = \mathrm{I},
\end{equation}

so that the probabilities of measurement outcomes sum to one
\begin{equation}\nonumber
\begin{split}
\sum_m p(m)&=\sum_m \bra{\psi}M_m^\dagger M_m\ket{\psi} = \bra{\psi}\left(\sum_m M_m^\dagger M_m\right)\ket{\psi} \\ &=\bra{\psi}\mathrm{I}\ket{\psi}=\braket{\psi|\psi}=1.
\end{split}
\end{equation}

A simple but important example of a measurement is the \textit{measurement of a qubit in the standard basis}. This measurement on a single qubit with two outcomes is defined by the two measurement operators\footnote{Note that $M_0$ and $M_1$ are idempotent ($M_0^2=M_0$,$M_1^2=M_1$) and Hermitian.} $M_0=\ket{0}\bra{0}$ and $M_1 = \ket{1}\bra{1}$. Using this measurement, the probabilities of the outcomes $0$ and $1$ for an arbitrary qubit $\ket{\psi}=\alpha\ket{0}+\beta\ket{1}$ are
\begin{equation}\nonumber
\begin{split}
p(0)&=\bra{\psi}M_0^\dagger M_0\ket{\psi} = \bra{\psi}M_0\ket{\psi}=(\alpha^*\bra{0}+\beta^*\bra{1})\ket{0}\bra{0}(\alpha\ket{0}+\beta\ket{1}) \\
&=(\alpha^*\braket{0|0}+\beta^*\braket{1|0})(\alpha\braket{0|0}+\beta\braket{0|1})=\alpha^*\alpha=|\alpha|^2,
\end{split}
\end{equation}
and
\begin{equation}\nonumber
\begin{split}
p(1)&=\bra{\psi}M_1^\dagger M_1\ket{\psi} = \bra{\psi}M_1\ket{\psi}=(\alpha^*\bra{0}+\beta^*\bra{1})\ket{1}\bra{1}(\alpha\ket{0}+\beta\ket{1}) \\
&=(\alpha^*\braket{0|1}+\beta^*\braket{1|1})(\alpha\braket{1|0}+\beta\braket{1|1})=\beta^*\beta=|\beta|^2.
\end{split}
\end{equation}

The states after measurement associated with each of the outcomes are
\begin{equation}\nonumber
\ket{\psi_0'}=\frac{M_0\ket{\psi}}{\sqrt{\bra{\psi}M_0^\dagger M_0\ket{\psi}}}=\frac{\ket{0}\bra{0}(\alpha\ket{0}+\beta\ket{1})}{|\alpha|
}=\frac{\ket{0}(\alpha\braket{0|0}+\beta\braket{0|1})}{|\alpha|}=
\frac{\alpha}{|\alpha|}\ket{0}
\end{equation}
and
\begin{equation}\nonumber
\ket{\psi_1'}=\frac{M_1\ket{\psi}}{\sqrt{\bra{\psi}M_1^\dagger M_1\ket{\psi}}}=\frac{\ket{1}\bra{1}(\alpha\ket{0}+\beta\ket{1})}{|\beta|
}=\frac{\ket{1}(\alpha\braket{1|0}+\beta\braket{1|1})}{|\beta|}=
\frac{\beta}{|\beta|}\ket{1}.
\end{equation}

Ignoring the phase multipliers $\frac{\alpha}{|\alpha|}$ and $\frac{\beta}{|\beta|}$, the two post-measurement states are $\ket{0}$ and $\ket{1}$, so the superposition state held by $\ket{\psi}$ has been \textit{destroyed} due to its measurement. This last observation leds to one of the most intriguing effects of quantum mechanics, being that the collapse of the quantum states after measurement. What this means is that when a quantum state is measured, the post-measurement state \textit{changes} from the superposition state to the specific state consistent with the measurement results. This conclusion will be important when introducing the theory of quantum error correction, as we will not be able to measure the states received because they would be destroyed by such operation.

Before presenting the last postulate of quantum mechanics, an special class of the general measurements presented in postulate \ref{post:postulate3} will be presented due to its importance in developing error correcting codes in the quantum paradigm.
\begin{definition}[Projective measurements]
A \textit{projective measurement} is described by an \textit{observable}, $M$, an Hermitian operator on the state space of the system being observed. The observable has a spectral decomposition
\begin{equation}\nonumber
M=\sum_m mP_m,
\end{equation}
where $P_m$ is the projector\footnote{A projector matrix $P$ is hermitian, $P^\dag=P$, and idempotent, $P^2=P$.} onto the eigenspace of $M$ with eigenvalue $m$. The possible outcomes of the measurement correspond to the eigenvalues, $m$, of the observable. Upon measuring the state $\ket{\psi}$, the probability of getting result $m$ is given by
\begin{equation}\nonumber
p(m)=\bra{\psi}P^\dag_mP_m\ket{\psi}=\bra{\psi}P_m\ket{\psi},
\end{equation}
and given that the outcome $m$ occurred, the state of the quantum system immediately after measurement is
\begin{equation}
\ket{\psi'_m}=\frac{P_m\ket{\psi}}{\sqrt{p(m)}}.
\end{equation}
\end{definition}

Postulate \ref{post:postulate4} deals with the description of composite quantum mechanical systems, that is, systems composed of the state spaces of the component quantum systems (see Appendix \ref{app:tensor} for the definition and properties of the tensor product).

\begin{postulate}[Composite systems]
The state space of a composite physical system is the tensor product of the state spaces of the component physical systems. Moreover, if we have systems numbered $1$ through $n$, and the system number $k$ is prepared in the state $\ket{\psi_k}$, then the joint state of the total system is $\ket{\psi_1}\otimes\ket{\psi_2}\otimes\cdots\otimes\ket{\psi_n}$.
\label{post:postulate4}
\end{postulate}

Postulate \ref{post:postulate4} then provides the tool to work with systems composed of several state vectors. From this postulate, one of the most surprising ideas associated with composite quantum systems can be defined, being that \textit{entanglement}.

\begin{definition}[Entangled state]
A composite quantum system is said to be an \textit{entangled system} if it cannot be written as a product of states of its component systems, that is
\begin{equation}\nonumber
\nexists \ket{\psi}\in\mathcal{H}_A,\ket{\varphi}\in\mathcal{H}_B : \ket{\psi}_{AB}=\ket{\psi}\otimes\ket{\varphi},\ket{\psi}_{AB}\in\mathcal{H}_A\otimes\mathcal{H}_B.
\end{equation}
\label{def:entanglement}
\end{definition}

One typical example of entangled systems is given by the \textit{EPR pairs\footnote{Named after Einstein, Podolsky and Rosen as they introduecd the so-called EPR paradox related with the strange phenomena of entanglement.}} or \textit{Bell states\footnote{They receive this name as they are subject to the \textit{Bell inequality}.}}. The state vectors that describe this set of entangled quantum systems are
\begin{equation}\nonumber
\ket{\Phi^+}=\frac{\ket{00}+\ket{11}}{\sqrt{2}},
\end{equation}
\begin{equation}\nonumber
\ket{\Phi^-}=\frac{\ket{00}-\ket{11}}{\sqrt{2}},
\end{equation}
\begin{equation}\nonumber
\ket{\Psi^+}=\frac{\ket{01}+\ket{10}}{\sqrt{2}},
\end{equation}
\begin{equation}\nonumber
\ket{\Psi^-}=\frac{\ket{01}-\ket{10}}{\sqrt{2}}.
\end{equation}

This states are two qubit composite systems that cannot be expressed as tensor products of single qubits. Even though this is the property that defines entanglement, the real interest of it lays on the fact that the measurement outcome of either qubit determines the outcome of the other qubit. To see why this happens consider EPR pair $\ket{\Phi^+}$ and measurement operators $M_0=\ket{00}\bra{00},M_1=\ket{01}\bra{01},M_2=\ket{10}\bra{10},M_3=\ket{11}\bra{11}$. This set of operators are consistent with postulate \ref{post:postulate3}, so the probabilities of each outcome are given by
\begin{equation}\nonumber
p(0)=\bra{\Phi^+}M_0^\dagger M_0\ket{\Phi^+}=\frac{1}{2},
\end{equation}
\begin{equation}\nonumber
p(1)=\bra{\Phi^+}M_1^\dagger M_1\ket{\Phi^+}=0,
\end{equation}
\begin{equation}\nonumber
p(2)=\bra{\Phi^+}M_2^\dagger M_2\ket{\Phi^+}=0,
\end{equation}
\begin{equation}\nonumber
p(3)=\bra{\Phi^+}M_3^\dagger M_3\ket{\Phi^+}=\frac{1}{2}.
\end{equation}

And so it can be seen that the only possible outcomes for the measurements are, up to a multiplicative factor, $\ket{00}$ and $\ket{11}$. What this means is that if one of the qubits is measured, then the outcome will be completely random between $\ket{0}$ and $\ket{1}$, but the outcome of the other qubit will be exactly the one for the first measured qubit\footnote{For $\ket{\Psi^\pm}$ the other qubits outcome will be exactly the contrary of the outcome of the first one.}. That is, measuring one of the qubits \textit{determines} the otucome of the other qubit. This is general for measuring $\ket{\Phi^+}$ in any basis. This kind of entangled states are a central part of the entanglement-assisted quantum error correction codes.

Entanglement may seem to be the same as distributed randomness between two parties in the classical setting, that is, the scenario where both parties share a common random variable which has the same realizations for both of them. However, and as it will be seen later, entanglement is a much more powerful asset than classical shared randomness since it has properties that cannot be replicated in the classical setting. A very insightful experiment where such advantage can be observed is the so-called CHSH game\footnote{CHSH stands for Clauser, Horne, Shimony and Holt.}, in which the probability of winning a thought game by the optimal strategy in the classical setting is overcomed by means of entanglement \cite{chsh}.

To finish with section \ref{subsub:postulates}, some two-qubit quantum gates will be presented. This gates are very relevant for quantum computing and error correction since they will be the ones used to make operations between the different individual elements that constitute the quantum system. The first two-qubit gate of interest is the \textit{controlled-NOT} or CNOT gate, where a NOT operation controlled by the first qubit is done in the second qubit. The unitary matrix that describes such operation is
\begin{equation}\nonumber
\mathrm{CNOT}=\begin{pmatrix}
1 & 0 & 0 & 0 \\
0 & 1 & 0 & 0 \\
0 & 0 & 0 & 1 \\
0 & 0 & 1 & 0
\end{pmatrix}.
\end{equation}

The CNOT gate then performs a NOT operation on the second qubit only if the first of the qubits is in the state $\ket{1}$, otherwise leaving it unchanged. The second two-qubit gate that will be presented is the general \textit{controlled-Unitary} or $C(U)$ gate. The controlled-$U$ gate performs a unitary gate controlled by the first qubit on the second qubit. The matricial expression for such an operation is
\begin{equation}\nonumber
C(U)=\begin{pmatrix}
1 & 0 & 0 & 0 \\
0 & 1 & 0 & 0 \\
0 & 0 & u_{00} & u_{01} \\
0 & 0 & u_{10} & u_{11}
\end{pmatrix},
\end{equation}
where $u_{ij}$ refer to the elements that constitute the unitary matrix $U$. As it happened with the CNOT gate, the controlled-$U$ gate will perform the unitary operation only if the first qubit is in state $\ket{1}$, otherwise leaving the second qubit unchanged. It is straightforward to see that the CNOT gate is just an specific case of the controlled-$U$ gate when $U=\mathrm{X}$. However, it has been introduced specifically because it is a widely used gate and will have an important role in the theory of quantum error correction.

The last two-qubit gate of interest that will be presented is the \textit{SWAP} gate, whose operation over the two qubits is described by the unitary matrix
\begin{equation}\nonumber
\mathrm{SWAP} = \begin{pmatrix}
1 & 0 & 0 & 0 \\
0 & 0 & 1 & 0 \\
0 & 1 & 0 & 0 \\
0 & 0 & 0 & 1
\end{pmatrix}.
\end{equation}
The SWAP gate over two qubits just swaps the two qubits of order in a quantum circuit. Consequently, the operation of such gate is
\begin{equation}\nonumber
\mathrm{SWAP}(\ket{\psi}\otimes\ket{\varphi})=\ket{\varphi}\otimes\ket{\psi}.
\end{equation}

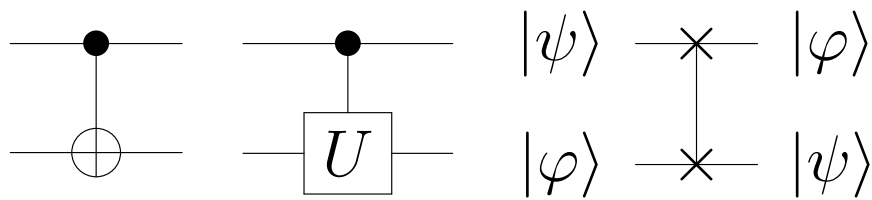
\begin{figure}[h]
\centering
\leavevmode
\Huge
\Qcircuit @C=1em @R=1.25em {
& \ctrl{1}  & \qw  \\
& \targ  & \qw
}
\Qcircuit @C=1em @R=1em {
& & \ctrl{1}  & \qw \\
& & \gate{U} & \qw \\
}
\Qcircuit @C=1em @R=2em @!{
& & & \lstick{\ket{\psi}}& \qswap  & \rstick{\ket{\varphi}}\qw \\
& & & \lstick{\ket{\varphi}}& \qswap\qwx &   \rstick{\ket{\psi}}\qw
}
\caption{Diagram representation of the CNOT, controlled-U and SWAP gates.}
\label{fig:controlledswap}
\end{figure}

\subsection{The density matrix formulation}\label{subsub:density}
In section \ref{subsub:postulates} a formulation of quantum mechanics has been given by means of state vectors. An equivalent formulation of the theory can be given by using the so-called \textit{density matrices} or \textit{density operators}. This reformulation of quantum mechanics in terms of density matrices provides a more convinient tool in order to describe systems whose state is not known completely in terms of state vectors. This description of quantum mechanics is equivalent to the one previously presented, but it presents an easier way to deal with some problems of the physical theory. In this section, such reformulation of quantum mechanics will be presented shortly so that its usage is clear when needed.

\begin{definition}[Density operator]
Suppose a quantum system is in one of a number states $\ket{\psi_k}$, where $k$ is an index, with respective probabilities $p_k$. We shall call $\{p_k,\ket{\psi_k}\}$ an \textit{ensemble of pure states}. The density operator or density matrix  for such system is defined by
\begin{equation}\nonumber
\rho\equiv\sum_k p_k\ket{\psi_k}\bra{\psi_k}.
\end{equation}
\label{def:density}
\end{definition}

The density operator description has been introduced as a way of describing ensembles of quantum states. A description that does not rely on the state vector will be given by the next theorem so that we can consistently introduce the reformulation of quantum mechanics.
\begin{theorem}
An operator $\rho$ is the density operator associated to some ensemble $\{p_k,\ket{\psi_k}\}$ if and only if it satisfies the conditions
\begin{itemize}
\item \textbf{Trace condition}: $\mathrm{Tr}(\rho)=1$.
\item \textbf{Positivity condition}: $\rho$ is a positive operator, that is $\bra{\psi}\rho\ket{\psi}\geq 0,\forall \ket{\psi}$.
\end{itemize}
\label{thm:density}
\end{theorem}

This way, any positive operator with trace equal to one is a density matrix describing some quantum state. By means of the description of the quantum system using the density operator, the postulates of quantum mechanics presented in \ref{subsub:postulates} can be reformulated like
\setcounter{postulate}{0}

\begin{postulate}[State space]
Associated to any isolated physical system is a complex vector space with inner product (that is, a Hilbert space $\mathcal{H}$) known as the \textit{state space} of the system. The system is completely described by its \textit{density operator}, which is a positive operator $\rho$ with trace one, acting on the state space of the system. If a quantum system is in the state $\rho_k$ with probability $p_k$, then the overall density operator for the system is $\sum_k p_k\rho_k$.
\end{postulate}

\begin{postulate}[Evolution]\label{densityevo1}
The evolution of a \textit{closed} quantum system is described by a \textit{unitary transformation}. That is, the state $\rho$ of the system at time $t_1$ is related to the state $\rho'$ of the system at time $t_2$ by a unitary operator $U$ which depends only on the times $t_1$ and $t_2$,
\begin{equation}\nonumber
\rho'=U\rho U^\dagger.
\end{equation}
\end{postulate}

\begin{postulate2*}[Continuous-time evolution]
The time evolution of a density matrix $\rho$ of a closed quantum system is described by the von Neumann equation
\begin{equation}\nonumber
i\hbar \frac{\partial \rho}{\partial t} = [\mathrm{\hat{H}},\rho],
\end{equation}
where $\hbar$ is the reduced Planck's constant, $[a,b]=ab-ba$ is the commutator operator and $\mathrm{\hat{H}}$ is a fixed Hermitian operator known as the Hamiltonian of the closed quantum system.
\end{postulate2*}

\begin{postulate}[Quantum measurement]
Quantum measurements are described by a collection $\{M_m\}$ of \textit{measurement operators}. These are operators acting on the state space of the system being measured. The index $m$ refers to the measurement outcome that may occur in the experiment. If the state of the quantum system is $\rho$ immediately before the measurement then the probability that result $m$ occurs is given by
\begin{equation}\nonumber
p(m)=\mathrm{Tr}(M_m^\dagger M_m\rho),
\end{equation}
and the state of the system $\rho_m'$ after the measurement is
\begin{equation}\nonumber
\rho_m'=\frac{M_m\rho M_m^\dagger}{\mathrm{Tr}(M_m^\dagger M_m\rho)}
\end{equation}
\end{postulate}

\begin{postulate}[Composite systems]
The state space of a composite quantum system is the tensor product of the state spaces of the component physical systems. Moreover, if we have systems numbered $1$ through $n$, and system $k$ is prepared in state $\rho_k$, then the joint state of the total system is $\rho_1\otimes\rho_2\otimes\cdots\otimes\rho_n$.
\end{postulate}

Postulate \ref{densityevo1} can be derived from postulate 2$'$ once again. The reformulation of quantum mechanics using the language of density matrices is mathematically equivalent to the description in terms of the state vector. However, this way of thinking about quantum mechanics specially shines for two applications: the description of quantum systems whose vector state is not completely known, and the description of composite quantum systems. Finally, the densisty operator allows distinguishing the so-called \textit{pure} and \textit{mixed states}.

\begin{definition}[Pure state]
A quantum system whose state $\ket{\psi}$ is perfectly known is said to be in a \textit{pure state}. In this case, the density operator is just $\rho=\ket{\psi}\bra{\psi}$. A criteria for determining if a state is pure is that its density matrix fulfills $\mathrm{Tr}(\rho^2)=1$
\label{def:purestate}
\end{definition}
\begin{definition}[Mixed state]
A quantum state is said to be in a \textit{mixed state} if the state vector of the system is not exactly known. In this case, the state is said to be the \textit{mixture} of the different pure states in the ensemble for $\rho$. A criteria for determining if a state is mixed is that its density matrix fulfills $\mathrm{Tr}(\rho^2)<1$.
\label{def:mixedstate}
\end{definition}

\subsection{No-cloning Theorem and Quantum Teleportation}\label{subsub:no-cloning}
This last section finishes this introduction to quantum mechanics required for understanding the basic theory of quantum information and error correction. Here, two shocking results of quantum mechanics will be presented, due to their importance when dealing with correction of information in quantum channels.

\textit{The no-cloning} theorem is a basic theorem of quantum mechanics that refers to the prohibition of cloning quantum states. This means that an arbitrary quantum state cannot be replicated in order to obtain redundancy of information and, thus, have several copies of the same data to protect the information from noise. This technique is very common in the classical world, and due to such result it cannot be done in the quantum paradigm.

\begin{theorem}[No-cloning theorem]
It is not possible to create an identical copy of an arbitrary quantum state. This means that there is no unitary operator $U$ that realizes the action
\begin{equation}\nonumber
\ket{\psi}\otimes\ket{\varphi}\rightarrow U(\ket{\psi}\otimes\ket{\varphi})=\ket{\psi}\otimes\ket{\psi},
\end{equation}
for any arbitrary qubit tuple $\{\ket{\psi},\ket{\varphi}\}$.
\label{thm:no-cloning}
\end{theorem}

The nonexistence of a universal qubit cloning machine is a result of special importance for the design of quantum error correction codes as this implies that a repetition code cannot be used in order to protect the information from quantum noise. Theorem \ref{thm:no-cloning} then limits how quantum codes have to be designed.

\textit{Quantum teleportation} refers to the technique of moving quantum states around, even in the absence of quantum communication channels linking the sender and the recipient of the quantum state. This phenomena is based on the usage of entanglement and some classical communication in order to transmit an arbitrary qubit between the two parties.

Imagine that Alice wants to send a qubit to Bob. Consider that the qubit Alice posses is
\[\ket{\psi}_{ A'}=\alpha|0\rangle_{A'}+\beta|1\rangle_{A'}\]
Suppose that she shares a EPR pair $\ket{\Phi^+}_{AB}$ with Bob. The joint state of the systems of Alice $(A',A)$ and Bob $B$ is $\ket{\psi}_{ A'}\ket{\Phi^+}_{AB}$.
In addition, Alice and Bob can send two classical bits across two uses of a classical channel.
Note that Alice does not know the state $\ket{\psi}_{ A'}$ to be sent, and it cannot determine it since $\ket{\psi}_{ A'}$ would be destroyed when measuring it (also a single measurement of the state cannot give enough information to determine its amplitudes $\alpha$ and $\beta$). Moreover, even if the sender would knew the state of the qubit, it would be impossible to send $\ket{\psi}_{ A'}$ by transmitting two classical bits as $\ket{\psi}_{ A'}$ operates on a \textit{continuous space} and, therefore, an infinite amount of classical bits would be required to describe such $\ket{\psi}_{ A'}$.

The key to obtain the desired \textit{teleportation} is to use the pre-shared entangled EPR pair, plus a little amount of classical information. The circuit used for the teleportation is shown in Figure \ref{fig:teleport}, and the procedure works as follows.
\begin{figure}[h]
\centering
\leavevmode
\Qcircuit @C=1em @R=0.5em @!{
& \lstick{\ket{\psi}_{ A'}} & \ctrl{1} & \gate{H} & \meter & \cw & \cw & \ustick{M_1}\cw \\
& \lstick{} &\targ & \qw  & \meter & \cw & \ustick{M_2}\cw & \cwx \\
& \lstick{\raisebox{3em}{$\ket{\Phi^+}_{AB}$\ }} & \dstick{\ket{\psi_0}\;}\qw & \dstick{\ket{\psi_1}\;}\qw & \dstick{\ket{\psi_2}}\qw & \dstick{\ket{\psi_3}}\qw
 &\gate{X^{M_2}} \cwx & \gate{Z^{M_1}} \cwx &  \rstick{\ket{\psi}}\qw
\gategroup{2}{2}{3}{2}{.8em}{\{}
}
\caption{Quantum teleportation scheme. $M_1$ and $M_2$ are the measurement outcomes of the first two qubits.}
\label{fig:teleport}
\end{figure}
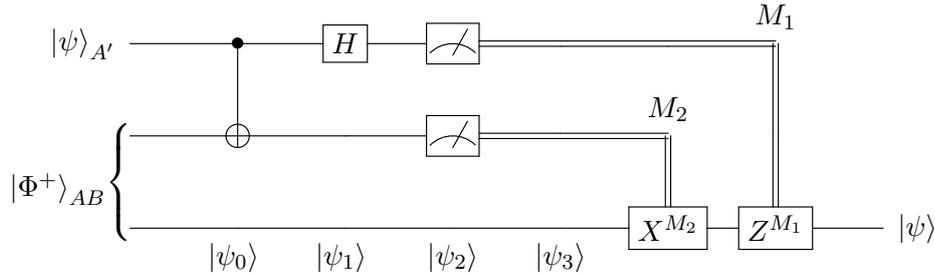

The initial composite state of the qubit to teleport and the EPR pair shared by the sender and the receiver $\ket{\psi_0}$ is
\begin{equation}\nonumber
\ket{\psi_0}_{A'AB}=\ket{\psi}_{A'}\ket{\Phi^+}_{AB}=\frac{1}{\sqrt{2}}[\alpha\ket{0}_{A'}(\ket{00}_{AB}+\ket{11}_{AB})+\beta\ket{1}_{A'}(\ket{00}_{AB}+\ket{11}_{AB})],
\end{equation}
After the CNOT operation, the state will be
\begin{equation}\nonumber
\ket{\psi_1}_{A'AB}=\frac{1}{\sqrt{2}}[\alpha\ket{0}_{A'}(\ket{00}_{AB}+\ket{11}_{AB})+\beta\ket{1}_{A'}(\ket{10}_{AB}+\ket{01}_{AB})]
\end{equation}
The next step is to apply a Hadamard gate to the first qubit of the composite state. This operation changes the state to
\begin{equation}\nonumber
\ket{\psi_2}_{A'AB}=\frac{1}{2}[\alpha(\ket{0}_{A'}+\ket{1}_{A'})(\ket{00}_{AB}+\ket{11}_{AB})+\beta(\ket{0}_{A'}-\ket{1}_{A'})(\ket{10}_{AB}+\ket{01}_{AB})],
\end{equation}
which after regrouping terms
\begin{equation}
\begin{split}
\ket{\psi_3}_{A'AB}=&\frac{1}{2}[\ket{00}_{A'A}(\alpha\ket{0}_B+\beta\ket{1}_B)+\ket{01}_{A'A}(\alpha\ket{1}_B+\beta\ket{0}_B) \\ +&\ket{10}_{A'A}(\alpha\ket{0}_B-\beta\ket{1}_B)+\ket{11}_{A'A}(\alpha\ket{1}_B-\beta\ket{0}_B)].
\end{split}
\label{eq:psi2}
\end{equation}

Expression \ref{eq:psi2} gives the composite state of Alice and Bob and is sum of four different terms. Then, if Alice measures her states, by the projective measurements $\{\ket{00}\bra{00}_{A'A}, \ket{01}\bra{01}_{A'A}, \ket{10}\bra{10}_{A'A}, \ket{11}\bra{11}_{A'A}\}$, one of the four possible outcomes $00,01,10,11$ will result. The post-measurement states of the composite system associated with each of the outcomes will be:
\begin{equation}\nonumber
|00\rangle_{A'A}\otimes \overbrace{(\alpha\ket{0}_B+\beta\ket{1}_B)}^{\ket{\psi_3}_B}=\ket{00}_{A'A}\otimes \ket{\psi}_B
\end{equation}
\begin{equation}\nonumber
|01\rangle_{A'A} \otimes \overbrace{(\alpha\ket{1}_B+\beta\ket{0}_B)}^{\ket{\psi_3}_B}=\ket{01}_{A'A}\otimes X\ket{\psi}_B
\end{equation}
\begin{equation}\nonumber
|10\rangle_{A'A} \otimes\overbrace{(\alpha\ket{0}_B-\beta\ket{1}_B)}^{\ket{\psi_3}_B}=\ket{10}_{A'A}\otimes Z\ket{\psi}_B
\end{equation}
\begin{equation}\nonumber
|11\rangle_{A'A} \otimes\overbrace{(\alpha\ket{1}_B-\beta\ket{0}_B)}^{\ket{\psi_3}_B}=\ket{01}_{A'A}\otimes XZ\ket{\psi}_B
\end{equation}
Note that these composite states are given by the product states of Alice and Bob. The two left qubits correspond to the state of Alice whereas the third qubit $\ket{\psi_3}_B$ belongs to Bob. Therefore, in order for Bob to recover the state $\ket{\psi}_B=\alpha\ket{0}_B+\beta\ket{1}_B$, Alice transmits two classical bits to Bob indicating which of the four results have occurred. After Bob receives the two bits, he knows what operation he has to perform in his qubit to recover the state. From the above results, the recover protocol will be:
\begin{itemize}
\item If Bob receives $00$, then no operation must be applied.
\item If Bob receives $01$, then he applies $\mathrm{X}$ gate.
\item If Bob receives $10$, then he applies $\mathrm{Z}$ gate.
\item If Bob receives $11$, then he applies a $\mathrm{X}$ gate first followed by a $\mathrm{Z}$ gate.
\end{itemize}

The recovery operation can be compactly written as $\ket{\psi}=Z^{M_1}X^{M_2}\ket{\psi_3}_B$, where $M_1$ and $M_2$ denote the measurement result of qubits $A'$ and $A$, respectively. Note that quantum teleportation does not violate Einstein's special relativity \cite{einstein} since the classical information that is needed in order to complete the teleportation protocol is limited by the speed of light.

\subsection{Qubit implementations}
We close this section by briefly introducing the most popular technologies that are being used nowadays for state-of-the-art qubit implementations. We will not go through the specific details of these technologies as such is out of the scope of this Ph.D. dissertation. The discussions about the results obtained in this thesis are hardware agnostic in the sense that we work with qubits as two-level coherent quantum systems without taking into account how they are physically constructed. However, many of the experimental results used to derive our research are based on some specific qubit construction. Therefore, we consider important to at least briefly introduce some of the most promising technologies for qubit implementation.

 The most relevant qubit implementations at the time of writing are:
\begin{itemize}
\item \textbf{Superconducting qubits:} Currently, this is the most common qubit implementation. The main idea behind superconducting qubits is to use a microwave signal in order to place a resistance-free current in a superposition state as it oscillates around a circuit loop. This is usually realized by using Josephson junctions \cite{josephson}, which are electrical elements consisted by two superconductors coupled by a weak link. The need of superconducting materials implies that the qubits need to be cooled down to temperatures near the absolute zero. There exist several types of superconducting qubits such as the transmon \cite{transmon2,transmon} or the Xmon \cite{xmon}.
\begin{figure}
\centering
\includegraphics[scale=0.7]{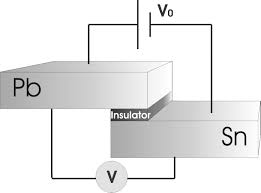}
\caption{Schematic representation of a Josephson junction. It can be seen that both of the superconductors (if both the lead (Pb) and tin (Sn) are cooled down) are linked with a thin insulator.}
\label{fig:jj}
\end{figure}
\item \textbf{Ion traps:} Ion traps are electrically trapped charged atoms (ions). The outermost electron orbiting the nucleus can be used as a qubit. Elements named Paul traps \cite{paultrap} are used in order to place the charged atoms as desired. This atoms must also be kept at very low temperatures so that the qubit can be formed.
\item \textbf{Photonic qubits:} Photonic qubits are based on the degrees of freedom that individual particles of light, named photons, present in order to construct the superposition state. The photonic qubits do not need extreme cooling and can be manufactured using silicon.
\end{itemize}

\begin{table}[h!]
\centering
\begin{tabular}{|c|c|c|}
\hline
\textbf{Technology} & \textbf{Players}                  \\
\hline
Superconducting     & IBM, Google, Rigetti, Intel, China  \\
\hline
Ion traps           & IonQ, Honeywell                    \\
\hline
Photonics           & Xanadu, PsiQuantum, China         \\
\hline
Neutral atoms       & Atom Computing, PASQAL             \\
\hline
NV centers          & Yale, Harvard                     \\
\hline
Quantum dots        & Intel                             \\
\hline
Topological         & Microsoft                         \\
\hline
\end{tabular}
\caption{Table summarizing state-of-the-art technologies used for qubit implementations.}
\label{tab:qubitTech}
\end{table}

Those are the main approaches to experimental quantum computing at the time. However, in table \ref{tab:qubitTech} we summarize several other technologies that are also being developed nowadays. Each of the technologies present several pros and cons and some of those will be presented throughout the dissertation, specially when decoherence is discussed.

\section{Decoherence and Quantum channels}\label{sec:deco}
Environmental decoherence is described as the undesired interaction\footnote{In quantum mechanical terms this interaction is more specifically called entanglement.} of a qubit with the environment, resulting in the perturbation of the coherent superposition, i.e., variations of the original amplitudes. Decoherence arises from several physical processes that should be mathematically modeled. In the previous section, postulates \ref{post:postulate2} and 2$'$ described the evolution of quantum systems. However, the Sch\"rodinger and von Neumann equations describe the evolution of \textit{closed} quantum mechanical systems. Consequently, due to the fact that actual implementable systems will never be closed\footnote{Any physical implementation of qubits will have some kind of interaction with its surrounding environment and, thus, decoherence.}, we have to describe the evolution of quantum information when it is subjected to interaction with the environment. Therefore, the theory of \textit{open} quantum systems must be considered in order to be able to understand how decoherence corrupts quantum information.

The physical processes that generate this quantum noise depend on the qubit technology being used to construct the quantum computer. As it was presented in table \ref{tab:qubitTech}, the most promising technologies currently in use are transmon superconducting qubits \cite{transmon}, trapped ion qubits \cite{trapped} and quantum dot qubits \cite{quantumdot} among others. Decoherence also takes place during the transmission of quantum states through fiber optics or via laser beams. Fortunately, despite the diversity of the quantum decoherence processes, they can be modeled by the same mathematical tools.

Next, we will discuss decoherence and the way it can be mathematically modelled by means of \textit{quantum channels}. It is important to point out that in this thesis we focus on quantum computation and error correction based on two-level physical constructions, in which the physical elements are realized by discrete two-level systems. However, other promising routes towards universal quantum computation and error correction exist, such as bosonic codes \cite{bosonic1,bosonic2}.

That being said, we begin this section describing some of the physical processes that compose decoherence. We follow by introducing the Lindblad master equation to describe the evolution of open quantum systems and how the set of interactions described before are encapsulated in such mathematical equation. Afterwards, we discuss how the results of the Lindblad master equation can be condensed into quantum channels, and present a pretty fair description of decoherence via the \textit{combined amplitude and phase damping channel}. We continue by discussing how to approximate such noise model into quantum channels that can be efficiently simulated in classical computers by using the quantum information technique named \textit{twirling}. We explain why the use of the approximate quantum channels is valid for designing and simulating QECCs that will work in real quantum hardware. We go on by briefly discussing the presence of memory in quantum channels. We finish the section by introducing the concept of quantum channel capacity and discussing it for some of the quantum channels of interest for this dissertation.

\subsection{Interaction of qubits with the environment}\label{sub:decWhy}
Quantum systems interact with their surrounding environment in a variety of ways. Thus, in the reality, the concept of closed quantum systems is just an idealisation. As a consequence of such leakage to the environment, quantum information suffers the so-called \textit{decoherence}, which is usually defined as the loss of coherence of the two-level coherent system that leds to the perturbation of the superposition state that comprises the qubit. The general problem described by the theory of open quantum systems is shown in Figure \ref{fig:openQ} in an schematic manner. There, it can be seen that both the environment and the quantum system form a complete system which is closed. However, we are interested in what happens with the quantum system, and so having to deal with the whole Hilbert space, which is often of infinite dimensions, does not make sense and is intractable. The theory of open quantum systems then studies the dynamics of the smaller quantum system of interest when subjected to interactions with the environment.

\begin{figure}[h!]
\centering
\includegraphics[scale=0.3]{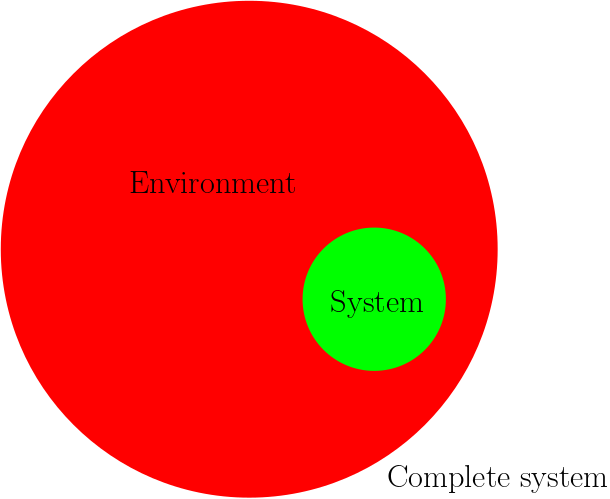}
\caption{Schematic representation of the problem treated by the theory of open quantum systems.}
\label{fig:openQ}
\end{figure}

Quantum systems suffer from \textit{relaxation} or \textit{energy dissipation}, which is the name given to the effects that quantum mechanical states suffer as a result of spontaneous energy losses. This can happen when an atom in an excited state emits a photon and returns to its ground state, when a spin system at high temperature approaches equilibrium with its environment, when a photon in an interferometer is subjected to scattering and attenuation; or when a photon is absorbed during its transmission through an optical fiber. An example of why this energy loss occurs can be appreciated by observing the energy profile of a transmon qubit, which is shown in Figure \ref{fig:transmon}. In the energy level chart presented for said qubits, it is easy to see that the states defining the qubits (or qudits if more than two energy levels of the transmon are taken) are given by the possible quantized energy states of such a superconducting transmon. Therefore, the $|0\rangle$ state corresponds to the ground state of the system, and the $|1\rangle$ state is related to an excited state. Due to the fact that an isolated system tends to drop into its lowest energy state, decoherence for this type of qubit takes place when energy is lost to the environment as a consequence of the excited state approaching the ground state.

\begin{figure}[h!]
\centering
\includegraphics[scale=0.3]{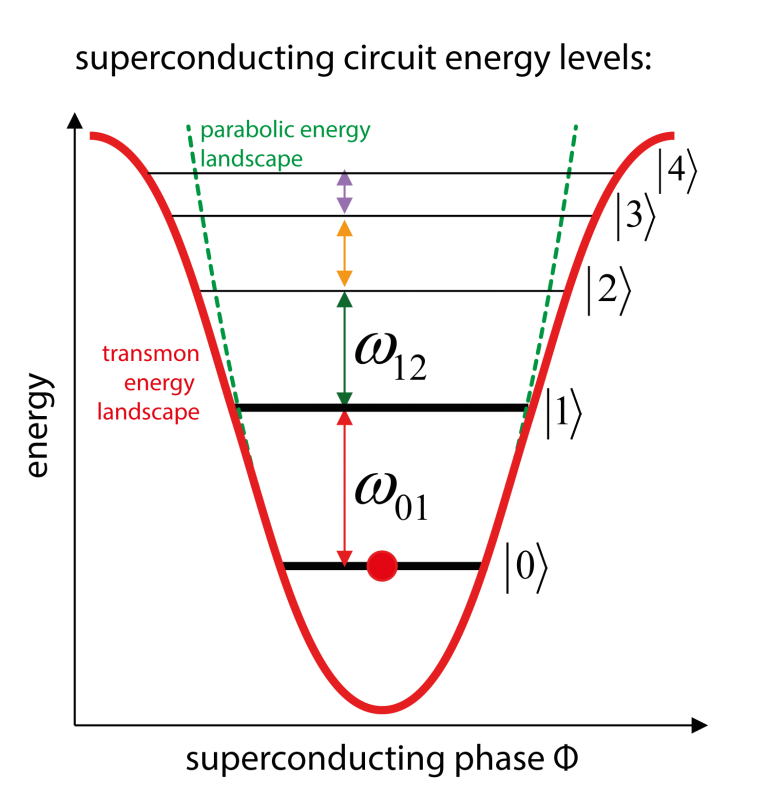}
\caption{Energy levels of typical transmon superconducting qubits. The $|0\rangle$ state corresponds to the ground state of the system while the $|1\rangle$ state is an excited state of such a transmon. Source: \cite{qutechTransmon}.}
\label{fig:transmon}
\end{figure}

Every energy dissipation process has its own unique features, but their effect in the qubits can be described mathematically in a similar manner. Consider the two-level transmon qubit described before, whose ground state is denoted by $|0\rangle_T$ and its excited state by $|1\rangle_T$. Additionally, consider the mathematical abstraction that the environment is also a two-level system with a vacuum state $|0\rangle_E$ and an excited state $|1\rangle_E$, that is always initialized in the vacuum state. Then energy dissipation processes can be described as
\begin{equation}\label{eq:ampdamp1}
\begin{split}
& |0\rangle_T|0\rangle_E \rightarrow |0\rangle_T|0\rangle_E, \\
& |1\rangle_T|0\rangle_E \rightarrow \sqrt{1 - \gamma}|1\rangle_T|0\rangle_E + \sqrt{\gamma} |0\rangle_T|1\rangle_E,
\end{split}
\end{equation}
where $\gamma$ refers to the damping probability or the probability that the system loses energy to the environment when it is in its excited state. The expressions in \eqref{ampdamp1} mathematically describe, in a general and unified way, what was previously discussed regarding the processes of energy dissipation that occur for the various methods used for the construction of qubits. The effect that energy dissipation induces in a general qubit\footnote{Note that the subscript $T$ is no longer used. The rationale now refers to every qubit technology, and not just transmon superconducting qubits. Its prior use was intended as a particular instance for explanatory purposes.} $|\psi\rangle = \alpha|0\rangle + \beta|1\rangle$ is described by the following transformation:
\begin{equation}\label{eq:ampdamp2}
|\psi\rangle|0\rangle_E \rightarrow \left(\alpha|0\rangle + \beta\sqrt{1-\gamma}|1\rangle\right)|0\rangle_E + \beta\sqrt{\gamma}|0\rangle|1\rangle_E.
\end{equation}
Following the rationale of \eqref{ampdamp1}, what \eqref{ampdamp2} is describing is the fact that state $|\psi\rangle$ decoheres to state $\left(\frac{\alpha}{\sqrt{1-\gamma\beta^2}}|0\rangle + \frac{\beta\sqrt{1-\gamma}}{\sqrt{1-\gamma\beta^2}}|1\rangle\right)$ with probability $1-\gamma\beta^2$ leaving the environment unchanged; or it decoheres to $|0\rangle$ with probability $\gamma\beta^2$ releasing a photon (or energy quanta) to the environment, and so leaving it in state $|1\rangle_E$.

It is important to remark that in the discussion above $|\psi\rangle$ has been assumed to be an isolated qubit. However, this is not true in general, since in a real system more than one qubit is present, and they will most assuredly interact with each other. This interaction is reflected in the form of entanglement, which implies that the decoherence processes affect the set of qubits in a combined fashion, and not in the isolated manner described before. Regardless, the assumption that each of the qubits interacts with the environment independently is still reasonable \cite{isomorphism}, simplifying the description to that of \eqref{ampdamp2}.

Another set of physical mechanisms that affect quantum information are those encompassed by the term \textit{dephasing}. These processes are uniquely quantum mechanical and describe the loss of quantum information without loss of energy. This kind of environmental decoherence can arise when a photon scatters randomly as it travels through a waveguide or when electronic states are perturbated due to the action of stray electrical charges. What happens during these events is that the system evolves for an amount of time which is not known with precision, and so partial information about its quantum phase is lost.

\begin{figure}[h!]
\centering
\includegraphics[scale=0.5]{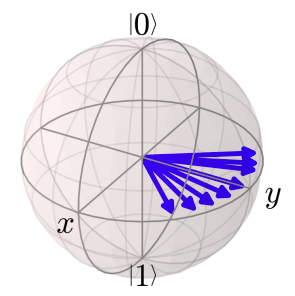}
\caption{Representation of the evolution of the phase of the qubit when it is involved in dephasing interactions with the environment.}
\label{fig:dephasing}
\end{figure}

A simple model to mathematically describe this particular set of processes can be constructed using \textit{phase kicks} on the qubits \cite{NielsenChuang}. Phase kicks are rotations applied to the qubit $|\psi\rangle$ with a random angle $\theta$. Assuming that the random variable $\theta$ follows a Gaussian distribution with mean $0$ and variance $2\lambda$ the phase of the qubit is perturbed with such random phase, obtaining the effects caused by dephasing mechanisms. Figure \ref{fig:dephasing} shows a graphical representation of such phase change.

Quantum information might be corrupted in other ways provoked by other physical interactions that have not been described before. An interesting example is the erasure and deletion errors that may happen caused by the loss of a photon when being processed in an optical quantum device \cite{erasure1,erasure2,deletion}. In those cases, the photons, that are used as qubits, are lost completely to the media, implying not that the information is corrupted but completely lost. In this thesis we will focus specifically on the decoherence processes related to relaxation and dephasing, so in the following we will discuss such as the mechanisms that corrupt quantum information.

\subsection{The Lindblad Master Equation}\label{sub:lindblad}
As explained in the previous section, we want to describe the time evolution of a quantum system when this is subjected to interaction with the environment. In order to do so, we could study the time evoultion of the whole quantum system following the rationale of Figure \ref{fig:openQ} which is subjected to closed evolution. This is generally done by means of the von Neumann equation \cite{SchlorPhD,Lindblad}
\begin{equation}
i\hbar \frac{\partial\rho_\mathrm{T}}{\partial t} = [\mathrm{\hat{H}}_\mathrm{T}, \rho_\mathrm{T}],
\end{equation}
where here $\rho_\mathrm{T}$ and $\mathrm{\hat{H}}_\mathrm{T}$ refer to the density matrix and the Hamiltonian of the complete system combining the quantum system of interest and the environment. Unfortunately, this approach is unfeasible and computationally intractable due to the size (usually infinite) of Hilbert space associated to the complete quantum system \cite{SchlorPhD}. This issue can be overcomed by means of the so-called \textit{Lindblad master equation}, which describes the evolution of a mixed state $\rho$ coupled to its environment. Such master equation describes the time evolution of the mixed state $\rho$ as \cite{SchlorPhD,Lindblad}
\begin{equation}\label{eq:lindblad}
i\hbar\frac{\partial \rho}{\partial t} = [\mathrm{\hat{H}},\rho] + \sum_k \Gamma_k \left(\mathrm{L}_k\rho\mathrm{L}_k^\dagger - \frac{1}{2}\left\lbrace \mathrm{L}_k^\dagger \mathrm{L}_k,\rho\right\rbrace\right),
\end{equation}
where $\rho$ is the density operator of the quantum system of interest, $\mathrm{\hat{H}}$ is its time-independent Hamiltonian\footnote{This Hamiltonian represents the coherent part of the dynamics \cite{NielsenChuang}.} including the so-called \textit{Lamb shift} Hamiltonian\footnote{Its role is to renormalize the system energy levels due to the interaction with the environment \cite{Lindblad}.}, $\mathrm{L}_k$ are the Lindblad or jump operators\footnote{The Lindblad operators are generally obtained by using the so-called Born-Markov approximations \cite{NielsenChuang}.} that describe the exchange of single quanta of the quantum system with the environment and $\Gamma_k$ is the interaction rate of the quanta for decoherence source $k$. $\{a,b\}=ab + ba$ refers to the anticommutator operator.

This way, using the master equation, we can study the evolution of the quantum state of interest, $\rho$, when it interacts with the environment by the decoherence mechanisms described by the Lindblad operators $\mathrm{L_k}$. As discussed in the previous section, we are interested in decoherence mechanisms named relaxation and dephasing, which usually affect and limit most of the qubits that are physically constructed. That said so, the main mechanisms of decoherence are:
\begin{itemize}
\item\textbf{Energy relaxation: }this relates to the energy loss processes described. The Lindblad operator associated to this set of decoherence mechanisms is $\mathrm{L}_1 = \frac{1}{2} (\mathrm{X}-i\mathrm{Y})$. The rate at which the quantum system interacts this way with the environment is defined by $\Gamma_1$, named relaxation rate.
\item\textbf{Pure dephasing: }this relates to the processes involving just change of phase of the quantum information. The jump operator associated to this set of interactions is $\mathrm{L}_\phi = \frac{1}{\sqrt{2}}\mathrm{Z}$. The pure dephasing rate is given by $\Gamma_\phi$.
\item\textbf{Thermal excitation: }this refers to the undesired excitation of the qubit caused by the finite temperature of the system. All real systems have finite temperature. However, the physical constructions are cooled down to very low temperatures ($T\approx 20$ mK), almost vanishing the rate of interaction with the environment \cite{SchlorPhD}. Consequently, we can theoretically neglect this decoherence mechanism\footnote{Note that photonic qubits are not usually cooled down to such extreme temperatures, even operating at room temperature, but photons are insensitive to thermal excitation \cite{photonsG}.}.
\end{itemize}

Once the decoherence mechanisms we are going to consider have been defined, we can follow up to solve the dynamics of the density matrix, $\rho$, by solving the Lindblad master equation. We know from postulate 1 of quantum mechanics in the density matrix picture that a $\rho$ is a unity trace positive operator, thus, it can be written at time $t=0$ as
\begin{equation}\label{eq:rhoMat}
\rho(0) = \begin{pmatrix}
1 - \rho_{11} & \rho_{01}\\
\rho_{01}^* &  \rho_{11}
\end{pmatrix},
\end{equation}
where $\rho_{ij}$ refer to the elements of the density matrix.

Now we can proceed to solve master equation \eqref{lindblad} to obtain the time evolution of the density matrix subjected to relaxation and dephasing. The solution for $\rho(t)$ is then \cite{SchlorPhD,NielsenChuang}
\begin{equation}\label{eq:stepLind1}
\rho(t) = \begin{pmatrix}
1 - \rho_{11}e^{-t\Gamma_1} & \rho_{01}e^{-t\Gamma_2} \\
\rho_{01}^*e^{-t\Gamma_2} &  \rho_{11}e^{-t\Gamma_1}
\end{pmatrix} = \begin{pmatrix}
1 - \rho_{11}e^{-\frac{t}{T_1}} & \rho_{01}e^{-\frac{t}{T_2}} \\
\rho_{01}^*e^{-\frac{t}{T_2}} & \rho_{11}e^{-\frac{t}{T_1}}
\end{pmatrix},
\end{equation}
where $\Gamma_2 = \Gamma_1/2 + \Gamma_\phi$ is the dephasing rate, $T_1 = 1/\Gamma_1$ is the relaxation time, $T_\phi = 1/\Gamma_\phi$ is the pure dephasing time and $T_2 = 1/\Gamma_2$ is the dephasing or Ramsey time. The solution of the master equation implies that the elements of the density matrix decay to the ground state with rate $\Gamma_1$. Note also that the off-diagonal elements, which encode the phase of the qubit, also vanish with the contribution of both $\Gamma_1$ and $\Gamma_\phi$. Consequently, the quantum information that is stored in density matrix $\rho$ is corrupted through the time, with rates $\Gamma_1$ and $\Gamma_2$, due to the relaxation and pure dephasing interactions with the environment.

It is important to emphasize the relation
\begin{equation}\label{eq:reldeph}
\frac{1}{T_2} = \frac{1}{2T_1} + \frac{1}{T_\phi},
\end{equation}
since it can be seen that relaxation always implies dephasing. This comes from the fact that the relaxation does change the off-diagonal elements of the density matrix, and those are the ones related with the phase of the quantum state. As said before, pure dephasing refers to effects that only imply phase change. Note that when the pure dephasing rate $\Gamma_\phi=0$, the actual dephasing rate is limited by the relaxation time as $T_2 = 2T_1$. The $T_2\approx 2T_1$ is called Ramsey limit, and qubits that show such relation are said to be $T_1$-limited.

Additionally, the parameters in \eqref{reldeph} are important as they are experimentally measurable and, thus, they are used to benchmark how robust are the qubits that are physically implemented when they face decoherence. We will not explain the experiments that need to be done in order to measure such parameters, as they are out of the scope of this dissertation.

To sum up, we have solved the master equation when relaxation and pure dephasing interactions are considered for open evolution, and we have seen that the quantum state decays to a classical state without phase information when the time is past. The rate at which such decay occurs is given by the experimentally measurable relaxation, pure dephasing and Ramsey times.  It is interesting to comment that the interaction of quantum states with the environment when the thermal contribution of decoherence is not negligible can be also studied via a master equation \cite{masterTemp}. However, as stated before, we will consider that the system under study is cooled down to very low temperatures.

\subsection{Quantum Channels}\label{sec:qChans}
The Lindblad master equation was presented in last section as a tool for describing the evolution of open quantum systems. We have solved it when the interactions are the energy relaxation and pure dephasing in order to see the decay of the density matrix. In general, solving a master equation gives the time dependency of a density matrix when subjected to interactions with its surrounding environment. Similarly to the fact that the continuous evolution described by the Sch\"odinger and the von Neumann equations is equivalent to unitary evolution (postulates 2 and 2$'$) when considering closed quantum systems; the open evolution described by the Lindblad master equation can be also interpreted by completely-positive, trace-preserving (CPTP) linear maps between the spaces of operators \cite{NielsenChuang}.

In general, the evolution of a quantum system is described by means of a \textit{quantum channel}, $\mathcal{N}$. Quantum channels are defined as \cite{wildeQIT}
\begin{definition}[Quantum channel]\label{def:qchan}
\textit{A quantum channel, $\mathcal{N}$, is a linear, completely-positive, trace-preserving map between spaces of operators, corresponding to a quantum physical evolution.}
\end{definition}

The three conditions that a map between the spaces of operators must fulfill in order to be a quantum channel are obtained in an axiomatic way \cite{wildeQIT}. The map must be linear so that it respects conves mixtures of states, it should be completely-positive so that the map is from quantum states to quantum states, and it should be trace-preserving again so that the outcome is a valid quantum state. This way, we know that quantum channels are indeed valid possible evolutions in the context of quantum mechanics.

The above definition of a quantum channel leads to the \textit{Choi-Kraus theorem}, which states that any map satisfying the above criteria can only take the expression named \textit{Choi-Kraus} decomposition or \textit{operator-sum} representation \cite{wildeQIT}.
\begin{theorem}[Choi-Kraus]
\textit{A map $\mathcal{N}:\mathcal{H}_A\rightarrow\mathcal{H}_B$ is linear, completely-positive, and trace-preserving if and only if it has a Choi-Kraus decomposition as}
\begin{equation}
\mathcal{N}(\rho) = \sum_k E_k\rho E_k^\dagger,
\end{equation}
\textit{where the $E_k$ matrices are operators on the state space of the principal system and receive the name of \textit{Kraus operators} or error operators of the channel. As quantum channels are trace-preserving operators, Kraus operators must fulfill $\sum_k E_k^\dagger E_k = \mathrm{I}$. The minimum number of Kraus operators is called the Kraus rank of the quantum channel $\mathcal{N}$. Channels with Kraus rank $1$ are called pure.}
\end{theorem}

Therefore, the evolution of any open quantum mechanical systems is described by maps that can only be represented by the operator-sum representation.

In the previous section we studied the evolution of open quantum mechanical systems by means of the Lindblad master equation. Both of the approaches are equivalent as the solutions of the master equations can be written as a Choi-Kraus decomposition \cite{NielsenChuang}. This way, one can obtain the error operators, $E_k(t)$ for the quantum channels by solving the master equation, and then work with the CPTP linear map acting on the density matrices. Both approaches are equally valid, and they have their own places. However, none of them is the most general approach since some problems can be worked out with one of the approaches but not with the other \cite{NielsenChuang}. In this thesis we will work with quantum channels, but as we will see now, the error operators used are the ones obtained from the solution that was obtained from the master equation.

\subsubsection{Amplitude damping channel}\label{sub:ampdampchan}
The amplitude damping channel, $\mathcal{N}_{\mathrm{AD}}$ is the quantum channel that models the evolution of a quantum system that suffers from energy loss \cite{NielsenChuang,isomorphism}. Therefore, this channel describes a quantum density matrix that can decay to the ground state from its excited state with probability $\gamma$. Such probability is usually named \textit{damping parameter}. The amplitude damping channel has Kraus rank $2$ and its error operators are
\begin{equation}
E_0 = \begin{pmatrix}
1 & 0 \\
0 & \sqrt{1- \gamma}
\end{pmatrix}\text{ and }
E_1 = \begin{pmatrix}
0 & \sqrt{\gamma} \\
0 & 0
\end{pmatrix}.
\end{equation}

If we apply each of the error operators to a pure state $|\psi\rangle=\alpha\ket{0}+\beta\ket{1}$, we will see that each of the possible outcomes that were present in \eqref{ampdamp2} are obtained. Applying $E_0$,
\begin{equation}\label{eq:krauspure1}
\begin{split}
E_0|\psi\rangle &= \begin{pmatrix}
1 & 0 \\
0 & \sqrt{1- \gamma}
\end{pmatrix}\begin{pmatrix}
\alpha \\ \beta
\end{pmatrix} \\ &=
\begin{pmatrix}
\alpha \\
\sqrt{1-\gamma}\beta
\end{pmatrix} \\ &=
\alpha\begin{pmatrix}
1 \\ 0
\end{pmatrix} +\sqrt{1-\gamma}\beta \begin{pmatrix}
0 \\ 1
\end{pmatrix} \\ & =
\alpha|0\rangle + \sqrt{1-\gamma}\beta |1\rangle,
\end{split}
\end{equation}
which, up to normalization, is the state that occurs with probability $|E_0|\psi\rangle|^2 = (1-\gamma\beta^2)$ in \eqref{ampdamp2} if the environment is left unchanged. If $E_1$ is applied,

\begin{equation}\label{eq:krauspure2}
\begin{split}
E_1|\psi\rangle &= \begin{pmatrix}
0 & \sqrt{\gamma} \\
0 & 0
\end{pmatrix}\begin{pmatrix}
\alpha \\ \beta
\end{pmatrix} \\ &=
\begin{pmatrix}
\sqrt{\gamma}\beta \\
0
\end{pmatrix} \\ &=
\sqrt{\gamma}\beta\begin{pmatrix}
1 \\ 0
\end{pmatrix}=
\sqrt{\gamma}\beta|0\rangle,
\end{split}
\end{equation}
which, up to normalization, is the state that occurs with probability $|E_1|\psi\rangle|^2 = \gamma\beta^2$ in \eqref{ampdamp2} if an energy quanta is lost to the environment. In consequence, we can see the equivalence relationship between the Kraus operators of the CPTP map of the amplitude damping channel and the description of energy loss that was given for pure states.

Now that the amplitude damping channel has been defined, a physical interpretation of the damping probability must be obtained. Note that $\mathcal{N}_{\mathrm{AD}}$ describes the open evolution of $\rho$ when subjected to relaxation. We know that the evolution of such density matrix is described by the solution of the Lindblad master equation with $\Gamma_1\neq 0$ and $\Gamma_\phi = 0$. This implies that $T_2=2T_1$. The evolution of the density matrix is then
\begin{equation}\label{eq:ampdampLind}
\rho(t) = \begin{pmatrix}
1 - \rho_{11}e^{-\frac{t}{T_1}} & \rho_{01}e^{-\frac{t}{2T_1}} \\
\rho_{01}^*e^{-\frac{t}{2T_1}} & \rho_{11}e^{-\frac{t}{T_1}}
\end{pmatrix}.
\end{equation}

Otherwise, expanding the Choi-Kraus representation of the amplitude damping channel, it can be seen that the evolution of the density matrix is
\begin{equation}\label{eq:ampdamprho}
\mathcal{N}_{\mathrm{AD}}(\rho) = \begin{pmatrix}
1-(1-\gamma)\rho_{11} & \sqrt{1-\gamma}\rho_{01} \\
\sqrt{1-\gamma}\rho_{01}^* & (1-\gamma)\rho_{11}
\end{pmatrix}.
\end{equation}

From both \eqref{ampdampLind} and \eqref{ampdamprho}, it is easy to see that by taking $\gamma(t)=1-e^{-t/T_1}$, both expressions will be the same. This way, we can see that the damping probability depends on how much time the system is let to evolve, implying that the longer it evolves, it is more probable that the quantum information is corrupted. This way we used the result of the master equation in order to obtain the relationship between the abstraction and the physical quantity. It is important to state that the amplitude damping channel is an accurate model in order to describe the decoherence of $T_1$-limited qubits, as it is considered that the pure dephasing rate is negligible and so the Ramsey limit is saturated.

\begin{table}[h!]
\centering
\begin{tabular}{|l|l|}
\hline
\textbf{Name}         & $\mathbf{T_1(\mu s)}$ \\ \hline
IBM Q System One \cite{IBM}    & $73.9$  \\ \hline
Rigetti 32Q Aspen-7 \cite{Rigetti}      & $41$               \\ \hline
Google Sycamore  \cite{EAQIRCC},\cite{Google}     & $\approx 1$         \\ \hline
Ion Q 11 Qubit \cite{IonQ}       & $>10^{10}$             \\ \hline
Intel Q (Spin Qubits) \cite{EAQIRCC},\cite{Intel} & $>10^9$          \\ \hline
\end{tabular}
\caption{Relaxation times $T_1$ for some of the quantum devices available nowadays. Each company uses a different technology to implement the qubits of the quantum machines.}
\label{tab:qdevicesT1}
\end{table}

We finish the presentation of the amplitude damping channel by showing some typical values of $T_1$. Table \ref{tab:qdevicesT1} shows the relaxation time of the quantum machines that some of the companies working with experimental quantum computers possess at the time of writing. The values of $T_1$ vary from several microseconds to a few hours. Such a large variation is the result of the different technologies that are used to construct the quantum information units of the devices. Having a longer relaxation time is better for decoherence purposes, as it will take longer for the damping probability to reach relevant values. However, the technologies shown in table \ref{tab:qdevicesT1} that present such long $T_1$ have some drawbacks (such as slow operation or size), and thus qubit technology selection can not be solely based on this parameter. Note that technology evolves quickly, so the values of table \ref{tab:qdevicesT1} will change accordingly.

\subsubsection{Dephasing channel}\label{sub:phasedampchan}
The dephasing or phase damping channel, $\mathcal{N}_{\mathrm{PD}}$, describes the evolution of the qubits when it is subjected to pure dephasing \cite{NielsenChuang,isomorphism}. We have seen in the solution of the Lindblad master equation that relaxation implies dephasing, but this channel refers to interactions that exclusively lead to dephasing. In this channel, the probability that the phase of the qubit suffers from decay occurs with probability $\lambda$, which is generally named as \textit{scattering probability}\footnote{It is usually interpreted as the scattering probability of a photon without energy loss \cite{isomorphism}.}. The dephasing channel, $\mathcal{N}_{\mathrm{PD}}$, has Kraus rank $2$ and acts on the qubit with density matrix $\rho$ followed by Kraus operators

\begin{equation}
E_0 = \begin{pmatrix}
1 & 0 \\
0 & \sqrt{1- \lambda}
\end{pmatrix}\text{ and }
E_1 = \begin{pmatrix}
0 & 0 \\
0 & \sqrt{\lambda}
\end{pmatrix}.
\end{equation}

In a similar fashion to what was done for the amplitude damping channel, we apply the error operators of the channel to an arbitrary pure state $|\psi\rangle = \alpha|0\rangle + \beta|1\rangle$. The application of $E_0$ results in

\begin{equation}\label{eq:phasekrauspure1}
\begin{split}
E_0|\psi\rangle &= \begin{pmatrix}
1 & 0 \\
0 & \sqrt{1- \lambda}
\end{pmatrix}\begin{pmatrix}
\alpha \\ \beta
\end{pmatrix} \\ &=
\begin{pmatrix}
\alpha \\
\sqrt{1-\lambda}\beta
\end{pmatrix} \\ &=
\alpha\begin{pmatrix}
1 \\ 0
\end{pmatrix} +\sqrt{1-\lambda}\beta \begin{pmatrix}
0 \\ 1
\end{pmatrix} \\ & =
\alpha|0\rangle + \sqrt{1-\lambda}\beta |1\rangle,
\end{split}
\end{equation}
with probability $|E_0|\psi\rangle|^2 = (1-\lambda\beta^2)$. The effect of $E_1$ produces

\begin{equation}\label{eq:phasekrauspure2}
\begin{split}
E_1|\psi\rangle &= \begin{pmatrix}
0 & 0 \\
0 & \sqrt{\lambda}
\end{pmatrix}\begin{pmatrix}
\alpha \\ \beta
\end{pmatrix} \\ &=
\begin{pmatrix}
0 \\
\sqrt{\lambda}\beta
\end{pmatrix} \\ &=
\sqrt{\lambda}\beta\begin{pmatrix}
0 \\ 1
\end{pmatrix}=
\sqrt{\lambda}\beta|1\rangle,
\end{split}
\end{equation}
which is the state $|1\rangle$, up to normalization, occurring with probability $\lambda\beta^2$.

Following the rationale done for the amplitude damping channel, now we will use the solution of the master equation in order to relate the scattering probability with some of the physical parameters that can be measured. Note that for the dephasing channel, the relaxation rate vanishes ($\Gamma_1 = 0$) while the pure dephasing rate ($\Gamma_\phi\neq 0$) is finite. Consequently, the time evolution of the qubit subjected to pure dephasing is
\begin{equation}\label{eq:pdampLind}
\rho(t) = \begin{pmatrix}
1 - \rho_{11} & \rho_{01}e^{-\frac{t}{T_\phi}} \\
\rho_{01}^*e^{-\frac{t}{T_\phi}} & \rho_{11}
\end{pmatrix},
\end{equation}
and so it can be seen that the decay only occurs for the off-diagonal elements, which are related to the phase of the qubit.

Expanding the operator-sum representation of the dephasing channel, we obtain
\begin{equation}\label{eq:pdamprho}
\mathcal{N}_{\mathrm{AD}}(\rho) = \begin{pmatrix}
1-\rho_{11} & \sqrt{1-\lambda}\rho_{01} \\
\sqrt{1-\lambda}\rho_{01}^* & \rho_{11}
\end{pmatrix}.
\end{equation}

From both \eqref{ampdampLind} and \eqref{ampdamprho}, it is easy to see that by taking $\lambda(t)=1-e^{-2t/T_\phi}$, both expressions will be the same. Using the relation in \eqref{reldeph}, the scattering probability can be rewritten as $\lambda(t) = 1 - e^{(t/T_1 - 2t/T_2)}$, which is the typical expression for $\lambda(t)$ found in the literature \cite{twirl1,isomorphism}. From this, it can be seen that the phase decay caused by pure dephasing is also more probable with time, as it happened with the damping parameter.

Table \ref{tab:qdevicesT2} shows the dephasing times $T_2$, in microseconds, for the experimental devices presented in table \ref{tab:qdevicesT1}.

\begin{table}[h!]
\centering
\begin{tabular}{|l|l|}
\hline
\textbf{Name}         & $\mathbf{T_2(\mu s)}$ \\ \hline
IBM Q System One \cite{IBM}    & $69.1$  \\ \hline
Rigetti 32Q Aspen-7 \cite{Rigetti}      & $35$               \\ \hline
Google Sycamore  \cite{EAQIRCC},\cite{Google}     & $\approx 0.1$         \\ \hline
Ion Q 11 Qubit \cite{IonQ}       & $>10^{6}$             \\ \hline
Intel Q (Spin Qubits) \cite{EAQIRCC},\cite{Intel} & $>10^3$          \\ \hline
\end{tabular}
\caption{Dephasing times $T_2$ for some of the quantum devices available nowadays. Each company uses a different technology to implement the qubits of the quantum machines.}
\label{tab:qdevicesT2}
\end{table}

The table shows that the values of $T_2$ display substantial variation among devices, since each company uses different qubit technologies for their machines. This is similar to what happened for the relaxation time.

\subsubsection{Combined amplitude and phase damping channel}\label{sub:ampphasdamp}
If the previously described amplitude and phase damping channels are combined, a fairly complete mathematical model of qubit decoherence is obtained, as several physical processes that corrupt quantum information are encompassed by this abstraction. We will designate this amalgamation of channels as the \textit{combined amplitude and phase damping channel}, $\mathcal{N}_{\mathrm{APD}}$.

The combined amplitude and phase damping channel can be obtained by the serial concatenation of the action of each of the individual channels. The serial concatenation of two quantum channels, $\mathcal{N}_1$ and $\mathcal{N}_2$, with Kraus operators $E_k^1$ and $E_l^2$ is another quantum channel, $\mathcal{N}_3 = \mathcal{N}_1\circ\mathcal{N}_2$, with Kraus operators $E_m^3=E_k^1E_l^2$ \cite{wildeQIT}.

The effect of the combined channel is then captured by the next operator-sum representation\footnote{Note that in order to obtain this set of Kraus operators, the composition done here is $\mathcal{N}_{\mathrm{APD}}=\mathcal{N}_{\mathrm{PD}}\circ \mathcal{N}_{\mathrm{AD}}$. The result for the composition with the inverse order is equivalent, but the set of Kraus operators obtained by doing so is different. Note that the Kraus representation of a quantum channel is not unique \cite{watrous}.} with Kraus rank $3$ \cite{twirl1}:
\begin{equation}\label{eq:combinedKraus}
\mathcal{N}_{\mathrm{APD}}(\rho) = E_0\rho E_0^\dagger + E_1\rho E_1^\dagger + E_2\rho E_2^\dagger,
\end{equation}
with the following error operators:
\begin{equation}\label{eq:ampphaseKraus}
\begin{aligned} E_0 &= \begin{pmatrix}
1 & 0 \\
0 & \sqrt{1-\gamma-(1-\gamma)\lambda}
\end{pmatrix} \\
 &= \frac{1+\sqrt{1-\gamma-(1-\gamma)\lambda}}{2} \mathrm{I} \\
 &+ \frac{1-\sqrt{1-\gamma-(1-\gamma)\lambda}}{2} \mathrm{Z}
\\E_1 &= \begin{pmatrix}
0 & \sqrt{\gamma} \\
0 & 0
\end{pmatrix} = \frac{\sqrt{\gamma}}{2}\mathrm{X} + i\frac{\sqrt{\gamma}}{2}\mathrm{Y}
\\E_2 &= \begin{pmatrix}
0 & 0 \\
0 & \sqrt{(1-\gamma)\lambda}
\end{pmatrix} = \frac{\sqrt{(1-\gamma)\lambda}}{2}\mathrm{I} - \frac{\sqrt{(1-\gamma)\lambda}}{2}\mathrm{Z},
\end{aligned}
\end{equation}
where $\mathrm{I,X,Y,Z}$ are the Pauli matrices. We use $\mathcal{P}_n$ to refer to the set of $n$-fold tensor products of the Pauli matrices $\mathcal{P}_n =\{ \mathrm{I,X,Y,Z}\}^{\otimes n}$ \cite{EAQECC}. The damping, $\gamma$, and the scattering, $\lambda$, probabilities, as well as their corresponding time evolutions, are defined as in earlier sections of this work.

If we now consider the combination of \eqref{combinedKraus} and \eqref{ampphaseKraus}, the time evolution of an arbitrary quantum state with density matrix $\rho$ can be obtained as

\begin{equation}\label{eq:combinedDensityMatrix}
\begin{split}
\rho \rightarrow \mathcal{N}_{\mathrm{APD}}(\rho) &=\begin{pmatrix}
1-(1-\gamma)\rho_{11} & \rho_{01}\sqrt{1-\gamma-(1-\gamma)\lambda} \\
\rho_{01}^*\sqrt{1-\gamma-(1-\gamma)\lambda}  & (1-\gamma)\rho_{11}
\end{pmatrix}\\ &= \begin{pmatrix}
1 - \rho_{11}e^{-\frac{t}{T_1}} & \rho_{01}e^{-\frac{t}{T_2}} \\
\rho_{01}^* e^{-\frac{t}{T_2}} & \rho_{11}e^{-\frac{t}{T_1}}
\end{pmatrix},
\end{split}
\end{equation}
where $\rho_{ij}$ corresponds to the element of the matrix $\rho$ in row $i$ and column $j$. This expression is the same as the solution to the Lindblad master equation in \eqref{stepLind1} and, thus, the combined amplitude and phase damping channel is a fairly complete mathematical model to describe the decoherence suffered by a qubit. This last expression shows that the qubits are likely to decohere if the operation time (either information transmission, processing or storage) $t$ is of the same order of magnitude as either the relaxation time ($t\approx T_1$) or the dephasing time ($t\approx T_2$). This means that the reliability time of a qubit is defined by $T_1$ and $T_2$, and so the run times of the algorithms that will be executed on such devices will be limited to the aforementioned reliability time\footnote{The motivation behind QECCs is to extend the reliability times of qubits so that time demanding algorithms can be effectively executed in the quantum devices.}.

\subsection{Twirl approximations of Quantum Channels}\label{sec:approximatedChannels}
The combined amplitude and phase damping channel introduced in the previous section embodies an accurate mathematical depiction of several decoherence processes that corrupt the quantum information that is processed in a quantum device, transmitted between machines or stored in a quantum memory. Thus, the application of this model to create QECCs that will guarantee the stability of quantum information in the presence of decoherence seems reasonable. Unfortunately, the fact that the dimensions of the Hilbert spaces of $n$-qubit composite systems scale with a factor of $2^n$ makes matters increasingly more complex, since this implies that the classical simulation of quantum algorithms or error correction mechanisms may ultimately turn into an intractable problem. For example, a simple surface code of distance $d=5$ needs $25$ physical qubits, which is a system that has a Hilbert space with $33$ millions of dimensions \cite{twirl1}. As a result, the error dynamics of QECC schemes cannot be efficiently modeled on classical computers\footnote{We assume that these QECC systems will be designed on classical computers due to the lack of practical quantum machines at the time of writing.} by using \eqref{combinedDensityMatrix},  meaning that an approximation that is manageable employing conventional methods must be found.

In this section we present how approximated channels that can be efficiently simulated on classical computers are obtained by using a quantum information theory technique called \textit{twirling} \cite{twirl1,twirl2,twirl3,twirl4,twirl5,twirl6}. Additionally, we justify why QECCs designed for the twirled channels will also be valid for the original channel, and why those approximated channels are indeed tractable as classical problems. This means that the quantum computing community has access to the necessary design tools working with algorithms that rely only on the classical machines that are available in this day and age.

\subsubsection{The Gottesmann-Knill theorem and the Pauli channel}\label{sub:gottKnill}
The decoherence model presented in section \ref{sub:ampphasdamp} in the form of the combined amplitude and phase damping channel cannot be efficiently implemented on a classical computer when multiqubit systems are considered. The quandary then becomes which quantum systems can be efficiently simulated using classical methods, if any at all.

Such a question is answered by the Gottesman-Knill theorem for stabilizer circuits\footnote{We will use the term ``stabilizer formalism" to refer to the quantum circuits that are constructed according to the restrictions that are described in the Gottesman-Knill theorem.} \cite{NielsenChuang,GotKnill}. The theorem states the following \cite{NielsenChuang}:

\begin{theorem}[Gottesman-Knill Theorem]\label{thm:gotknill}
\textit{Suppose a quantum computation is performed using only the following elements: state preparations in the computational basis, Hadamard gates, phase gates, controlled-NOT gates, Pauli gates, and measurements of observables in the Pauli group (which includes measurement in the computational basis as a special case), together with the possibility of classical control conditioned on the outcomes of such measurements. Such computation can be efficiently simulated on a classical computer.}
\end{theorem}

This theorem shows that certain quantum computations that involve very complex and highly entangled states are actually tractable problems for classical computers. The stabilizer formalism does not describe all possible quantum computations, but it does so for a sizeable quantity of these operations.

As a consequence of Theorem \ref{thm:gotknill}, it is desirable to operate with a quantum channel whose dynamics are described by the stabilizer formalism, so that the error model of the system can be efficiently included in the classical simulation. In quantum information theory, the \textit{Pauli channel}, $\mathcal{N}_{\mathrm{P}}$, is the CPTP map that transforms the quantum state with density matrix $\rho$ as
\begin{equation}\label{eq:Pauli}
\begin{split}
\mathcal{N}_\mathrm{P} (\rho) =& (1 - p_\mathrm{x} - p_\mathrm{y} - p_\mathrm{z})\mathrm{I}\rho\mathrm{I} + p_\mathrm{x}\mathrm{X}\rho\mathrm{X}+p_\mathrm{y}\mathrm{Y}\rho\mathrm{Y} \\
& +p_\mathrm{z}\mathrm{Z}\rho\mathrm{Z},
\end{split}
\end{equation}
where the constants $\{p_\mathrm{x},p_\mathrm{y},p_\mathrm{z}\}$ are interpreted as the probabilities of each specific operator impinging on the state $\rho$. The Kraus operators of this channel are $\sqrt{p_\mathrm{A}}\mathrm{A}$, with $\mathrm{A}\in\mathcal{P}_1$.

The operation of each of the error operators on a qubit $|\psi\rangle$ is given by the operation of the Pauli matrices described before.

Following \eqref{Pauli}, we see that $\mathcal{N}_\mathrm{P}$ maps qubits $|\psi\rangle$ onto a linear combination of the original state ($\mathrm{I}$), the phase-flipped state ($\mathrm{Z}$), the bit-flipped state ($\mathrm{X}$) and the phase-and-bit-flipped state ($\mathrm{Y}$), where the sum is weighted by the probabilities $p_\mathrm{A},\mathrm{A}\in\mathcal{P}_1$. The symmetric instance of this channel\footnote{The scenario where $p_\mathrm{x}=p_\mathrm{y}=p_\mathrm{z}=p/3$ and $p_\mathrm{I} = 1- p$.} is a widely used model for decoherence in quantum information theory \cite{QRM,bicycle,qldpc15,jgf,patrick,QCC,QTC,EAQTC,josu1,josu2,toric,QEClidar} and it is known as the \textit{depolarizing channel} $\mathcal{N}_\mathrm{D}$. The expression \eqref{Pauli} then simplifies to
\begin{equation}\label{eq:depolarizing}
\mathcal{N}_\mathrm{D} (\rho) = (1-p)\rho + \frac{p}{3}\left(\mathrm{X}\rho\mathrm{X}+\mathrm{Y}\rho\mathrm{Y}+\mathrm{Z}\rho\mathrm{Z}\right),
\end{equation}
and $p$ receives the name of depolarizing probability. This symmetric instance of the Pauli channel represents the worst case scenario for the family of channels, as all the possible errors are equally likely to happen \cite{isomorphism}.

The generalization of the Pauli channel from \eqref{Pauli} to an $n$-qubit system is mathematically described by
\begin{equation}\label{eq:NPauli}
\mathcal{N}_\mathrm{P}^{(n)}(\rho) = \sum_{\mathrm{A}\in\mathcal{P}_n} p_\mathrm{A} \mathrm{A}\rho\mathrm{A},
\end{equation}
where $p_\mathrm{A}$ is the probability for each of the $n$-fold Pauli operators, $\mathrm{A}\in\mathcal{P}_n$, to occur.

Observing \eqref{Pauli}, \eqref{depolarizing} and \eqref{NPauli}, it is clear that the dynamics of these channels are described by only using Pauli operators. By the Gottesman-Knill theorem (theorem \ref{thm:gotknill}), these quantum channel instances can be efficiently simulated on classical computers, and consequently, it is worthwhile to relate the decoherence processes described in the previous sections, and mathematically modeled via the combined amplitude and phase damping channel, to the transformations described by the Pauli channel. It is clear that the Pauli channel is incapable of capturing the exact manner in which the combined amplitude and damping channel corrupts input states, as there is no combination of $p_\mathrm{x}$, $p_\mathrm{y}$ and $p_\mathrm{z}$ that makes $\mathcal{N}_\mathrm{P}(\rho) = \mathcal{N}_{\mathrm{APD}}(\rho)$. However, we can use techniques of quantum information theory to engender approximations of $\mathcal{N}_{\mathrm{APD}}$ that are in the form of a Pauli channel.

\subsubsection{Approximating quantum channels with Twirling}\label{sub:twirling}
In the prior section, we discussed the fact that quantum operations require an exponential number of parameters to completely describe them. Naturally, this leads to the conclusion that the simulation of such quantum constructs using classical resources cannot be efficiently performed. However, tools aiming at overcoming the issue of exponential parameter growth\footnote{This
issue is not limited to the classical simulation of quantum dynamics. For example, quantum metrology requires an exponential number of experiments (e.g. Nuclear Magnetic Resonances (NMR)) to determine the physical parameters of a quantum system.}, by successfully extracting useful partial information of the quantum dynamics, have been developed in the literature. This information extraction can be performed using a technique called \textit{twirling}. Twirling is an extensively used method in quantum information theory to study the average effect of general quantum noise models via their mapping to more symmetric versions of themselves \cite{twirl1,twirl2,twirl3,twirl4,twirl5,twirl6}. Generally, twirling a quantum channel $\mathcal{N}$ comes down to averaging the composition $\mathcal{U^\dagger\circ N\circ U}$ for unitary operators $\mathcal{U}(\rho) = \mathrm{U}\rho \mathrm{U}^\dagger$ that are randomly chosen according to some probability measure $\mu( \mathrm{U})$ \cite{twirl2,twirl5}. The resulting channel
\begin{equation}\label{eq:twirledchannelGeneral}
\begin{split}
\bar{\mathcal{N}}^{\mathcal{U}}(\rho) &= \int_{\mathcal{U}}d\mu(\mathrm{U})\mathcal{U^\dagger\circ N\circ U}(\rho) \\
& =\int_{\mathcal{U}}d\mu(\mathrm{U}) \mathrm{U}^\dagger\mathcal{N}(\mathrm{U} \rho\mathrm{U}^\dagger)\mathrm{U},
\end{split}
\end{equation}
receives the name of \textit{twirled channel}. A particularly important scenario, and the one we will consider in this thesis, is the case where the distribution over the unitaries $\mu(\mathrm{U})$ is discrete. For such a case, the dynamics of the twirled channel are given by
\begin{equation}\label{eq:twirledchannelDiscrete}
\begin{split}
\bar{\mathcal{N}}^{\mathcal{U}}(\rho) &= \sum_l \mu(\mathrm{U}_l)\mathcal{U}_l^\dagger\circ \mathcal{N}\circ \mathcal{U}_l(\rho) \\ & = \sum_{l} \mu(\mathrm{U}_l) \mathrm{U}_l^\dagger\mathcal{N}(\mathrm{U}_l \rho\mathrm{U}_l^\dagger)\mathrm{U}_l,
\end{split}
\end{equation}
where $\mu(\mathrm{U}_l)$ is a probability distribution over $\mathcal{U}_l$.

Twirling a channel over some unitary is useful to analyze the average effects that said quantum channel produces. However, the serviceability of the information that the twirled channel provides is not clear, nor is the concept of how an approximated more symmetric channel constructed via twirling can be employed to design QECC schemes capable of fighting realistic decoherence. This conundrum is answered by the following lemma \cite{twirl2}, which introduces a fundamental way in which the approximations of decoherence models can be realised to construct QECCs with the twirling method presented in this thesis.
\begin{lemma}\label{lem:QECCtwirl}
\textit{Any correctable code for the twirled channel $\bar{\mathcal{N}}$ is a correctable code for the original channel $\mathcal{N}$ up to an additional unitary correction.}
\end{lemma}

The proof of lemma \ref{lem:QECCtwirl} is given in \cite{twirl2}. The importance of the lemma resides in the idea that designing QECCs that correct errors for an approximated channel obtained by twirling will also correct errors of the original channel up to unitary correction. This allows us to design error correction codes for the twirled channel, which, due to the lemma, will still be valid for the original channel. The analysis conducted for the approximated channels will not be precise due to the fact that parts of the evolutions of the density matrices will be lost when operating with the twirled channels instead of the original ones. Nevertheless, work in the literature has shown that the estimations obtained using the twirling methods presented here are accurate \cite{honest}.

\paragraph*{Pauli Twirl Approximation (PTA)}\label{subsub:PTA}
An extensively used type of twirl is the Pauli twirl. This twirl consists in averaging the original channel $\mathcal{N}$ with the elements of the $n$-fold set of Pauli operators $\mathcal{P}_n$ weighted in an equiprobable manner. Since the cardinality of $\mathcal{P}_n$ is $4^n$, the Pauli twirled channel deduced from \eqref{twirledchannelDiscrete} will be \cite{twirl2}

\begin{equation}\label{eq:PauliTwirlGen}
\bar{\mathcal{N}}^{\mathcal{P}_n}(\rho) = \frac{1}{4^n}\sum_{l=1}^{4^n}\mathrm{P}_l\mathcal{N}\left(\mathrm{P}_l \rho \mathrm{P}_l\right)\mathrm{P}_l
=\frac{1}{4^n}\sum_{l=1}^{4^n}\mathrm{P}_l\left(\sum_{k} E_k\mathrm{P}_l\rho \mathrm{P}_lE^\dag_k\right)\mathrm{P}_l,
\end{equation}
where the $E_k$ are the Choi-Kraus error operators of the channel under study and $\mathrm{P}_l\in\mathcal{P}_n$ are the $n$-fold Pauli operators.

Next we show that the above expression reduces to 
\begin{equation}
\bar{\mathcal{N}}^{\mathcal{P}_n}(\rho) = \sum_{l=1}^{4^n} \chi_{l} \mathrm{P}_l \rho \mathrm{P}_l,
\end{equation}
where the coefficients of the expansion are given by \cite{twirl3}
\begin{equation}\label{eq:PaulitwirlCoeff}
\chi_{l} = \sum_{k} \frac{|\alpha_{l}^{(k)}|^2}{2^n}, \;\;\mbox{with } \alpha_{l}^{(k)} = \frac{\mathrm{Tr}[E_k \mathrm{P}_l]}{2^{n/2}}.
\end{equation}

To that end, we expand the Choi-Kraus error operators $E_k$ using the Pauli basis, that is, 
\begin{equation}\label{eq:errorOpsPaulibasis}
E_k = \sum_{l=1}^{4^n} \alpha_{l}^{(k)} \frac{\mathrm{P}_l}{2^{n/2}}.
\end{equation}
By the orthogonality of the of the $\{\mathrm{P}_l\}$ under the trace, i.e., $\mathrm{Tr} (\mathrm{P}_l\mathrm{P}_k)=0$, if $l\neq k$ and $2^n$, otherwise, the coefficients $\alpha_{l}^{(k)}$ can be calculated as
\begin{equation}\label{eq:alphacoeffs}
\alpha_{l}^{(k)} = \frac{\mathrm{Tr}[E_k \mathrm{P}_l]}{2^{n/2}}.
\end{equation}
Substituting, \eqref{errorOpsPaulibasis} into \eqref{PauliTwirlGen} yields,

\begin{equation}\label{eq:pedro1}
\bar{\mathcal{N}}^{\mathcal{P}_n}(\rho) = \frac{1}{4^n}\sum_{l=1}^{4^n}\sum_{j=1}^{4^n} \sum_{i=1}^{4^n}\left(\sum_k\frac{\alpha_{j}^{(k)}\alpha_{i}^{(k)*}}{2^{n}}\right) \mathrm{P}_l\mathrm{P}_j\mathrm{P}_l\rho \mathrm{P}_l\mathrm{P}_i\mathrm{P}_l,
\end{equation}
where the fact that elements of the $n$-fold Pauli operators are hermitian has been used. Note that
\[\mathrm{P}_l\mathrm{P}_j\mathrm{P}_l = \left\{ \begin{array}{ll}
         \mathrm{P}_j & \mbox{if $l=j$};\\
        \pm \mathrm{P}_j & \mbox{if $l\neq j$}.\end{array} \right. \] 
        Furthermore, when $i\neq j$ and $l$ sweeps from 1 to $4^n$, half of the terms $\mathrm{P}_l\mathrm{P}_j\mathrm{P}_l$ are $\mathrm{P}_j$ and half $-\mathrm{P}_j$. Therefore, by denoting $\Upsilon_{j,i}^{+}$ and $\Upsilon_{j,i}^{-}$ the sets of Pauli operators that
make $\mathrm{P}_l\mathrm{P}_j\mathrm{P}_l\rho \mathrm{P}_l\mathrm{P}_i\mathrm{P}_l=\mathrm{P}_j\rho \mathrm{P}_i$ and $\mathrm{P}_l\mathrm{P}_j\mathrm{P}_l\rho \mathrm{P}_l\mathrm{P}_i\mathrm{P}_l=-\mathrm{P}_j\rho \mathrm{P}_i$, respectively, expression \eqref{pedro1} can be splited as
\begin{equation}\label{eq:twirledchannelDiscretePauli}
\begin{split}
 \bar{\mathcal{N}}^{\mathcal{P}_n}(\rho) = \frac{2^n}{4^n}\left[\sum_{j,i:\Upsilon_{j,i}^{+}, i\neq j}\left(\sum_k\frac{\alpha_{j}^{(k)}\alpha_{i}^{(k)*}}{2^{n}}\right)\mathrm{P}_j\rho \mathrm{P}_i \right. & \left.- \sum_{j,i:\Upsilon_{j,i}^{-} , i\neq j}\left(\sum_k \frac{\alpha_{j}^{(k)}\alpha_{i}^{(k)*}}{2^{n}}\right)\mathrm{P}_j\rho \mathrm{P}_i\right] \\+\frac{4^n}{4^n}\sum_{j=1}^{4^n}\left(\frac{\sum_k|\alpha_{j}^{(k)}|^2}{2^{n}}\right)\mathrm{P}_j\rho \mathrm{P}_j=\sum_{j=1}^{4^n}\chi_j\mathrm{P}_j\rho \mathrm{P}_j.
\end{split}
\end{equation}
Note that as $ \mathrm{Tr}(\bar{\mathcal{N}}^{\mathcal{P}_n}(\rho))=1$ so that the outcome of the channel is a valid density matrix, then the $\chi_l$ elements must sum to one, i.e., $\sum_{l=1}^{4^n} \chi_l=1$. This can be easily checked by using the completion of the Choi-Kraus error operators, i.e., $\sum_k E^\dag_kE_k=I$ and the expansion in \eqref{errorOpsPaulibasis}.

In other words, the off-diagonal elements\footnote{By off-diagonal elements we refer to the cross products of Pauli matrices $\mathrm{P}_j\rho \mathrm{P}_i$, with $j\neq i$, that arise when writing the Choi-Kraus operators in the Pauli basis as in \eqref{errorOpsPaulibasis}.} of the original channel are eliminated \cite{twirl2}. For this reason, applying the Pauli twirl to a general quantum channel $\mathcal{N}$ will result in a twirled channel with the structure of the Pauli channel. The probabilities for each of the elements of the $n$-qubit Pauli operators can be calculated using \eqref{PaulitwirlCoeff} and \eqref{alphacoeffs}. Additionally, we know from the Gottesman-Knill theorem that quantum CPTP maps with such a structure can be simulated efficiently on classical computers. Lemma \ref{lem:QECCtwirl} implies that we can design QECCs using this twirled approximation in order to construct QECCs for the original channel $\mathcal{N}$, and that such error correction methods will maintain their error correcting capabilities over the original channel. As a result, we will be able to design QECC families for realistic quantum computers by designing and simulating them on classical computers based on approximating the channels by Pauli twirling.

We now proceed with the application of the Pauli twirl to the model of decoherence presented in section \ref{sub:ampphasdamp} in the form of the combined amplitude and phase damping channel $\mathcal{N}_{\mathrm{APD}}$. For this purpose, we first write the transformation of the density matrix induced by the channel presented in \eqref{combinedDensityMatrix}, as a function of the Pauli matrices
\begin{equation}\label{eq:combinedDensityPauliFcn}
\begin{split}
\mathcal{N}_{\mathrm{APD}}(\rho) =& \frac{2 - \gamma + 2\sqrt{1- \gamma -(1-\gamma)\lambda}}{4} \mathrm{I}\rho\mathrm{I} \\
& + \frac{\gamma}{4}\mathrm{X}\rho\mathrm{X} + \frac{\gamma}{4}\mathrm{Y}\rho\mathrm{Y} \\
& +\frac{2- \gamma - 2\sqrt{1-\gamma - (1-\gamma)\lambda}}{4}\mathrm{Z}\rho\mathrm{Z} \\
& - \frac{\gamma}{4}\mathrm{I}\rho\mathrm{Z} - \frac{\gamma}{4}\mathrm{Z}\rho\mathrm{I} + \frac{\gamma}{4i}\mathrm{X}\rho\mathrm{Y}-\frac{\gamma}{4i}\mathrm{Y}\rho\mathrm{X},
\end{split}
\end{equation}
where \eqref{ampphaseKraus} is used to express the Kraus operators in the Pauli basis. As stated in \eqref{PauliTwirlGen}, twirling the combined amplitude and phase damping channel by the Pauli operators $\mathcal{P}_1$ will remove the off-diagonal elements of \eqref{combinedDensityPauliFcn}, resulting in the twirled channel
\begin{equation}\label{eq:PTAdensity}
\begin{split}
\bar{\mathcal{N}}^{\mathcal{P}_1}_{\mathrm{APD}}(\rho) =
& \frac{2 - \gamma + 2\sqrt{1- \gamma -(1-\gamma)\lambda}}{4} \mathrm{I}\rho\mathrm{I} \\
& + \frac{\gamma}{4}\mathrm{X}\rho\mathrm{X} + \frac{\gamma}{4}\mathrm{Y}\rho\mathrm{Y} \\
& +\frac{2- \gamma - 2\sqrt{1-\gamma - (1-\gamma)\lambda}}{4}\mathrm{Z}\rho\mathrm{Z}.
\end{split}
\end{equation}

It is obvious that the channel shown in \eqref{PTAdensity} exhibits the form of a Pauli channel \eqref{Pauli}. The probabilities for each of the possible Pauli errors in \eqref{Pauli} are
\begin{equation}\label{eq:PTAprobs}
\begin{split}
& p_\mathrm{x} = p_\mathrm{y} = \frac{\gamma}{4} = \frac{1}{4}\left(1 - e^{\frac{-t}{T_1}}\right) \\
& p_\mathrm{z} = \frac{2- \gamma - 2\sqrt{1-\gamma - (1-\gamma)\lambda}}{4}  \\ & \quad = \frac{1}{4}\left(1 + e ^{\frac{-t}{T_1}}-2e^{\frac{-t}{T_2}}\right),
\end{split}
\end{equation}
where we used the time dependencies of the damping and scattering probabilities obtained from the solution to the master equation \cite{twirl1,twirl6}. This asymmetric channel takes the form of the depolarizing channel \eqref{depolarizing} when the relaxation time and the dephasing time are the same $T_1= T_2$.

Consequently, twirling a general quantum channel by the set of $n$-fold Pauli operators $\mathcal{P}_n$ results in an approximated channel with the same form as a Pauli channel that will be efficiently implementable on a classical machine thanks to the Gottesman-Knill theorem (theorem \ref{thm:gotknill}). We will use the term \textit{Pauli Twirl Approximation} (PTA) \cite{twirl1} to refer to the family of channels that are obtained using this symmetrization process. The process is simple, as it is based on rewriting the Kraus operators of the original channel in the Pauli basis and then eliminating the off-diagonal terms of the density matrix evolution equation.

When performing a Pauli twirling to the combined amplitude and phase damping channel, the resulting probabilities $p_\mathrm{x}$ and $p_\mathrm{y}$ for the bit-flip ($\mathrm{X}$) and the phase-and-bit-flip ($\mathrm{Y}$), respectively, turn out to be equal, reducing the degrees of freedom of the PTA to two. By defining the \textit{asymmetry coefficient} \cite{EAQIRCC,twirl6}
\begin{equation}\label{eq:asymmetryCoeff}
\alpha = \frac{p_\mathrm{z}}{p_\mathrm{x}} = 1 + 2 \frac{1- e^{\frac{t}{T_1}\left(1-\frac{T_1}{T_2}\right)}}{e^{\frac{t}{T_1}}-1}.
\end{equation}
one can rewrite the expression for the Pauli channel \eqref{Pauli} as
\begin{equation}\label{eq:PauliAlpha}
\begin{split}
\mathcal{N}_{\mathrm{P}}(\rho) =& (1-p)\mathrm{I}\rho\mathrm{I} + \frac{p}{\alpha + 2}\mathrm{X}\rho\mathrm{X} \\
&  + \frac{p}{\alpha + 2}\mathrm{Y}\rho\mathrm{Y} + \frac{\alpha p}{\alpha + 2}\mathrm{Z}\rho\mathrm{Z},
\end{split}
\end{equation}
where $p=p_\mathrm{x}+p_\mathrm{y}+p_\mathrm{z}$. Note that $\alpha$ in \eqref{asymmetryCoeff} is a time-dependent parameter. However, if the coherent time duration $t$ of the task is assumed to be negligible compared to the relaxation time $T_1$, it can be approximated as
\begin{equation}\label{eq:approximAlpha}
\alpha \xrightarrow{t<<T_1} \alpha\approx 2\frac{T_1}{T_2} - 1.
\end{equation}

Based on this approximation, the data from tables \ref{tab:qdevicesT1} and \ref{tab:qdevicesT2} can be used to determine the degree of asymmetry of existing quantum devices under the PTA.
\begin{table}[h!]
\centering
\begin{tabular}{|l|l|}
\hline
\textbf{Name}         & $\mathbf{\alpha}$ \\ \hline
IBM Q System One \cite{IBM}    & $\approx 1$  \\ \hline
Rigetti 32Q Aspen-7 \cite{Rigetti}      & $\approx 1$               \\ \hline
Google Sycamore  \cite{EAQIRCC},\cite{Google}     & $\approx 10$         \\ \hline
Ion Q 11 Qubit \cite{IonQ}       & $\approx 10^4$             \\ \hline
Intel Q (Spin Qubits) \cite{EAQIRCC},\cite{Intel} & $\approx 10^6$          \\ \hline
\end{tabular}
\caption{Asymmetry coefficients, $\alpha$, for some existing quantum devices. Different companies use different technologies to implement the qubits in their quantum machines.}
\label{tab:qdevicesAlphas}
\end{table}
This is shown in Table \ref{tab:qdevicesAlphas}. Note that PTAs obtained for the different quantum technologies give rise to large variations on the asymmetry parameter, ranging from the depolarizing channel ($\alpha \approx 1$) to channels presenting very strong asymmetry levels ($\alpha\approx 10^6$).

\paragraph*{Clifford Twirl Approximation (CTA)}\label{subsub:CTA}
Another twirling operation used for the purpose of channel approximation is known as Clifford twirling. For this twirl, an arbitrary channel $\mathcal{N}$ is averaged with the elements of the $n$-fold Clifford group $\mathcal{C}_1^{\otimes n}$ \cite{OzolsClifford} weighted equiprobably, as in the Pauli twirl. Using the discrete twirling case as in \eqref{twirledchannelDiscrete}, the Clifford twirled channel will be \cite{twirl3}
\begin{equation}\label{eq:cliffTwirlGen}
\begin{split}
\bar{\mathcal{N}}^{\mathcal{C}_1^{\otimes n}}(\rho) &= \frac{1}{|\mathcal{C}_1^{\otimes n}|}\sum_{\mathrm{C}_l\in\mathcal{C}_1^{\otimes n}} \mathrm{C}_l^\dagger\mathcal{N}\left(\mathrm{C}_l\rho\mathrm{C}_l^\dagger\right)\mathrm{C}_l \\
& = \frac{1}{|\mathcal{C}_1^{\otimes n}|}\sum_{\mathrm{C}_l\in\mathcal{C}_1^{\otimes n}} \sum_k \mathrm{C}_l^\dagger E_k\mathrm{C}_l\rho\mathrm{C}_l^\dagger E_k^\dagger\mathrm{C}_l,
\end{split}
\end{equation}
where as before, the Choi-Kraus operator-sum representation of $\mathcal{N}$ has been used, and $|\cdot|$ refers to the cardinality of a set.

Expression \eqref{cliffTwirlGen} can be further expanded by noting that the $n$-fold Clifford group $\mathcal{C}_1^{\otimes n}$ is the semi-direct product of the $n$-fold Pauli operators, $\mathcal{P}_n$, and the symplectic group, $\mathcal{S}_1^{\otimes n}$, of dimension $n$ \cite{CliffIsPauliSymp,SympGroup}. Thus, the elements $\mathrm{C}_l\in\mathcal{C}_1^{\otimes n}$ can be written as $\mathrm{C}_l = \mathrm{P}_j \mathrm{S}_m$ with $\mathrm{P}_j\in\mathcal{P}_n$ and $\mathrm{S}_m\in\mathcal{S}_1^{\otimes n}$ \cite{twirl3,CliffIsPauliSymp}. This way, the Clifford twirl can be written as
\begin{equation}\label{eq:cliffTwirlPauliSymp}
\begin{split}
&\bar{\mathcal{N}}^{\mathcal{C}_1^{\otimes n}} (\rho)  \\ &= \frac{1}{|\mathcal{C}_1^{\otimes n}|} \sum_{\mathrm{S}_m} \sum_{j=1}^{4^n} \sum_k \mathrm{S}_m^\dagger \mathrm{P}_j E_k \mathrm{P}_j \mathrm{S}_m \rho \mathrm{S}_m^\dagger \mathrm{P}_j E_k^\dagger \mathrm{P}_j \mathrm{S}_m,
\end{split}
\end{equation}
with $\mathrm{S}_m\in\mathcal{S}_1^{\otimes n}$ and $\mathrm{P}_j\in\mathcal{P}_n$. The cardinality of the Clifford group is $|\mathcal{C}_1^{\otimes n}| = |\mathcal{P}_n||\mathcal{S}_1^{\otimes n}|$.

By observing expression \eqref{cliffTwirlPauliSymp}, it can be concluded that Clifford twirling a quantum channel $\mathcal{N}$ is equivalent to first Pauli twirling the channel and then performing a symplectic twirl \cite{twirl3}. That is
\begin{equation}\label{eq:twostepCliff}
\mathcal{N} \xrightarrow{\mathcal{C}_1^{\otimes n}} \bar{\mathcal{N}}^{\mathcal{C}_1^{\otimes n}} \equiv \mathcal{N}\xrightarrow{\mathcal{P}_n}
\bar{\mathcal{N}}^{\mathcal{P}_n}\xrightarrow{\mathcal{S}_1^{\otimes n}}\bar{\mathcal{N}}^{\mathcal{C}_1^{\otimes n}}.
\end{equation}

The action of each step can then be summarized as \cite{twirl3}:
\begin{enumerate}
\item $\mathcal{P}_n$-twirl: the effect of the Pauli twirl has been extensively studied in the last section. Twirling an arbitrary quantum channel $\mathcal{N}$ by $\mathcal{P}_n$ will lead to a Pauli channel as in \eqref{PauliTwirlGen}.
\item $\mathcal{S}_1^{\otimes n}$-twirl: the effect of the symplectic twirl on the PTA obtained in step 1 is the mapping of each of the non-identity Pauli operators to a uniform sum over the 3 non-identity Pauli operators. The application of the symplectic twirl to the PTA channel obtained in step one is:
\begin{equation}\label{eq:PTAsymp}
\bar{\mathcal{N}}^{\mathcal{C}_1^{\otimes n}}(\rho) = \frac{1}{|\mathcal{S}_1^{\otimes n}|}\sum_{\mathrm{S}_m} \sum_{j} \chi_{j}\mathrm{S}_m^\dagger \mathrm{P}_j \mathrm{S}_m \rho \mathrm{S}_m^\dagger \mathrm{P}_j \mathrm{S}_m,
\end{equation}
where the coefficients $\chi_{j}$ are the PTA coefficients. In order to analyze the effect of the symplectic twirl, we rewrite \eqref{PTAsymp} as \cite{twirl3}
\begin{equation}\label{eq:PTAsympRewrite}
\begin{split}
\bar{\mathcal{N}}^{\mathcal{C}_1^{\otimes n}}(\rho)
=& \frac{1}{|\mathcal{S}_1^{\otimes n}|} \sum_{\mathrm{S}_m}\sum_{\omega=0}^n \sum_{\nu_\omega}^{\binom{n}{\omega}} \sum_{\mathbf{i}_\omega = 1}^{3^\omega} \chi_{\omega,\nu_\omega,\mathbf{i}_\omega}\mathrm{S}_m^\dagger \\ &\mathrm{P}_{\omega,\nu_\omega,\mathbf{i}_\omega} \mathrm{S}_m
\rho\mathrm{S}_m^\dagger \mathrm{P}_{\omega,\nu_\omega,\mathbf{i}_\omega}\mathrm{S}_m,
\end{split}
\end{equation}
where the Pauli operators have been split based on indexes $\omega$ which gives the weight\footnote{Number of non-identity factors of the tensor product.} of a Pauli operator $\mathrm{P}_j$, $\nu_\omega$ which counts the number of distinct ways that $\omega$ non-identity Pauli operators can be placed in the corresponding $n$-tensor product, and $\mathbf{i}_\omega=\{i_1, \cdots, i_\omega \}$ with $i_\beta = \{1,2,3\}$ denotes which of the three non-identity Pauli operators occupies the position $\beta^{\mathrm{th}}$ inside the $n$-tensor product. Now fixing $\omega$ and $\nu_\omega$, the action of the symplectic twirl for each term is
\begin{equation}\label{eq:sympEachTerm}
\begin{split}
&\frac{1}{|\mathcal{S}_1^{\otimes n}|}\sum_{\mathbf{i}_\omega}^{3^\omega} \chi_{\omega,\nu_\omega,\mathbf{i}_\omega} \sum_{\mathrm{S}_m}\mathrm{S}_m^\dagger \mathrm{P}_{\omega,\nu_\omega,\mathbf{i}_\omega} \mathrm{S}_m
\rho\mathrm{S}_m^\dagger \mathrm{P}_{\omega,\nu_\omega,\mathbf{i}_\omega}\mathrm{S}_m \\
& = \left( \frac{1}{3^\omega}\sum_{\mathbf{i}_\omega}^{3^\omega}\chi_{\omega,\nu_\omega,\mathbf{i}_\omega}\right) \sum_{\mathbf{i}_\omega}^{3^\omega} \mathrm{P}_{\omega,\nu_\omega,\mathbf{i}_\omega}\rho\mathrm{P}_{\omega,\nu_\omega,\mathbf{i}_\omega} \\
& = \xi_{\omega, \nu_\omega}\sum_{\mathbf{i}_\omega}^{3^\omega} \mathrm{P}_{\omega,\nu_\omega,\mathbf{i}_\omega}\rho\mathrm{P}_{\omega,\nu_\omega,\mathbf{i}_\omega},
\end{split}
\end{equation}
so one can observe that the symplectic twirl maps each of the non-identity Pauli operators to a uniform sum over the $3$ non-identity Pauli operators, as stated earlier. For this reason, the expression for the Clifford twirled channel in \eqref{PTAsympRewrite} is
\begin{equation}\label{eq:finalClifftwirled}
\bar{\mathcal{N}}^{\mathcal{C}_1^{\otimes n}}(\rho) = \sum_{\omega = 1}^n \sum_{\nu_\omega=1}^{\binom{n}{\omega}} \xi_{\omega,\nu_\omega}\sum_{\mathbf{i}_\omega}^{3^\omega} \mathrm{P}_{\omega,\nu_\omega,\mathbf{i}_\omega}\rho\mathrm{P}_{\omega,\nu_\omega,\mathbf{i}_\omega},
\end{equation}
with coefficients $\xi_{\omega,\nu_\omega}=\frac{1}{3^\omega}\sum_{\mathbf{i}_\omega}^{3^\omega} \chi_{\omega,\nu_\omega,\mathbf{i}_\omega}$.
\end{enumerate}

From equation, \eqref{finalClifftwirled} it is easy to deduce that by applying the Clifford twirl to an arbitrary quantum channel $\mathcal{N}$, a Pauli channel will be obtained. However, this case differs from the PTA case that an additional degree of symmetrization is performed, since the probabilities of Pauli operators that share the same weight and spatial position regarding non-identity operators will be uniformly summed over the $3$ non-identity Pauli operators. This symmetrization process results on a quantum channel with fewer parameters, as the probabilities of the error operators that share the weight and spatial distribution of non-identity factors will now have the same probabilities of occurring. In general, this was not true for the PTA channels. Note that as in the PTA case, the Clifford twirled channels will fulfill the Gottesman-Knill theorem, and so it will be possible to efficiently implement them on classical computers.

We now apply the Clifford twirl to the combined amplitude and phase damping channel, $\mathcal{N}_{\mathrm{APD}}$. As done for the PTA channels, we use the expansion of the channel in the Pauli basis (refer to \eqref{combinedDensityPauliFcn}), and then apply the two previous steps to this expansion:
\begin{enumerate}
\item $\mathcal{P}_n$-twirl: the application of the Pauli twirl has been addressed in section \ref{subsub:PTA}. The obtained $\mathcal{N}_{\mathrm{APD}}$ is the Pauli channel given in \eqref{PTAdensity}.
\item $\mathcal{S}_1^{\otimes n}$-twirl: from the previous discussion, the symplectic twirl maps each of the non-identity Pauli operators to the uniform sum over the three non-identity Pauli operators. In consequence, and following \eqref{finalClifftwirled}, the Pauli channel of \eqref{PTAdensity} results in
\begin{equation}\label{eq:CTAdensity}
\begin{split}
\bar{\mathcal{N}}_{\mathrm{APD}}^{\mathcal{C}_1^{\otimes n}}(\rho)=&\frac{2-\gamma + 2\sqrt{1-\gamma-(1-\gamma)\lambda}}{4}\mathrm{I}\rho\mathrm{I} + \\ & \frac{2+\gamma - 2\sqrt{1-\gamma - (1-\gamma)\lambda}}{12}( \mathrm{X}\rho\mathrm{X} \\ &+ \mathrm{Y}\rho\mathrm{Y}+ \mathrm{Z}\rho\mathrm{Z}),
\end{split}
\end{equation}
where $p = \frac{2+\gamma - 2\sqrt{1-\gamma - (1-\gamma)\lambda}}{4}$ is the depolarizing probability given by $p=p_\mathrm{x}+p_\mathrm{y}+p_\mathrm{z}$, where $p_\mathrm{x},p_\mathrm{y},p_\mathrm{z}$ are given in \eqref{PTAprobs}.
\end{enumerate}

The channel \eqref{CTAdensity} has the form of a symmetric Pauli channel. In other words, it has the structure of the so-called depolarizing channel in \eqref{depolarizing}. Moreover, by using the temporal expressions of the probabilities of the PTA in \eqref{PTAprobs}, the depolarizing probability of the Clifford twirled channel is related to the physical parameters of a quantum device as
\begin{equation}\label{eq:CTAprob}
\begin{split}
p =& \frac{2+\gamma - 2\sqrt{1-\gamma - (1-\gamma)\lambda}}{4}\\
 = & \frac{3}{4} - \frac{1}{4}e^{\frac{-t}{T_1}}-\frac{1}{2}e^{\frac{-t}{T_2}},
\end{split}
\end{equation}
where $T_1$ and $T_2$ are the relaxation and dephasing times, respectively.

Consequently, twirling an arbitrary quantum noise channel $\mathcal{N}$ by the $n$-fold Clifford group $\mathcal{C}_1^{\otimes}$ will result in an approximated channel having the form of a Pauli channel whose principal characteristic is that the error operators that have the same weight and non-identity operator distribution in the tensor product will have the same probability of occurring. For the one qubit case, this implies that the resulting channel is a depolarizing channel. We will use the term \textit{Clifford Twirl Approximation} (CTA) to refer to the family of channels obtained by this symmetrization method.

As already mentioned, the most important aspect of the PTA and CTA channels is the fact that they are efficiently implementable on classical computers, since they fulfill the  Gottesman-Knill theorem (theorem \ref{thm:gotknill}). This is especially interesting in the field of QECC design, given that from lemma \ref{lem:QECCtwirl} we will be able to design these codes without the need for actual quantum computers.

\subsection{Quantum memoryless channels and channels with memory}\label{sec:memory}
In the previous section, we presented a symmetrization method called twirling that can be used to approximate general quantum channels to PTA and CTA Pauli channels that can be efficiently implemented on classical computers. Generally speaking, PTA and CTA channels act collectively on $n$ qubits, see \eqref{PTAdensity} and \eqref{CTAdensity}, and they are both reduced to $n$-qubit Pauli channels with different degrees of symmetrization. More degrees of symmetrization can be obtained for the $n$-qubit twirled channels with extra twirling operations such as permutation twirls \cite{twirl2,twirl3}. These twirled channels are also efficiently implementable, but in order to obtain the probability distributions that define the specific Pauli channel that comes out as a PTA or CTA, the original channel $\mathcal{N}$ must be modeled. The interaction processes between the elements of an $n$-qubit system can be quite subtle, and, as stated in section \ref{sub:decWhy}, it is reasonable to assume that, if certain conditions apply, each of the qubits of the system will interact with the environment in an independent and similar way. Additionally, simplified memory effects can be employed to express the interactions that each of the qubits has with its counterparts. Consequently, we will be able to use the expressions obtained for the PTA and CTA of the combined amplitude and phase damping channel for $1$ qubit in order to simulate decoherence over complex $n$-qubit quantum states.

Most of the research related to quantum channels considers a scenario in which the corruption of the input quantum states occur independently and identically \cite{QRM,bicycle,qldpc15,jgf,QCC,QTC,EAQTC,toric,QEClidar,EAQIRCC,twirl6}. If two different quantum states are corrupted by the same channel, the noise applied to the first transmitted state is assumed to be independent of the noise applied to the second state. This can be seen as the channel having no ``memory'' of previous events, which is why such a configuration is known as a memoryless quantum channel. The memoryless assumption simplifies the noise induced input-output mapping of quantum states, and provides an accurate portrayal of particular physical events. For instance, communication schemes whose signalling rate is low to allow the environment to reset, or scenarios in which a magnetic field is applied to reset the memory of the channel, can be accurately modeled using the memoryless configuration \cite{qchannelswithmem}.

However, memoryless quantum channels cannot be used to model all quantum communication systems. In optical fibers (a typical medium to transmit quantum information), sufficiently high signalling rates cause the environment to respond to each successive transmission in a way that is correlated to previous ones \cite{Ball1}, \cite{Ball2}. In quantum information processors, quantum bits can be so closely spaced that cross-talk may occur and the channel noise might be correlated \cite{qubits1}, \cite{qubits2}. These two scenarios exemplify situations in which applying the popular memoryless noise configuration would not provide a realistic description of the corresponding physical events. Instead, these systems require a quantum channel model that integrates memory effects.
We introduce such a model in the sequel, where we begin by discussing memoryless quantum channels and proceed by extending them to the more general model of quantum channels exhibiting memory.

\subsubsection{Memoryless quantum channels}\label{sub:memoryless}
In \cite{memless_model}, the authors describe a specific scenario in which applying a memoryless quantum channel is appropriate. They assume that the sequence of signals is transmitted over the physical medium in an ordered fashion and at a constant speed and that the environment undergoes a dissipative process which resets it to a stable configuration on a timescale $\tau$, known as the channel relaxation time\footnote{This refers to the time the channel requires to come back to an idle state.} of the medium. Then, if the rate at which each successive signal is transmitted is much lower than the inverse of the channel relaxation time, i.e $R < \frac{1}{\tau}$, the physical event can be represented by a memoryless quantum channel.

Generally speaking, a memoryless channel affecting a system of $n$ qubits is comprised by the tensor product of the channels affecting each of the individual qubits of the system. Mathematically, this is expressed as
\begin{equation}\label{eq:memorylessGeneral}
\mathcal{N}^{(n)} = \bigotimes_{j=1}^{n}\mathcal{N}^j,
\end{equation}
where $\mathcal{N}^j$ refers to the noise channels that individually affect each of the qubits. Furthermore, it is a common assumption when dealing with this type of channels that the interaction that each qubit has with the environment is identical. Thus, the memoryless expression \eqref{memorylessGeneral} is reduced to
\begin{equation}\label{eq:memorylessIdenticalGen}
\mathcal{N}^{(n)}(\rho) = \mathcal{N}^{\otimes n}(\rho),
\end{equation}
where now $\mathcal{N}$ is the one-qubit decoherence model that will equally affect each of the elements of the quantum system. Therefore, for the decoherence model based on the combined amplitude and phase damping channel, the memoryless channel for an $n$-qubit system will be expressed as
\begin{equation}\label{eq:memlessAmpPhaseDamp}
\mathcal{N}^{(n)}(\rho) = \mathcal{N}_{\mathrm{APD}}^{\otimes n}(\rho),
\end{equation}
where $\mathcal{N}_\mathrm{APD}$ is described by the error operators in  \eqref{ampphaseKraus}.

However, as explained in section \ref{sec:approximatedChannels}, we will not be able to efficiently implement this channel on a classical computer. The solution is to work with the corresponding PTA and CTA channels, which are known to fulfill the Gottesman-Knill theorem and to have the structure of $n$-qubit Pauli channels.

As seen in sections \ref{subsub:PTA} and \ref{subsub:CTA}, the one-qubit channels obtained by twirling the combined amplitude and phase damping channel are the Pauli channel for the PTA case and its symmetric instance or depolarizing channel for the CTA. Accordingly, the $n$-qubit channels that arise in an identical and memoryless manner from these channels will have the structure of the $n$-qubit Pauli channels in \eqref{NPauli}
\begin{equation}\label{eq:memorylessPauli}
\mathcal{N}^{(n)}(\rho) = \mathcal{N}_\mathrm{P}^{\otimes n}(\rho) = \mathcal{N}_\mathrm{P}^{(n)}(\rho) = \sum_{\mathrm{A}\in\mathcal{P}_n} p_\mathrm{A} \mathrm{A}\rho\mathrm{A},\;\;\mathrm{A}\in \mathcal{P}_n
\end{equation}
where $\mathrm{A} = \mathrm{A}_{1}\otimes\cdots\otimes\mathrm{A}_{n-1}\otimes\mathrm{A}_n$ with $\mathrm{A}_j\in\mathcal{P}_1$ denotes each of the possible $n$-fold Pauli error operators, with probability distribution $p_\mathrm{A}$
\begin{equation}\label{eq:memorylessProbDist}
p_{\mathrm{A}} = \prod_{j=1}^{n}p_{\mathrm{A}_j}.
\end{equation}
The last probability expression comes from the fact that the channel acts independently on each of the qubits and the fact that $\mathrm{A} = \bigotimes_{j=1}^n \mathrm{A}_j$, $\mathrm{A}_j\in\mathcal{P}_1$. Therefore, the probability of the total error event $\mathrm{A}$ will be given by the product of the probabilities of each of the individual qubits error events.

\subsubsection{Quantum Channels with Memory}\label{sub:memorychannels}
Quantum information can sometimes be exposed to physical events that exhibit spatial or temporal correlations. The outcome of the general twirling approximations of $n$-qubit channels also shows features related to these phenomena. To obtain simplified channels capable of representing such interactions, twirled channels with the ability to integrate memory effects are necessary. For this purpose, we must find how to model correlations in PTA and CTA channels. To that end, we need to model the conditioned probability distribution $\mathrm{P}(\mathrm{A}_n|\mathrm{A}_{n-1}\otimes\mathrm{A}_{n-2}\otimes\cdots\otimes\mathrm{A}_1),\mathrm{A}_j\in\mathcal{P}_1$, that is, the probability of the current $1$-qubit PTA or CTA error event conditioned by all the previous error events. This way, the expression for the general $n$-fold Pauli channel described in \eqref{NPauli} can be written as
\begin{equation}\label{eq:generalMemCTAPTA}
\begin{split}
\mathcal{N}_{\mathrm{P}}^{(n)}(\rho) &=\sum_{\mathrm{A}\in\mathcal{P}_n} p_\mathrm{A} \mathrm{A}\rho\mathrm{A}\\
&= \sum_{A\in\mathcal{P}_n}\left(\prod_{j=1}^n \mathrm{P}(\mathrm{A}_j|\mathrm{A}_{j-1}\otimes\cdots\otimes\mathrm{A}_1)\right) \mathrm{A}\rho\mathrm{A}.
\end{split}
\end{equation}

The most studied class of quantum channels with memory is the family of channels with Markovian correlated noise \cite{markov_noise1,markov_noise2,markov_noise3}, which considers noise models in which quantum objects are transformed via the application of maps whose elements are randomly generated by a classical Markov process. Markov processes describe sequences of events in which the probability of each event depends only on the previous event, i.e. $\mathrm{P}(\mathrm{A}_j|\mathrm{A}_{j-1}\otimes\mathrm{A}_{j-2}\otimes\cdots\otimes\mathrm{A}_1) = \mathrm{P}(\mathrm{A}_j|\mathrm{A}_{j-1})$. For this type of channel, the correlation between the single qubit Pauli operators of the set $\mathcal{P}_1$ can be described by means of a $4$-state Markov chain \cite{markov1,markov2,markov3,MemQTC}. This can be seen in Figure \ref{fig:markov_chain}, where each state corresponds to one of the single qubit Pauli operators. The transition probability from a previous error state $\mathrm{A}_{j-1}$ to the current error state $\mathrm{A}_j$ is denoted by $q_{\mathrm{A}_j|\mathrm{A}_{j-1}} = \mathrm{P}(\mathrm{A}_j|\mathrm{A}_{j-1})$, where $\mathrm{A}_j,\mathrm{A}_{j-1}\in \mathcal{P}_1$.

\begin{figure}
\centering
\includegraphics[scale=0.7]{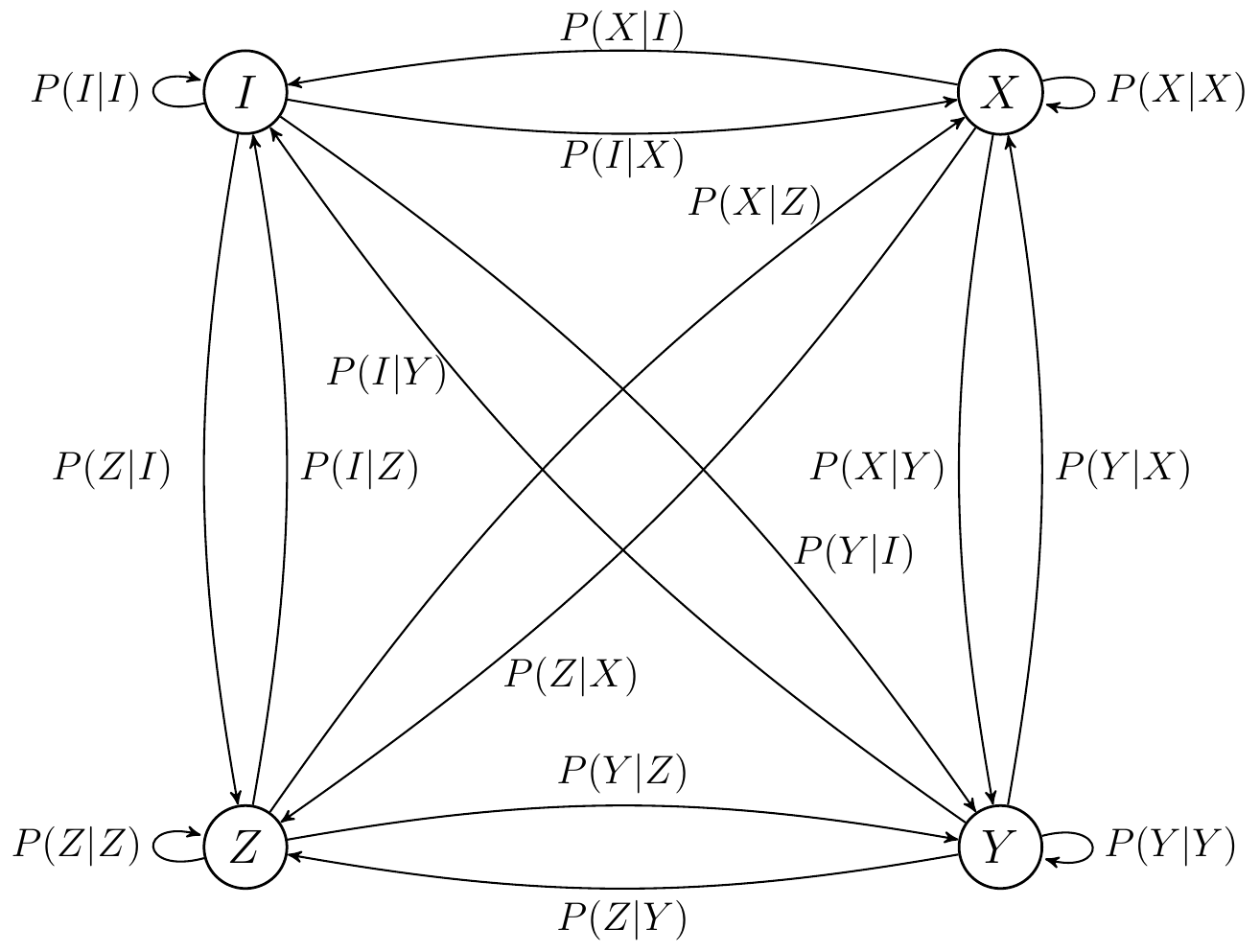}
\caption{State diagram of the Markov chain that describes the correlation between two consecutive qubit Pauli operators in a channel with Markovian memory.}
\label{fig:markov_chain}
\end{figure}

A possible way to capture the temporal correlation of Pauli operators over a channel with memory is to introduce the so called \textit{correlation parameter} $\mu \in [0, 1]$, where $\mu = 0$ indicates zero correlation and $\mu = 1$ indicates perfect correlation. This correlation parameter was originally presented in \cite{first_model}, where the authors derive the first comprehensive quantum information characterization of memory effects in a continuous variable setup. Specifically, a continuous variable model of quantum memory channels that accurately portrays the transmission of quantum objects (photons encoded with information) through optical fibers characterized by finite relaxation times is proposed. The model describes each channel use as an independent bosonic mode, and along with two parameters, enables the description of multiple communication scenarios. In terms of the time delay between each successive transmission of a photon $\Delta t$ and the relaxation time $\tau$ of the optical fiber, the model enables the representation of a memoryless configuration ($\Delta t \gg \tau$), intersymbol interference memory ($\Delta t \sim \tau$) \cite{isimem}, or a perfect memory configuration ($\Delta t \ll \tau$) \cite{perf_mem}. As shown in \cite{first_model} and \cite{model2}, the transmissivity of the optical fiber $\epsilon$ (in the beam-splitter configuration considered for the model) plays the role of a memory parameter and it can be approximated by $\epsilon \approx e^{-\frac{\Delta t}{\tau}}$. For more abstract memory modeling, we can define the correlation parameter $\mu$ and assume, as in previous works \cite{memless_model,MemQTC,model2}, that it can be quantified in terms of the time delay between successive channel uses and the relaxation time of the environment: $\mu \approx e^{-\frac{\Delta t}{\tau}}$. 

A model introduced in \cite{markov_noise3} considers that two successive channel uses are related by the transition probability $q_{\mathrm{A}_j|\mathrm{A}_{j-1}}$ defined as
\begin{equation} \label{transition}
    q_{\mathrm{A}_j|\mathrm{A}_{j-1}} = (1 - \mu)p_{\mathrm{A}_j} + \mu\delta_{\mathrm{A}_{j-1}\mathrm{A}_j},
\end{equation}
where $\delta_{\mathrm{A}_{j-1}\mathrm{A}_j}$ is the Kronecker delta function and $p_{\mathrm{A}_j}$ is the probability that Pauli operator $\mathrm{A}_j \in \mathcal{P}_1$ is imposed on the transmitted quantum state. As a consequence, considering a quantum channel with memory defined by Markovian correlated noise, the joint probability $p_{\mathrm{A}}$ is given by 

\begin{equation}\label{eq:jointMarkovprob}
p_{\mathrm{A}} = p_{\mathrm{A}_1}q_{\mathrm{A}_2|\mathrm{A}_1}\cdots q_{\mathrm{A}_n|\mathrm{A}_{n-1}},
\end{equation}
for each $n$-fold Pauli operator $\mathrm{A}\in\mathcal{P}_n$. The resulting Markovian Pauli channel, $\mathcal{N}_{\mathrm{MP}}$, will thus have the following structure
\begin{equation} \label{eq:Pauli_with_mem} 
\mathcal{N}_{\mathrm{MP}}^{(n)}(\rho)= \sum_{{\mathrm{A}\in\mathcal{P}_n}} p_{\mathrm{A}_1}q_{\mathrm{A}_2|\mathrm{A}_1}\cdots q_{\mathrm{A}_n|\mathrm{A}_{n-1}}\mathrm{A}\rho\mathrm{A},
\end{equation}
where $\mathrm{A} = \mathrm{A}_{n}\otimes\mathrm{A}_{n-1}\otimes\cdots\otimes\mathrm{A}_1$ with $\mathrm{A}_j\in\mathcal{P}_1$.

In this manner, Markovian Pauli memory channels, $\mathcal{N}_{\mathrm{MP}}$, can be created from the PTA and CTA approximations, which allow for  efficient implementation on a classical computer.

Arbitrary Markovian quantum channels with memory can be derived by applying the memory effects to the error-sum operator representation of general quantum channels \cite{memless_model,memreview}. These models will accurately capture the general effects that decoherence processes with memory generate over quantum information. We expect that in the near future the research community will show increasing interest in this topic, including the development of memory models that go beyond the Markovian paradigm. Nonetheless, at the time of writing, the discussion provided in this section completely describes the current state of affairs with regard to channels that are efficiently implementable on classical computers and that have been used for QECC construction. In the following of this dissertation, the noise models considered will be memoryless.

\subsection{Quantum Channel Capacity}\label{sec:qCap}
One of the main results of the work of Claude Shannon \cite{shannon} is the noisy-channel coding theorem, where the concept of \textit{channel capacity} was introduced. In the context of classical information theory, a rate\footnote{The rate of a code is the ratio between the number of information bits, $k$, and the number of codeword bits, $n$, i.e. $R=k/n$.}, $R$, is said to be achievable for a noisy channel if there exists a sequence of codes of such rate such that the probability of error of the code goes to zero as $n\rightarrow \infty$. The channel capacity is the supremum of all achievable rates for a noisy channel. That is, it is the maximum rate $R$, in bits per channel use, at which information can be reliably\footnote{Reliably here means with asymptotically vanishing error probability} transmitted over the channel. Therefore, codes with $R<C$ will exist such that the decoding error probability can asymptotically vanish as the length of the codewords increases. On the contrary, no such codes exist when $R>C$.
Note the importance of channel capacity for the coding theorist in order to construct channel codes.

The capacity is mathematically defined as
\begin{equation}\label{eq:classicalCap}
C = \sup_{p_x} I(X;Y),
\end{equation}
where $X$ and $Y$ refer to the random variables associated to the channel input and the channel output, respectively, $p_x$ refers to the distribution of the channel input random variable and $I(X;Y)$ refers to the mutual information of the two random variables. The mutual information gives a measure of the mutual dependence between the two random variables. From \eqref{classicalCap}, the channel capacity is then given as the supremum of mutual information with respect to all possible input distributions $p_x$.

Quantum channel capacity is a fundamental result of quantum information theory and it can be interpreted as the quantum analogue of the classical result obtained by Shannon \cite{shannon}. The quantum capacity, $C_\mathrm{Q}$, is the maximum rate at which quantum information can be communicated/corrected over many independent uses of a noisy quantum channel. In other words, the concept of the quantum capacity establishes the quantum rate\footnote{The quantum coding rate, $R_\mathrm{Q}$, of an $[[n,k]]$ quantum code is measured in terms of the number of qubits transmitted per channel use, i.e. we have $R_\mathrm{Q} = k/n$, where this means that $k$ logical qubits are encoded per $n$ physical qubits. A rate $R_\mathrm{Q}$ is said to be achievable for a quantum channel, $\mathcal{N}$, if there exists a sequence of $[[n,k]]$ quantum codes such that the probability of error of the codes goes to zero as $n\rightarrow\infty$ \cite{wildeQIT}.}, $R_\mathrm{Q}$, limit for which reliable (i.e., with a vanishing error rate) quantum communication/correction is asymptotically possible. Note that, traditionally, the concept of quantum channel capacity is understood in the context of quantum communications. In the realm of communication, it is convenient to think of a sender (Alice) who wants to relay qubits to a receiver (Bob). For memory or processing devices, Alice and Bob simply label the input and output. In this way, the noise suffered by qubits due to decoherence can be thought as the transmission of the information through a ``virtual" noisy channel \cite{capComm}. Hence, we can also apply the concept of quantum channel capacity to this framework by defining it as the maximum achievable rate by QEC that can make the stored or processed quantum information errorless.

The definition of quantum channel capacity, $C_\mathrm{Q}(\mathcal{N})$, is similar to its classical analogue, that is, the supremum of all achievable quantum rates for a noise channel $\mathcal{N}$ \cite{wildeQIT}. The following theorem, often referred to as the Lloyd-Shor-Devetak (LSD) theorem, relates quantum channel capacity with the regularized coherent information of the channel \cite{quantumcap,wildeQIT}.

\begin{theorem}[LSD capacity]\label{thm:quantumcap}
\textit{The quantum capacity $C_\mathrm{Q}(\mathcal{N})$ of a quantum channel $\mathcal{N}$ is equal to the regularized coherent information of the channel
\begin{equation}\label{eq:cqreg}
C_\mathrm{Q}(\mathcal{N}) = Q_{\mathrm{reg}}(\mathcal{N}),
\end{equation}
where
\begin{equation}\label{eq:coherentreg}
Q_{\mathrm{reg}}(\mathcal{N}) = \lim_{n\rightarrow \infty} \frac{1}{n} Q_{\mathrm{coh}}(\mathcal{N}^{\otimes n}).
\end{equation}
The channel coherent information $Q_{\mathrm{coh}}(\mathcal{N})$ is defined as
\begin{equation}\label{eq:coh}
Q_{\mathrm{coh}}(\mathcal{N}) = \max_{\rho} (S(\mathcal{N}(\rho)) - S(\rho_\mathrm{E})),
\end{equation}
where $S(\rho) = -\mathrm{Tr}(\rho\ln{\rho})$ is the von Neumann entropy and $S(\rho_\mathrm{E})$ measures how much information the environment has.}
\end{theorem}
We will not discuss the subtleties of the LSD theorem as it is out of the scope of the dissertation. For arbitrary channels, there is no closed-form analytical expression of the quantum capacity given in theorem \ref{thm:quantumcap}. However, the amplitude damping channel and its twirl approximations have either closed-form expressions or bounds for their LSD capacities. The pure dephasing channel does also have an analytical solution.

\subsubsection*{Amplitude damping channel}\label{sub:CqAD}
The quantum capacity of an AD channel with damping parameter $\gamma \in [0,1]$ is equal to \cite{quantumcap,wildeQIT}
\begin{equation}\label{eq:ADcap}
C_\mathrm{Q}(\gamma) = \max_{\xi\in[0,1]}  \mathrm{H}_2((1-\gamma) \xi) - \mathrm{H}_2(\gamma \xi),
\end{equation}
whenever $\gamma \in [0, 1/2]$, and zero for $\gamma\in [1/2,1]$. $\mathrm{H}_2(p)$ is the binary entropy in bits, i.e., $\mathrm{H}_2(p) = -p\log_2(p)-(1-p)\log_2(1-p)$.
\begin{figure}[!h]
\centering
\includegraphics[scale=0.8]{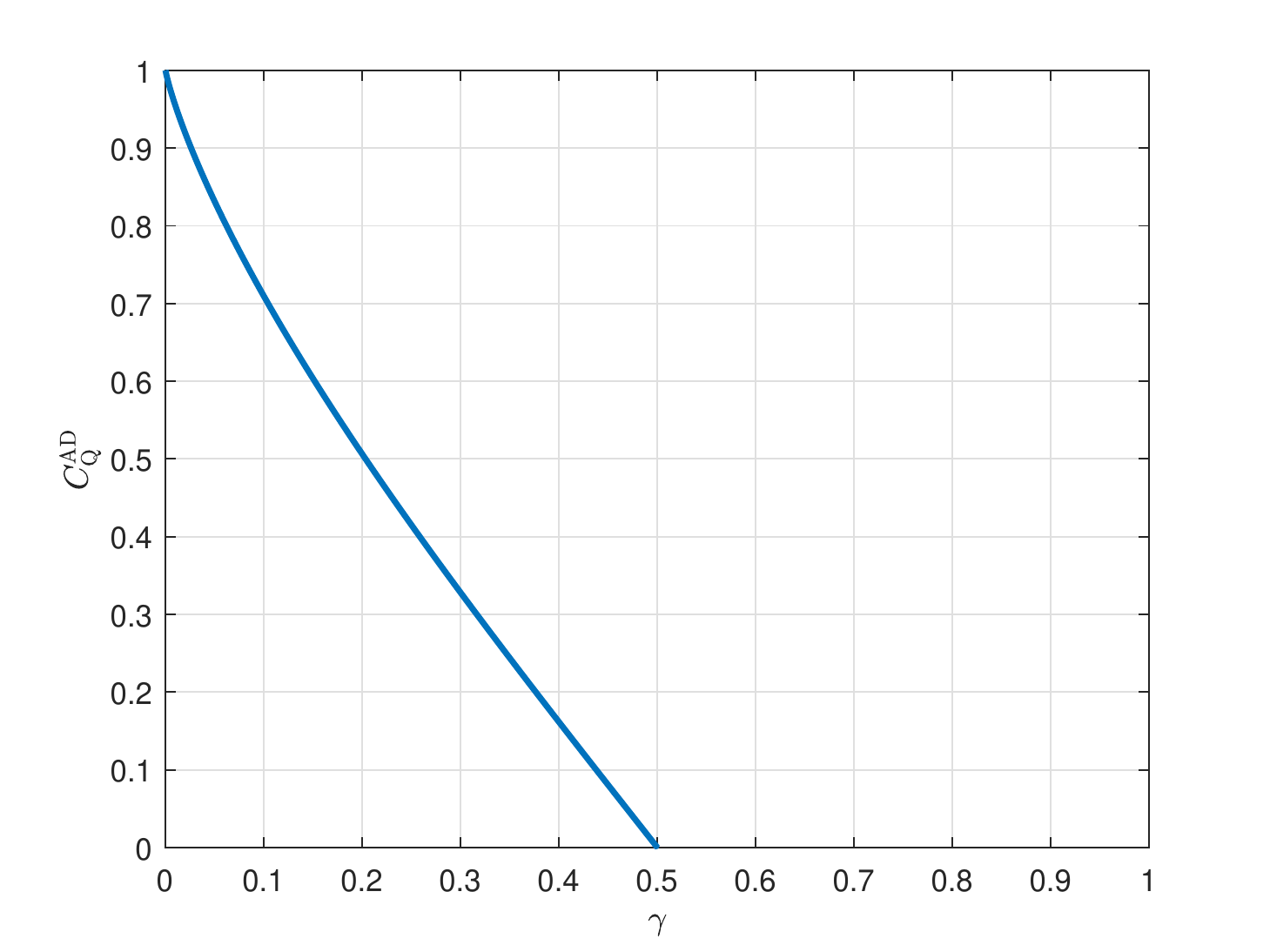}
\caption{Quantum capacity of the amplitude damping channel as a function of the damping parameter.}
\label{fig:CqAD}
\end{figure}

\subsubsection*{Pure dephasing channel}\label{sub:CqPD}
The quantum capacity of a PD channel with scattering parameter $\lambda \in [0,1]$ is equal to \cite{wildeQIT}
\begin{equation}\label{eq:PDcap}
C_\mathrm{Q}(\lambda) = 1- \mathrm{H}_2 \left(\frac{1-\sqrt{1-\lambda}}{2}\right),
\end{equation}
whenever $\lambda \in [0, 1]$, and again $\mathrm{H}_2(p)$ is the binary entropy.

\begin{figure}[!h]
\centering
\includegraphics[scale=0.8]{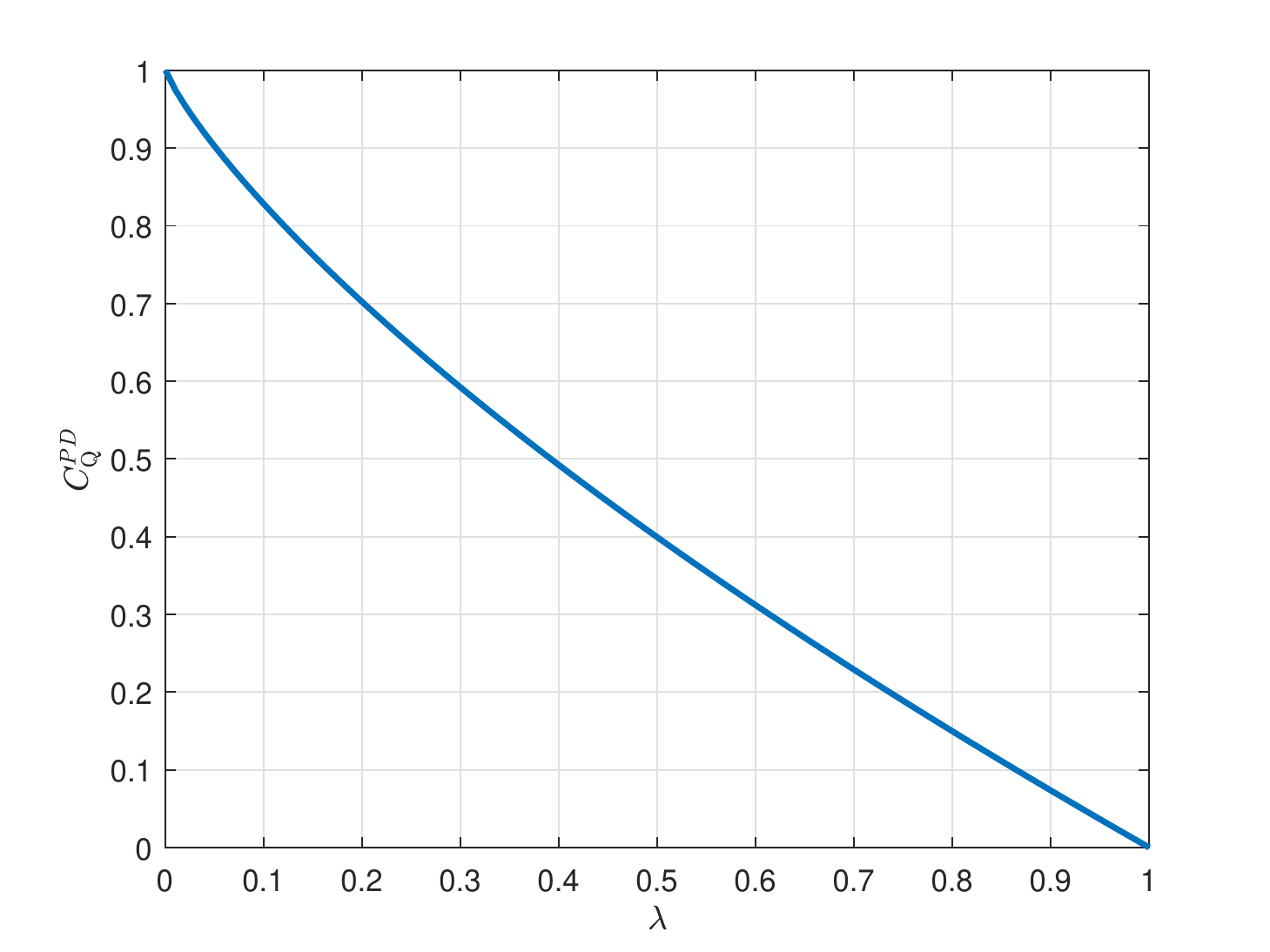}
\caption{Quantum capacity of the pure dephasing channel as a function of the scattering parameter.}
\label{fig:CqPD}
\end{figure}
\subsubsection*{Combined amplitude and phase damping channel}\label{sub:CqAPD}
There is no single-letter formula for the combined amplitude and phase damping channel at the time of writing. However, an upper bound can be obtained by using the bottleneck inequality \cite{bottleneck}
\begin{equation}
C_\mathrm{Q}(\mathcal{N}_1\circ\mathcal{N}_2)\leq \min\{C_\mathrm{Q}(\mathcal{N}_1), C_\mathrm{Q}(\mathcal{N}_2)\}.
\end{equation}

\subsubsection*{Pauli channels}\label{sub:CqP}
An expression for the quantum capacity of the widely used Pauli channels remains unknown \cite{wildeQIT}. However, a lower bound that can be achieved by stabilizer codes, the hashing bound, $C_\mathrm{H}$, \cite{wildeQIT} is known. The reason why the quantum capacity of a Pauli channel can be higher than the hashing bound, i.e. $C_\mathrm{Q} \geq C_\mathrm{H}$, is the degenerate nature of quantum codes \cite{degenPRL}, which arises from the fact that several distinct channel errors affect quantum states in an indistinguishable manner.

The hashing bound for a Pauli channel defined by the probability mass function $\mathbf{p} = (p_\mathrm{I},p_\mathrm{x},p_\mathrm{y},p_\mathrm{z})$ is given by \cite{wildeQIT}
\begin{equation}\label{eq:hash}
C_\mathrm{H}(\mathbf{p}) = 1 - \mathrm{H}_2(\mathbf{p}).
\end{equation}
$\mathrm{H}_2(\mathbf{p}) = -\sum_j p_j\log_2(p_j)$ is the entropy in bits of a discrete random variable with probability mass function given by $\mathbf{p}$.

To sum up, quantum channel capacity gives the maximum achievable quantum rate that the quantum coding theorist can aim in order to construct codes with asymptotically vanishing error rates. Closed-form expressions for such quantity exist for the amplitude damping channel and the pure dephasing channels. Both for the combined amplitude and phase damping channel and the family of Pauli channels, the best we can do at the time of writing are bounds for their capacity.

\section{Quantum error correction}\label{sec:QEC}
Previously we have described how to model decoherence processes, and how to obtain twirled approximated channels from these general models so that they can be appropriately simulated on classical computers. In consequence, the twirled channels we construct will be implementable in run-of-the-mill classical machines available at this moment in time. This is a significant step forward in the endeavour of the quantum information community to construct QECC families for the NISQ and post-NISQ-era, before quantum machines that can take advantage of these error correction methods can be constructed.

Design and simulation of error correction methods for the Pauli channel have been exhaustively studied by the quantum information community and are well documented \cite{QRM,bicycle,qldpc15,jgf,QCC,QTC,EAQTC,toric,
QEClidar,EAQECC,EAQIRCC,twirl6,MemQTC,catalytic}. All those codes are constructed within the Quantum Stabilizer Code framework \cite{QSC} and, thus, will fulfill the conditions of the Gottesman-Knill theorem (including the entanglement-assisted versions), and so we will be able to implement them in regular computers.  The key to implement Pauli channels on classical machines is the Pauli-to-binary isomorphism, which maps elements of the Pauli operator set onto binary strings \cite{isomorphism,catalytic}. Establishing this isomorphism between these sets enables researchers to assess the performance of designed QECC families via Monte Carlo simulations in a fashion reminiscent of classical coding theory.

In this section we describe the Pauli-to-binary isomorphism and the way in which this mapping can be used to simulate the performance of QECCs on classical computers. {With this goal in mind, we also discuss basic stabilizer coding, syndrome measurements, and error discretization. Additionally, we provide a simple example to see how the Word Error Rate (WER) and the QuBit Error Rate (QBER) of a quantum error correction method can be obtained via Monte Carlo simulation. Both WER and QBER are used as figures of merit when gauging the error correcting capability of the QECC schemes.

Note that in the remainder of this dissertation when we write QEC we refer to the family of stabilizer codes. Other approaches to quantum computation and error correction exist, such as bosonic codes \cite{bosonic1,bosonic2}, but here we focus our research and discussions on stabilizer coding.

\subsection{Knill-Lafflame theorem and Error discretization}\label{sub:errorDisc}
Before presenting the theory of stabilizer codes, some important results about the general theory of quantum error correction must be discussed, so that the construction of the stabilizer family can be introduced in a logical and coherent manner.

We begin with the following theorem, which defines the conditions that a quantum error correcting code must fulfill to protect quantum information from a particular noise process $\mathcal{N}$. We refer to this theorem as the \textit{quantum error correction conditions} or the \textit{Knill-Lafflame theorem} \cite{NielsenChuang,knillLaff}.

\begin{theorem}[Knill-Lafflame theorem]\label{thm:QECCconditions}
\textit{Let $\mathrm{C}\subset\mathcal{H}_2^{\otimes n}$ be a quantum code defined as a subspace of the $2^n$-dimensional Hilbert space, that is, the state space of the $n$-qubit system to be protected, and let $P$ be the projector onto $\mathrm{C}$. Let $\mathcal{N}$ be a noise operation defined by Kraus operators $\{E_j\}$. A necessary and sufficient condition for the existence of an error-correction operation $\mathcal{R}$ correcting $\mathcal{N}$ on $\mathrm{C}$ is that
\begin{equation}\label{eq:QECCconditions}
P E_j^\dagger E_k P = \alpha_{jk}P,
\end{equation}
for some Hermitian matrix of complex numbers $\alpha$. }

\textit{We call the Kraus operators $\{E_j\}$ errors. If the recovery operation $\mathcal{R}$ exists, we say that $\{E_j\}$ constitutes a correctable set of errors}.
\end{theorem}

Theorem \ref{thm:QECCconditions} states the conditions that any quantum error correcting code must verify in order to correct a specific error operator $\mathcal{N}$. At first glance, it seems that the ability that any code has to correct errors will be limited to a finite set $\{E_j\}$ errors operators, something that appears to be an important limitation. However, QECCs are actually significantly more powerful due to the following theorem known as \textit{discretization of errors} \cite{NielsenChuang,knillLaff}.

\begin{theorem}[Discretization of errors]\label{thm:discretization}
\textit{Suppose $\mathrm{C}$ is a quantum code and $\mathcal{R}$ is the error-correction operation to perform recovery from a noise process $\mathcal{N}$ with Kraus operators $\{E_j\}$. Suppose that $\mathcal{F}$ is another noise process with Kraus operators $\{F_k\}$ which are linear combinations of the $E_j$ operators, i.e. $F_k= \sum_j \xi_{kj}E_j$ for some matrix $\xi$ of complex numbers. Then, the error correction operation $\mathcal{R}$ also corrects for the effects of the noise process $\mathcal{F}$ on the code $\mathrm{C}$.}
\end{theorem}

This result is momentous, since quantum codes designed to correct a finite number of Kraus operators, they will also be able to correct an infinite set of noise processes, i.e., any linear combinations these of Kraus operators. In particular, when considering the previous twirled channels, the discretization of errors means that any code that is able to correct a subset of $n$-fold Pauli operators in $\mathcal{P}_n$ will also be able to correct any linear combinations of these Pauli operators.

This result is invaluable when designing and simulating QECCs efficiently on a classical computer, since it allows us to do so by just considering discrete sets of errors represented by $n$-fold Pauli operators. Considering that the set of errors that have to be tested is discrete, it will be possible to represent the errors using binary notation, i.e. in a way that classical computers can manage.

\subsection{Stabilizer Codes}\label{sub:stabCodes}
The set of $n$-fold Pauli operators $\mathcal{P}_n$ together with the overall factors $\{\pm 1,\pm i\}$ forms a group under multiplication. This group is known as the $n$-fold Pauli group $\mathcal{G}_n$ \cite{EAQECC}. A stabilizer code is defined by an abelian subgroup $\mathcal{S}\subset\mathcal{G}_n$. For a $[[n,k,d]]$ unassisted\footnote{In this section we discuss stabilizer coding without entanglement-assistance for the sake of brevity. However, the methods presented here are equally valid for the entanglement-assisted formulation.} stabilizer code encoding $k$ logical qubits into $n$ physical qubits with distance\footnote{The distance of a stabilizer code is defined as the minimum weight of the Pauli operators that belong to the normalizer of the stabilizer. The normalizer of the stabilizer consists of the elements of the $n$-fold Pauli group that commute with all the generators of the stabilizer but do not belong to the stabilizer.} $d$, the stabilizer set has $2^{n-k}$ distinct elements up to an overall phase, and it is generated by $n-k$ independent generators\footnote{Stabilizer codes are represented by $n-k$ generators, given that the rest are combinations of them. This provides a compact representation of the code.}. The stabilizer code $\mathrm{C}(\mathcal{S})$ associated with the stabilizer set $\mathcal{S}$ is then defined as
\begin{equation}\label{eq:stabilizerCode}
\mathrm{C}(\mathcal{S}) = \{|\bar{\psi}\rangle\in\mathcal{H}_2^{\otimes n} : M|\bar{\psi}\rangle\ = |\bar{\psi}\rangle, \forall M \in \mathcal{S}\},
\end{equation}
where $\mathcal{H}_2^{\otimes n}$ denotes the complex Hilbert space of dimension $2^n$, which is the state space of systems formed by $n$ qubits.
Note that $\mathrm{C}(\mathcal{S})$ is given by the simultaneous $+1$-eigenspace defined by the elements\footnote{The simultaneous eigenspace is generated by the $n-k$ independent generators, and so the code is defined by them.} of $\mathcal{S}$.
\begin{figure}[!h]
\centering
\includegraphics[scale=0.8]{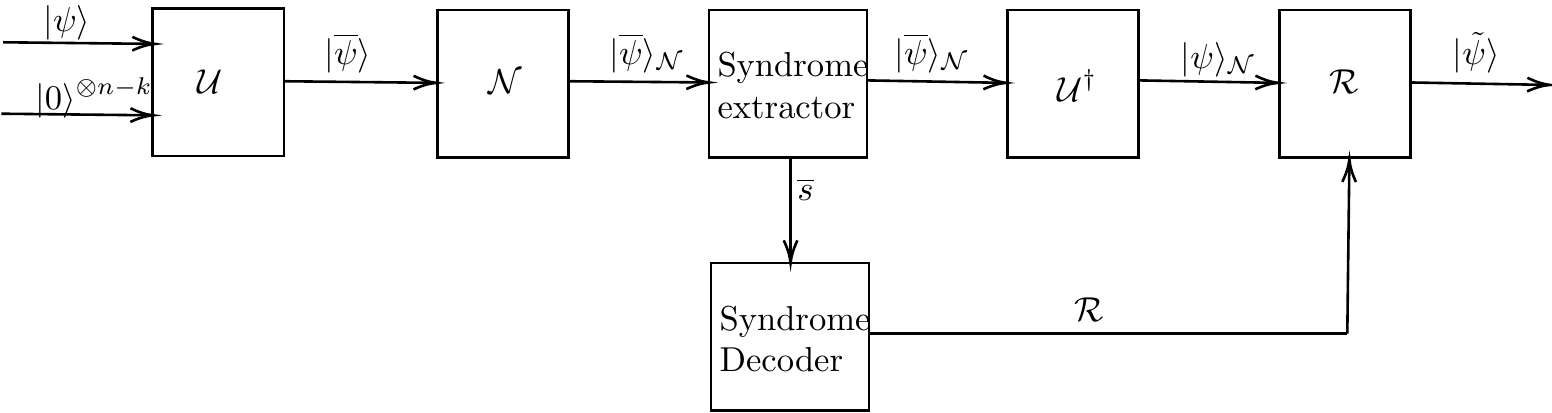}
\caption{General schematic of a stabilizer QECC.}
\label{fig:StabCode}
\end{figure}

Figure \ref{fig:StabCode} presents the general scheme for stabilizer codes that correct a noise operation $\mathcal{N}$. The unitary $\mathcal{U}$ maps the input information word $|\psi\rangle\in\mathcal{H}_2^{\otimes k}$ to the codespace $|\bar{\psi}\rangle\in\mathrm{C}(\mathcal{S})\subset\mathcal{H}_2^{\otimes n}$ with the aid of the ancilla qubits $|0\rangle^{\otimes n-k}$. The existence of such a unitary that takes the arbitrary input states $|\psi\rangle|0\rangle^{\otimes n-k}$ to the codespace is guaranteed\footnote{Such unitary can be formed up to a global phase by using $\mathcal{O}(n^2)$ generators of the Clifford group $\mathrm{H}$, $\mathrm{P}$ and $\mathrm{CNOT}$ \cite{EAQECC,OzolsClifford}.} \cite{EAQECC}. The encoded state then experiences the action of a quantum channel $\mathcal{N}$, which is described as one of the approximated twirl channels with the structure of a Pauli channel that have been introduced earlier in this work. The noisy quantum state, $|\bar{\psi}\rangle_\mathcal{N}$, must then be corrected by a recovery operation. The fundamentals of quantum mechanics establish that measuring a quantum state forces its superposition state to collapse, which causes the loss of the quantum information contained in the original state. Therefore, the theory of quantum error correction codes must circumvent this issue. This is achieved by measuring in an indirect way the so-called \textit{syndrome} of the error $\bar{s}$. Such a measure avoids the destruction of the quantum state and allows us to garner information regarding the error that will be used to estimate the best recovery operation $\mathcal{R}$.

The error syndrome $\bar{s}$ is defined as a binary vector of length $n-k$, i.e., $\bar{s}\in (\mathbb{F}_2)^{n-k}$, that captures the commutation relationship between the generators of the stabilizer set $\mathcal{S}$ and the Kraus error operators of the channel. By the twirled approximation the error operators $\{\mathcal{E}_k\}$  are  elements of the Pauli group, i.e., $\mathcal{E}_k\in\mathcal{G}_n$. Furthermore, as already stated, the designed codes will also correct any linear combinations of the error operators\footnote{For stabilizer codes, the effect of error discretization comes from the syndrome measurement \cite{NielsenChuang,isomorphism}.} $\{\mathcal{E}_k\}$. It is common knowledge that any two elements of the $n$-fold Pauli group either commute or anticommute. For this reason, any error operator $\mathcal{E}$ will either commute or anticommute with each of the generators $\mathrm{S}_j, j\in\{1,\cdots,n-k\}$, of the stabilizer set, and the components $s_j$ of the error syndrome capture their relationship as
\begin{equation}\label{eq:syndrome}
\mathcal{E}\mathrm{S}_j = (-1)^{s_j}\mathrm{S}_j\mathcal{E}.
\end{equation}
That is, $s_j$ will be zero or one depending if $\mathcal{E}$ and $\mathrm{S}_j$ commute or anticommute, respectively.

To construct a circuit that can extract the syndrome for each of the generators $\mathrm{S}_j$, consider the received noisy quantum state  $|\bar{\psi}\rangle_\mathcal{N}=\mathcal{E}|\bar{\psi}\rangle$. Then the vector $|\bar{\psi}\rangle_\mathcal{N}$ is an eigenstate of each of the generators of the stabilizer set $\mathcal{S}$ associated to the $\pm 1$ eigenvalues, i.e.
\begin{equation}\label{eq:syndromeMeasure}
\begin{split}
\mathrm{S}_j|\bar{\psi}\rangle_\mathcal{N} &= \mathrm{S}_j\mathcal{E}|\bar{\psi}\rangle = (-1)^{s_j}\mathcal{E}\mathrm{S}_j |\bar{\psi}\rangle = (-1)^{s_j}\mathcal{E}|\bar{\psi}\rangle \\
 &= (-1)^{s_j}|\bar{\psi}\rangle_\mathcal{N},
\end{split}
\end{equation}
where the particular value out of the two possible eigenvalues depends on the commutation relationship between the channel error and the stabilizer generator. As a result, in order to determine the syndrome of the error that takes place in a specific channel instance, the eigenvalue of $\mathrm{S}_j$ must be measured \cite{EAQECC}. Figure \ref{fig:syndromeMeas} presents a circuit built to perform such measurements.

\begin{figure}[!h]
\centering
\includegraphics[scale=1]{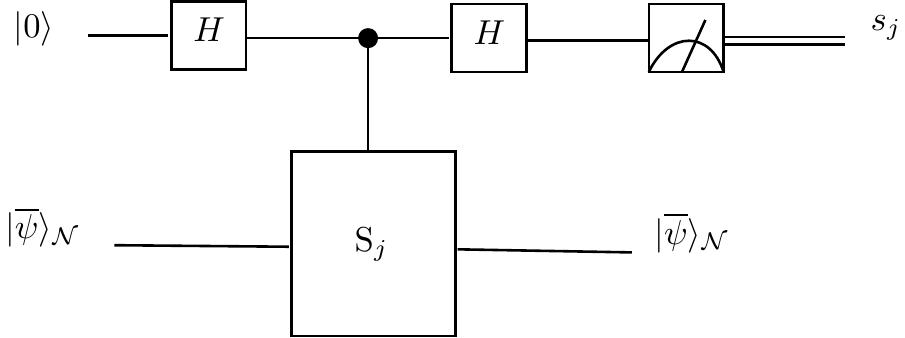}
\caption{Quantum circuit that measures the syndrome associated to each of the stabilizer generators $\mathrm{S}_j$. The $\mathrm{H}$ blocks stand for Hadamard gates and the $\mathrm{S}_j$ is a controlled unitary gate with the stabilizer generator as its unitary.}
\label{fig:syndromeMeas}
\end{figure}

Once the error syndrome is obtained, it is used to decode or estimate the recovery operation that will correct the corrupted quantum information state. Syndrome decoders of quantum stabilizer codes depend on the specific QECC family construction \cite{QRM,bicycle,qldpc15,jgf,QCC,QTC,EAQTC,toric,
QEClidar,EAQECC,EAQIRCC,twirl6,MemQTC,catalytic}. In general, optimal decoding of quantum stabilizer codes must consider the phenomenon known as \textit{degeneracy} \cite{degenPRL,hardness}, a feature which has no equivalence in classical coding. In degenerate quantum codes, there are sets of different error operators that have the same effect in the transmitted codeword. Although it should be possible to exploit this property in the decoding process (and also for the design of better codes), the design of decoders that can efficiently perform the so-called Degenerate Quantum Maximum Likelihood Decoding (DQMLD) remains unsolved for stabilizer codes in general\footnote{Quantum turbo codes are decoded exploiting degeneracy \cite{QTC,EAQTC}.}. Current decoding of these codes is approached in terms of Quantum Maximum Likelihood Decoding (QMLD), which ignores error degeneracy and undertakes the decoding task as in the field of classical linear block codes. Building effective degeneracy-exploiting decoders is out of the scope of this dissertation, so it will not be discussed further.

The last steps in the error correction operation are the application of the inverse of the encoding operation (unitary $\mathcal{U}^\dagger$), followed by the use of the recovery operation $\mathcal{R}$, which depends on the results of the syndrome decoder.

\subsection{Pauli-to-binary isomorphism}\label{sub:paulibinary}
In this section we present the Pauli-to-binary isomorphism as a way to express the operation of QECCs as binary vectors, which are generally referred to as \textit{symplectic strings}. We begin by defining a one-to-one mapping between the set of Pauli matrices $\mathcal{P}_1$ of one qubit and the set of binary strings of length two, and then extrapolating this mapping to the case of more dimensions (i.e., more qubits). To that end, we first review some facts regarding the set of Pauli matrices $\mathcal{P}_1$, Pauli group $\mathcal{G}_1$ and the effective Pauli group $[\mathcal{G}_1]$. Pauli matrices are Hermitian unitary matrices with eigenvalues equal to $\pm 1$ and Table \ref{tab:PauliMult} shows their product operation.

\begin{table}[h!]
\centering
\begin{tabular}{|c|cccc|}
\hline
$\times$      & $\mathrm{I}$ & $\mathrm{X}$   & $\mathrm{Y}$   & $\mathrm{Z}$   \\ \hline
$\mathrm{I}$ & $\mathrm{I}$ & $\mathrm{X}$   & $\mathrm{Y}$   & $\mathrm{Z}$   \\
$\mathrm{X}$ & $\mathrm{X}$ & $\mathrm{I}$   & $i\mathrm{Z}$  & $-i\mathrm{Y}$ \\
$\mathrm{Y}$ & $\mathrm{Y}$ & $-i\mathrm{Z}$ & $\mathrm{I}$   & $i\mathrm{X}$  \\
$\mathrm{Z}$ & $\mathrm{Z}$ & $i\mathrm{Y}$  & $-i\mathrm{X}$ & $\mathrm{I}$\\
\hline
\end{tabular}
\caption{Multiplication table for the Pauli matrices. Pauli matrices either commute or anticommute.}
\label{tab:PauliMult}
\end{table}

Note that the elements of the table are elements of the Pauli group $\mathcal{G}_1$. Given that neglecting the overall phase has no observable consequences \cite{catalytic}, it makes perfect physical sense to ignore it and construct an equivalence class of matrices $[\mathcal{G}_1] = \{[\mathrm{I}],[\mathrm{X}],[\mathrm{Y}],[\mathrm{Z}]\}$, where $[\mathrm{A}]$  refers to the equivalence class of matrices equal to $\mathrm{A}$ up to an overall phase. This equivalence class forms an Abelian group under the multiplication operation defined as $[\mathrm{A}][\mathrm{B}]=[\mathrm{AB}]$. The product relationships of this equivalence class are shown in table \ref{tab:equivTable}. In the literature, this equivalence class $[\mathcal{G}_1]$ receives the name of \textit{effective Pauli group}.

\begin{table}[h!]
\centering
\begin{tabular}{|c|cccc|}
\hline
$\times$      & $\mathrm{I}$ & $\mathrm{X}$   & $\mathrm{Y}$   & $\mathrm{Z}$   \\ \hline
$\mathrm{I}$ & $\mathrm{I}$ & $\mathrm{X}$   & $\mathrm{Y}$   & $\mathrm{Z}$   \\
$\mathrm{X}$ & $\mathrm{X}$ & $\mathrm{I}$   & $\mathrm{Z}$  & $\mathrm{Y}$ \\
$\mathrm{Y}$ & $\mathrm{Y}$ & $\mathrm{Z}$ & $\mathrm{I}$   & $\mathrm{X}$  \\
$\mathrm{Z}$ & $\mathrm{Z}$ & $\mathrm{Y}$  & $\mathrm{X}$ & $\mathrm{I}$\\
\hline
\end{tabular}
\caption{Multiplication table for the effective Pauli group. With a slight abuse of notation we refer to the equivalence class $[\mathrm{A}]$ as $\mathrm{A}$.}
\label{tab:equivTable}
\end{table}

We now consider the Abelian group of binary strings of length two $(\mathbb{F}_2)^2 = \{00,01,10,11\}$, with the usual modulo-$2$ addition (refer to Table \ref{tab:f2addition}). Note that if the modulo-$2$ product is also considered, $(\mathbb{F}_2)^2$ becomes a vector space over the field $\mathbb{F}_2$.

\begin{table}[h!]
\centering
\begin{tabular}{|c|cccc|}
\hline
$+$      & $00$ & $01$   & $11$   & $10$   \\ \hline
$00$ & $00$ & $01$   & $11$   & $10$   \\
$01$ & $01$ & $00$   & $10$  & $11$ \\
$11$ & $11$ & $10$ & $00$   & $01$  \\
$10$ & $10$ & $11$  & $01$ & $00$\\
\hline
\end{tabular}
\caption{Addition table for the binary vectors of length two under modulo-$2$ addition of their elements.}
\label{tab:f2addition}
\end{table}
Note that any element $u\in (\mathbb{F}_2)^2$ can be partitioned\footnote{Note that the isomorphism presented here can be also done taking $u=(x|z)$, which is also common in the quantum information theory literature. Both are equivalent.} as $u = (z|x)$, where $z,x\in\mathbb{F}$. Having said that, we now define the bilinear form \cite{catalytic} called \textit{symplectic form} or \textit{symplectic product} $\odot$, in the vector space $(\mathbb{F}_2)^2$. The symplectic product, $\odot:(\mathbb{F}_2)^2\times(\mathbb{F}_2)^2\rightarrow \mathbb{F}_2$, is defined as: 

\begin{equation}
u\odot v = (zx' + z'x) \text{ mod 2},
\end{equation}
where $u = (z|x)$ and $v=(z'|x')$. Table \ref{tab:sympProd} shows the values of this symplectic product for the elements of $(\mathbb{F}_2)^2$.

\begin{table}[h!]
\centering
\begin{tabular}{|c|cccc|}
\hline
$\odot$      & $00$ & $01$   & $11$   & $10$   \\ \hline
$00$ & $0$ & $0$   & $0$   & $0$   \\
$01$ & $0$ & $0$   & $1$  & $1$ \\
$11$ & $0$ & $1$ & $0$   & $1$  \\
$10$ & $0$ & $1$  & $1$ & $0$\\
\hline
\end{tabular}
\caption{Computation of the symplectic product for the binary vectors of length two.}
\label{tab:sympProd}
\end{table}

We are now ready to consider the one-to-one map $\Upsilon:(\mathbb{F}_2)^2\rightarrow \mathcal{P}_1$ defined in Table \ref{tab:map}. Note that any element $u=(z|x)$ in $(\mathbb{F}_2)^2$ is mapped into $\mathrm{Z}^z\mathrm{X}^x$ up to a phase factor, i.e. is mapped to the equivalence class $\Upsilon_{(z|x)}=[\mathrm{Z}^z\mathrm{X}^x]$, previously defined.

\begin{table}[h!]
\centering
\begin{tabular}{|c|c|}
\hline
$(\mathbb{F}_2)^2$   & $\mathcal{P}_1$  \\ \hline
$00$ & $\mathrm{I}$    \\
$01$ & $\mathrm{X}$    \\
$11$ & $\mathrm{Y}$   \\
$10$ & $\mathrm{Z}$ \\
\hline
\end{tabular}
\caption{Mapping $\Upsilon$ of the Pauli matrices to the elements of the set of length two binary vectors.}
\label{tab:map}
\end{table}

The following two important properties of $\Upsilon$ also hold \cite{catalytic}:
\begin{itemize}
\item From tables \ref{tab:equivTable} and \ref{tab:f2addition}, the map $[\Upsilon]: (\mathbb{F}_2)^2 \rightarrow [\mathcal{G}_1]$ is an isomorphism. That is,
\begin{equation}\label{eq:isomorph}
[\Upsilon_u] [\Upsilon_v] = [\Upsilon_{u+v}].
\end{equation}
\item From tables \ref{tab:PauliMult} and \ref{tab:sympProd}, the commutation relationships between the Pauli matrices are captured by the symplectic product for binary vectors of length two. That is,
\begin{equation}\label{eq:commmasSymp}
\Upsilon_u \Upsilon_v = (-1)^{u\odot v}\Upsilon_v \Upsilon_u.
\end{equation}
\end{itemize}

Note that by using $\Upsilon$, we may represent Pauli operators as binary strings of length two, and still be able to capture the commutation relationships via the symplectic product.

The next step is to extend the isomorphism to the $n$-fold set of Pauli operators in order to cope with systems that have $n$ qubits. The elements of the $n$-fold set of Pauli operators $\mathcal{P}_n$ are tensor products of individual Pauli matrices, and so the equivalence class $[\cdot]$ for these elements of the Pauli group will be defined in the same way, i.e. the set $[\mathcal{G}_n]=\{[\mathrm{A}]:\mathrm{A}\in\mathcal{P}_n\}$. The operation of this equivalence class will be the product defined as $[\mathrm{A}][\mathrm{B}] = [\mathrm{A}_1\mathrm{B}_1]\otimes[\mathrm{A}_2\mathrm{B}_2]\otimes\cdots\otimes[\mathrm{A}_n\mathrm{B}_n]=[\mathrm{AB}]$, and this set will also be an Abelian group under multiplication, as it was the case for the $1$ qubit scenario.

Having defined the $n$-fold effective Pauli group, we now proceed, as in the previous case, by considering the vector space $(\mathbb{F}_2)^{2n}$ under the usual sum and product mod-2  operations. Any element $\bar{u}\in (\mathbb{F}_2)^{2n}$ will be written as $\bar{u}=(\bar{z}|\bar{x})$, where $\bar{z}=z_1\cdots z_n\in(\mathbb{F}_2)^n$ and $\bar{x}=x_1\cdots x_n\in(\mathbb{F}_2)^n$. Now, the symplectic product for $\bar{u}, \bar{v}\in (\mathbb{F}_2)^{2n}$ is the application $\odot: (\mathbb{F}_2)^{2n}\times (\mathbb{F}_2)^{2n}\rightarrow \mathbb{F}_2$ defined as
\begin{equation}\label{eq:symplect2n}
\bar{u}\odot \bar{v} = \left(\bar{z}\bar{x}'^{T} + \bar{z}'\bar{x}^{T}\right)\text{ mod 2} = \left(\sum_j u_j\odot v_j\right)\text{mod 2},
\end{equation}
where $\bar{u}=(\bar{z}|\bar{x})$, $\bar{v}=(\bar{z}'|\bar{x}')$, $u_j = (z_j|x_j)\in(\mathbb{F}_2)^2$, $v_j = (z_j'|x_j')\in(\mathbb{F}_2)^2$ and $T$ denotes the transpose. Notice that we are representing the strings as row vectors. Consequently, the symplectic inner product of the $2n$ length binary string will be the Boolean sum of the symplectic product of the $n$ $(\mathbb{F}_2)^2$ elements that form such a string.

The map $\Upsilon$ is now performed individually for each of the elements of the $n$-fold tensor product that form each of the elements of $\mathcal{P}_n$, namely, the map $\Upsilon: (\mathbb{F}_2)^{2n}\rightarrow \mathcal{P}_n$ is taken as $\Upsilon_{\bar{u}} = \Upsilon_{u_1}\otimes\Upsilon_{u_2}\otimes\cdots\otimes\Upsilon_{u_n}$. It is trivial to see that this map has been selected so that $\Upsilon_{(\bar{z}|\bar{x})}$ is equal to $\mathrm{Z}^{\bar{z}} \mathrm{X}^{\bar{x}} = \mathrm{Z}^{z_1}\mathrm{X}^{x_1}\otimes \mathrm{Z}^{z_2}\mathrm{X}^{x_2}\otimes\cdots\otimes \mathrm{Z}^{z_n}\mathrm{X}^{x_n}$ up to an overall phase.

The $n$-fold map, $\Upsilon$, verifies de following two properties \cite{catalytic}:
\begin{itemize}
\item The map $[\Upsilon]:(\mathbb{F}_2)^{2n}\rightarrow [\mathcal{G}_n]$ is an isomorphism as
\begin{equation}\label{eq:isomorphN}
[\Upsilon_{\bar{u}}][\Upsilon_{\bar{v}}]=[\Upsilon_{\bar{u}+\bar{v}}].
\end{equation}
\item The commutation relationships of the $n$-fold Pauli matrices are captured by the symplectic product as
\begin{equation}\label{eq:commuAsSympN}
\Upsilon_{\bar{u}} \Upsilon_{\bar{v}} = (-1)^{\bar{u}\odot \bar{v}}\Upsilon_{\bar{v}} \Upsilon_{\bar{u}}.
\end{equation}
\end{itemize}

Based on these results, the $n$-fold Pauli error operators of the Pauli channels can be represented by binary vectors of length $2n$. Similarly, since the generators of a stabilizer set $\mathcal{S}$ of a code are $n$-fold Pauli matrices, by using the $n$-fold map, $\Upsilon$, any stabilizer code can also be described by a binary matrix known as \textit{parity check matrix} (PCM), as in classical coding. More explicitly, the PCM of a stabilizer code is the $(n-k)\times 2n$ matrix:
\begin{equation}\label{eq:PCM}
H = \left(H_z  |  H_x\right),
\end{equation}
where each row of the matrix is obtained by applying the map $\Upsilon$ to each of the $(n-k)$ stabilizer generators. Moreover, since the symplectic product captures the commutation relationships that exist among the $n$-fold Pauli matrices, the PCM $H$ of a code can be used to obtain the error syndrome arising from the Pauli error operators of the channel. The syndrome $\bar{s}(\bar{\epsilon})$ of a particular error $\mathcal{E}$, with binary representation $\bar{\epsilon}$, is calculated as
\begin{equation}\label{eq:syndromeSymp}
\bar{s}(\bar{\epsilon}) = H \odot \bar{\epsilon},
\end{equation}
where the symplectic product between the matrix and the error vector is performed row-wise, that is to say, each of the components of the syndrome is obtained by computing the symplectic product of the associated row of the PCM and the error operator.
As a direct result of this, we will be able to evaluate the performance of the target stabilizer code by conducting Monte Carlo simulations as it is detailed in Appendix \ref{app:montecarlo}. An error correction round of the system can be broken down into the following steps:
\begin{enumerate}
\item Generate a binary error pattern, $\bar{\epsilon}$, of length $2n$, with a probability distribution derived from the particular Pauli channel under consideration and the isomorphism $\Upsilon$. 
\item Calculate the syndrome associated to such an error $\bar{s}(\bar{\epsilon})$ by using the symplectic product in \eqref{syndromeSymp}.
\item Run the syndrome decoder that has been implemented for the particular QEC stabilizer code under study. Although the syndrome decoder depends on the stabilizer code under consideration \cite{QRM,bicycle,qldpc15,jgf,QCC,QTC,
EAQTC,toric,QEClidar,EAQECC,EAQIRCC,twirl6,MemQTC,catalytic}, its decoding algorithm is always entirely classical. The decoding algorithms that will be used in this thesis are described in Appendix \ref{app:decoding}. The quantum operation implemented for final recovery depends on the result of the syndrome decoder.
\item Check if the estimated recovery operation in this error correction round has been successful in its attempt to revert the channel error. This depends on whether error degeneracy is taken into account when defining the ``success" of the decoding scheme \cite{degenPRL,hardness}. This happens because for QECCs, the physical error that is sometimes estimated does not always match the error that actually occurred during transmission, however, the corresponding estimated logical error does indeed match the logical error produced by the actual physical error. This implies that even if the physical error is not correctly estimated, quantum information will be successfully recovered.
\end{enumerate}

Multiple error correction rounds will be required to obtain performance metrics with sufficient accuracy. These metrics will quantify the error correcting capability of the code under consideration. The most common metrics used to evaluate the performance of quantum error correcting codes are:
\begin{itemize}
\item Quantum Word Error Rate: it is the probability that at least one qubit in the block is incorrectly decoded, i.e., a decoding failure is accounted for if just one of the operators of the error operators is incorrectly estimated.
\item QuBit Error Rate: it is the probability that an individual qubit is incorrectly decoded, i.e. in each transmission round, each of the operators of the error operator is individually considered.
\end{itemize}

As explained in point $4$, these operational figures will depend on whether the physical errors or their associated logical errors are considered. Both figures of merit are used in the literature when assessing the performance of stabilizer codes  \cite{QRM,bicycle,qldpc15,jgf,QCC,QTC,EAQTC,toric,
QEClidar,EAQECC,EAQIRCC,twirl6,MemQTC,catalytic}, and in both cases the WER and QBER obtained without considering degeneracy will be upper bounds for the WER and QBER that could be achieved if degeneracy were exploited. The minimum number of rounds required to obtain precise estimates of these metrics is given by the theory of Monte Carlo simulations as explained in Appendix \ref{app:montecarlo}.

\subsection{Example: [[5,1,3]] stabilizer code}\label{sub:examp}
In this section we present a simple example of an error correction round, to understand how classical simulation of QECCs works.

Consider the $[[5,1,3]]$ stabilizer code defined by the following stabilizer generators:
\begin{equation}\label{eq:513stabs}
\begin{array}{cccccc}
\mathrm{S}_1 = & \mathrm{Z} & \mathrm{Z} & \mathrm{Z} & \mathrm{Z} & \mathrm{I} \\
\mathrm{S}_2 = & \mathrm{Z} & \mathrm{X} & \mathrm{Y} & \mathrm{I} & \mathrm{Z} \\
\mathrm{S}_3 = & \mathrm{X} & \mathrm{X} & \mathrm{X} & \mathrm{X} & \mathrm{I} \\
\mathrm{S}_4 = & \mathrm{X} & \mathrm{Y} & \mathrm{Z} & \mathrm{I} & \mathrm{X}
\end{array}.
\end{equation}

Figure \ref{fig:encodercirc513} shows a possible encoding circuit for this quantum error correction code. The encoding unitary is constructed by using elements of the Clifford group \cite{OzolsClifford}. We will not discuss here how the unitaries are obtained from the stabilizer generators, but its derivation is based on a Gaussian-like elimination of the parity check matrix \cite{QSC}. Note that for the current discussion of how classical simulation of this code works, the encoding circuit is unnecessary. However, the encoder circuit should be considered if one wants to introduce gate errors in the simulations\footnote{Note that the errors introduced by the gates will be circuit-dependent, and so obtaining the encoding circuit with the lower error probability is important for obtaining the best possible QECC.}, which is not our case.

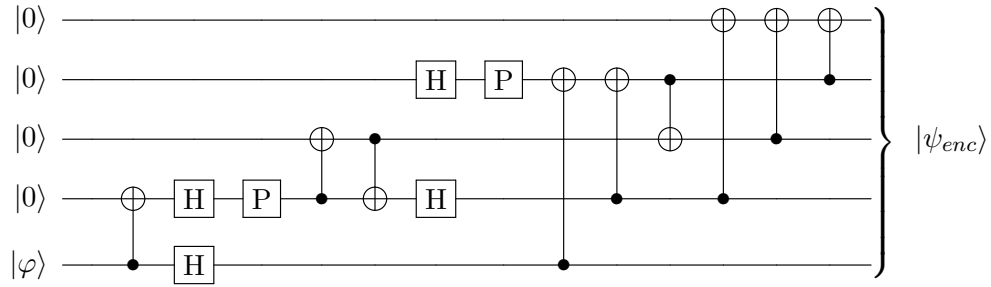
\begin{figure}[h]
\centering
\leavevmode
\Qcircuit @C=1em @R=1em {
& \lstick{\ket{0}} & \qw & \qw & \qw & \qw & \qw & \qw & \qw & \qw & \qw & \qw & \qw & \targ & \targ & \targ & \qw \\
& \lstick{\ket{0}} & \qw & \qw & \qw & \qw & \qw & \qw & \gate{\mathrm{H}} & \gate{\mathrm{P}} & \targ & \targ & \ctrl{1} & \qw & \qw & \ctrl{-1} & \qw \\
& \lstick{\ket{0}} & \qw & \qw & \qw & \qw & \targ & \ctrl{1} & \qw & \qw & \qw & \qw & \targ & \qw  & \ctrl{-2} & \qw & \qw \\
& \lstick{\ket{0}} & \qw & \targ & \gate{\mathrm{H}} & \gate{\mathrm{P}} & \ctrl{-1} & \targ & \gate{\mathrm{H}} & \qw & \qw & \ctrl{-2} & \qw & \ctrl{-3} & \qw & \qw & \qw \\
& \lstick{\ket{\varphi}} & \qw & \ctrl{-1} & \gate{\mathrm{H}} & \qw & \qw & \qw & \qw & \qw & \ctrl{-3} & \qw & \qw & \qw & \qw & \qw & \qw & \rstick{\raisebox{8.9em}{$\ket{\psi_{enc}}$\ }}
\gategroup{1}{17}{5}{17}{.8em}{\}}
}
\caption{Quantum circuit for encoding the $[[5,1,3]]$ stabilizer code. $\mathrm{H}$ and $\mathrm{P}$ stand for the Hadamard and Phase gates, respectively.}
\label{fig:encodercirc513}
\end{figure}

By using the map $\Upsilon$ defined in section \ref{sub:paulibinary}, the parity check matrix of this stabilizer code is
\begin{equation}\label{eq:513PCM}
H = \left(\begin{array}{ccccc|ccccc}
1 & 1 & 1 & 1 & 0 & 0 & 0 & 0 & 0 & 0   \\
1 & 0 & 1 & 0 & 1 & 0 & 1 & 1 & 0 & 0 \\
0 & 0 & 0 & 0 & 0 & 1 & 1 & 1 & 1 & 0  \\
0 & 1 & 1 & 0 & 0 & 1 & 1 & 0 & 0 & 1
\end{array}\right).
\end{equation}

The next step in the simulation is to generate binary error patterns and then obtain their syndromes so that error estimations can be obtained from these syndromes. The probability distributions used to simulate the channel are the distributions for the twirled approximated channels presented in sections \ref{sec:approximatedChannels} and \ref{sec:memory}, where the $n$-fold Pauli operators are mapped to binary vectors via the map $\Upsilon$. For our particular example we consider a memoryless depolarizing channel. In this example, a decoding round will be performed for each of the following error patterns:
\begin{itemize}
\item $\bar{\epsilon}_1 = \begin{pmatrix}
0 & 0 & 0 & 0 & 0 & 0 & 0 & 1 & 0 & 0
\end{pmatrix}$ which maps to the physical error $\mathcal{E}_1 = \mathrm{I}\otimes\mathrm{I}\otimes\mathrm{X}\otimes\mathrm{I}\otimes\mathrm{I}$.
\item $\bar{\epsilon}_2 = \begin{pmatrix}
1 & 0 & 0 & 0 & 1 & 1 & 1 & 0 & 0 & 0
\end{pmatrix}$ which maps to the physical error $\mathcal{E}_2 = \mathrm{Y}\otimes\mathrm{X}\otimes\mathrm{I}\otimes\mathrm{I}\otimes\mathrm{Z}$.
\item $\bar{\epsilon}_3 = \begin{pmatrix}
1 & 0 & 1 & 0 & 1 & 0 & 1 & 1 & 0 & 0
\end{pmatrix}$ which maps to the physical error $\mathcal{E}_3 = \mathrm{Z}\otimes\mathrm{X}\otimes\mathrm{Y}\otimes\mathrm{I}\otimes\mathrm{Z}$.
\end{itemize}

The corresponding error syndromes $\bar{s}(\bar{\epsilon}_j)$ are calculated from equation \eqref{syndromeSymp}, and are given by:
\begin{itemize}
\item $\bar{s}(\bar{\epsilon}_1) = H\odot \bar{\epsilon}_1 = \begin{pmatrix}
1 & 1 & 0 & 1
\end{pmatrix}$.
\item $\bar{s}(\bar{\epsilon}_2) = H\odot \bar{\epsilon}_2 = \begin{pmatrix}
0 & 1 & 1 & 1
\end{pmatrix}$.
\item $\bar{s}(\bar{\epsilon}_3) = H\odot \bar{\epsilon}_3 = \begin{pmatrix}
0 & 0 & 0 & 0
\end{pmatrix}$.
\end{itemize}

Once the error syndrome is found, it has to be fed to the decoding algorithm so that the best recovery operator may be found. Decoding stabilizer codes is substantially nuanced given that the optimal task differs from the classical decoding version due to degeneracy \cite{hardness}. In general, quantum decoding strongly depends on the code and the channel model under consideration. Nonetheless, QMLD is still a valid decoding algorithm, even though it is not optimal. Here, we will ignore degeneracy, and consequently, the QMLD syndrome decoder reduces to the lookup table shown in Table \ref{tab:513lookup}.

\begin{table*}[!h]
\centering
\begin{tabular}{|c|c|c|c|c|}
\hline
$\bar{s}$     & $\begin{pmatrix}0&0&0&0\end{pmatrix}$ & $\begin{pmatrix}0&0&0&1\end{pmatrix}$     & $\begin{pmatrix}0&0&1&0\end{pmatrix}$                      & $\begin{pmatrix}0&0&1&1\end{pmatrix}$     \\ \hline
$\mathcal{R}$ & $\mathrm{I}^{\otimes 5}$              & $\mathrm{I}^{\otimes 4}\otimes\mathrm{Z}$ & $\mathrm{I}^{\otimes 3}\otimes\mathrm{Z}\otimes\mathrm{I}$ & $\mathrm{Z}\otimes\mathrm{I}^{\otimes 4}$  \\ \hline
\end{tabular}
\begin{tabular}{|c|c|c|c|c|}
\hline
$\bar{s}$     & $\begin{pmatrix}0&1&0&0\end{pmatrix}$     & $\begin{pmatrix}0&1&0&1\end{pmatrix}$     & $\begin{pmatrix}0&1&1&0\end{pmatrix}$                                                                      & $\begin{pmatrix}0&1&1&1\end{pmatrix}$  \\ \hline
$\mathcal{R}$ & $\mathrm{I}^{\otimes 4}\otimes\mathrm{X}$ & $\mathrm{I}^{\otimes 4}\otimes\mathrm{Y}$ & $\mathrm{I}^{\otimes 2}\otimes\mathrm{Z}\otimes\mathrm{I}^{\otimes 2}$ & $\mathrm{I}\otimes\mathrm{Z}\otimes\mathrm{I}^{\otimes 3}$ \\ \hline
\end{tabular}
\begin{tabular}{|c|c|c|c|c|}
\hline
$\bar{s}$                           & $\begin{pmatrix}1&0&0&0\end{pmatrix}$                      & $\begin{pmatrix}1&0&0&1\end{pmatrix}$                      & $\begin{pmatrix}1&0&1&0\end{pmatrix}$                      & $\begin{pmatrix}1&0&1&1\end{pmatrix}$                                   \\ \hline
$\mathcal{R}$  & $\mathrm{I}^{\otimes 3}\otimes\mathrm{X}\otimes\mathrm{I}$ & $\mathrm{I}\otimes\mathrm{X}\otimes\mathrm{I}^{\otimes 3}$ & $\mathrm{I}^{\otimes 3}\otimes\mathrm{Y}\otimes\mathrm{I}$ & $\mathrm{I}^{\otimes 2}\otimes\mathrm{Y}\otimes\mathrm{I}^{\otimes 2}$  \\ \hline
\end{tabular}
\begin{tabular}{|c|c|c|c|c|}
\hline
$\bar{s}$& $\begin{pmatrix}1&1&0&0\end{pmatrix}$            & $\begin{pmatrix}1&1&0&1\end{pmatrix}$                                  & $\begin{pmatrix}1&1&1&0\end{pmatrix}$                      & $\begin{pmatrix}1&1&1&1\end{pmatrix}$  \\ \hline
$\mathcal{R}$ & $\mathrm{X}\otimes\mathrm{I}^{\otimes 4}$ & $\mathrm{I}^{\otimes 2}\otimes\mathrm{X}\otimes\mathrm{I}^{\otimes 2}$ & $\mathrm{I}\otimes\mathrm{Y}\otimes\mathrm{I}^{\otimes 3}$ & $\mathrm{Y}\otimes\mathrm{I}^{\otimes 4}$ \\
\hline
\end{tabular}
\caption{Lookup table used to perform syndrome decoding of the stabilizer code considered in the example. Since the channel is assumed to be a memoryless depolarizing channel and degeneracy is ignored, the most probable error for each syndrome will be the one with the lowest weight that produces said syndrome.}
\label{tab:513lookup}
\end{table*}

Note that the recovery operators (called representatives) in the lookup Table (table \ref{tab:513lookup}) are themselves Pauli operators. This is because $n$-fold Pauli operators are Hermitian and unitary matrices, so that $\mathrm{A}^2 = \mathrm{I},\forall \mathrm{A}\in\mathcal{P}_n$. Therefore, the $n$-Pauli representative operator for a particular syndrome $\bar{s}$ in the lookup table of the QMLD decoder is a Pauli error operator of syndrome $\bar{s}$ with the highest channel occurrence probability.  After knowing the syndrome, the decoder multiplies the corresponding representative operator by the noisy quantum information $|\bar{\psi}\rangle_\mathcal{N}$, so that the action of the channel may be reversed \footnote{Note that Pauli matrix product is the modulo $2$ sum for the binary representation under the map $\Upsilon$.}. For the current example, the decoder will render the following output states \ref{tab:513lookup}:
\begin{itemize}
\item $\bar{s}(\bar{\epsilon}_1) = \begin{pmatrix}
1 & 1 & 0 & 1
\end{pmatrix}\rightarrow \mathcal{R} = \mathrm{I}^{\otimes 2}\otimes\mathrm{X}\otimes\mathrm{I}^{\otimes 2}$. The resulting state after recovery will be $(\mathrm{I}^{\otimes 2}\otimes\mathrm{X}\otimes\mathrm{I}^{\otimes 2})|\bar{\psi}\rangle_\mathcal{N}=(\mathrm{I}^{\otimes 2}\otimes\mathrm{X}\otimes\mathrm{I}^{\otimes 2})\mathcal{E}_1|\bar{\psi}\rangle=|\bar{\psi}\rangle$. As a result, the code will succeed in correcting the channel error $\mathcal{E}_1$.
\item$\bar{s}(\bar{\epsilon}_2) = \begin{pmatrix}
0 & 1 & 1 & 1
\end{pmatrix}\rightarrow \mathcal{R} = \mathrm{I}\otimes\mathrm{Z}\otimes\mathrm{I}^{\otimes 3}$. The resulting state after recovery will be $(\mathrm{I}\otimes\mathrm{Z}\otimes\mathrm{I}^{\otimes 3})|\bar{\psi}\rangle_\mathcal{N}=(\mathrm{I}\otimes\mathrm{Z}\otimes\mathrm{I}^{\otimes 3})\mathcal{E}_2|\bar{\psi}\rangle=(\mathrm{I}\otimes\mathrm{Z}\otimes\mathrm{I}^{\otimes 3})(\mathrm{Y}\otimes\mathrm{X}\otimes\mathrm{I}^{\otimes 2}\otimes\mathrm{Z})|\bar{\psi}\rangle=(\mathrm{Y}\otimes\mathrm{Y}\otimes\mathrm{I}\otimes\mathrm{I}\otimes\mathrm{Z})|\bar{\psi}\rangle$. Therefore, for the channel error $\mathcal{E}_2$, decoding is unsuccessful. Note, that for this decoding round, a word error will be produced with $3$ qubits out of 5 in error.
\item $\bar{s}(\bar{\epsilon}_3) = \begin{pmatrix}
0 & 0 & 0 & 0
\end{pmatrix}\rightarrow \mathcal{R} = \mathrm{I}^{\otimes 5}$. The state after recovery will be $(\mathrm{I}^{\otimes 5})|\bar{\psi}\rangle_\mathcal{N}=(\mathrm{I}^{\otimes 5})\mathcal{E}_3|\bar{\psi}\rangle=(\mathrm{Z}\otimes\mathrm{X}\otimes\mathrm{Y}\otimes\mathrm{I}\otimes\mathrm{Z})|\bar{\psi}\rangle$. At first glance, it looks like the decoder has failed to recover the correct quantum information state, since $(\mathrm{Z}\otimes\mathrm{X}\otimes\mathrm{Y}\otimes\mathrm{I}\otimes\mathrm{Z})\neq \mathrm{I}^{\bigotimes 5}$. This would produce a word error, with $4$ qubits in error. However, in turns out that $(\mathrm{Z}\otimes\mathrm{X}\otimes\mathrm{Y}\otimes\mathrm{I}\otimes\mathrm{Z})$ lays in the codespace since is equal to the stabilizer generator $\mathrm{S}_2$. Therefore, no error will be produced since $(\mathrm{Z}\otimes\mathrm{X}\otimes\mathrm{Y}\otimes\mathrm{I}\otimes\mathrm{Z})|\bar{\psi}\rangle=|\bar{\psi}\rangle$. This last error $\mathcal{E}_3$ belongs to the class of errors known as degenerate errors. Hence, if degeneracy is being considered for the computation of the performance metrics, as it should for optimal decoding, this third decoding round will be considered to be successful.
\end{itemize}

Note that by randomly generating error operators according to the probability distributions of the approximated twirl channels and then performing the above decoding operations, we will be able to simulate the behavior of stabilizer codes under channels that realistically model the decoherence processes.

 \label{chapter2}  
\clearemptydoublepage
\part{Quantum Information Theory: Decoherence modelling and asymptotical limits} \label{part1}
\clearemptydoublepage
Decoherence effects experienced by the qubits in a quantum processor have been discussed in Chapter \ref{ch:preliminary}. As stated there, decoherence is often characterized by  the relaxation time ($T_1$) and the dephasing time ($T_2$) parameters. At the time of writing this thesis, the proposed quantum channel models in the literature assume that these parameters are fixed and invariant. This implies that the noise dynamics experienced by the qubits in a quantum device are identical for each quantum information processing task, independently of when the task is performed. As a consequence, the state-of-the-art quantum error correction codes have been constructed and evaluated under static noise channels \cite{QRM,bicycle,qldpc15,jgf,QCC,QTC,EAQTC,toric,
QEClidar,EAQECC,EAQIRCC,twirl6,MemQTC,catalytic}. However, recent experimental studies have shown that the decoherence parameters exhibit a time-varying behaviour \cite{decoherenceBenchmarking,klimov,fluctAPS,fluctApp,temperature,fluctGoogle}. These studies showed that $T_1$ and $T_2$ can experience time variations of up to $50\%$ of their mean value in the sample data, suggesting that qubit-based QECCs implemented in superconducting circuits will not perform, in average, as predicted when static channels models are considered. A possible solution to cope with channel fluctuations is to consider a worst parameters scenario. In this way one can assure that the QECCs will operate reliably for any realization. However, the resource consumption of QECCs designed in this manner would be far from optimal, since a higher number of physical qubits (due to the lower rate or longer blocklength needed to correct the worst case scenario) would be required.
Furthermore, the quantum capacity for the quantum channels discussed in Chapter \ref{ch:preliminary} were valid for the case of static quantum channels. Consequently, including time fluctuations into the framework of quantum channels implies that the asymptotically achievable limits for quantum error correction must be reinterpreted.

Therefore, in this part of the thesis, we propose time-varying quantum channels (TVQCs) for superconducting qubits, $\mathcal{N}(\rho,\omega,t)$, that correctly model the experimentally observed time variations in the parameters. We also focus on the important subject of how to characterize channel capacity under time variant conditions. To that end, we propose the quantum outage probability as the asymptotically achievable error rate for QECCs when operating over the proposed TVQCs. Finally, it is important to study how the error correction codes perform when they operate over the proposed TVQCs. This is why in this part we provide a qualitative analysis regarding the implications of such decoherence model on QECCs by simulating Kitaev toric codes and QTCs.

\chapter{Time-varying quantum channels} \label{cp3}

In this chapter we propose a time-varying quantum channel (TVQC) model that includes the experimentally observed fluctuations of the decoherence paramters of qubits. We present the experimental studies, based on superconducting qubits, where the fluctuations of the relaxation and dephasing time have been observed. We study the stochastic processes that describe the time dynamics of those parameters that define how the qubits suffer from errors. Based on such analysis, we provide the definition of the time-varying quantum channel model, $\mathcal{N}(\rho,\omega,t)$. The proposed model is characterized by the fact that the Kraus operators describing its dynamics are random matrices whose behaviour is governed by the random processes $T_1$ and $T_2$. We use the diamond norm \cite{diamondNat,FanoDiamond} distance between quantum channels $||\mathcal{N}_1 - \mathcal{N}_2||_\diamond$ to show that neglecting the fluctuating nature of $T_1$ and $T_2$ may result in an unrealistic model for quantum noise. By analyzing the statistical properties of the decoherence parameters, we assess when the TV noise channel model will substantially deviate from the static channel assumption.

\section{Experimental evidence}\label{sec:experFluct}
Here, we present some of the recent experiments found in the literature regarding decoherence parameter fluctuation in superconducting qubits. These experiments were the motivation for our time-varying quantum channel proposal.

\subsection*{Burnett et al.}
In \cite{decoherenceBenchmarking}, Burnett et al. studied the fluctuations of the decoherence parameters $T_1$ and $T_2$ for superconducting qubits. The authors mainly focused on the relaxation time. For doing so, they measured the value of these parameters a large number of counts ($N>1000$) into a calibration cycle of the system, which is usually in the order of several hours. In each calibration cycle, the superconducting processor is restarted and cooled down to the desired temperature. For their measurements, the authors considered two different qubits in the quantum processor.

\begin{figure}[!h]
\centering
\includegraphics[scale=1]{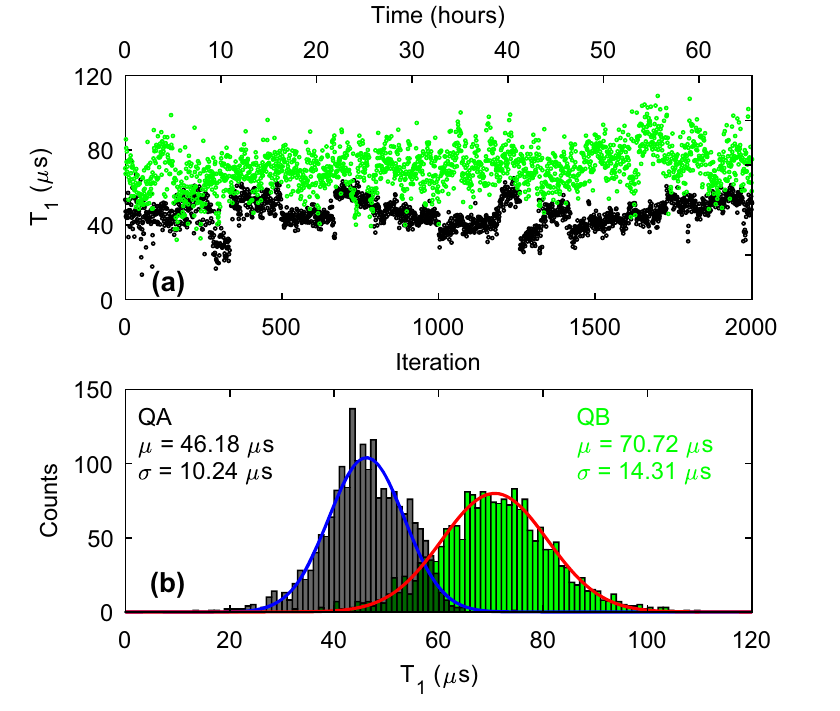}
\caption{(a) Measurements of $T_1$ for both qubits considered in \cite{decoherenceBenchmarking}, labelled as A (black) and B (green) (b) Fit of the data to Gaussian distributions. This data corresponds to cooldown $6$ of their dataset. Image from \cite{decoherenceBenchmarking}.}
\label{fig:bylanderT1fluct}
\end{figure}

In all the cooldowns they performed, substantial variations of the relaxation time were observed. The parameter showed variations of up to $50\%$ of its mean value in the sample data for every scenario they consider and for both of the probed qubits.

Figure \ref{fig:bylanderT1fluct} shows the data for the fluctuations of the relaxation time of one of the experiments they conducted for both of their qubits. It can be seen that the relaxation time fluctuates in a significant manner. The authors fit the data to Gaussian distributions. The stochastic process that describes the dynamics of the relaxation time was modelled as the sum of two Lorentzian processes with the contribution of a white noise process. Burnett et al. conducted Allan deviation and Welch-method spectral analyses in order to obtain such conclusions. We will further discuss this in the following sections.

The authors of \cite{decoherenceBenchmarking} made additional measurements to characterize the fluctuations of the dephasing time $T_2$. The qubit pure dephasing was studied by observing the fluctuations of the qubit frequency ($f_{01}$). The slow integrated fluctuations of the qubit frequency give rise to the pure dephasing time \cite{bylander}, $T_\phi$. It was concluded that the pure dephasing time also fluctuates substantially. 

Furthermore, the authors made Ramsey measurements in order to see the fluctuations of the dephasing time, $T_2$. Note that from the relation (refer to \eqref{reldeph})
\[\frac{1}{T_2}=\frac{1}{2T_1}+\frac{1}{T_{\phi}},\]
the pure dephasing parameter $T_{\phi}$ can also be obtained \cite{bylander}. 

\begin{figure}[!h]
\centering
\includegraphics[scale=1]{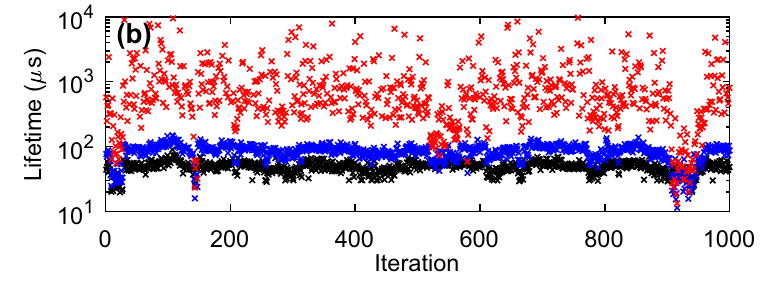}
\caption{Measurements of $T_1$ (black), $T_2$ (blue) and $T_\phi$ (red) for qubit A in cooldown 3. Image from \cite{decoherenceBenchmarking}.}
\label{fig:bylanderTphifluct}
\end{figure}

Figure \ref{fig:bylanderTphifluct} shows the fluctuation of the three parameters through time. It can be seen that the pure dephasing time also fluctuates substantially. The correlation between $T_1$ and $T_2$ is also observable due to the similar trend both have through the time. It is important to state that the authors in \cite{decoherenceBenchmarking} claim that their qubits are in the Ramsey limit, that is, $T_2\approx 2T_1$, so that $T_\phi >> 2T_1$.

\subsection*{Klimov et al.}
In \cite{klimov}, Klimov et al. studied the fluctuations of the relaxation time in superconducting qubits. Dephasing time $T_2$ fluctuation was not measured. The authors considered different qubit frequencies for the $T_1$ measurements, that is, they studied how the relaxation rate changes through time depending on which frequency the qubit is tuned at. They conducted experiments that last several hours by measuring the relaxation time every few minutes.

\begin{figure}[!h]
\centering
\includegraphics[scale=1]{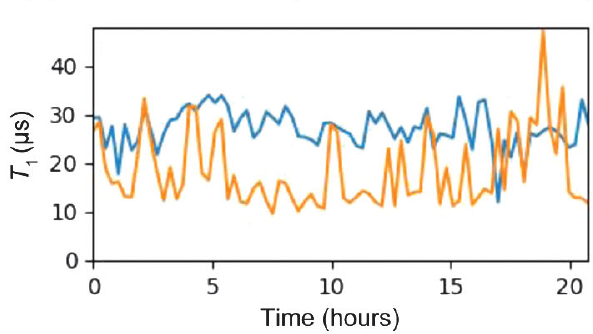}
\caption{Measurements of $T_1$ at different qubit frequencies $f_{01}$ with $100$ MHz of difference. Image from \cite{klimov}.}
\label{fig:klimovT1fluct}
\end{figure}

Figure \ref{fig:klimovT1fluct} shows measurements of $T_1$ at two different qubit frequencies that are separated by $100$ MHz. The figure clearly shows that the relaxation time of the superconducting qubit fluctuates substantially through time for both of the frequencies. The authors of \cite{klimov} concluded that $T_1$ can vary by up to an order of magnitude. Furthermore, they discuss that the observed fluctuations present an obstacle for quantum computing since the relaxation time limits the quantum gate fidelity (in general both relaxation and dephasing limit gate fidelity \cite{limitGate}) in quantum systems.

\subsection*{Schl\"or et al.}

In \cite{fluctAPS}, Schl\"or et al. thoroughly studied the fluctuations of the relaxation time and the dephasing time for superconducting transmon qubits. To that end, they measured the decoherence parameters of one superconducting qubit for several hours with a time resolution of $10$ seconds. The authors also measured the qubit frequency shift, $\Delta\omega_q$, which is also referred as the Ramsey detuning. As stated for the Burnett et al. experiments, the pure dephasing time is related to the frequency shifts of the qubits.

\begin{figure}[!h]
\centering
\includegraphics[scale=1]{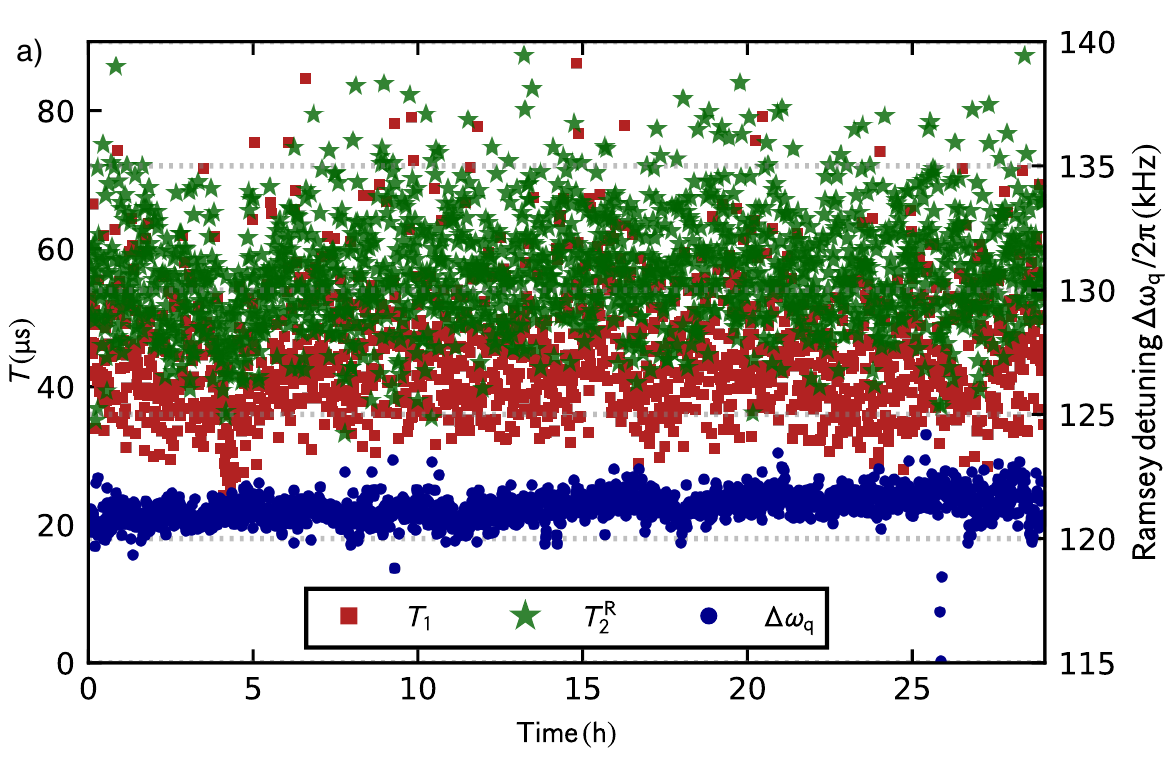}
\caption{Measurements of $T_1$ (red), $T_2$ (green) and $\Delta \omega_q$ (blue). The axis on the left refers to the times and the one in the right refers to the frequency. Image from \cite{fluctAPS}.}
\label{fig:schlorFluct}
\end{figure}

The results of their experiments are shown in Figure \ref{fig:schlorFluct}. It is again observed that the parameters of the qubit vary substantially through time, posing a problem for quantum computing. The authors concluded that the decoherence parameters of their superconducting transmon qubits vary up to an order of magnitude.

It should be mentioned that, in contrast to the qubits tested in \cite{decoherenceBenchmarking}, the qubits of Schl\"or et al. did not saturate the Ramsey limit. Thus, they studied the mechanism of pure dephasing for non-$T_1$-limited qubits and modelled the fluctuations of the qubit frequency as a random process based on the combination of a Lorentzian stochastic process and a $1/f^{1.1}$ process.

More experiments regarding parameter stability have been done in the literature. In \cite{fluctApp}, Stehli et al. measured the decoherence times of four superconducting qubits approximately every minute for several hours and observed substantial fluctuations of those parameters. In \cite{temperature}, Wang et al. studied the dependency of the parameters with the temperature of their superconducting qubits and discussed that their values also vary in a significant manner when the chips go to temperatures that cannot be considered negligible. Finally, in \cite{fluctGoogle}, Carrol et al. measured the relaxation time of a 20 qubit superconducting processor from IBM each day for several months and concluded that the relaxation time significantly fluctuates in a daily basis. It should be mentioned that in contrast to \cite{decoherenceBenchmarking,klimov,fluctAPS,fluctApp} where fluctuations that occur in a calibration cycle were studied, this last study is more related to qubit aging and variations from calibration to calibration.

Note that all the factors and studies that consider parameter stability are relevant for the construction of quantum computers. However, we will focus on modelling the fluctuations of the qubits in a calibration cycle of the hardware. Thus, the experimental studies on which we have based the derivation of our time-varying quantum channel model are \cite{decoherenceBenchmarking,klimov,fluctAPS,fluctApp}.

\section{Modelling $T_1$ and $T_2$}\label{sec:modelingt1t2}

In this section we study the behaviour of the stochastic processes that describe the fluctuations of the decoherence parameters of superconducting qubits.

\subsection*{Relaxation time}
In \cite{decoherenceBenchmarking}, the authors model $\{T_1(\omega,t)-\mu_{T_1}\}$ as the sum of three independent zero mean stationary Gaussian random processes, two of which have a Lorentzian power spectral density\footnote{Power Spectral Density, (PSD, the Fourier transform of the covariance function), $K(\tau)=\mathrm{E}[(T(\omega,t+\Delta t)-\mu_{T})(T(\omega,t)-\mu_{T})]$, where $\mathrm{E}$ denotes the expected value.} whereas the third process has a flat PSD of level $h_0$ (i.e., white noise). That is,
\begin{equation}\label{eq:t1fluct}
T_1(\omega,t) = \mu_{T_1}+\sum_{i=1}^2\text{Lor}_i(\omega,t) + N(\omega,t),
\end{equation}
 where $N(\omega,t)$ denotes the white Gaussian process and $\text{Lor}_i(\omega,t)$, $i=1,2$ represents the two Gaussian random processes with Lorentzian PSD. The authors in \cite{decoherenceBenchmarking} showed that modelling $T_1$ this way with fitted parameters exactly matched the experimental data.

The noise PSD of a Lorentzian process is given by
\begin{equation}\label{eq:lorentzianProfile}
S_i(f) = \frac{4(A^i)^2\tau^i_0}{1+(2\pi f\tau^i_0)^2}= 4(A^i)^2\tau^i_0 \left\lvert\frac{1}{1 + j2\pi f\tau^i_0} \right\rvert^2,
\end{equation}
where $A^i$ is the Lorentzian noise amplitude and $\tau^i_0$ is the characteristic timescale, $i=1,2$.

From the above expression, a Lorentzian gaussian random process can be simulated by filtering a white gaussian noise process \cite{Lorentzcomp}, with PSD level $4(A^i)^2\tau^i_0$, using a filter with a frequency transfer function $F = 1/(1 + j2\pi f\tau^i_0)$.
Therefore, since the Lorentzian random processes $\text{Lor}_i(\omega,t)$, $i=1,2$ are independent, their to the simulation of $T_1$ is obtained by adding the two filtered processes. Finally, by adding a white noise process of power $h_0$, one can simulate the whole $T_1$ process.

\begin{table*}[h!]
\centering
\resizebox{\textwidth}{!}{
\begin{tabular}{|cccccccc|}
\hline
\multicolumn{1}{|c}{Name} & \multicolumn{1}{c}{$\mu_{T_1}$ ($\mu$s)} & \multicolumn{1}{c}{$\sigma_{T_1}$ ($\mu$s)} & \multicolumn{1}{c}{$h_0$ ($\mu \mathrm{s}^2  \mathrm{Hz}^{-1}$)} & \multicolumn{1}{c}{$A^1$ ($\mu$s)} & \multicolumn{1}{c}{$1/\tau_0^1$ ($ \mu$Hz)} & \multicolumn{1}{c}{$A^2$ ($\mu$s)} & \multicolumn{1}{c|}{$1/\tau_0^2$ ($\mu$Hz)} \\ \hline
QA\_C5 & $44.49$ & $11.7$ & $ 2\times 10^{-3}$ & $5.2$ & $142.9$ & $2.6$ & $83.3$\\
QB\_C5 & $81.63$ & $17.01$ & $ 1.4\times 10^{-2}$ & $3.2$ & $1000$ & $6.6$ & $90.9$\\
QA\_C6                     & $46.64$                                    & $10.24$                                       & $1.2\times10^{-3}$                          & $4.5$                                & $333.3$                                       & $1.8$                                & $71.4$                                       \\
       QB\_C6                    &                $71.22$                            &                                              $14.31$ &                $5.7\times 10^{-3}$                             &                                     $4.2$  &                      $1111.1$                         &                                     $2.2$ &                        $76.9$
 \\\hline
\end{tabular}}
\caption{Parameters to model $T_1$ fluctuations in the superconducting qubits of \cite{decoherenceBenchmarking}. $\mu_{T_1}$ refers to the mean relaxation time and $\sigma_{T_1}$ to its standard deviation. The nomenclature for each case is the same as in \cite{decoherenceBenchmarking} and is interpreted as QX\_CY, being X the superconducting qubit and Y the cooldown (in \cite{decoherenceBenchmarking} several cooldowns were performed and data about $T_1$ was measured in each of them) of the quantum chip.}
\label{tab:variationParameters}
\end{table*}

Table \ref{tab:variationParameters} shows the values obtained for the parameters of the noise processes in \eqref{t1fluct} from experimental data measured for the superconducting qubits used in \cite{decoherenceBenchmarking}. Note from this table that the ratio between the PSD level at $f=0$ of the Lorentzian processes and the level of the white noise, $4(A^i)^2 \tau^i/h_0$,  is of the order of $10^{8}$. Therefore, the contribution of the white noise is negligible when compared to the contribution of the $\text{Lor}_i(\omega,t)$ processes. This means that the PSD bandwidth of the process $T_1$ is approximately given by $\mathrm{BW}=\max_{i=1,2}\{1/\tau_0^i\}$. This also implies that its covariance function $K^1(\Delta t)$ is approximately constant for $|\Delta t|<<T_\mathrm{c} = \frac{1}{\mathrm{BW}}$, with $\Delta t$ in the order of microseconds and $T_\mathrm{c}$ in the order of minutes\footnote{As is conventional in the context of fading channels \cite{fading}, we refer to $T_c$ as coherence time.}. Taking into account that the processing time of current state-of-the-art quantum algorithms\footnote{We refer to the simple NISQ algorithms that may be implemented in run-of-the-mill hardware or to a single round of an error correcting code. The implementation of algorithms such as the Shor algorithm would take longer processing times. Note that without error correction, algorithms with processing times higher than the decoherence times of the qubits will not work.}, $t_{\mathrm{algo}}$, is around a few microseconds, it is reasonable to assume that during the execution of a quantum algorithm, the realization of the stochastic process $T_1$, remains constant. In other words, $T_1(\omega,t)$ can be modelled as a random variable $T_1(\omega)=T_1(\omega,t)|_{t=0}$ where $t=0$ has been chosen without any loss of generality since the random process $T_1$ is stationary. Each realization of such random variable will remain constant for $t\in[0,t_{\mathrm{algo}}],t_{\mathrm{algo}} << T_\mathrm{c}$. Specifically, $T_1(\omega)$ can be considered to be a truncated Gaussian random variable (the truncation is necessary as negative $T_1$ values do not make physical sense) with probability density function $T_1(\omega)\sim \mathcal{GN}_{[0,\infty]}(\mu_{T_1},\sigma^2_{T_1})$, where the variance is obtained, based on the parameters of table \ref{tab:variationParameters}, by integrating the PSD of \eqref{t1fluct} in the frequency band $[-\mathrm{BW,BW}]$. We use $\mathcal{GN}_{[a,b]}(\mu,\sigma^2)$ to denote a truncated normal random variable with mean $\mu$ and variance $\sigma^2$, truncated in the time interval $[a,b]$.

\subsection*{Dephasing time}
Recall that the qubits studied in \cite{decoherenceBenchmarking} were showed to be $T_1$-limited ($T_2\approx 2T_1$). This fact implies that the pure dephasing noise will be negligible since from equation \eqref{reldeph}, $T_2\approx 2T_1 \rightarrow 1/T_\phi \approx 0$. Therefore, for $T_1$-limited superconducting qubits, neglecting the contribution of $T_\phi$ is a reasonable assumption and Burnett et al. did not discuss any further the nature of the fluctuations of $T_\phi$.

However, most of the state-of-the-art quantum processors cannot achieve the $T_2\approx 2T_1$ limit. In these cases, the existence of pure dephasing must be taken into account. To that end, the authors of \cite{SchlorPhD,fluctAPS} studied the mechanism of pure dephasing for non-$T_1$-limited and modelled the fluctuations of the qubit frequency as a random process based on the combination of a Lorentzian stochastic process and a $1/f^{1.1}$ process\footnote{This refers to a stochastic process showing a $S(f)\propto 1/f^{1.1}$ PSD.}. It is easy to find that the coherence time of the this stochastic processes is in the order of magnitude of minutes \cite{fluctAPS} ($\mathrm{BW}\approx 1$ mHz set by the Lorentzian process \cite{fluctAPS}). Therefore, following the same reasoning used for the relaxation time fluctuations, and taking into account that pure dephasing time arises from the qubit frequency process, it is reasonable to model $T_\phi$ as a random variable $T_\phi(\omega)=T_\phi(\omega,t)\rvert_{t=0}$ whose realizations will remain constant for $t\in[0,t_{\mathrm{algo}}],t_{\mathrm{algo}}<<T_\mathrm{c}$. Moreover, $T_\phi(\omega)$ can be considered to be a truncated Gaussian random variable with probability density function $T_\phi(\omega)\sim \mathcal{GN}_{[0,\infty]}(\mu_{T_\phi},\sigma^2_{T_\phi})$. Note that by modelling $T_\phi(\omega)$ this way, and $T_1$ as $T_1(\omega)\sim \mathcal{GN}_{[0,\infty]}(\mu_{T_1},\sigma^2_{T_1})$, the depahsing time $T_2(\omega)$ is now obtained  from \eqref{reldeph}.

\section{Time-varying quantum channels}\label{res:TVQC}
We define a time-varying quantum channel, $\mathcal{N}(\rho,\omega,t)$, as
\begin{equation}\label{eq:TVQCgen}
\mathcal{N}(\rho,\omega,t) = \sum_k E_k(\omega,t) \rho E_k^\dagger(\omega,t),
\end{equation}
where the $E_k(\omega,t)$ linear operators are the so-called Kraus operators of the operator-sum representation of a quantum channel. Note that the matrices $\{E_k(\omega,t)\}$ are continuous-time random processes that will determine the time variations of the TVQC. Consequently, the channel dynamics will experience temporal fluctuations determined by the random processes $\{E_k(\omega,t)\}$.

As described in Chapter \ref{ch:preliminary}, decoherence arises from a wide range of physical processes involved in the interaction of the qubits with their environment. Nonetheless, we have seen that a fairly complete mathematical model of these harmful noise effects can be obtained by combining the amplitude damping channel (AD), and the phase damping or dephasing channel (PD) into  the combined amplitude and phase damping channel (APD), $\mathcal{N}_{\mathrm{APD}}$. Therefore, we propose the time-varying amplitude damping channel (TVAD), the time-varying dephasing channel (TVPD) and the time-varying combined amplitude and phase damping channel (TVAPD) as the TVQCs associated to decoherence processes whose parameters fluctuate through time. The Kraus operators of these time-varying quantum channels are:
\begin{itemize}
\item Time-varying amplitude damping channels (TVAD), $\mathcal{N}_{\mathrm{AD}}(\rho,\omega,t)$,
\begin{equation}
\begin{split}
&E_0(\omega,t) = \begin{pmatrix}
1 & 0 \\
0 & \sqrt{1-\gamma(\omega,t)}
\end{pmatrix} \text{ and } \\ &
E_1(\omega,t) = \begin{pmatrix}
0 & \sqrt{\gamma(\omega,t)} \\
0 & 0
\end{pmatrix},
\end{split}
\end{equation}
\item Time-varying dephasing channels (TVPD), $\mathcal{N}_{\mathrm{PD}}(\rho,\omega,t)$,
\begin{equation}
\begin{split}
&
E_0(\omega,t) = \begin{pmatrix}
1 & 0 \\
0 & \sqrt{1-\lambda(\omega,t)}
\end{pmatrix} \text{ and } \\&
E_1(\omega,t) = \begin{pmatrix}
0 & 0 \\
0 & \sqrt{\lambda(\omega,t)}
\end{pmatrix}.
\end{split}
\end{equation}
\item Time-varying combined amplitude and phase damping channel (TVAPD), $\mathcal{N}_{\mathrm{APD}}(\rho,\omega,t)$,
\begin{equation}
\begin{split}
& E_0(\omega,t) = \begin{pmatrix}
1 & 0 \\
0 & \sqrt{1-\gamma(\omega,t) - (1-\gamma(\omega,t))\lambda(\omega,t)}
\end{pmatrix}, \\
& E_1(\omega,t) = \begin{pmatrix}
0 & \sqrt{\gamma(\omega,t)}\\
0 & 0
\end{pmatrix}\text{ and } \\
& E_2(\omega,t) = \begin{pmatrix}
0 & 0 \\
0 & \sqrt{(1-\gamma(\omega,t))\lambda(\omega,t)}
\end{pmatrix},
\end{split}
\end{equation}
\end{itemize}
where the damping $\{\gamma(\omega, t)\}$ and scattering $\{\lambda(\omega,t)\}$ stochastic processes are functions of the qubit relaxation time $\{T_1(\omega,t)\}$ and the qubit dephasing time $\{T_2(\omega,t)\}$ stochastic processes. They are given by

\begin{equation}\label{eq:gammatime}
\gamma(\omega,t) = 1 - \mathrm{e}^{-\frac{t}{T_1(\omega,t)}} \text{ and}
\end{equation}

\begin{equation}\label{eq:lambdatime}
\lambda(\omega,t) = 1 - \mathrm{e} ^{\frac{t}{T_1(\omega,t)} - \frac{2t}{T_2(\omega,t)}}.
\end{equation}

$\{T_1(\omega,t)\}$ and $\{T_2(\omega,t)\}$ are modelled as wide-sense stationary random processes with means $\mu_{T1}$ and $\mu_{T2}$, respectively, and with experimentally measured PSDs as described in the previous section. As noted before, the time realizations of those processes can be considered constant for time intervals less than their coherence times, $T_c = 1/\mathrm{BW}$. Therefore, the realizations of $T_1$ and $T_\phi$ can be modelled by truncated Gaussian random variables. In this way, the TVQCs will remain constant for those time intervals, and their Kraus operators will be fixed by the realizations of the relaxation and pure dephasing time random variables.

\subsection{Twirled approximations of time-varying quantum channels}

As discussed in Chapter \ref{ch:preliminary}, twirling an arbitrary quantum channel provides us with approximated quantum channels that can be efficiently implemented in classical computers. 

The time-varying Pauli twirl approximation (TVPTA), $\mathcal{N}_{\mathrm{PTA}}(\rho,\omega,t)$, is the Pauli channel \cite{twirl6} obtained by twirling a time-varying quantum channel by the $n$-fold Pauli operators $\mathcal{P}_n$. Twirling the TVAD channel will lead to the Pauli channel (TVADPTA) described by the probabilities that each of the Pauli matrices has of taking place. Note that in this context these probabilities are realizations of the random processes \cite{twirl6}:
\begin{equation}\label{eq:TVADPTA}
\begin{split}
& p_\mathrm{I}(\omega,t) = 1 - p_\mathrm{x}(\omega,t) - p_\mathrm{y}(\omega,t) - p_\mathrm{z}(\omega,t), \\
& p_\mathrm{x}(\omega,t) = p_\mathrm{y}(\omega,t) = \frac{1}{4}(1 - \mathrm{e}^{-\frac{t}{T_1(\omega,t)}})\text{ and} \\
& p_\mathrm{z}(\omega,t) = \frac{1}{4}(1 + \mathrm{e}^{-\frac{t}{T_1(\omega,t)}} - 2\mathrm{e}^{-\frac{t}{2T_1(\omega,t)}}).
\end{split}
\end{equation}

For the TVAPD channel, the TVAPDPTA approximation is described by the realizations of the following stochastic processes for each of the Pauli matrices \cite{twirl6}
\begin{equation}\label{eq:TVPTA}
\begin{split}
& p_\mathrm{I}(\omega,t) = 1 - p_\mathrm{x}(\omega,t) - p_\mathrm{y}(\omega,t) - p_\mathrm{z}(\omega,t), \\
& p_\mathrm{x}(\omega,t) = p_\mathrm{y}(\omega,t) = \frac{1}{4}(1 - \mathrm{e}^{-\frac{t}{T_1(\omega,t)}})\text{ and} \\
& p_\mathrm{z}(\omega,t) = \frac{1}{4}(1 + \mathrm{e}^{-\frac{t}{T_1(\omega,t)}} - 2\mathrm{e}^{-\frac{t}{T_2(\omega,t)}}),
\end{split}
\end{equation}
where, once again, $T_1(\omega,t)$ and $T_2(\omega,t)$ are stochastic processes.

Another twirled channel of interest is the time-varying Clifford twirl approximation (TVCTA), $\mathcal{N}_{\mathrm{CTA}}(\rho,\omega,t)$ , which for the TVAD channel will be a depolarizing channel with depolarizing parameter \cite{twirl3}
\begin{equation}\label{eq:TVADCTA}
p(\omega,t) = \frac{3}{4} - \frac{1}{4}\mathrm{e}^{-\frac{t}{T_1(\omega,t)}} - \frac{1}{2}\mathrm{e}^{-\frac{t}{2T_1(\omega,t)}},
\end{equation}
and for the TVAPD channel a depolarizing channel with depolarizing parameter
\begin{equation}\label{eq:TVAPDCTA}
p(\omega,t) = \frac{3}{4} - \frac{1}{4}\mathrm{e}^{-\frac{t}{T_1(\omega,t)}} - \frac{1}{2}\mathrm{e}^{-\frac{t}{T_2(\omega,t)}},
\end{equation}
where, once more, $T_1(\omega,t)$ and $T_2(\omega,t)$ are stochastic processes.

It should be pointed out that although the derived time-variant channel models are based on the statistical characterization of the parameters $T_1$ and $T_2$ from the experimental results in \cite{SchlorPhD,decoherenceBenchmarking,fluctAPS,fluctApp}, they are also applicable to any superconducting quantum processor whose decoherence parameters exhibit slow fluctuations. Moreover, the model is also applicable to any quantum-coherent two-level system that presents similar time dependencies, regardless of its physical implementation. For example, the photonic systems such as the one studied in \cite{bylanderPhoton} with strong pure dephasing fluctuations.

\section{$\diamond$-norm distance numerical analysis}

In this section we assess the difference between the widely employed static channel models and the proposed TVQCs. We make this comparison by using a metric known as the diamond norm distance (see Appendix \ref{app:diamond}) $||\mathcal{N}(\mu_{T_1},\mu_{T_2}) - \mathcal{N}(\rho,\omega,t)||_\diamond$. More concretely, we perform an extensive numerical analysis over the parameter space, for the different types of TVQCs discussed in the previous section (the particular results for the superconducting qubits of \cite{decoherenceBenchmarking} are presented later). We will begin by studying the mean value of such metric for different scenarios and then proceed to study its variability in order to completely understand how the static and TV channels differ.

The parameters required in the computation of the diamond norm distance between the static and time-varying channels are the mean and the standard deviation of the random variables $T_1(\omega)$, $T_2(\omega)$ and $T_{\phi}(\omega)$. Note from \eqref{reldeph} that $T_2(\omega)$ is a function of the independent random variables $T_1(\omega)$ and $T_{\phi}(\omega)$. Table \ref{tab:numericalsim} shows these parameters. The coefficients of variation, $c_\mathrm{v}=\sigma/\mu$, defined as the ratios between the standard deviations and means of these random variables are also shown in the table.
Note that the other parameters defining the stochastic random processes $T_1(t,\omega)$ and $T_\phi(t,\omega)$ are not involved in the computation of the mean diamond norm distance. However, they were used to derive the parameters in Table \ref{tab:numericalsim} and the probability distributions of random variables $T_1(\omega)$ and $T_{\phi}(\omega)$.

\begin{table*}[h!]
\centering
\resizebox{\textwidth}{!}{
\begin{tabular}{|cccccccc|}
\hline
\multicolumn{1}{|c}{Scenario} & \multicolumn{1}{c}{$\mu_{T_1}$} & \multicolumn{1}{c}{$\sigma_{T_1}$} & \multicolumn{1}{c}{$c_\mathrm{v}(T_1)$} & \multicolumn{1}{c}{$\mu_{T_\phi}$} & \multicolumn{1}{c}{$\sigma_{T_\phi}$} & \multicolumn{1}{c}{$c_\mathrm{v}(T_\phi)$} & \multicolumn{1}{c|}{$\mu_{T_2}$} \\ \hline
$T_1$-limited & $100$ & $\{1,10,25\}$ & $ \{1,10,25\}\%$ & - & - & - & -\\
$T_1\approx T_2$ & $100$ & $\{1,10,25\}$ & $ \{1,10,25\}\%$ & $200$ & $\{2,20,50\}$ & $\{1,10,25\}\%$ & $100$\\
$T_2$-dominated                     & $100$                                    & $\{1,10,25\}$                                       & $\{1,10,25\}\%$                          & $100$                                & $\{1,10,25\}$                                       & $\{1,10,25\}\%$                                & $66.67$
 \\\hline
\end{tabular}}
\caption{Parameters used for different superconducting scenarios. We consider the following scenarios for superconducting qubits: $T_1$-limited \cite{decoherenceBenchmarking}, $T_1\approx T_2$ \cite{melbourne} and $T_2$-dominated ($T_2<T_1$) \cite{yorktown}. $c_\mathrm{v}=\sigma / \mu$ refers to the coefficient of variation of the random variables. $\mu_{T_2}$ is calculated via expression \eqref{reldeph}.}
\label{tab:numericalsim}
\end{table*}

The computation of the static channels, denoted as $\mathcal{N}(\mu_{T_1},\mu_{T_2})$, is done by substituting the parameters $T_1$ and $T_2$ in the Kraus operators of the channels (refer to section \ref{res:TVQC}) by the mean $\mathrm{E}\{T_1\} = \mu_{T_1}$ and $\mathrm{E}\{T_2\} = \mu_{T_2}$ of the random variables $T_1(\omega)$ and $T_{2}(\omega)$.
On the other hand, the realizations for the TV quantum channels are obtained by replacing the parameters $T_1$ and $T_2$ in the Kraus operators by the realizations of the random variables $T_1(\omega)$ and $T_2(\omega)$, respectively.

\subsection{Mean value of the diamond norm distance}\label{sub:meandiam}
In the previous sections we have analyzed how the decoherence parameters of superconducting qubits fluctuate through time, and propose TV channel models that describe these varying conditions. We concluded that it is reasonable to assume that these parameters do not fluctuate over the runtime of a QECC encoding-decoding round (or a short quantum algorithm), but they do change from round to round if the rounds are sufficiently separated in time. Note that running error corrected quantum algorithms requires a large number of QECC rounds. For example, the quantum algorithm in \cite{rsaRounds} requires 25 billion surface code cycles in an 8 hour runtime.

Let $L$ be the number of rounds produced during the quantum processor operation, and consider round $k$, $k\in\{1,\ldots L\}$. If $T_1^{(k)}(\omega)$ and $T_2^{(k)}(\omega)$ denote the relaxation and dephasing times associated to this round, then the sequence of random variables $\{T_j^{(k)}(\omega)\}_{k=1}^L$, $j=1,2$ would be independent and identically distributed, as explained in section \ref{sec:modelingt1t2}. Note that the duration of a round is what we have previously called the algorithm time, $t_{\mathrm{algo}}$, and its value is upperbounded by $\min\{\mu_{T_1},\mu_{T_2}\}$ since for times longer than those the superconducting qubit will be very likely in its equilibrium state and, therefore, will be useless as a resource. Later in this section, we will illustrate that QECCs do in fact have cycle times lower than $\min\{\mu_{T_1},\mu_{T_2}\}$ by obtaining the runtime of a Shor code round in the IBM\_Q\_16\_Melbourne processor. To compute the TV channel operator at round $k$, one needs to obtain realizations $t_1^{(k)}$ and $t_2^{(k)}$, of the two random variables $T_1^{(k)}(\omega)$ and $T_2^{(k)}(\omega)$, respectively, and compute the corresponding  Kraus operators at time $t=t_{algo}$. We denote by $\mathcal{N}(\rho,t_1^{(k)},t_2^{(k)},t=t_{algo})$ this channel operator at round $k$. On the other hand, the static channel operator, sets $T_1^{(k)}(\omega)$ and $T_2^{(k)}(\omega)$ equal to their mean values $\mu_{T_1}$ and $\mu_{T_2}$, and therefore they will be constant for all rounds.

We are interested in computing the average value of the diamond norm distance for all rounds. That is,
\begin{equation}\label{eq:aver}
\frac{1}{L}\sum_{k=1}^L||\mathcal{N}(\mu_{T_1},\mu_{T_2},t=t_{\mathrm{algo}})-\mathcal{N}(\rho,t_1^{(k)},t_2^{(k)},t=t_{\mathrm{algo}})||_\diamond.
\end{equation}
Note that as $L$ goes to infinity, \eqref{aver} converges to the expected value $\mathrm{E}\{||\mathcal{N}(\mu_{T_1},\mu_{T_2},t=t_{\mathrm{algo}})-\mathcal{N}(\rho,T_1^{(k)}(\omega),T_2^{(k)}(\omega),t=t_{\mathrm{algo}})||_\diamond\}$.

Figure \ref{fig:diamondMean} shows the average diamond norm distance versus algorithm time, $t$, for the different superconducting qubit scenarios of table \ref{tab:numericalsim}. The range of values of $t$ is set to $t\in[0,\min\{\mu_{T_1},\mu_{T_2}\}]$ and  the number of rounds to $L=20000$. Note that the algorithm time has been normalized with respect to $\min\{\mu_{T_1},\mu_{T_2}\}$. The reason behind this normalization is to decouple the simulations from the fact that qubits with longer decoherence parameters will take longer to have significant mean diamond norms. In this way, it is easier to compare in a single figure the effect of time-variation for different qubit characteristics.

\begin{figure}[h!]
\centering
\includegraphics[width=\linewidth]{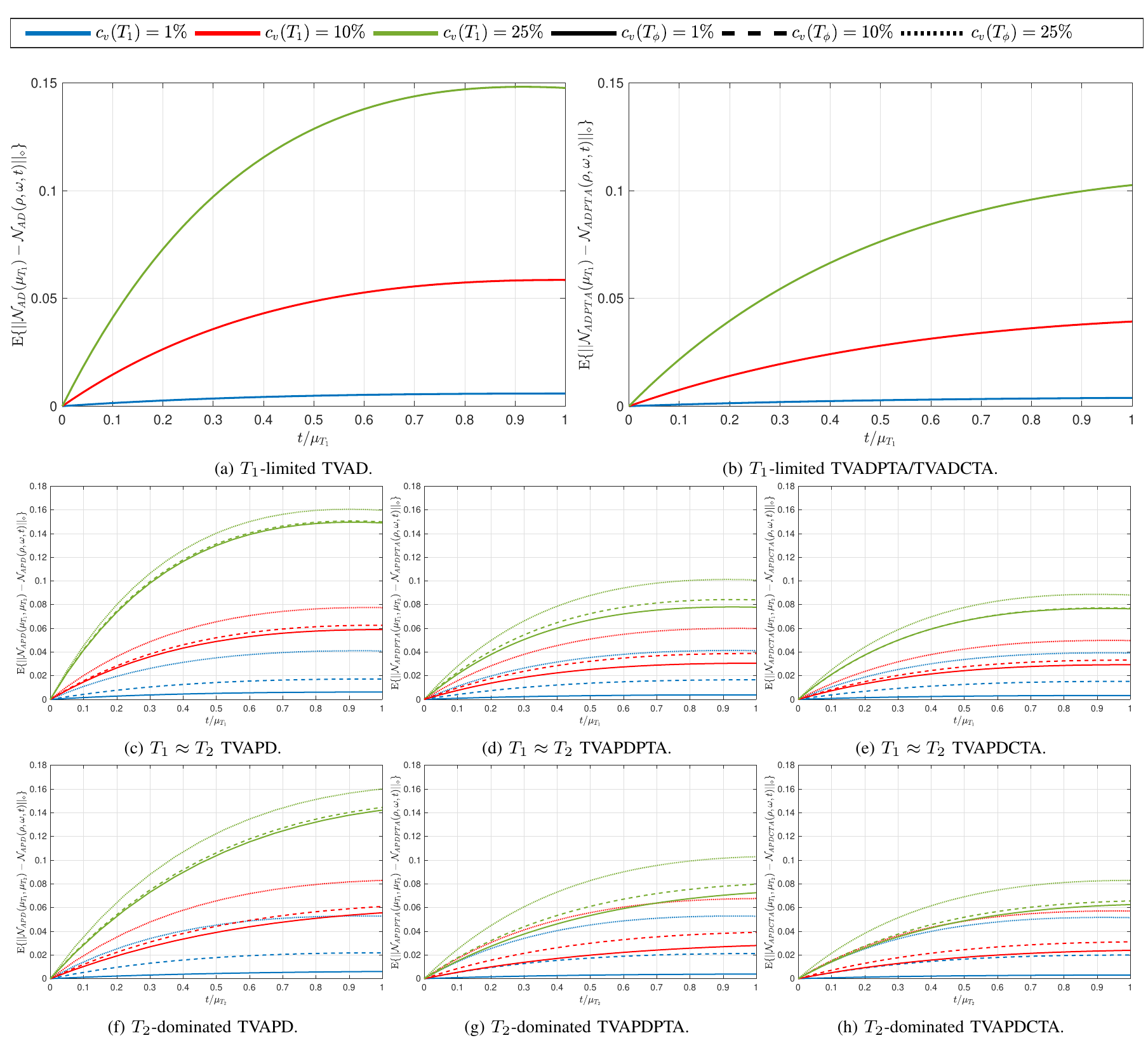}
\caption{Mean diamond norm distance $\mathbf{\mathrm{E}\{||\mathcal{N}(\mu_{T_1},\mu_{T_2})-\mathcal{N}(\rho,\omega,t)||_\diamond\}}$ between the static and the proposed time-varying quantum channels. The simulation for the calculation is done by using the parameters for the different scenarios for superconducting qubits in table \ref{tab:numericalsim}. \textbf{a} $T_1$-limited TVAD. \textbf{b} $T_1$-limited TVADPTA/TVADCTA. \textbf{c} $T_1\approx T_2$ TVAPD. \textbf{d} $T_1\approx T_2$ TVAPDPTA. \textbf{e} $T_1\approx T_2$ TVAPDCTA. \textbf{f} $T_2$-dominated TVAPD. \textbf{g} $T_2$-dominated TVAPDPTA. \textbf{h} $T_2$-dominated TVAPDCTA.}
\label{fig:diamondMean}
\end{figure}

Observe that for all the scenarios considered\footnote{It can be proven that the diamond norm distance between ADPTA and ADCTA channels obtained from AD channels with parameters $\gamma_1,\gamma_2$ is the same, i.e., $||\mathcal{N}_{\mathrm{ADPTA}}(\gamma_1)-\mathcal{N}_{\mathrm{ADPTA}}(\gamma_2)||_\diamond = ||\mathcal{N}_{\mathrm{ADCTA}}(\gamma_1)-\mathcal{N}_{\mathrm{ADCTA}}(\gamma_2)||_\diamond$ (proof in Appendix \ref{app:diamond}), and thus Figure \ref{fig:diamondMean}b will be the exact same for the mean diamond norm distance $\mathrm{E}\{||\mathcal{N}_{\mathrm{ADCTA}}(\mu_{T_1})-\mathcal{N}_{\mathrm{ADCTA}}(\rho,\omega,t)||_\diamond\}$ of the TVADCTA channels.}, the mean diamond norm distance increases with the normalized algorithm time. An interesting conclusion that can be drawn from Figure \ref{fig:diamondMean} is that the impact that time fluctuations of the decoherence parameters have on the behavior of the mean diamond norm distance can be singled out by the coefficients of variation, $c_\mathrm{v}(T_1)$ and $c_\mathrm{v}(T_\phi)$.
The higher the coefficients of variation are, the higher the mean diamond norm distance will be. Therefore, if $c_\mathrm{v}$ are sufficiently low, a static approach might accurately represent the decoherence processes suffered by the superconducting qubits. Figure \ref{fig:diamondMean} also shows that the fluctuations of $T_1$ are more relevant than the fluctuations of $T_\phi$. This can be observed for the APD channels and their twirled approximations since $c_\mathrm{v}(T_1)$ has more impact on the mean diamond norm distance increase than $c_\mathrm{v}(T_\phi)$. This is not surprising considering the fact that $T_2$ is a function not only of $T_\phi$ but also of $T_1$ (refer equation to \eqref{reldeph}), so its fluctuations affect both relaxation and dephasing.

Given that QECCs are used to protect quantum information from decoherence, it is desirable for one encoding-decoding round of these quantum information processing tasks to be short, since the longer they take the noisier the channel becomes. Therefore, it is convenient to have algorithm times $t<< \min\{\mu_{T_1},\mu_{T_2}\}$, which also render small mean diamond norm distances, meaning that in average the use of static channel models may be reasonable. However, this last conclusion is misleading since as shown later in this section, even if the average norm is small, the spread of this metric for different realizations is considerable when the coefficients of variation are large.

Next, we justify that bounding the algorithm times for QECCs by $\min\{\mu_{T_1},\mu_{T_2}\}$ is a reasonable assumption for current quantum processors.
\begin{table}[!ht]
\centering
\begin{tabular}{|cccc|}
\hline
\multicolumn{1}{|c}{$t_{\mathrm{1Q}}$ (ns)} & \multicolumn{1}{c}{$t_{\mathrm{2Q}}$ (ns)} & \multicolumn{1}{c}{$t_\Delta$ (ns)} & \multicolumn{1}{c|}{$t_{\mathrm{meas}}$ (ns)} \\ \hline
$100$ & $\approx 500$ & $20$ & $\approx 300$  \\
\hline
\end{tabular}
\caption{IBM\_Q\_16\_Melbourne values for 1-qubit gate time ($t_{\mathrm{1Q}}$), 2-qubit gate time ($t_{\mathrm{2Q}}$), delay time (buffer) between operations ($t_\Delta$) and measurement time ($t_{\mathrm{meas}}$) \cite{melbourne}. }
\label{tab:IBMq16}
\end{table}

To that end, let us consider one of the simplest QECCs known: the Shor code \cite{decoherenceShor}. Note that in \cite{SchlorPhD,decoherenceBenchmarking,klimov,fluctAPS,fluctApp,fluctGoogle}, which are used to obtain the parameters that describe the fluctuations of $T_1$ and $T_2$, no gates are implemented.  Thus, we select a state-of-the-art experimental quantum computer, the IBM\_Q\_16\_Melbourne, to get the gate times \cite{melbourne}. IBM quantum computers use superconducting technology for their qubits so we expect the behaviour of their decoherence parameters to be similar to the ones considered in \cite{decoherenceBenchmarking,fluctAPS}. From the data in table \ref{tab:IBMq16} for the IBM\_Q\_16\_Melbourne quantum computer, implementing the Shor code would require approximately $t\approx 13.52$ $\mu \mathrm{s}<\min\{\mu_{T_1},\mu_{T_2}\}$ for a full encoding-decoding operation. Therefore, it is reasonable to assume that more advanced error correction methods will have processing time durations similar to that or even of $t\approx \min\{\mu_{T_1},\mu_{T_2}\}$, depending on the code construction.

To close our discussion on $\mathrm{E}\{||\mathcal{N}(\mu_{T_1},\mu_{T_2})-\mathcal{N}(\rho,\omega,t)||_\diamond\}$, we compare the behaviour of this metric for the AD and APD channels to its behaviour for the approximated channels. Figure \ref{fig:diamondMean} shows how the value of the mean diamond norm distance is smaller for the approximated channels obtained via twirling than for the original channels. This stems from the fact that performing the twirl operation eliminates the off-diagonal elements of the original channels (due to the symmetrization of the channel), and thus their contribution to the diamond norm distance is also eliminated. In any case, Figure \ref{fig:diamondMean} shows that the approximated channels also experience time-variation in the decoherence parameters. Furthermore, observe that even though the value of the diamond norm distance is lower for the twirled channels, the shape of the curve that embodies the time dependance of this metric is similar to the curves obtained for the AD and APD channels.

\subsection{Variability of the diamond norm distance}\label{sub:boxplots}

The mean diamond norm distance $\mathrm{E}\{||\mathcal{N}(\mu_{T_1},\mu_{T_2})-\mathcal{N}(\rho,\omega,t)||_\diamond\}$ has enabled us to analyze the average difference between static and TV quantum channels. However, while studying the average distance corroborates that static channels and TVQCs differ, it fails to portray the variability of $||\mathcal{N}(\mu_{T_1},\mu_{T_2}) - \mathcal{N}(\rho,\omega,t)||_\diamond$. Thus, looking solely at the average of this metric may be shortsighted in the sense that additional trends and results regarding the nature of TVQCs may be overlooked. Consequently, a boxplot analysis of the realizations of the diamond norm distance is performed with the goal of studying the variability of $||\mathcal{N}(\mu_{T_1},\mu_{T_2}) - \mathcal{N}(\rho,\omega,t)||_\diamond$ (see Appendix \ref{app:boxSkew} for description of adjusted boxplots for skewed distributions).

\subsubsection*{$T_1$-limited qubits}
Figure \ref{fig:boxplots} presents the boxplots for the AD and the ADPTA channels, considering different coefficients of variation of the relaxation time random variable $T_1(\omega)$. When $c_\mathrm{v}(T_1)=1\%$ (Figure \ref{fig:boxplots}a and \ref{fig:boxplots}d), the size of the boxplots are negligible for all the range of algorithm times. This means that when the fluctuations of $T_1$ around its mean are small, the dispersion of $||\mathcal{N}(\mu_{T_1},\mu_{T_2}) - \mathcal{N}(\rho,\omega,t)||_\diamond$ around its mean is nearly zero and, thus, a static channel approach will be a valid assumption. Increasing the coefficient of variation to $c_\mathrm{v}(T_1)=10\%$ (Figure \ref{fig:boxplots}b and \ref{fig:boxplots}e) the box size and the whisker length begins to widen as the algorithm time increases, meaning that a significantly large number of realizations of $||\mathcal{N}(\mu_{T_1},\mu_{T_2}) - \mathcal{N}(\rho,\omega,t)||_\diamond$ will be higher than its mean given by Figure \ref{fig:diamondMean}. Thus, a static noise assumption may start to fail when algorithm times are sufficiently large. Finally, for $c_\mathrm{v}(T_1)=25\%$ (Figure \ref{fig:boxplots}c and \ref{fig:boxplots}f) one obtains the large box sizes and whiskers of the boxplots, showing a large variation of the diamond norm distance around its mean. Therefore, adopting static noise modelling strategies in such scenarios represents an erroneous approach.

\begin{figure}[h!]
\centering
\includegraphics[width=\linewidth]{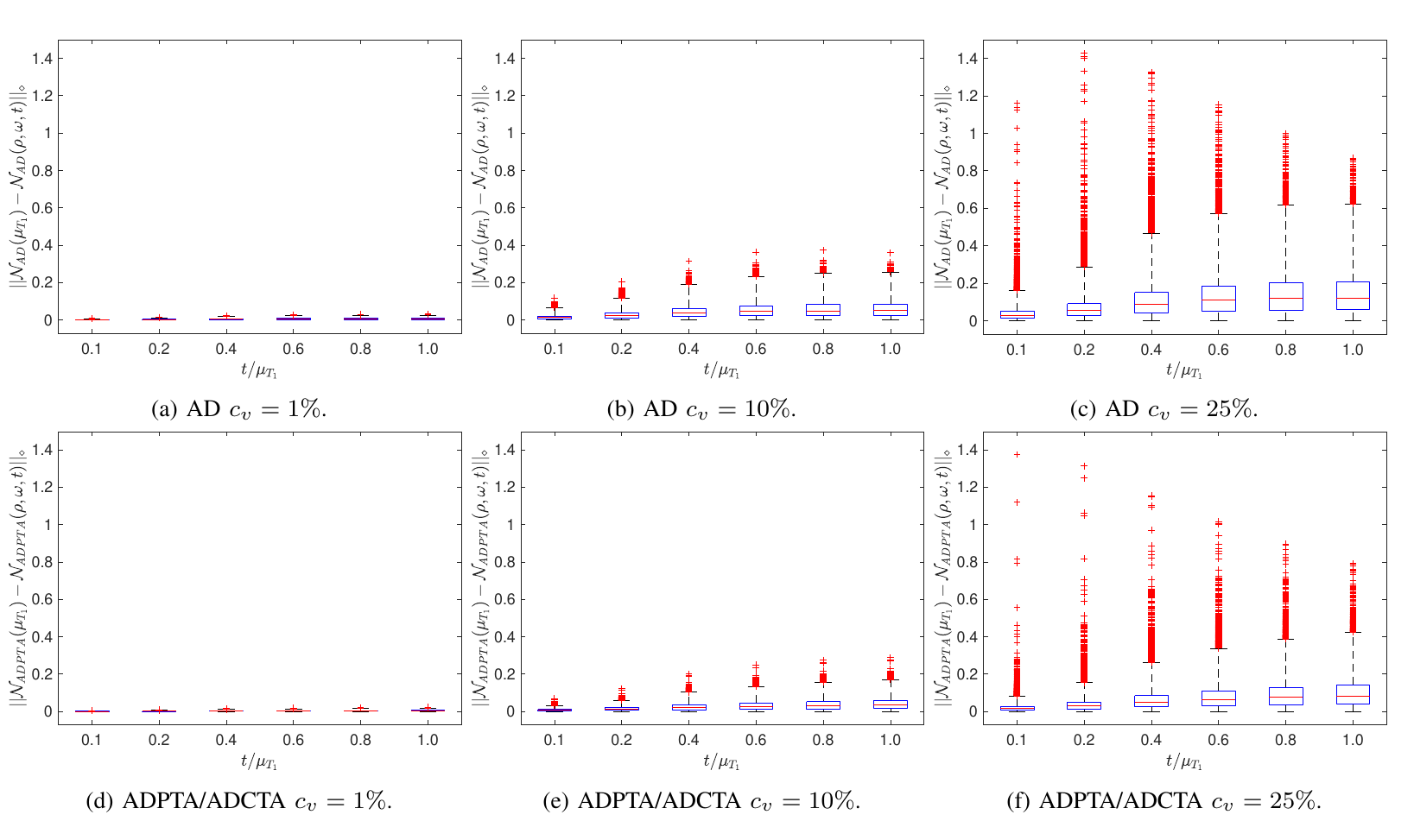}
\caption{Boxplot representation of the simulations of the diamond norm distance $\mathbf{||\mathcal{N}(\mu_{T_1})-\mathcal{N}(\rho,\omega,t) ||_\diamond}$ for the $\mathbf{T_1}$-limited scenario. \textbf{a-c} AD channels with $c_\mathrm{v}(T_1) = \{1,10,25\}\%$. \textbf{d-f} ADPTA/ADCTA channels with $c_\mathrm{v}(T_1) = \{1,10,25\}\%$. Outlier data is represented with red crosses.}
\label{fig:boxplots}
\end{figure}

It is also important to discuss the presence of outliers on these boxplots. These events appear to be significant in number when the coefficient of variation $c_\mathrm{v}(T_1)$ is large. Consequently, even for algorithms that require short processing times, where the mean of the diamond norm distance is very small, the use of a static channel when $c_\mathrm{v}(T_1)$ is large might be misleading. For example, for a normalized algorithm time of $0.1$, the mean value, $\mathrm{E}\{||\mathcal{N}(\mu_{T_1})-\mathcal{N}(\rho,\omega,t)||_\diamond\}$, is $0.04$ and $0.02$ for the AD and ADPTA channels, respectively (refer to Figures \ref{fig:diamondMean}a and \ref{fig:diamondMean}b) considering a $c_\mathrm{v}(T_1)=25\%$. Therefore, by looking at these values it will be reasonable to assume that a static channel approach when simulating QECCs will be valid. However, realizations of $||\mathcal{N}(\mu_{T_1})-\mathcal{N}(\rho,\omega,t)||_\diamond \approx 1.2$ may occur for both the AD and ADPTA channels (refer to Figure \ref{fig:boxplots}c and \ref{fig:boxplots}f), and the frequency of outlier events seems to be important. Thus, the above assumption may fail.

From the above results, it is clear that when the coefficients of variation are large, the diamond norm distance $||\mathcal{N}(\mu_{T_1})-\mathcal{N}(\rho,\omega,t)||_\diamond$ exhibits a significant amount of spread. This entails that events where the TV channel and the static channel substantially differ will be pretty common, even if the mean value of the diamond norm is not very high. Thus, adopting static noise modelling strategies in such scenarios will be an erroneous approach.

\subsubsection*{$T_1\approx T_2$ and $T_2$-dominated qubits}
We also provide boxplots for the $T_1\approx T_2$ and $T_2$-dominated superconducting qubit scenarios.

\begin{figure*}[!h]
\centering
\includegraphics[width=\linewidth]{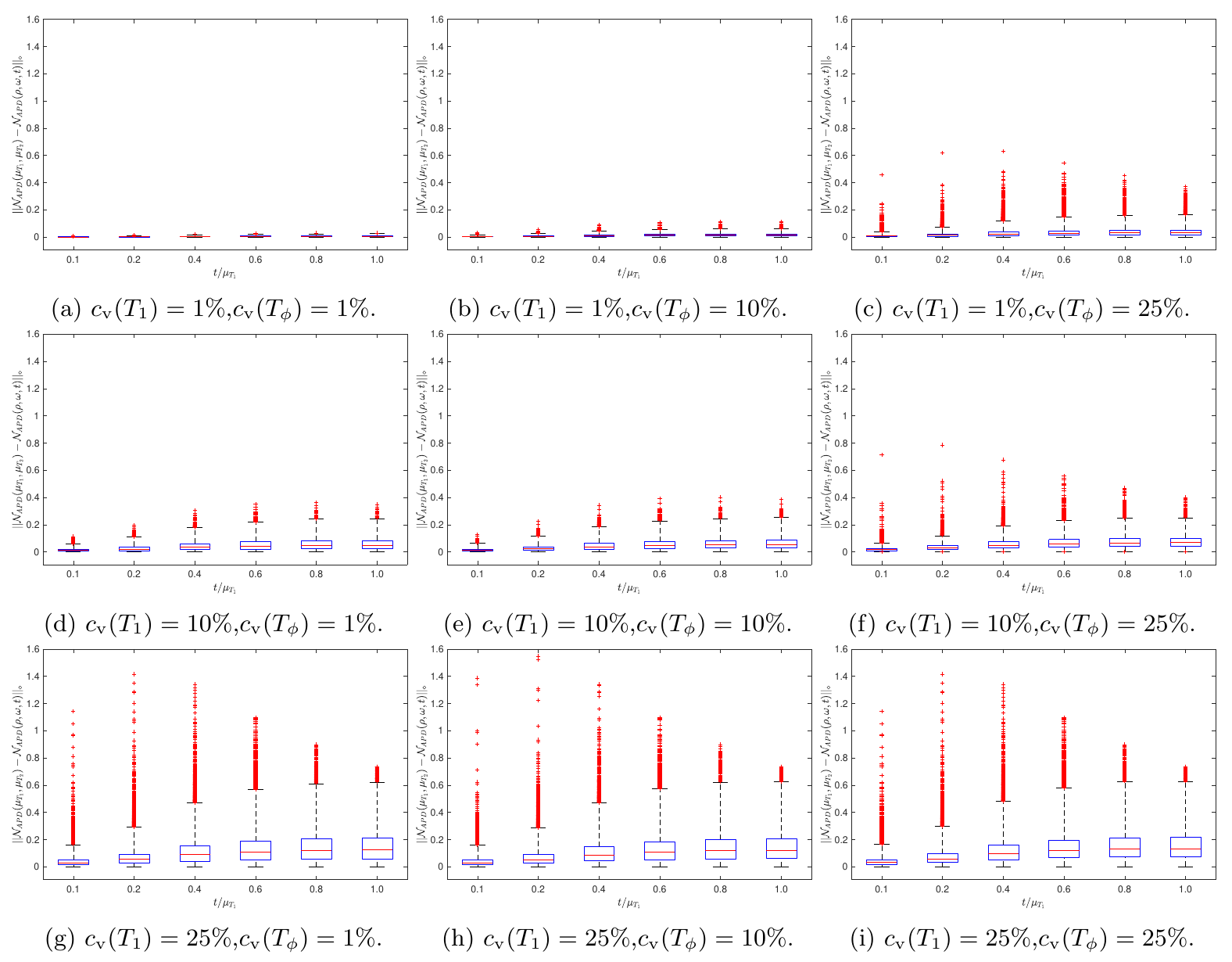}
\caption{Boxplot representation of the simulations of the diamond norm distance $||\mathcal{N}(\mu_{T_1},\mu_{T_2})-\mathcal{N}(\rho,\omega,t) ||_\diamond$ for the $T_1\approx T_2$ scenario. \textbf{a-i} APD channels with $c_\mathrm{v}(T_1)=\{1,10,25\}\%$ and $c_\mathrm{v}(T_\phi)=\{1,10,25\}\%$. Outlier data is represented with red crosses.}
\label{fig:T1T2}
\end{figure*}

Figures \ref{fig:T1T2} and \ref{fig:T2dom} show the resulting diamond norm distance boxplots for the $T_1\approx T_2$ and $T_2$-dominated scenarios studied in section \ref{res:TVQC} considering the combined amplitude and phase damping channel (we omit the results for their twirled channels since they will be similar). The trend observed for $T_1$-limited superconducting qubits can also be appreciated in these figures. For both figures, increasing the coefficient of variation of the random variables $T_1$ and $T_\phi$ results in a higher spread of the diamond norm distance between the static and TV channels. Consequently, for low coefficients of variation of the decoherence parameters (see Figures \ref{fig:T1T2}a, \ref{fig:T1T2}b, \ref{fig:T2dom}a and \ref{fig:T2dom}b), considering a static channel a is valid assumption. For milder coefficients of variation (see Figures \ref{fig:T1T2}c, \ref{fig:T1T2}d, \ref{fig:T1T2}e,\ref{fig:T2dom}c, \ref{fig:T2dom}d and \ref{fig:T2dom}e) the static assumption starts to diverge from the TV channel model, and implications for the QECCs will start to occur. Finally, for the high values of $c_\mathrm{v}(T_1)$ and $c_\mathrm{v}(T_\phi)$ (see Figures \ref{fig:T1T2}f, \ref{fig:T1T2}g, \ref{fig:T1T2}h, \ref{fig:T1T2}i, \ref{fig:T2dom}f, \ref{fig:T2dom}g, \ref{fig:T2dom}h and \ref{fig:T2dom}i), both channel models differ substantially and the non-dynamic assumption is flawed.  Note that as discussed before, the variation of the relaxation time has more impact in the channel fluctuations than the variation of the pure dephasing (recall that relaxation is part of the dephasing time $T_2$). Therefore, the scenarios with low $c_\mathrm{v}(T_1)$, but moderately large $c_\mathrm{v}(T_\phi)$ can still be considered static (as it can be seen in Figures \ref{fig:T1T2}b and \ref{fig:T2dom}b).

\begin{figure*}[!h]
\centering
\includegraphics[width=\linewidth]{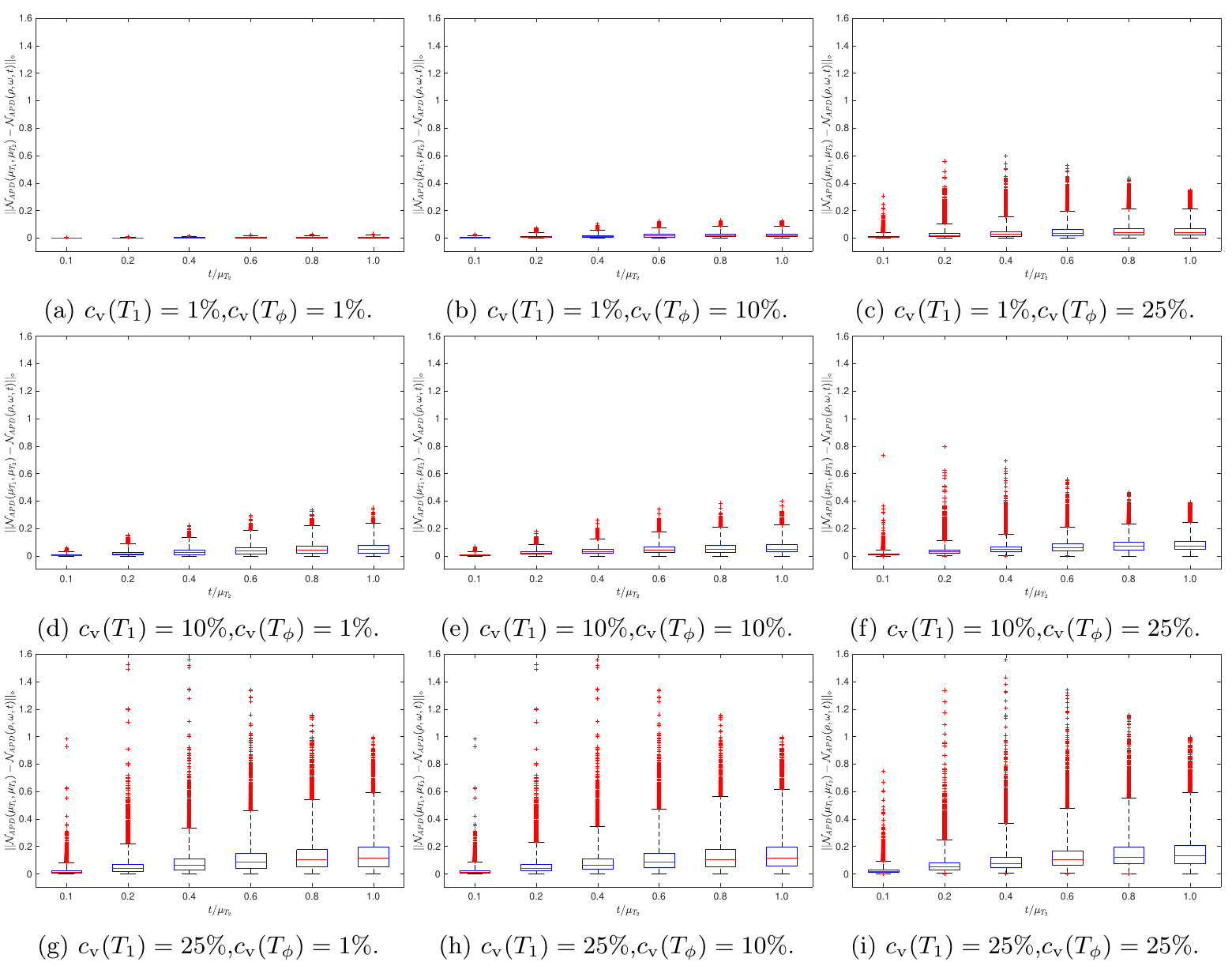}
\caption{Boxplot representation of the simulations of the diamond norm distance $||\mathcal{N}(\mu_{T_1},\mu_{T_2})-\mathcal{N}(\rho,\omega,t) ||_\diamond$ for the $T_2$-dominated scenario. \textbf{a-i} APD channels with $c_\mathrm{v}(T_1)=\{1,10,25\}\%$ and $c_\mathrm{v}(T_\phi)=\{1,10,25\}\%$. Outlier data is represented with red crosses.}
\label{fig:T2dom}
\end{figure*}

The above boxplot analysis for $T_1\approx T_2$ and $T_2$-dominated qubits show similar results as the one obtained for the $T_1$-limited scenario. In this way, superconducting qubits that present pure dephasing channels that require attention are also affected by the variation of the decoherence parameters. Thus, the TVQC model is relevant for all types of superconducting qubits that present fluctuations of $T_1$ and $T_\phi$ with sufficiently large coefficients of variations.

\FloatBarrier

\section{Discussion}\label{cp3_section6}
In this chapter we have presented a time-varying (TV) quantum channel model as a way to include the time dependant fluctuations of the parameters that define the decoherence effects experienced by superconducting qubits. To that end, we have used state-of-the-art experimental results to model the time-variations of these parameters and provided an extensive numerical analysis of the proposed TVQC models. Therefore, the proposed model serves to coalesce recent experimental results regarding the fluctuations of the relaxation time, $T_1$, and the dephasing time, $T_2$, into a quantum channel model that provides a more realistic portrayal of decoherence effects to the quantum error correction community.
The proposed channel models are sufficiently general to include any type of slow variation in the random processes $T_1(\omega,t)$ and $T_2(\omega,t)$. Therefore, even though the present analysis is motivated on superconducting qubits, the proposed model is also applicable to other quantum-coherent two-level systems that present similar time-dependant parameter oscillations, such as trapped ions or quantum dots.
In our analysis, the criterium chosen to measure the deviation between time-variant and time invariant quantum channels has been the diamond norm distance $||\mathcal{N}(\mu_{T_1},\mu_{T_2}) - \mathcal{N}(\rho,\omega,t)||_\diamond$. Based on the statistical characterization of the random processes $T_1(\omega,t)$, $T_{\phi}(\omega,t)$ and $T_2(\omega,t)$, it is found that their stochastic process coherence times, $T_c$, are much smaller than the algorithm times. Therefore, these decoherence parameters can be modelled as random variables taking independent values from round to round. This in turn results in $||\mathcal{N}(\mu_{T_1},\mu_{T_2}) - \mathcal{N}(\rho,\omega,t)||_\diamond$ being independent random variables from round to round. Computing the diamond norm distance average over all rounds in a quantum processing task reveals that by not considering time fluctuations in quantum channel models, an integral part of the decoherence effects are neglected.

In order to asses the dispersion of the realizations of the random variable $||\mathcal{N}(\mu_{T_1},\mu_{T_2}) - \mathcal{N}(\rho,\omega,t)||_\diamond$, a boxplot analysis has been carried out. The primary conclusion drawn from this boxplot study is that the diamond norm distance exhibits a high degree of variability if the coefficients of variation of the decoherence times are sufficiently large, say larger than $\approx 10\%$. This is due to the fact that under this condition, the whisker length from the median and the number of outliers in this region is considerably high. It should be mentioned that the state-of-the-art superconducting processors in \cite{decoherenceBenchmarking} have a $c_\mathrm{v}(T_1)\approx 20\%$ (refer to Tables \ref{tab:variationParameters}). Therefore, the TV quantum channel models proposed in this article are important to describe the decoherence effects suffered by their qubits when studying QECCs requiring a large number of error correction cycles (e.g., a quantum memory or a long quantum algorithm such as the one in \cite{rsaRounds}). If a QECC is run just once or a few number of times, the parameters will not fluctuate and the effects of the proposed channel model will not be noticeable.

These results also speak towards the importance of superconducting qubit construction and cooldown, in the sense that if they are optimized correctly, the fluctuations relative to the mean will be milder. Traditionally, long mean decoherence times are seeked for the qubits to be able to handle long time duration algorithms, however, this should be done jointly with lowering the dispersion in such parameter. In other words, the results presented here point out that qubit benchmarking, at least for superconducting qubits at the moment, should include the variation parameters of their decoherence times. However, the current experimental research on the topic usually only focuses at the mean or best case scenarios, which is far from being accurate when describing the decoherence dynamics of the qubits.

Finally, note that in this chapter, we did not discuss how the proposed time-varying channel models will affect the performance of quantum error correction codes. This will be done in Chapter \ref{cp5}, where QTCs and Kitaev toric codes operating over the TVQCs will be simulated.

\chapter{Quantum outage probability} \label{cp4}
The time-varying quantum channel model was presented in Chapter \ref{cp3}. In this chapter we focus on channel capacity. Quantum channel capacity gives the quantum rate limit at which a QECC can operate with asymptotically vanishing error rate. However, in the literature, quantum capacity has been derived for static quantum channels. Here, we study the asymptotical limits of error correction for time-varying quantum channel models. Motivated by the similarity between the TVQCs and classical slow or block fading scenarios, i.e., when the channel remains constant over the duration of the coded block \cite{tse}; we define the quantum outage probability of a TVQC as the asymptotically achievable error rate for a QECC with quantum rate $R_\mathrm{Q}$ that operates over the aforementioned noise model. Additionally, we also introduce the concept of the quantum hashing outage probability to provide an upper bound on the asymptotically achievable error rate for TVQC channels whose quantum capacity (of their static counterparts) is unknown, but for which a lower bound known as the hashing limit exists. Based on the experimentally determined statistical distribution of $T_1$ discussed in Chapter \ref{cp3}, we provide closed-form expressions for the following TVQCs: time-varying amplitude damping channel (TVAD), time-varying amplitude damping Pauli twirl approximated channel (TVADPTA) and time-varying amplitude damping Clifford twirl approximated channel (TVADCTA). We analyze the quantum outage probability and quantum hashing outage probabilities of the aforementioned TVQCs for different scenarios.

\section{Classical outage probability for the block fading additive Gaussian noise channel}\label{subsub:classicalpout}
The TVQC model proposed in this dissertation clearly resembles the paradigm of classical block fading \cite{tse}. This occurs because the algorithm processing times \cite{decoherenceBenchmarking} are much smaller than the coherence times of the stochastic processes that define the dynamic behaviour of the relaxation and dephasing times. Consequently, the realization of $T_i(\omega)$, $i=1,2$ remains constant for all the qubits in the block, where we have also assumed that all qubits in the codeword are affected equally by the noise. With this we mean that the $T_1(\omega)$ and $T_2(\omega)$ of each of the qubits of the QECC block are completely correlated implying that the realizations of each of the decoherence times are the same for each of the qubits\footnote{In general this is not the only possibility, but here we study such scenario. We study this case following the usual reasoning in the QECC literature that the error probability for each of the qubits of the system is the same \cite{QRM,bicycle,qldpc15,jgf,QCC,QTC,EAQTC,toric,
QEClidar,EAQECC,EAQIRCC,twirl6,MemQTC,catalytic}.}. Mathematically, this can be written as
\begin{equation}\label{eq:correlationFull}
\rho_{T_i^j,T_i^k}=\frac{\mathrm{E}\{(T_i^j(\omega)-\mu_{T_i^j})(T_i^k(\omega)-\mu_{T_i^k})\}}{\sigma_{T_i^j}\sigma_{T_i^k}}=1,\forall i\in\{1,2\},\forall j,k\in\{1,\cdots,n\}
\end{equation}
where $\rho_{X,Y}$ is named the Pearson correlation coefficient and $n$ is the number of qubits of the error correction block.

In this way, the realization of the $n$-qubit TVQC at block $l$, $\mathcal{N}_l^{(n)}(\rho)$, is equal to
\begin{equation}\label{eq:nqubitTVQC}
\mathcal{N}_l^{(n)}(\rho) = \mathcal{N}^{\otimes n}(\rho,t_1^{l},t_2^{l}),
\end{equation}
where $t_1^{l}$ and $t_2^{l}$ refer to the realizations of random variables $T_1(\omega)$ and $T_2(\omega)$ at block $l$, realizations that are assumed to be equal for each of the $n$ qubits of the block.

This is reminiscent of classical block fading scenarios in which a realization of  the fading random process, $\alpha(\omega,t)$ (channel gain), is considered to be constant for the entire codeword length \cite{tse}. We use this similarity to develop the information theoretical concepts for TVQCs. To that end, we begin by explaining in more detail the slow block fading channel model use in classical communications and define the concept of outage probability for these channels.

The general input-output relation of a continuous-time, additive gaussian noise fading channel is given by
\begin{equation}\label{eq:eee}
Y(\omega,t)=\alpha(\omega,t)X_{cd}(t)+N(\omega,t)
\end{equation}
where $X_{cd}(t)$ denotes the input codeword, and $N(\omega,t)$ and $\alpha(\omega,t)$ are the wide-sense stationary (WSS) random processes modeling the gaussian noise channel and the fading process, respectively. Whenever the duration of the codewords, $X_{cd}(t)$, is much smaller than the coherence time, $T_c$, of the fading gain $\alpha(\omega,t)$, (i.e., the time where $\alpha$ can be considered constant) the channel is said to be operating under slow fading conditions. Therefore, the value of the channel gain during the transmission of a codeword can be considered to be approximately constant and given by a realization of the random variable $\alpha(\omega)$. Furthermore, when considering the transmission of a sequence of $K$ codewords which are sufficiently separated in time, say $T$ seconds then the output sequence will be
\[Y(\omega,t)=\sum_{j=0}^{K-1} \alpha_j(\omega)X^{(j)}_{cd}(t-jT) +N(\omega,t),\]
where the sequence of random variables $\{\alpha_j(\omega)\}_{j=0}^{K-1}$ can be considered independent and identically distributed (iid).

The operational concept of channel capacity for these channels is meaningless and has to be replaced by what is known the outage probability \cite{tse}.
Let us now define this parameter. Note that conditional on a realization of the channel gain $\alpha(\omega)$, channel \eqref{eee} reduces to an additive white Gaussian noise (AWGN) channel with received signal-to-noise ratio (SNR) $|\alpha(\omega)|^2\mathrm{SNR}$. By the Shannon channel coding theorem, the maximum rate of reliable communication supported by this channel is $\log_2(1 + |\alpha(\omega)|^2 \mathrm{SNR})$ in bits per channel use. This quantity is a function of the gain $\alpha(\omega)$ and is therefore random. Channel capacity is defined as the maximum coding rate that is achievable (an error correction scheme of such rate with vanishing error rate exists) for a noisy channel. This, however, cannot be done for the slow fading scenario considered here since there will be a non-zero probability that the channel realization has a capacity (realization) lower than the rate considered. Hence, the quantum capacity for such channel in the strict sense is zero. Now suppose that the transmitter encodes data at a rate $R$ bits per channel use. If the channel realization $\alpha$ is such
that $C(\omega)=\log(1+ |\alpha|^2\mathrm{SNR})<R$, then whatever code that was used by the transmitter, the
decoding error probability cannot be made arbitrarily small. The system is said to be
in outage, and the outage probability is given by \cite{tse}
\begin{equation}\label{eq:poutFading}
\begin{split}
p_{\mathrm{out}}(R,\mathrm{SNR}) &= \mathrm{P}(\{\omega\in\Omega : C(\omega) < R\}) \\  &= P(\{\omega\in\Omega : \log_2(1+|\alpha(\omega)|^2 \mathrm{SNR}) < R\}).
\end{split}
\end{equation}

The outage probability will depend on the probability distribution of the channel gain random variable, $\alpha(\omega)$. For the widely used Rayleigh fading model, for which the channel gain follows a circularly symmetric complex normal distribution, $\mathcal{CN}(0,1)$, the outage probability can be shown to be \cite{tse}
\begin{equation}\label{eq:rayleighPout}
p_{\mathrm{out}}(R,\mathrm{SNR}) = 1 - \mathrm{e}^{\frac{-(2^R - 1)}{\mathrm{SNR}}}.
\end{equation}

\section{Quantum outage probability for TVQCs}\label{sub:poutQdef}
Let us now look for the outage probability of time-varying quantum channels. As stated before, during the duration of a QECC codeword, the realization of the random processes $T_1$ and $T_2$, defining the dynamics of the quantum channel, can be considered constant for the duration of the algorithm processing time. Therefore, we can assume that for each of the realizations of the qubit relaxation and dephasing times, $T_1$ and $T_2$, associated to a codeword, a static quantum channel will result. This channel will be valid for the current codeword under consideration, and it will render a specific value for its quantum channel capacity, $C_\mathrm{Q}$ qubits per channel use. Here is where one can observe the similarities with the case slow fading in classical communications.

Therefore, if the realization of the decoherence parameters leads to a channel capacity lower that the quantum coding rate, $R_\mathrm{Q}$, then the quantum bit error rate (QBER) will not vanish asymptotically with the blocklength, independently of the selected QECC. For such realizations we say the channel is in outage, and the probability of such outage event will be
\begin{equation}\label{eq:qPout}
p_{\mathrm{out}}^\mathrm{Q}(R_\mathrm{Q}) = \mathrm{P}(\{\omega\in\Omega : C_\mathrm{Q}(\omega) < R_\mathrm{Q}\}),
\end{equation}
which we name as quantum outage probability.

In other words, with probability $p_{\mathrm{out}}^\mathrm{Q}(R_\mathrm{Q})$, the capacity of the channel $C_\mathrm{Q}(\omega)$ will be lower than the rate of the code, and thus, the error rate will not vanish asymptotically. Conversely, with probability $1-p_{\mathrm{out}}^\mathrm{Q}(R_\mathrm{Q})$, reliable quantum correction will be possible. Thus, the quantum outage probability will be the asymptotically achievable error rate for quantum error correction when the rate is $R_\mathrm{Q}$.

Note that the quantum outage probability is the metric that has to be taken into account in order to benchmark the performance of QECCs when they operate over TVQCs that fulfill the assumptions stated in this chapter.

\section{Computation of the quantum outage probability for the family of time-varying amplitude damping channels}
In Chapter \ref{ch:preliminary}, we provided analytic expressions of the quantum capacity of the amplitude damping (AD) channel. We also mentioned that for its Pauli/Clifford twirl approximations a closed-form expression for its capacity remains unknown, but it can be upper bounded by the hashing bound. Therefore, for the family of TVAD channels, we derive a closed-form expression for the quantum outage probability (refer to Theorem \ref{thm:poutTVAD}). On the other hand, for the Pauli/Clifford twirl approximated channels of the TVAD we propose (refer to Corollary \ref{cor:pouHPauli}) the quantum hashing outage probability, $p_{\mathrm{out}}^\mathrm{H}$, as an upper bound of their quantum outage probability. Note that since we are studying AD channels, we deal with $T_1$-limited qubits, that is, qubits that do not present pure dephasing.

\subsection{Outage probability for the time-varying amplitude damping channel}\label{sec:poutTVAD}
 Before presenting the main result of this section, Theorem \ref{thm:poutTVAD}, we begin by stating some facts. From expression \eqref{ADcap}, the quantum capacity $C_\mathrm{Q}$ of the AD channel is a monotonically decreasing function of the damping parameter, $\gamma$. Therefore, there will be a unique $\gamma^*(R_\mathrm{Q})$ that makes the value of the channel capacity $C_\mathrm{Q}$ equal to $R_\mathrm{Q}$, i.e., $C_\mathrm{Q}(\gamma^*(R_\mathrm{Q})) = R_\mathrm{Q}$. That is,
\begin{equation}\label{josu2}
C_\mathrm{Q}(\gamma^*(R_\mathrm{Q})) = R_\mathrm{Q} \Leftrightarrow\gamma^*(R_\mathrm{Q})=C^{-1}_\mathrm{Q}(R_\mathrm{Q}).
\end{equation}
We will refer to $\gamma^*(R_\mathrm{Q})$ as the noise limit. Note that codes of rate $R_\mathrm{Q}$ operating over channels with a damping parameter $\gamma>\gamma^*(R_\mathrm{Q})$ will not be reliable since $R_\mathrm{Q}> C_\mathrm{Q}(\gamma)$.

In addition, from \eqref{gammatime} we define the critical relaxation time $T_1^*(R_\mathrm{Q}, t_{\mathrm{algo}})$ as
\begin{equation}\label{eq:citicalT1}
T_1^*(R_\mathrm{Q},t_{\mathrm{algo}}) = \frac{-t_{\mathrm{algo}}}{\ln{(1-\gamma^*(R_\mathrm{Q}))}},
\end{equation}
which is a function of the algorithm time, $t_{\mathrm{algo}}$. In order to compare quantum channels with different mean relaxation times, $\mu_{T_1}$, we would like rewrite the critical time as a function of the damping parameter, $\gamma$, that results when assuming a static channel\footnote{This is similar to the normalization done in Chapter \ref{cp3} for the diamond norm analysis, as the damping rate is a function of the algorithm time and the relaxation time.} with mean $\mu_{T_1}$. The reason for it is that if the calculations were done as a function of the algorithm time, the comparison between qubits with different mean relaxation times would not be fair, since for a given $t_{algo}$, the higher the values of $\mu_{T_1}$ are, the lower the values of $\gamma $ will be, so that we will be comparing channels with different degrees of noise. It is obvious that longer mean relaxation times are more favorable for computing applications, as they allow for longer algorithm times. However, we are interested in calculating the quantum outage probability versus the noise level of the channel, i.e, we want to know how much noise a qubit is able to tolerate.

To that end, given a damping parameter $\gamma$ and a $\mu_{T_1}$, the algorithm time will be set to
\begin{equation}\label{eq:talgo}
t_{\mathrm{algo}} = - \mu_{T_1}\ln{(1 - \gamma)}.
\end{equation}
That is, for quantum channels with different $\mu_{T_1}$'s, the above expression gives the corresponding algorithm times that will produce the same damping parameter $\gamma$ in all of them.

Substituting such value in \eqref{citicalT1}, yields, 
\begin{equation}\label{josu12}
T^*_{1}(R_\mathrm{Q},\gamma)=\frac{\mu_{T_1}\ln(1-{\gamma})}{\ln(1-{\gamma^*(R_\mathrm{Q})})}.
\end{equation}
In this way, the critical relaxation time in \eqref{citicalT1} is independent of the algorithm time, and is only a function of the damping parameter $\gamma$ associated to the static AD channel.

We are now ready to introduce the theorem that derives the quantum outage probability for TVAD channels for $T_1$-limited qubits \cite{decoherenceBenchmarking}.
\begin{theorem}[TVAD quantum outage probability]\label{thm:poutTVAD}
\textit{The quantum outage probability for the time-varying amplitude damping channels associated to the damping parameter,  $\gamma\in[0,1-{e}^{-1}]$, is equal to
\begin{equation}\label{eq:poutTVAD}
p_{\mathrm{out}}^\mathrm{Q}(R_\mathrm{Q}, \gamma) = 1 - \frac{\mathrm{Q}\left(\frac{1}{c_\mathrm{v}(T_1)}\left[\frac{\ln(1-{\gamma})}{\ln(1-{\gamma^*(R_\mathrm{Q})})} - 1\right]\right)}{1 - \mathrm{Q}\left(\frac{1}{c_\mathrm{v}(T_1)}\right)},
\end{equation}
where $\mathrm{Q}(\cdot)$ is the Q-function, $c_\mathrm{v}(T_1)$ is the coefficient of variation of $T_1$, $\gamma^*(R_\mathrm{Q})$ is the noise limit, $\mu_{T_1}$ is the mean relaxation time and $\sigma_{T_1}$ is the standard deviation of the relaxation time.}
\end{theorem}
\begin{proof}
In order to compute the outage probability $p_{\mathrm{out}}^\mathrm{Q}(R_\mathrm{Q},\gamma)$, we use the decreasing monotonicity of $C_\mathrm{Q}$ and $T_1$ with respect to $\gamma$. This implies that the events $\{\omega\in\Omega : C_\mathrm{Q}(\gamma(\omega))< R_\mathrm{Q}\}$, $\{ \omega\in\Omega : \gamma(\omega)< \gamma^*(R_\mathrm{Q})\}$ and $\{\omega\in\Omega : T_1(\omega) < T_1^*(R_\mathrm{Q},\gamma)\}$ are all the same.
Therefore,
\begin{equation}\label{eq:TVADstep1}
\begin{split}
p_{\mathrm{out}}^\mathrm{Q}(R_\mathrm{Q}, \gamma) &= \mathrm{P}(\{\omega\in\Omega : C_\mathrm{Q}(\gamma(\omega)) < R_\mathrm{Q}\}) \\&=  \mathrm{P}(\{\omega\in\Omega : T_1(\omega) < T_1^*(R_\mathrm{Q},\gamma)\}).
\end{split}
\end{equation}

Next, we compute \eqref{TVADstep1} based on the fact that the random variable $T_1(\omega)\sim \mathcal{GN}_{[0,\infty]}(\mu_{T_1},\sigma^2_{T_1})$ has a probability density function

\begin{equation}\label{eq:pdf}
f_{T_1}(t_1)\begin{cases}
\frac{1}{\sigma_{T_1}\sqrt{2\pi}}\frac{\mathrm{e}^{-\frac{(t_1 - \mu_{T_1})^2}{2\sigma_{T_1}^2}}}{1-\mathrm{Q}\left(\frac{\mu_{T_1}}{\sigma_{T_1}}\right)} &\text{ if } t_1\geq 0 \\
0 &\text{ if } t_1<0
\end{cases},
\end{equation}
where in the above expression, $\mathrm{Q}(\cdot)$ is the Q-function defined as
\begin{equation}\label{josu6}
\mathrm{Q}(x) = \frac{1}{\sqrt{2\pi}}\int_x^{\infty}\mathrm{e}^{-\frac{x^2}{2}} dx.
\end{equation}

The outage probability of the TVAD channel can be calculated as

\begin{equation}\label{eq:TVADstep3}
\begin{split}
p_{\mathrm{out}}^\mathrm{Q}(&R_\mathrm{Q},\gamma) =  \mathrm{P}(\{\omega\in\Omega : T_1(\omega) < T_1^*(R_\mathrm{Q}, \gamma)\})
 = \int_{-\infty}^{T_1^*(R_\mathrm{Q}, \gamma)}f_{T_1}(t_1)dt_1 \\
 &= \int_{0}^{T_1^*(R_\mathrm{Q},\gamma)}\frac{1}{\sigma_{T_1}\sqrt{2\pi}}\frac{\mathrm{e}^{-\frac{(t_1 - \mu_{T_1})^2}{2\sigma_{T_1}^2}}}{1-\mathrm{Q}\left(\frac{\mu_{T_1}}{\sigma_{T_1}}\right)} dt_1 \\
&=\frac{1}{1-\mathrm{Q}\left(\frac{\mu_{T_1}}{\sigma_{T_1}}\right)}\left(\int_{-\infty}^{T_1^*(R_\mathrm{Q}, \gamma)}\frac{1}{\sigma_{T_1}\sqrt{2\pi}}\mathrm{e}^{-\frac{(t_1 - \mu_{T_1})^2}{2\sigma_{T_1}^2}}dt_1\right. \\ & \left.- \int_{-\infty}^{0}\frac{1}{\sigma_{T_1}\sqrt{2\pi}}\mathrm{e}^{-\frac{(t_1 - \mu_{T_1})^2}{2\sigma_{T_1}^2}}dt_1\right)\\
& =\frac{1}{1-\mathrm{Q}\left(\frac{\mu_{T_1}}{\sigma_{T_1}}\right)}\left(\int_{-\infty}^{\frac{T_1^*(R_\mathrm{Q}, \gamma)-\mu_{T_1}}{\sigma_{T_1}}}\frac{1}{\sqrt{2\pi}}\mathrm{e}^{-\frac{\eta^2}{2}}d\eta\right. \\ & \left. - \int_{-\infty}^{\frac{-\mu_{T_1}}{\sigma_{T_1}}}\frac{1}{\sqrt{2\pi}}\mathrm{e}^{-\frac{\eta^2}{2}}d\eta\right)\\
& = \frac{1}{1-\mathrm{Q}\left(\frac{\mu_{T_1}}{\sigma_{T_1}}\right)}\left(1-\int_{\frac{T_1^*(R_\mathrm{Q}, \gamma)-\mu_{T_1}}{\sigma_{T_1}}}^{\infty}\frac{1}{\sqrt{2\pi}}\mathrm{e}^{-\frac{\eta^2}{2}}d\eta\right. \\ & \left. - \int_{\frac{\mu_{T_1}}{\sigma_{T_1}}}^{\infty}\frac{1}{\sqrt{2\pi}}\mathrm{e}^{-\frac{\eta^2}{2}}d\eta\right) \\
& =\frac{1-\mathrm{Q}\left(\frac{T_1^*(R_\mathrm{Q}, \gamma)-\mu_{T_1}}{\sigma_{T_1}}\right) - \mathrm{Q}\left(\frac{\mu_{T_1}}{\sigma_{T_1}}\right)}{1-\mathrm{Q}\left(\frac{\mu_{T_1}}{\sigma_{T_1}}\right)}  = 1 - \frac{\mathrm{Q}\left(\frac{T_1^*(R_\mathrm{Q}, \gamma)-\mu_{T_1}}{\sigma_{T_1}}\right)}{1-\mathrm{Q}\left(\frac{\mu_{T_1}}{\sigma_{T_1}}\right)}\\&= 1 - \frac{\mathrm{Q}\left(\frac{\mu_{T_1}}{\sigma_{T_1}}\left[\frac{\ln(1-{\gamma})}{\ln(1-{\gamma^*(R_\mathrm{Q})})} - 1\right]\right)}{1 - \mathrm{Q}\left(\frac{\mu_{T_1}}{\sigma_{T_1}}\right)} =1 - \frac{\mathrm{Q}\left(\frac{1}{c_\mathrm{v}(T_1)}\left[\frac{\ln(1-{\gamma})}{\ln(1-{\gamma^*(R_\mathrm{Q})})} - 1\right]\right)}{1 - \mathrm{Q}\left(\frac{1}{c_\mathrm{v}(T_1)}\right)},
\end{split}
\end{equation}
as we wanted to prove.

Note that the quantum outage probability as a function of the damping parameter $\gamma$ does not depend on the absolute value of the mean relaxation time, but on the coefficient of variation of $T_1$. This way, we decouple the time-varying effects from the fact that longer mean relaxation times admit longer quantum algorithm processing times. Consequently, we present a result agnostic to the impact that longer coherence times have and we can provide conclusions for all superconducting qubits.

Finally, we need to find a valid range of $\gamma$ that will result in $t_{\mathrm{algo}}$ that are much smaller than the coherence time, $T_\mathrm{c}$, of the random process $T_1(t,\omega)$, (i.e., verifying the initial hypothesis of slow time variations). To that end, note that algorithm times longer than $\mu_{T_1}$ makes no sense since for such a time frame, the qubit will be with high probability in an equilibrium state, which is useless as a resource. Therefore, we select $t_{\mathrm{algo}}< \mu_{T_1}$. Since $\mu_{T_1}$ is of order of microseconds for superconducting qubits (as explained in Chapter \ref{cp3}), and $T_\mathrm{c}$ is of order of minutes \cite{decoherenceBenchmarking} the hypothesis $t_{\mathrm{algo}}\ll T_\mathrm{c}$ holds. From \eqref{talgo}, we conclude that valid for the range for $\gamma$ is $\gamma\in[0,1-e^{-1}]$.
\end{proof}

\subsection{Quantum hashing outage probability for the time-varying twirl approximated channels}\label{sub:TVtwirls}
Since the LSD capacity of the Pauli channels that result from applying the Pauli/Clifford twirled approximation to the AD channel is not known, the quantum outage probability for this family of approximated channels cannot be calculated. Nevertheless, by means of the hashing bound, $C_\mathrm{H}(\mathbf{p})$, (refer to \eqref{hash}), we define the quantum hashing outage probability for the time-varying Pauli channels defined by the probability mass
function $\mathbf{p}(\omega) = (p_\mathrm{I}, p_\mathrm{x}, p_\mathrm{y}, p_\mathrm{z})$ random vector as
\begin{equation}\label{eq:poutPauli}
p_{\mathrm{out}}^\mathrm{H}(R_\mathrm{Q}) = \mathrm{P}(\{\omega\in\Omega : C_\mathrm{H}(\mathbf{p}(\omega))< R_\mathrm{Q}\}).
\end{equation}

Note that since the hashing limit is a lower bound of the LSD capacity $C_Q$, we have that
\[ \{\omega\in\Omega : C_Q(\mathbf{p}(\omega))< R_\mathrm{Q}\}\subseteq \{\omega\in\Omega : C_\mathrm{H}(\mathbf{p}(\omega))<R_\mathrm{Q}\}\]
Consequently, $p_{\mathrm{out}}^\mathrm{H}(R_\mathrm{Q})$ gives an upper bound of the actual quantum outage probability $p_{\mathrm{out}}^\mathrm{Q}(R_\mathrm{Q})$, i.e., $p_{\mathrm{out}}^\mathrm{H}(R_\mathrm{Q})\geq p_{\mathrm{out}}^\mathrm{Q}(R_\mathrm{Q})$.

It is important to realise that the Hashing bound in \eqref{hash}
\begin{equation}
C_\mathrm{H}(\mathbf{p}(\gamma)) =1+\sum_{k\in\{\mathrm{I,x,y,z}\}}p_k(\gamma)\log_2p_k(\gamma) = 1- \mathrm{H}_2(\mathbf{p}(\gamma)),
\end{equation}
 for the twirled approximated channels of the AD channel with the probability distributions given in \eqref{PTAprobs} and \eqref{CTAprob}\footnote{Note that expressions \eqref{PTAprobs} and \eqref{CTAprob} are for the twirl approximations of the complete amplitude and phase damping channel. However, it is easy to obtain the expressions for the ADPTA and ADCTA channels by considering $\lambda = 0$, i.e., the absence of pure dephasing channels.}, is a monotonic decreasing function of the damping probability, $\gamma$. This is justified by the fact that, as $\gamma\in[0,1]$ increases, the values of $p_\mathrm{x},p_\mathrm{y},p_\mathrm{z}$ from either \eqref{PTAprobs} or \eqref{CTAprob} also increase. This results in the uncertainty of the discrete random variables associated to each of these distributions, and consequently their corresponding entropy values, becoming higher. Therefore, as we did for the AD channel, we define the noise limit for these Pauli channels as the unique value of the damping parameter, $\gamma_{\mathrm{T}}^*(R_\mathrm{Q})$, such that:
\begin{equation}\label{eq:josu11}
1-C_\mathrm{H}(\mathbf{p}(\gamma^*_{\mathrm{T}}(R_\mathrm{Q})) =R_\mathrm{Q}\Leftrightarrow \gamma^*_{\mathrm{T}}(R_\mathrm{Q})=C^{-1}_\mathrm{H}(1-R_\mathrm{Q}).
\end{equation}
From \eqref{gammatime}, the critical relaxation time (note that we have added the subindex $\mathrm{T}$ to indicate we are twirling the AD channel) is
\begin{equation}\label{josu122}
T^*_{1,\mathrm{T}}(R_\mathrm{Q},t_{\mathrm{algo}})=\frac{-t_{\mathrm{algo}}}{\ln(1-\gamma^*_{\mathrm{T}}(R_\mathrm{Q}))},
\end{equation}
 where the probability mass function $\mathbf{p}$ in \eqref{josu11} should be taken as $\mathbf{p}_{\mbox{\tiny{ADPTA}}}$ or $\mathbf{p}_{\mbox{\tiny{ADCTA}}}$ when considering the twirled ADPTA or ADCTA channels, respectively. Similarly to the TVAD channel, we can write the critical relaxation time as a function of the damping parameter
 \begin{equation}\label{josu12T}
T^*_{1,\mathrm{T}}(R_\mathrm{Q},\gamma)=\frac{\mu_{T_1}\ln(1-{\gamma})}{\ln(1-{\gamma^*_{\mathrm{T}}(R_\mathrm{Q})})}.
\end{equation}

The following corollary yields the hashing outage probability for the twirled TVADPTA and TVADCTA Pauli channels that consider $T_1$-limited qubits as the ones of \cite{decoherenceBenchmarking}.

\begin{corollary}[Quantum hashing outage probability]\label{cor:pouHPauli}
\textit{The quantum hashing outage probability for the time-varying twirled approximated channels associated to the damping parameter,  $\gamma\in[0,1-e^{-1}]$, is equal to
\begin{equation}\label{eq:poutTVADTAs}
p_{\mathrm{out}}^\mathrm{H}(R_\mathrm{Q}, \gamma) = 1 - \frac{\mathrm{Q}\left(\frac{1}{c_\mathrm{v}(T_1)}\left[\frac{\ln(1-{\gamma})}{\ln(1-{\gamma^*_{\mathrm{T}}(R_\mathrm{Q})})} - 1\right]\right)}{1 - \mathrm{Q}\left(\frac{1}{c_\mathrm{v}(T_1)}\right)},
\end{equation}
where $\mathrm{Q}(\cdot)$ is the Q-function, $c_\mathrm{v}(T_1)$ is the coefficient of variation of $T_1$, $\gamma^*_{\mathrm{T}}(R_\mathrm{Q})$ is the noise limit that depends on the considered twirled approximation, $\mu_{T_1}$ is the mean relaxation time, and $\sigma_{T_1}$ is the standard deviation of the relaxation time.}
\end{corollary}
\begin{proof}
In order to compute the hashing outage probability $p_{\mathrm{out}}^\mathrm{H}(R_\mathrm{Q}, \gamma)$, we use the fact that $C_\mathrm{H}$ and $T_1$ are monotonically decreasing functions of $\gamma$. This implies that the events $\{\omega\in\Omega : C_\mathrm{H}(\omega)< R_\mathrm{Q}\}$, $\{ \omega\in\Omega : \gamma(\omega)< \gamma^*_{\mathrm{T}}(R_\mathrm{Q})\}$ and $\{\omega\in\Omega : T_1(\omega) < T_{1,\mathrm{T}}^*(R_\mathrm{Q}, \gamma)\}$ are all same.
Therefore,
\begin{equation}\label{eq:TVADTAstep2}
\begin{split}
p_{\mathrm{out}}^\mathrm{H}(R_\mathrm{Q}, \gamma) &= P(\{\omega\in\Omega : C_\mathrm{H}(\omega) < R_\mathrm{Q}\})= \\&  \mathrm{P}(\{\omega\in\Omega : T_1(\omega) < T_{1,\mathrm{T}}^*(R_\mathrm{Q}, \gamma)\}.
\end{split}
\end{equation}
Thus, the hashing outage corresponds to events where the realization of the relaxation time is lower than the critical relaxation time.

The calculation of the quantum hashing outage probability now follows the same steps as in the proof of Theorem \ref{thm:poutTVAD}, since equation \eqref{TVADTAstep2} is the same as \eqref{TVADstep1}.
\end{proof}

Even though the final expression for the quantum hashing outage probability looks the same as the one of Theorem \ref{thm:poutTVAD}, what makes their values different is that their corresponding noise limit values are computed in a different way.

\subsection{Numerical simulations}\label{sub:numsim}

We now discuss the behaviour of the quantum outage probability and the quantum hashing outage probability. Following the reasoning of Chapter \ref{cp3}, we will compare scenarios set by different coefficients of variation of the random variable $T_1$. In particular, we consider the following values of the coefficient of variation; $c_\mathrm{v}(T_1)=\{1,10,15,20,25\}\%$.

\subsubsection{Quantum outage probability of the TVAD channel}\label{subsub:tvadPoutQ}

Figure \ref{fig:ADpoutQ} plots the quantum outage probability versus the damping parameter, $10^{-3}\leq \gamma\leq 0.6$, for a quantum rate $R_\mathrm{Q}=1/9$, and for all the coefficients of variation of the relaxation time random variable.

\begin{figure}[!ht]
\centering
\includegraphics[width=\linewidth]{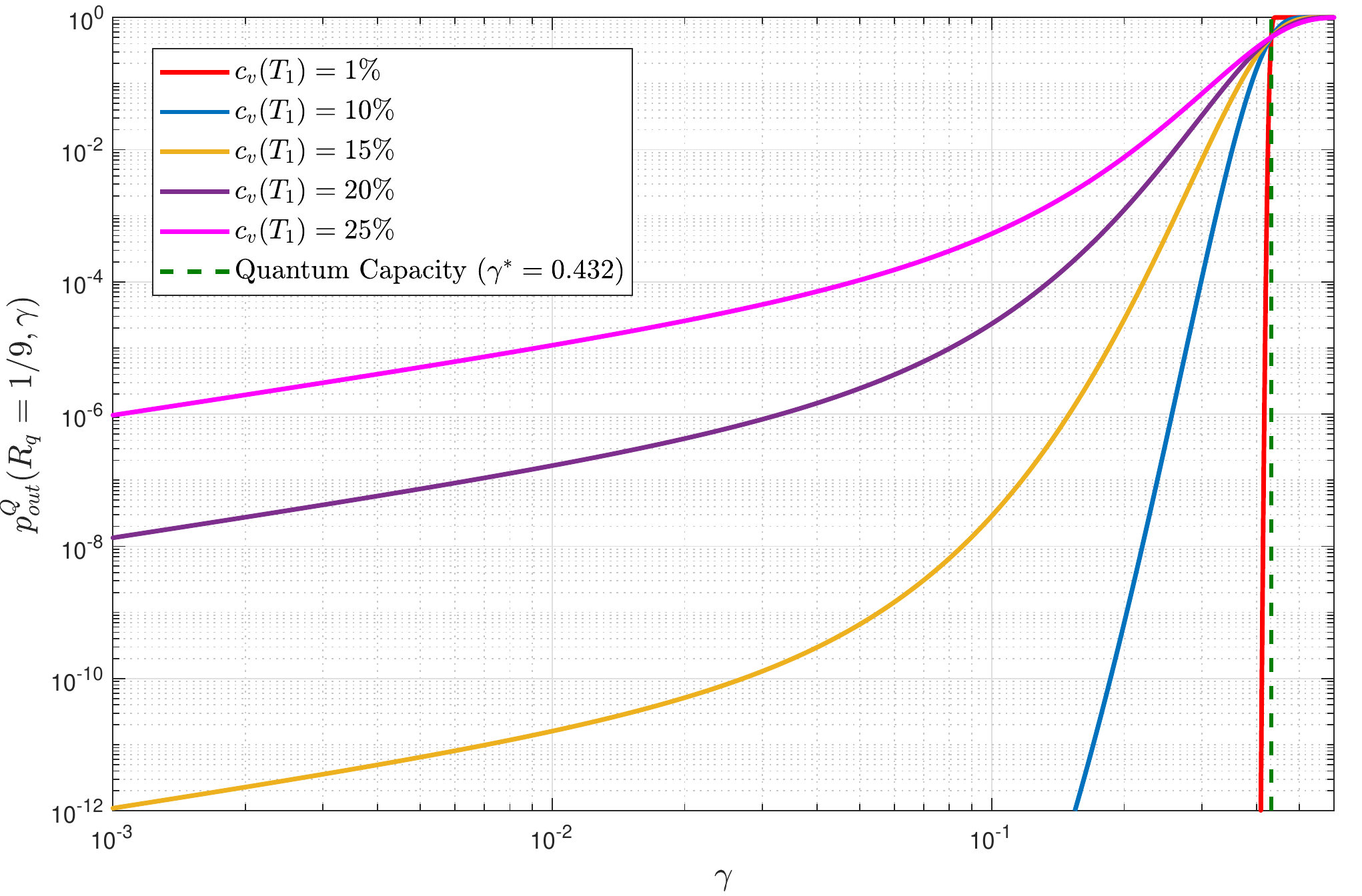}
\caption{Quantum outage probability of the TVAD channel. The metric is calculated for TVADs with $c_\mathrm{v}(T_1)=\{1,10,15,20,25\}\%$ and for a quantum rate of $R_\mathrm{Q}=1/9$.}
\label{fig:ADpoutQ}
\end{figure}

Figure \ref{fig:ADpoutQ} further cements the conclusions derived in Chapter \ref{cp3}, that is, the impact of the fluctuations of the decoherence parameters can be accurately quantified by the coefficient of variation of $T_1$. Notice that when the coefficient of variation is very low, i.e. $c_\mathrm{v}(T_1)=1\%$, the quantum outage probability of the TVAD channel almost coincides with the quantum capacity (represented herein by the noise limit, $\gamma^*$). Consequently, QECCs operating over TVQCs that present low coefficients of variation will behave asymptotically in a similar manner to static channels. However, by increasing the variability of the relaxation time around its mean causes the outage probability to diverge from the static capacity. In this case, the asymptotical bounds flatten and the achievable error rate of QECCs operating over TVQCs does not vanish. Therefore, the higher $c_\mathrm{v}(T_1)$ is, the worse the achievable error rate will be.

The previous discussion indicates that the coefficient of variability of the random variable $T_1(\omega)$ can be used to describe the effect that decoherence parameter time fluctuations will produce on the aymptotical limits of QECCs. These results show the importance of qubit construction and cooldown: if optimized correctly, the fluctuations relative to the mean will be mild, and the outage events will be significantly less likely. Obviously, it is desirable for qubits to exhibit large $\mu_{T_1}$, so that algorithms with longer lifespans can be handled appropriately. However, as important as that is the minimization of the dispersion of $T_1$ around its $\mu_{T_1}$, so the outage probability can be reduced.

\begin{figure}[!h]
\centering
\includegraphics[width=\linewidth]{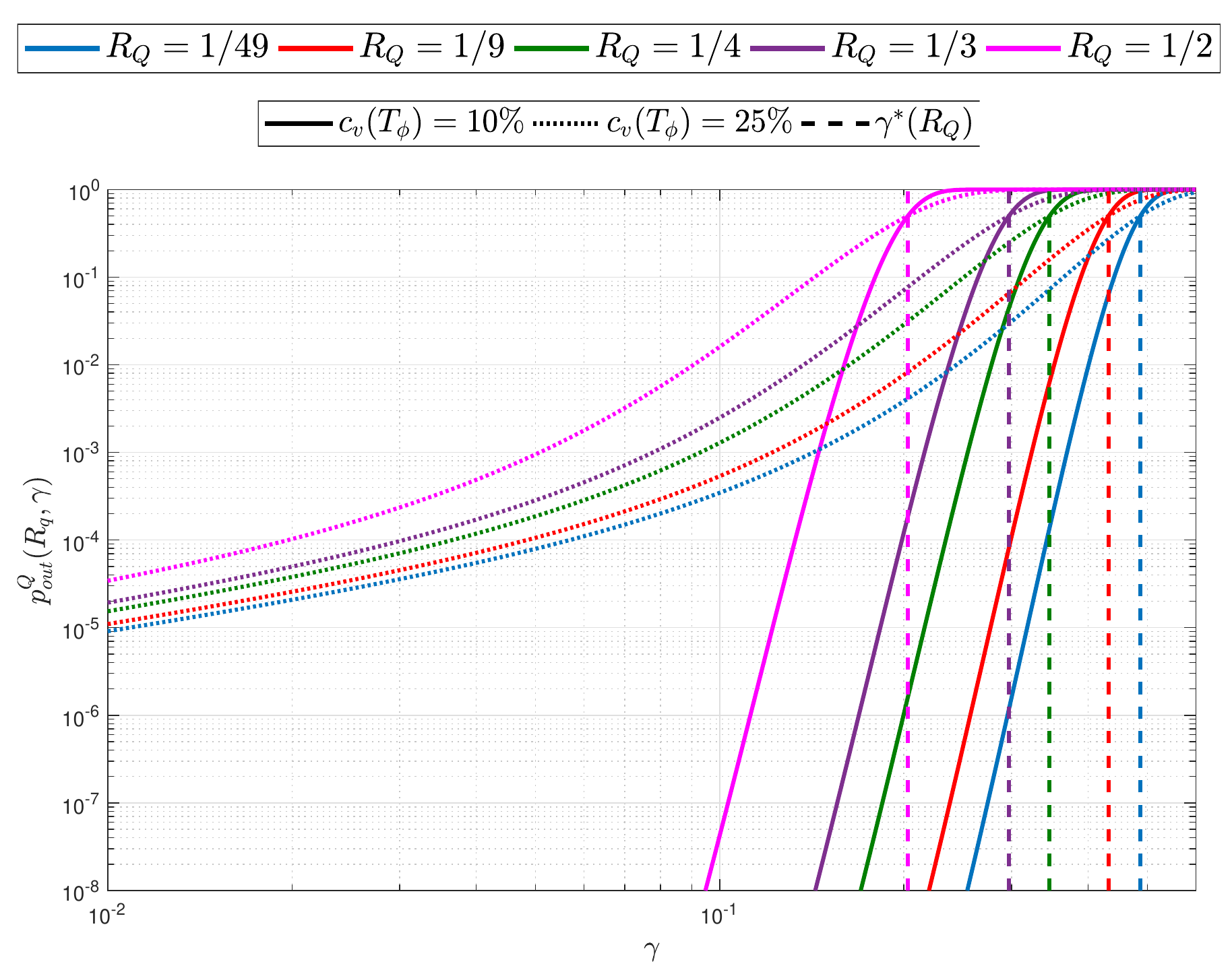}
\caption{Quantum outage probability of the TVAD channel for different quantum rates. The considered quantum rates are $R_\mathrm{Q}\in\{1/49,1/9,1/4,1/3,1/2\}$. We plot the quantum outage probability for $c_\mathrm{v}(T_1) =\{10,25\}\%$. The quantum capacities (noise limits, $\gamma^*$) for the static quantum channels are also represented.}
\label{fig:ADpoutRates}
\end{figure}

Let us now discuss how the quantum rate affects the quantum outage probability. Figure \ref{fig:ADpoutRates} shows the quantum outage probability versus $\gamma$ for $c_\mathrm{v}(T_1) = \{10,25\}\%$, and for the set of quantum rates $R_\mathrm{Q}\in\{1/49,1/9,1/4,1/3,1/2\}$. The figure also shows the corresponding noise limits $\gamma^*(R_Q)$ for the different rates. As expected, the results portrayed in this figure show that increasing the rate leads to an increase of $p_{\mathrm{out}}^\mathrm{Q}(R_\mathrm{Q},\gamma)$. Note also that for the same coefficient of variation, the shape of outage probability versus $\gamma$ curves remains fairly the same. The main conclusion is that, by increasing the rate of an error correction code reduces the overall resource consumption, however, this occurs at the expense of a degradation in the asymptotical error correction performance. Furthermore, given a $R_Q$, the larger the coefficient of variability $c_\mathrm{v}(T_1)$ is, the larger the gap between the outage probability curve and the noise limits $\gamma^*(R_Q)$ curve will be. It is important to mention that this degradation does not occur because there is higher sensitivity to time fluctuations at higher rates. Further inspection of Figure \ref{fig:ADpoutRates} reveals that the noise limits for each rate change as expected, and that the outage probabilities behave similarly according to those noise limits. Thus, similarly to classical coding, the quantum rate does indeed impact the quantum outage probability, but not due to a higher sensitivity to time-variance. Furthermore, similar to what happens in static channels, there is trade-off between resource consumption and how demanding (in terms of noise) the quantum channel is.

\subsubsection{Quantum hashing outage probability of the time-varying twirl approximated channels}\label{subsub:PauliPoutH}
We proceed by comparing the outage of the TVAD channel and its twirled approximated channels. Figure \ref{fig:comparison} plots the hashing outage probability results of the TVADPTA and TVADCTA channels for a coding rate $R_Q=1/9$. Note that the x-axis is still $\gamma$, despite the fact that the defining parameter for the TVADPTA and TVADCTA channels is $\mathbf{p}$. However, the $\gamma$ associated to a given $\mathbf{p}$ can be easily obtained from \eqref{PTAprobs} and \eqref{CTAprob}. The quantum outage probability of the TVAD channel from Figure \ref{fig:ADpoutQ} are also shown. Note that the hashing outage probabilities for the ADPTA and ADCTA channels are slightly higher than the quantum outage probability of the TVAD channel. This is due to the fact the noise limits for those channels are smaller than the one for the TVAD.

Also note that even though the hashing outage probabilities for these approximated channels are higher than the quantum outage probability for the TVAD, one can not conclude that the actual quantum outage probability for these twirled channels will be worse than the one for the AD channel (recall that the hashing outage probability provides an upper bound on the actual outage probability).

\begin{figure}[!h]
\centering
\includegraphics[width=\linewidth]{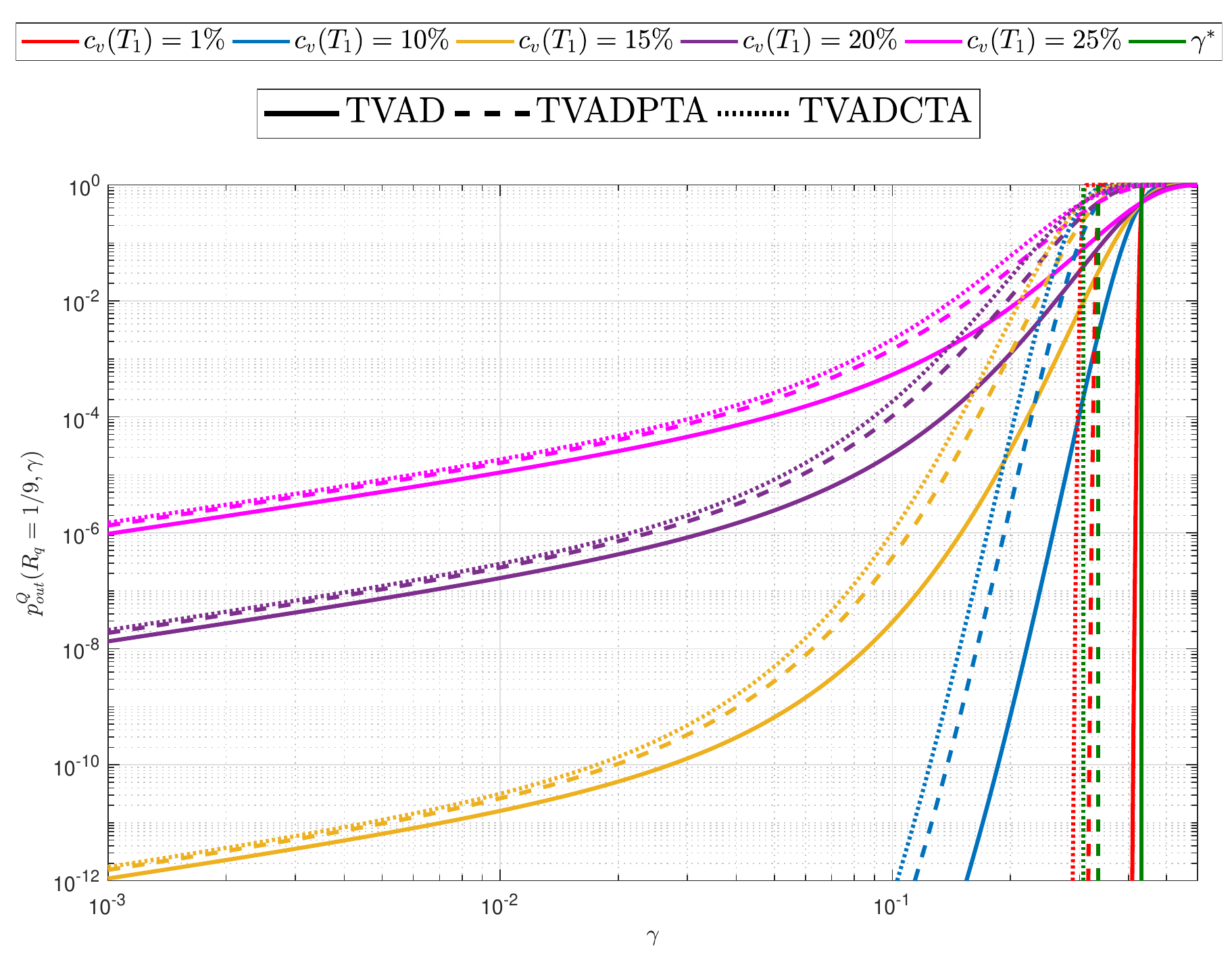}
\caption{Quantum outage and hashing outage probabilities for TVAD, TVADPTA and TVADCTA when $\mathbf{R_\text{\textbf{Q}}=1/9}$. The noise limits are: $\gamma^*_{\mathrm{AD}} = 0.432$, $\gamma^*_{\mathrm{ADPTA}} = 0.3354$ and $\gamma^*_{\mathrm{ADCTA}}=0.3065$.}
\label{fig:comparison}
\end{figure}

\section{Discussion} \label{cp4_conclusion}

In this chapter, we have introduced the concept of quantum outage probability as the asymptotically achievable error rate for quantum error correction when slow time-varying quantum channels are considered. This probability is based on the fact that the realizations of the $T_1$ and $T_2$ remain constant and the hypothesis that they are the same for all the qubits in a quantum codeword (they are fully correlated). We have also introduced the quantum hashing outage probability as an upper bound on the quantum outage probability useful when Pauli channels are considered, since the actual quantum capacity of these channels is not known. We have provided closed-form expressions of these probabilities for the TVAD, TVADPTA and the TVADCTA channels.  The information theoretical analysis presented in this chapter is essential to benchmark the behaviour of quantum error correction codes in TV scenarios. This comes from the fact that for slow time-varying quantum channels the operational meaning of quantum capacity does not give information about the channel and must be replaced by the quantum outage probability.
 
Results show that given a damping parameter, $\gamma$, the quantum outage probability is a monotonic increasing function of the coefficient of variation, $c_\mathrm{v}(T_1)$. Therefore, it is important to experimentally look for qubits that, not only have a large mean relaxation time, $\mu_{T_1}$, but also exhibit a low standard deviation relaxation time, $\sigma_{T_1}$. By doing so, the asymptotically achievable error correction performance of QECCs under time-varying conditions will approach the performance achieved by static channels, since the largest rate that would make $p_{\mathrm{out}}^\mathrm{Q}(R_\mathrm{Q}, \gamma)\approx 0$ (or $p_{\mathrm{out}}^\mathrm{H}(R_\mathrm{Q}, \gamma)\approx 0$) will be near to the quantum capacity of the corresponding static channel (or to the hashing bound).

In summary, we have shown that the time-varying nature of the decoherence parameters may have a significant impact on the performance of future QECCs that will be used to protect quantum information. To further corroborate this fact, in the next chapter, numerical simulations on the performance of QECCs over TVQCs will be done.

\chapter{Time-varying quantum channels and quantum error correction codes} \label{cp5}
In Chapter \ref{cp3}, we showed that the diamond norm distance between  TVQCs and static channel models deviates significantly. As discussed later, the effect that this norm divergence has on error correction is not straightforward. To that end, we will numerically study the performance of QECCs when they operate over the proposed time-varying quantum channels. The error correction codes used for the numerical simulations are quantum turbo codes (QTCs) \cite{QTC,EAQTC,EXITQTC,EAQIRCC} and Kitaev toric codes \cite{toric}. We will asses the performance of these codes by means of the quantum outage probability of Chapter \ref{cp4}. Note that this is equivalent to what is done in the literature when benchmarking the performance of QECC codes operating over static channels respect to the quantum capacity.

We begin by briefly introducing the quantum error correction codes considered in the simulations. We follow by presenting the numerical results obtained from these codes when operating over the proposed time-varying channel models, and discuss how the time-variations of these channels degrade their performance. We close the chapter by using the quantum outage probability to benchmark their performance.

\section{Considered Quantum Error Correction Codes}\label{sec:qeccscp5}
The next section introduces the Kitaev toric codes and the unassisted QTCs that will be numerically simulated. The decoding of these codes is discussed in Appendix \ref{app:decoding} and the Monte Carlo numerical simulation methods in Appendix \ref{app:montecarlo}. As it was explained in Chapter \ref{ch:preliminary}, the classical simulation of QECCs on classical computers is based on the twirled approximations of the combined amplitude and phase damping channel. Therefore, the channels considered in the current simulations are the static and time-varying Pauli channels.

\subsection*{Kitaev Toric codes}\label{sub:kittoriccp5}
The Kitaev toric code is a quantum error correction code that belongs to the family of surface codes \cite{toric,surface}. It is defined as a 2D lattice of qubits with periodic boundaries and, thus, it has the shape of a torus. Figure \ref{fig:toricRep} shows a graphical representation of the toric code.

\begin{figure}[!h]
\centering
\includegraphics[width=0.75\linewidth]{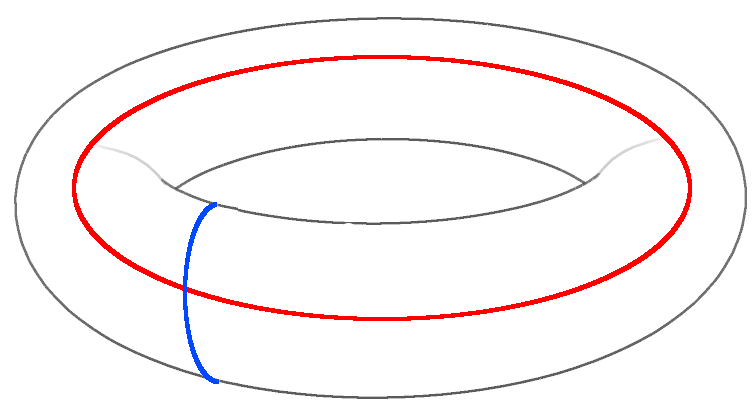}
\caption{Graphical representation of the toric code. The blue and red lines represent logical $\mathrm{X}$ and $\mathrm{Z}$ operators acting on the encoded logical qubit.}
\label{fig:toricRep}
\end{figure}

In the toric code, the stabilizers are implemented in the vertices and in the plaquettes (vertices of the dual lattice) of the lattice that constitutes the code. By using those vertex and plaquette qubits, the syndrome of the errors that occur in the toric code can be obtained. The errors of each of the grid qubits occur with the probabilities given by the PTA or CTA in consideration. Finally, the syndrome is fed to the minimum-weight perfect matching (MWPM) decoder that estimates the logical error that has occurred in an error correction round of the toric code (see Appendix \ref{app:decoding}).

A $d\times d$ Kitaev toric code is said to be a $[[2d^2,2,d]]$ quantum error correction code \cite{surface} since it is able to protect two logical qubits by means of $2d^2$ physical qubits with code distance $d$. For the numerical simulations that will be conducted in this chapter, Kitaev toric codes with distances $d\in\{3,5,7,9\}$ will be considered.

\subsection*{Quantum Turbo Codes}\label{sub:qtccp5}
The Quantum Turbo Codes considered in this chapter consist of the interleaved serial concatenation of unassisted\footnote{In this context, unassisted refers to error correction codes that are not aided by entanglement.} quantum convolutional codes acting as outer and inner codes \cite{EAQIRCC}. Figure \ref{fig:QTCsystemcp5} presents the full schematic representation of such a quantum error correction system. The $k$ input logical qubits that make up the information word $\ket{\psi_1}$ are first fed to the outer convolutional encoder $\mathcal{V}_1$, and encoded into $n'$ physical qubits with the help of ancilla qubits and memory qubits.

\begin{figure}[!h]
\centering
\includegraphics[width=\linewidth]{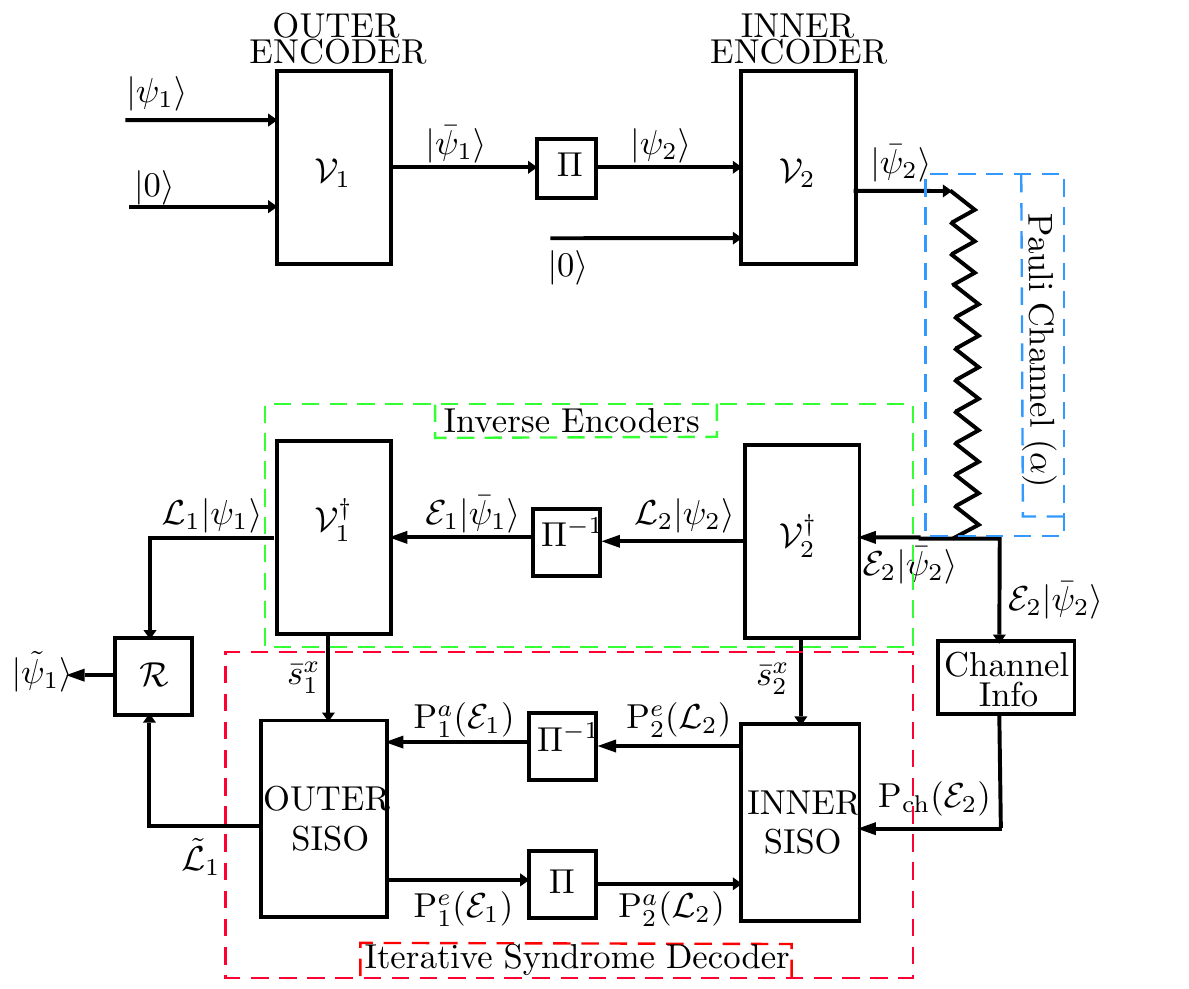}
\caption{Schematic of the QTC. $\mathrm{P}_i^a(.)$ and $\mathrm{P}_i^e(.)$ denote the a-priori and extrinsic probabilities related to each of the SISO decoders used for turbo decoding (see Appendix \ref{app:decoding}).}
\label{fig:QTCsystemcp5}
\end{figure}

The codeword $\ket{\bar{\psi}_1}$ consists of $n'$ physical qubits generated by the first encoder, which are then passed through a quantum interleaver $\Pi$, before being used as the input to the inner convolutional encoder $\mathcal{V}_2$. Such an encoder is an unassisted device that encodes the interleaved sequence of $n'$ qubits $\ket{\psi_2}$ into the codeword $\ket{\bar{\psi}_2}$ of length $n$, aided by ancilla qubits and memory qubits. The codeword $\ket{\bar{\psi}_2}$ is then transmitted through a quantum Pauli channel (static or time-varying) with overall error probability $p=p_\mathrm{x}+p_\mathrm{y}+p_\mathrm{z}$, which inflicts an $n$-qubit Pauli error $\mathcal{E}_2\in\mathcal{G}_n$ on the codeword. The Pauli channel is independently applied to each of the qubits of the stream $\ket{\bar{\psi}_2}$, and, consequently, each of the qubits may experience a bit-flip ($\mathrm{X}$ operator) with probability $p_\mathrm{x}$, a phase-flip ($\mathrm{Z}$ operator) with probability $p_\mathrm{z}$ or a combination of both ($\mathrm{Y}$ operator) with probability $p_\mathrm{y}$. Those probabilities will be defined depending on the static or TV twirl approximation being in consideration.

At the output of the Pauli channel, the state $\mathcal{E}_2\ket{\bar{\psi}_2}$ is fed to the inverse of the inner encoder $\mathcal{V}_2^\dagger$, which outputs the decoded state $\mathcal{L}_2\ket{\psi_2}$, where $\mathcal{L}_2\in\mathcal{G}_{n'}$ refers to the logical error suffered by the decoded state due to the operation of the channel; and the classical syndrome bits\footnote{Note that we emphasize the fact that when the QTC is unassisted, the syndrome measured just gives information of the $\mathrm{X}$ operators applied to the ancillas. This terminology is common for QTCs \cite{EAQTC}.} $\bar{s}_2^x$ obtained from $\mathrm{Z}$ basis measurements on the ancilla qubits. The corrupted logical qubits are then passed through a de-interleaver $\Pi^{-1}$ resulting in the state $\mathcal{E}_1\ket{\bar{\psi}_1}$, which is supplied to the inverse of the outer encoder $\mathcal{V}_1^\dagger$. The resulting output is the state $\mathcal{L}_1\ket{\psi_1}$, which corresponds to the information quantum state corrupted by a logical error $\mathcal{L}_1\in\mathcal{G}_k$; and the classical syndrome bits $\bar{s}_1^x$ obtained after measuring the ancilla qubits on the $\mathrm{Z}$ basis. The classical syndromes $\bar{s}_2^x$ and $\bar{s}_1^x$, obtained in the inverse decoders $\mathcal{V}_2^\dagger$ and $\mathcal{V}_1^\dagger$, respectively, are then provided to the iterative syndrome decoder made up of two serially concatenated SISO decoders, as shown in Figure \ref{fig:QTCsystemcp5}. Based on $\bar{s}_2^x$ and $\bar{s}_1^x$, the SISO decoders estimate the most probable logical error (see Appendix \ref{app:decoding}). Based on the aforementioned estimation, a recovery operation $\mathcal{R}$ is applied to the corrupted state $\mathcal{L}_1\ket{\psi_1}$, yielding the recovered output $\ket{\tilde{\psi}_1}$.

The QCCs used for the quantum turbo code in Figure \ref{fig:QTCsystemcp5} are the ``PTO1R'' unassisted convolutional encoders introduced in \cite{EAQIRCC}. The seed transformation\footnote{The seed transformation is a smaller encoding circuit that is repeatedly executed in order to construct the whole convolutional encoding unitary. The seed transformation $\mathcal{U}$ is represented using the decimal representation presented in \cite{EAQTC}.} associated with the convolutional encoders is

\begin{equation}
\begin{split}
\mathbf{\mathcal{U}} = \{1355, 2847, 558, 2107, 3330, 739, \\
 2009, 286, 473, 1669, 1979, 189  \}.
\end{split}
\end{equation}

Following previous work in the literature, we will refer to the serial concatenation of two similar interleaved PTO1R QCCs as the \textit{``PTO1R-PTO1R'' configuration}. Both QCC encoders have rate $1/3$ resulting in an error correction scheme of rate $R_Q = R_{in}\times R_{out} = 1/3 \times 1/3 = 1/9$. Each of the ``PTO1R'' encoders is aided by $3$ memory qubits in order to perform the convolutional encoding operation. The turbo code in consideration encodes blocks of $k=1000$ qubits.

\section{QECCs operating over TVQCs}\label{sec:qeccstvqcs}
Next, we provide the numerical results obtained for the performance of QECCs when they operate over TVQCs. Similar of what we did in Chapter \ref{cp4}, it is assumed that the realizations of the decoherence times for the TVQCs are the same for each qubits of a block of error correction. Thus, we are considering that the processes are fully correlated.

We begin by providing an example\footnote{In what follows, we select a set of realistic operating scenarios to asses the performance of error correction codes when operating on real hardware that exhibits time fluctuations \cite{decoherenceBenchmarking}.} to asses the impact that the deviations of the diamond norm distance may have on quantum error correction codes when implemented in the superconducting qubits of \cite{decoherenceBenchmarking}. Consider an ADCTA channel\footnote{Note that as we are considering qubits that are $T_1$-limited \cite{decoherenceBenchmarking}, the family of amplitude damping channels describe the decoherence experienced by those qubits.} like the one of scenario QA\_C5 ($c_\mathrm{v}(T_1)\approx 25\%$) at $t/\mu_{T_1}= 0.1$ (refer to Chapter \ref{cp3} for specific data regarding the scenarios discussed). The static version of this channel will have a depolarizing probability of $p\approx 0.05$. Knowing that the diamond norm for depolarizing channels can be simplified to $||\mathcal{N}_{D}(p_1)-\mathcal{N}_D(p_2)||_\diamond = 2|p_1-p_2|$, it can be shown that for a channel instance of the TVADCTA, when the diamond norm distance between the static and TV channels is $||\mathcal{N}_{\mathrm{ADCTA}}(\mu_{T_1})-\mathcal{N}_{\mathrm{ADCTA}}(\rho,\omega,t)||_\diamond = 0.29$, the corresponding depolarizing probability for the TV channel will be $p\approx 0.15$. Let us now consider the toric codes with $d\in\{7,9\}$ \cite{toric} and the QTC from \cite{EAQIRCC} designed for the depolarizing channel (the WER for those scenarios when the channels are static are the red curves in figures \ref{fig:TVvsStattoric} and \ref{fig:TVvsStatQTC}). These error correction codes are capable of operating at a Word Error Rate of $\mathrm{WER}\approx \{5\cdot 10^{-3},10^{-3},10^{-6}\}$, respectively, for $p=0.05$, while for $p=0.15$ they fail completely, resulting in $\mathrm{WER}\approx 1$. This shows that the selected QECCs will be operating with a significantly worse $\mathrm{WER}$ than that which would be expected based on a static channel model. Although there will also be channel instances for which the particular realization of $\{T_1(\omega,t)\}$ will result in operation within a region where the $\mathrm{WER}$ of the QECC is actually smaller than what is assumed for the static channel, on average, the effect of these realizations on the overall performance of the code will be much smaller than that of the events where the real channel action is worse than what is reflected by the static model. This winds up degrading the performance of the QECCs, which is reflected by a flattening of the $\mathrm{WER}$ curve of the code as a result of its waterfall region falling at a lower rate. Given that the existence of channel parameter fluctuations with time has been proven by recent experimental outcomes, it is likely that current error correction schemes will actually yield worse performance than what would be expected based on the results that have been obtained for static quantum channels. It is worth mentioning that the diamond norm distance we have chosen for this this example is pretty low. In fact,  Figure \ref{fig:boxplots} in Chapter \ref{cp3} portrays that values of up to $||\mathcal{N}_{\mathrm{ADCTA}}(\mu_{T_1})-\mathcal{N}_{\mathrm{ADCTA}}(\rho,\omega,t)||_\diamond \approx 1.4$ can occur for a scenario with $c_\mathrm{v}(T_1)\approx 25\%$ at $t/\mu_{T_1}= 0.1$. Note that events with high diamond norm will all be due to events where the channel instance is very noisy since as $||\mathcal{N}_{\mathrm{ADCTA}}(\mu_{T_1})-\mathcal{N}_{\mathrm{ADCTA}}(\rho,\omega,t)||_\diamond = 2|p(\mu_{T_1})-p(\omega)|$, $p(\mu_{T_1})$ is fixed and $p(\omega)\geq 0$, the realizations $||\mathcal{N}_{\mathrm{ADCTA}}(\mu_{T_1})-\mathcal{N}_{\mathrm{ADCTA}}(\rho,\omega,t)||_\diamond \geq 2|p(\mu_{T_1})|$ must occur because $p(\omega)\geq 2p(\mu_{T_1})$. For the example we are considering, the events $||\mathcal{N}_{\mathrm{ADCTA}}(\mu_{T_1})-\mathcal{N}_{\mathrm{ADCTA}}(\rho,\omega,t)||_\diamond \geq 0.1$ will all belong to this kind of TV channel realizations.

\begin{figure*}[!h]
\centering
\includegraphics[width=\linewidth]{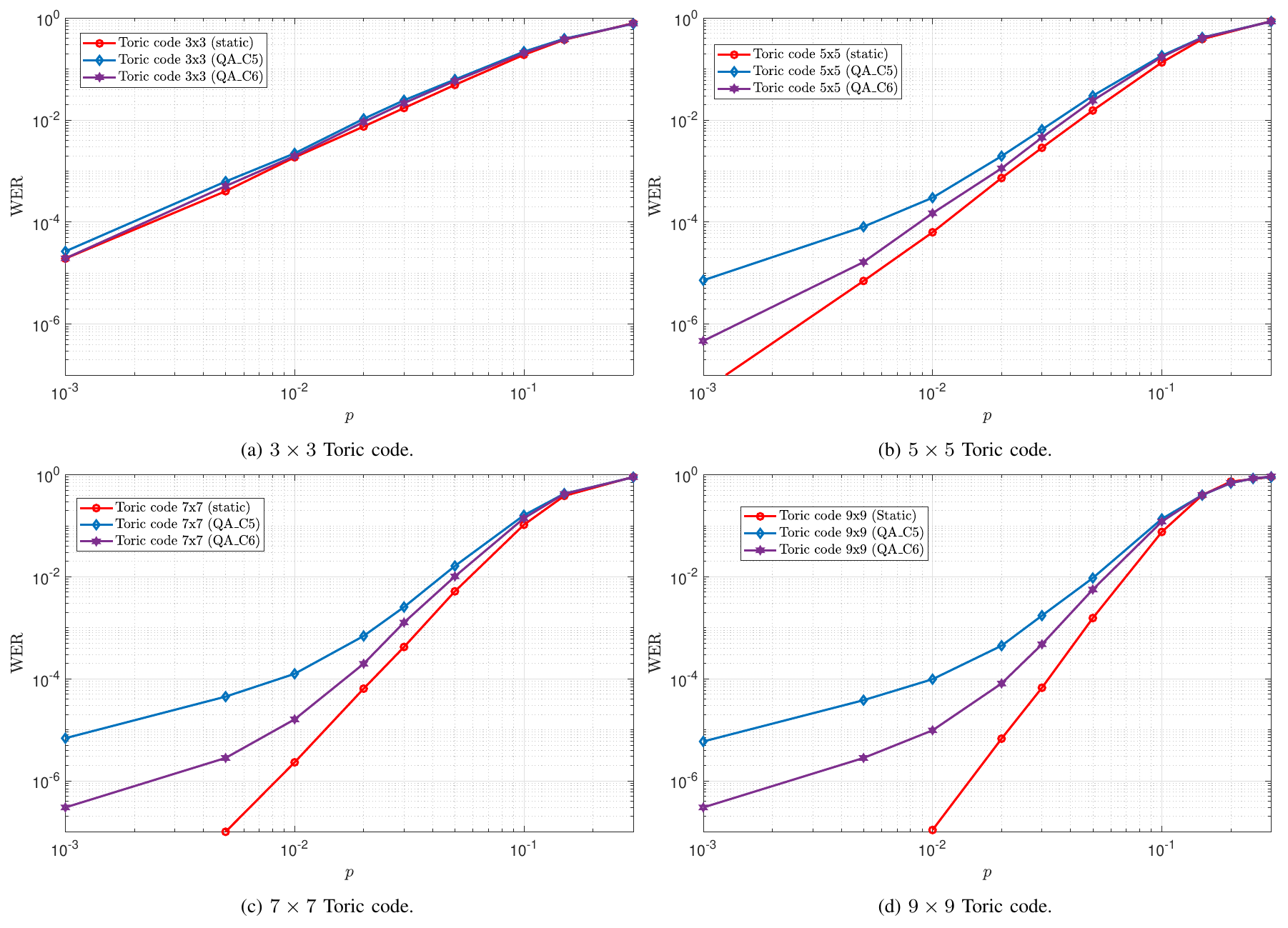}
\caption{Performance of the $d\times d$ Kitaev toric codes with $d\in\{3,5,7,9\}$ \cite{toric}.  Each panel presents performance for i) the static channel (red), ii) the TVADCTA for scenario QA\_C5 ($c_\mathrm{v}(T_1)\approx 25\%$) (blue) and iii) the TVADCTA for scenario QA\_C6 ($c_\mathrm{v}(T_1)\approx 22\%$) (purple).}
\label{fig:TVvsStattoric}
\end{figure*}

Let us now consider the toric codes with $d\in\{3,5\}$ \cite{toric} (the WER for those scenarios when the channels are static are the red curves in Figure \ref{fig:TVvsStattoric}). These surface codes operate at a $\mathrm{WER}\approx \{8\cdot 10^{-2},2\cdot 10^{-2}\}$ for $p=0.05$ and at a $\mathrm{WER}\approx 1$ for $p=0.15$. Considering that the error rate at the static point is not low enough and given that the waterfall region is pretty flat, for these codes the $\mathrm{WER}$ for a good channel realization will not be very different from that of a bad channel realization. Thus, although these codes will experience time-variation effects, these effects will not be as significant as for the previosuly mentioned codes.

\begin{figure*}[!h]
\centering
\includegraphics[width=\linewidth]{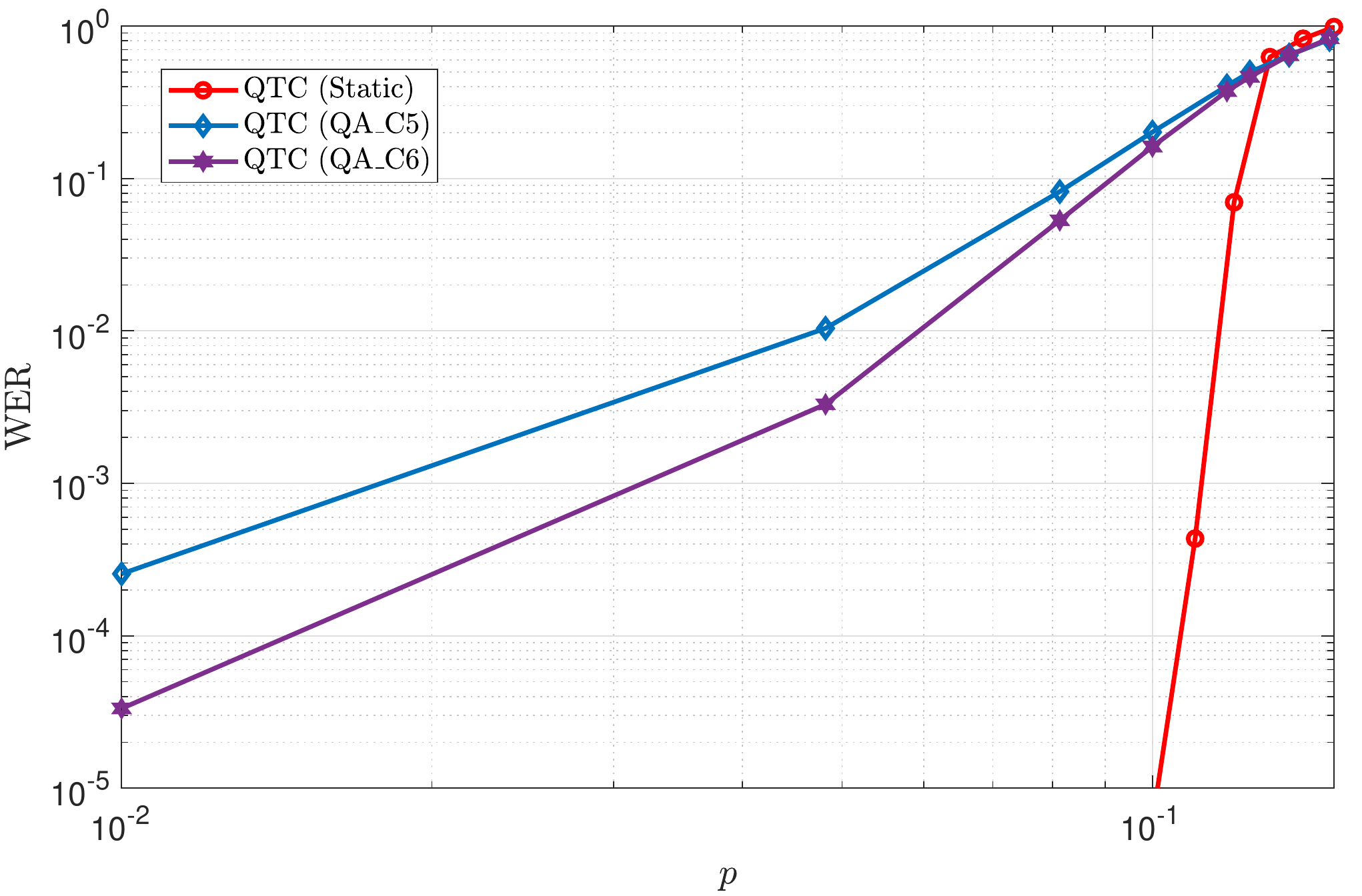}
\caption{Performance of the QTC discussed in \cite{EAQIRCC}. Each panel presents performance for i) the static channel (red), ii) the TVADCTA for scenario QA\_C5 ($c_\mathrm{v}(T_1)\approx 25\%$) (blue) and iii) the TVADCTA for scenario QA\_C6 ($c_\mathrm{v}(T_1)\approx 22\%$) (purple).}
\label{fig:TVvsStatQTC}
\end{figure*}

We prove this intuition by simulating the discussed QECCs for the time-varying amplitude damping Clifford twirl approximated (TVADCTA) channel obtained for the configurations QA\_C5 ($c_\mathrm{v}(T_1)\approx 25\%$) and QA\_C6 ($c_\mathrm{v}(T_1)\approx 22\%$) and comparing the results with the corresponding static channel. Figures \ref{fig:TVvsStattoric} and \ref{fig:TVvsStatQTC} show the results of these simulations. The blue curves are for the QA\_C5 scenario and the purple ones for the QA\_C6 scenario. A quick inspection of this figures reveals the expected outcomes. It can be observed that the low distance toric codes ($d\in\{3,5\}$, figures \ref{fig:TVvsStattoric}a and \ref{fig:TVvsStattoric}b) do not exhibit a significant degradation of their error rate performance (especially the $3\times 3$ toric code; the $5\times 5$ toric code shows signs of performance degradation, but at very low depolarizing probabilities). However, the higher distance toric codes ($d\in \{7,9\}$, figures \ref{fig:TVvsStattoric}c and \ref{fig:TVvsStattoric}d) and the QTC (Figure \ref{fig:TVvsStatQTC}) are significantly affected by the time fluctuations. This is reflected by a reduction in the steepness of the waterfall region (it becomes flattened) in the case of the time-varying channel, which compromises the excellent performance that QECCs exhibit for static channel models. The impact of the coefficient of variation can also be observed, as the scenario with higher $c_\mathrm{v}$ is flattened more.

The above discussion suggests that the manner in which QECCs are affected by the proposed time-varying quantum channel model is a combination of two factors: how much the decoherence parameters fluctuate (which is quantified by the coefficients of variation) and the steepness of the $\mathrm{WER}$ curve when the QECCs operate over the static channel. This occurs because if the code presents a rather flat $\mathrm{WER}$ curve (as a function of the depolarizing probability) in the static channel, the difference in $\mathrm{WER}$ between benign and malign channel realizations of the time-varying channel is not significant, and thus, on average, the overall $\mathrm{WER}$ will be similar to that of the corresponding static channel. Since, everything else being equal, shorter error correction codes present flatter waterfall regions, they will be less susceptible to the mismatch. On the other hand, the contrary effect is observed for error correction codes that exhibit a steep waterfall region in the static channel. In these cases, the difference between a good and a bad channel realization might be either complete decoding success or utter failure, implying that on average, the overall performance in the time-varying channel will experience substantial degradation. Thus, the better the code performance in the static channel, the more it will be affected in the time-varying channel model.

\section{Benchmarking the performance of QECCs operating over TVQCs}\label{sec:outQECCcp5}
We finish this chapter by proposing the quantum outage probability as the information theoretical metric to asses the performance of QECCs when operating over the TVQCs proposed when there is full correlation between the realizations of the decoherence parameters for each of the qubits of a code block. Since the ADCTA depolarizing static channel and its time-varying version TVADCTA belong to the family of Pauli channels, the information theoretical benchmarks that will be considered are the hashing bound (for the static channels) and the quantum hashing outage probability (for the TV channels).

The following simulations are based on a $9\times 9$ Kitaev toric code with coding of rate $R_\mathrm{Q}=1/81$ and on the previous QTC of $R_\mathrm{Q}=1/9$. We have selected these two codes since they experience a large degradation (with respect to the static channel) under channel time variations.
\begin{figure}[!h]
\centering
\includegraphics[width=\linewidth]{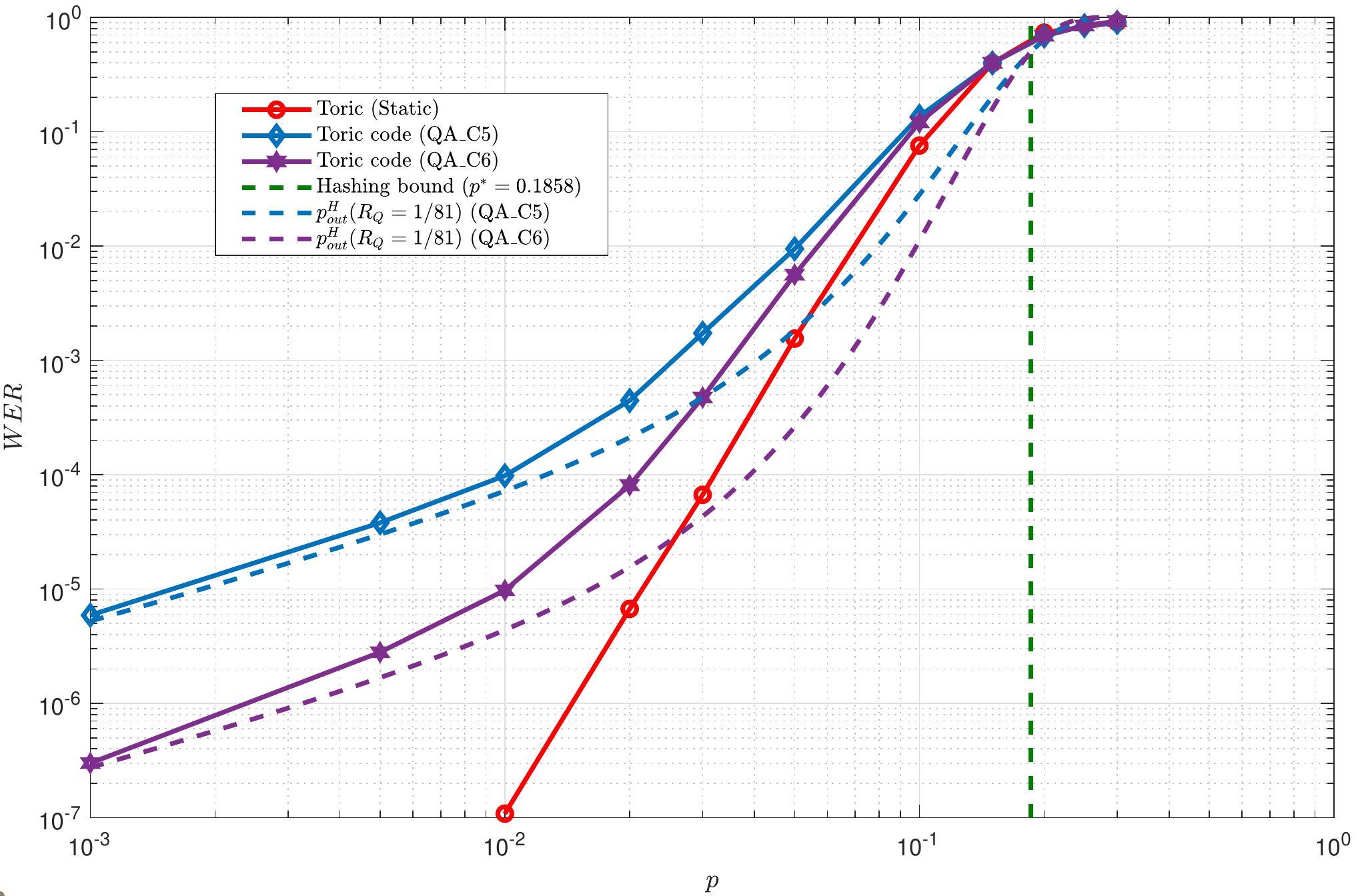}
\caption{Performance of the $9\times 9$ Kitaev toric code \cite{toric}. The $[[162,2,9]]$ quantum error correction code has rate $R_\mathrm{Q}=1/81$. The hashing bound and the hashing outage probabilities are also plotted.}
\label{fig:ToricTVout}
\end{figure}
By considering the same previous channel scenarios, Figures \ref{fig:ToricTVout} and \ref{fig:QTCTVout} show the $\mathrm{WER}$ versus $p$ of these two codes when they operate on the static channels (red curve) and on the TVQCs (blue and purple). The figures also plot their respective quantum hashing outage probabilities. Both figures show that the quantum hashing outage probability presents a shape similar to the performance of the QECCs. Therefore, the quantum outage probability curves are a good indicator of the performance that some QECC will have when it operates over a TVQC with certain coefficients of variation of their decoherence parameters. Also, the closer the $\mathrm{WER}$ curves are to the corresponding quantum outage probability, the better the performance of the codes will be.

\begin{figure}[!h]
\centering
\includegraphics[width=\linewidth]{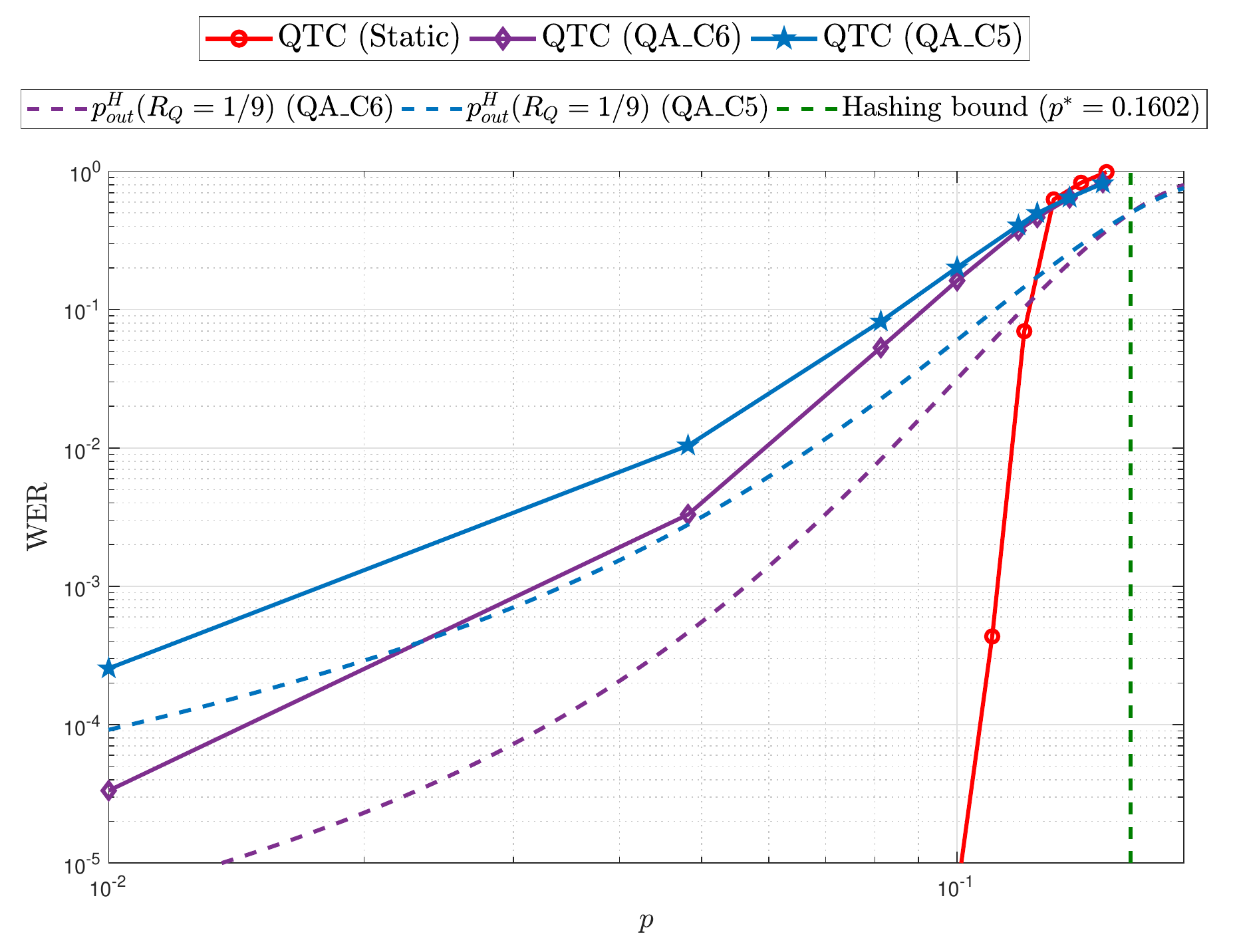}
\caption{Performance of the quantum turbo code from \cite{EAQIRCC}. The quantum error correction code has rate $R_\mathrm{Q}=1/9$ and encodes blocks of $1000$ qubits. The hashing bound and the hashing outage probabilities are also plotted.}
\label{fig:QTCTVout}
\end{figure}

It is also interesting to point out that at low depolarizing probabilities, the $\mathrm{WER}$ curves approach to the quantum hashing outage probability curves. This effect is clearly observed for the $9\times 9$ Kitaev toric code in Figure \ref{fig:ToricTVout}. For the QTC in Figure \ref{fig:QTCTVout} such saturation is not observed mainly due to the range of values of $\mathrm{WER}$ and $p$ considered in the simulation. A possible explanation for this effect is that at low depolarizing probabilities, $p$, the errors produced by the channel are all mainly of low weight (i.e., these errors will affect only a few qubits in the codeword), and the codes will be capable of correcting all of them. Therefore, the code for this scenario of low depolarizing probabilities behaves as an ideal code as the one used to compute the quantum hashing outage probability.

To quantify the distance to the hashing outage bound, we use a metric similar to the one proposed in \cite{CQiso}. It measures the gap in dBs between the depolarizing probabilities of a code and the hashing outage at a given $\mathrm{WER} = \chi$:
\begin{equation}\label{eq:distancPout}
\delta_{\mathrm{out}}(@\chi) = 10\log\left(\frac{p(p_{\mathrm{out}}^\mathrm{H}=\chi)}{p(\mathrm{WER}_{\mathrm{code}} = \chi)}\right).
\end{equation}

Therefore, the lower $\delta_{\mathrm{out}}$ is, the better the performance of the code will be since its performance will be more similar to the asymptotically achievable error rate. Table \ref{tab:hashdist} shows the $\delta_{\mathrm{out}}(@10^{-3})$ for the two codes in consideration.

\begin{table}[!ht]
\centering
\begin{tabular}{|cccc|}
\hline
\multicolumn{1}{|c}{Code} & \multicolumn{1}{c}{Static} & \multicolumn{1}{c}{QA\_C5} & \multicolumn{1}{c|}{QA\_C6} \\ \hline
\multicolumn{1}{|c}{$9\times 9$ Toric} & $3.93$ & $2.04$ & $2.81$  \\
\multicolumn{1}{|c}{QTC} & $1.79$ & $1.75$ & $1.67$  \\ \hline
\end{tabular}
\caption{Distance to the quantum hashing outage probability of the codes in consideration when WER = $10^{-3}$. $\delta_{\mathrm{out}}(@10^{-3})$ is given in dBs.}
\label{tab:hashdist}
\end{table}

It can be seen that the $9\times 9$ Kitaev toric codes are farther from the quantum hashing outage probability than the QTC in consideration at WER=$10^{-3}$. This fact is also true for the distance of the codes to their hashing bounds when they operate over static channels so the result is not surprising. Thus the QTC seems to be a better error correction method than the Toric code since its performance is closer to the hashing bound and its quantum rate is lower. However, it is important to state that this QTC requires codewords of $9000$ physical qubits while the $9\times 9$ Toric code just requires $162$.

Finally, note that for the toric code, the value of $\delta_{\mathrm{out}}(@10^{-3})$ for scenario QA\_C5 is smaller than for scenario QA\_C6. At a first glance, this seems counterintuitive, since the coefficient of variation is higher for QA\_C5. However, this can be explained by observing the $\mathrm{WER}$ curve for scenario QA\_C5 approaches earlier the quantum hashing outage probability curve than for scenario QA\_C6.

\FloatBarrier
\section{Discussion} \label{cp5_sec: conclusion}

In light of the discussion presented in this chapter, and based on the experimental evidence that decoherence parameters vary with time, the proposed time-varying quantum channel models will be very important when studying and optimizing error correction codes capable of protecting qubit-based quantum computers from the dynamics of the decoherence phenomenon. However, it should be pointed out that the proposed time-varying quantum channel model would be relevant for protocols involving a large number of error correction cycles (e.g., a quantum memory or a long quantum algorithm such as the one in \cite{rsaRounds}). If a very short algorithm or if a QECC is run just once or a few number of times, the parameters will not fluctuate as much, and the effects of the proposed channel model will not be as noticeable. Even if the current state-of-the-art experimentally implemented QECCs will not be significantly affected by the TVQC (as they are generally short codes such as the $3\times 3$ toric code discussed above), the proposed channel model will be essential when better QECCs are implemented, such as the $7\times 7$ and $9\times 9$ toric codes for the near term, or the more advanced QTCs or QLDPCs for the long term.

In addition, we have proposed the quantum outage probability as a benchmark of the performance of the QECCs operating over TVQCs. We conclude that the quantum hashing outage probability does represent well the behaviour of QECCs on those dynamic noise scenarios. It has also been observed that the performance of the codes approach the quantum hashing outage probability limit when the depolarizing noise is very small.

\clearemptydoublepage
\part{Quantum Error Correction: Optimization of Quantum Turbo Codes} \label{part2}
\clearemptydoublepage

In this second part of the dissertation we focus on quantum error correction schemes. Sepecifically, the objective is to optimize the error correction capabilities of Quantum Turbo Codes (QTC) when they operate under different scenarios. First, we lower the error floor performance of QTCs by considering interleavers with some construction rather than just performing a random scrambling of the coded random information that comes out of the first convolutional encoder of the QTC. Additionally, we study the performance of this family of quantum error correction methods when the Channel State Information (CSI) is unknown for the decoder. We consider this scenario to be a more realitic one since, in the reality, the decoder of a QECC must operate with estimates of the level of noise of the quantum channel. In particular, we study how the error correction capabilities of QTCs are compromised when there exists a mismatch between the true channel probability and the estimated channel probability and discuss existing estimation methods that aid the decoder to operate properly. Moreover, we propose an on-line estimation method that uses the information being processed in the decoder in order to obtain accurate estimates of the channel probability so that the decoder of QTCs is able to operate ``blind'' to the true CSI without losing performance. Both the depolarizing channel (Clifford twirl approximation) and the general Pauli channel (Pauli twirl approximation) are considered. Before delving into the proposed schemes and with the aim of providing some context, we begin by including a general overview of QTCs and their development.

Classical turbo codes, introduced by Berrou et al. in \cite{berrouglavieux}, marked a breakthrough in classical coding theory as the first practical codes closely approaching Shannon's channel capacity \cite{shannon}. Motivated by their success, QTCs were proposed by Poulin et al. as the interleaved serial concatenation of Quantum Convolutional Codes (QCC) \cite{QTC}. The use of a serial structure in \cite{QTC} is explained by the laws of quantum mechanics, in particular by the underlying classical-to-quantum isomorphism \cite{isomorphism} applied in the design of quantum error correction codes (QECC): Classical turbo codes are generally systematic, interleaved parallel concatenations of convolutional codes but in the quantum domain such a systematic parallel structure is prohibited by the no-cloning theorem. 

The performance of the QTCs proposed in \cite{QTC} is not on par with that of classical serial turbo codes. The reason is that unassisted QCCs cannot be both non-catastrophic and recursive at the same time, properties that are required for the inner code of a turbo code so that the minimum distance grows with the block length and the iterative decoding algorithm achieves convergence. Consequently, non-catastrophic, non-recursive QCCs were used in \cite{QTC}, resulting in codes with bounded minimum distance. Utilizing the entanglement-assistance (EA) techniques\footnote{{Entanglement-assistance techniques utilize entangled qubits pre-shared between the sender and the receiver, which simplifies the construction of stabilizer QECCs by relaxing the stringent commutativity conditions that the stabilizers need to satisfy.}} for QECCs presented in \cite{EAQECC}, this issue was addressed by Wilde et al. in \cite{EAQTC}, which proposed QTCs operating within  1 dB of their corresponding EA hashing bounds. Taking a step further, Babar et al. proposed in \cite{EXITQTC} an EXtrinsic Information Transfer (EXIT) chart technique to narrow the gap to the Hashing bound, resulting in codes with a performance as close as 0.3 dB  to the hashing bounds.\begin{table*}[h!]
\centering
\begin{tabular}{r |@{\foo} l}
2009 & QTCs \cite{QTC}. Non-catastrophic, non-recursive. $\delta \gtrapprox 2$ dB.\\
2014 & EAQTCs \cite{EAQTC}. Non-catastrophic, recursive. $\delta \approx 1$ dB. \\
2015 & EXIT optimized \cite{EXITQTC}. $\delta \approx 0.3$ dB. High error floor.\\
 & QIRCC \cite{QIrCC}. Like \cite{EXITQTC} with lower error floor.\\
2016 & EAQIRCC \cite{EAQIRCC}. Optimized for the PTA channel.\\
 & QURC \cite{QURC}. Higher rate maintaining performance.\\
 & Fully parallel decoder for QTCs \cite{paralDec}. Lower complexity.\\
2018 & QTC operating over channels with memory \cite{MemQTC}.\\
2019 & Short block, high-rate QTCs \cite{shortQTC}. 
\end{tabular}
\caption{Timeline of the evolution of QTCs.}
\label{tab:timelineQTC}
\end{table*} However, these QTC codes present error floors that are higher than in the original distance-spectra-based codes in \cite{EAQTC}. Based on the use of Quantum IrRegular Convolutional Codes (QIRCC) as outer codes, the authors in \cite{QIrCC} showed a performance in the turbo-cliff region similar to that of the EXIT optimized codes in \cite{EXITQTC} but with lower error floors. An entanglement-assisted version of QIRCCs was introduced in \cite{EAQIRCC} with the aim of designing efficient concatenated QECCs for asymmetric depolarizing channels. A Quantum Unity Rate Code (QURC) aided QTC scheme was proposed in \cite{QURC} in order to improve the performance of the outer code without experiencing code rate reduction due to the inner code. A modified decoder for QTCs that relies on parallel computing was developed in \cite{paralDec} in order to reduce the decoding computation time. Later, QTCs were studied in \cite{MemQTC} when they operate over Pauli channels with memory and a modified decoding algorithm was proposed in order to improve the performance of the codes in such scenarios. Finally, short block, high-rate QTC schemes with multiple rates presenting performances near the hashing bound were proposed and studied in \cite{shortQTC}.

Table \ref{tab:timelineQTC} presents the evolution of QTCs presented before. In spite of their short history, quantum turbo codes are a promising family of error correction methods for the quantum computers of the post-NISQ-era or for quantum communications due to their closeness to the hashing bound, degenerate decoding and non-complex decoding algorithm. 

Note that in Part \ref{part1} of this thesis we proposed quantum channels that vary in time and discussed how those degrade the performance shown by quantum error correction codes. However, it is important to point out that in this part of the dissertation we consider that the QTCs operate over quantum channels that are static.

\chapter{Optimization of the error floor performance of QTCs via interleaver design} \label{cp6}
Classical turbo codes are constructed by the interleaved serial or parallel concatenation of two convolutional codes \cite{berrouglavieux}. The interleaver has a central role in turbo coding due to the fact that it is an element that permutes the symbols of the coded block so that the error locations are randomized. This way, the bursts of errors that are generated in the decoding of the first convolutional code can be more efficiently corrected due to their scrambling. Initially, the concatenation of the convolutional codes was done by means of a \textit{random interleaver}, that is, an interleaver whose permutation pattern is selected randomly. However, due to the importance of such code concatenation element, the construction of classical interleavers with certain structure has been extensively studied in the literature \cite{SrandomConstruction, Lazcano, twoStep, IntDesignforTC, Kovaci, Vafi, turboCoding, SrandomFlexible}. This way, the authors seeked the optimization of the performance of classical turbo codes by designing permutation patterns that scramble the information in the best possible way. 

Quantum turbo codes are also constructed by concatenating two quantum convolutional codes \cite{QTC}. However, the concatenation must be done in a serial way since the no-clonning theorem prohibits the copy of information in the quantum realm that is needed for the parallel concatenation of QCCs. The QTCs proposed in the literature \cite{QTC, isomorphism, EAQECC, EAQTC, EXITQTC, QIrCC, EAQIRCC, QURC, MemQTC} use the quantum analogue of the random interleaver. Regarding QTCs, the only preliminary regarding practical quantum interleaver design was addressed in \cite{paralDec}, where the authors discussed the implementation of an odd-even interleaver for reducing the complexity of the parallel QTC decoding algorithm they proposed. Nevertheless, no performance improvement was observed by the implementation of such scrambling method. Consequently, and based on the classical-to-quantum isomorphism \cite{isomorphism} that is typically applied in the design of quantum error correction codes, we aim to integrate the classical interleaving design methods into the paradigm of quantum turbo coding. In this manner, we expect to improve the error correction capabilities of those codes, specially in the error floor region. Simulation results show that the QTCs designed using the proposed interleavers present the same behavior in the turbo-cliff region as the codes with random interleavers in \cite{EXITQTC}, while the performance in the error floor region is improved by up to two orders of magnitude. Simulations also show reduction in memory consumption, while the performance is comparable to or better than that of QTCs with random interleavers.

\section{System Model}\label{sub:system}

The Quantum Turbo Codes considered in this chapter consist of the interleaved serial concatenation of unassisted QCCs acting as outer codes and entanglement-assisted QCCs as inner codes\footnote{Note that this QTC is similar to the QTC considered in Chapter \ref{cp5} but with an EA inner QCC.}. Figure \ref{fig:QTCsystem} presents the full schematic representation of such quantum error correction system. The $k$ input logical qubits that compose the information word $\ket{\psi_1}$ are first fed to the outer $[[n_1,k_1,m_1]]$ unassisted convolutional encoder $\mathcal{V}_1$ and encoded into $n'=k\frac{n_1}{k_1}$ physical qubits with the help of $(n_1-k_1)$ ancilla qubits and $m$ memory qubits. The $n'$ physical qubits that form the codeword $\ket{\bar{\psi_1}}$ generated by the first encoder are then passed through a quantum interleaver $\Pi$ before being input to the inner convolutional encoder $\mathcal{V}_2$. Such an encoder is an $[[n_2,k_2,m_2,c]]$ entanglement-assisted encoder that encodes the interleaved sequence of $n'$ qubits $\ket{\psi_2}$ into the codeword $\ket{\bar{\psi_2}}$ of length $n = n'\frac{n_2}{k_2}=k\frac{n_2}{k_2}\frac{n_1}{k_1}$, aided by $(n_2-k_2-c)$ ancilla qubits, $c$ pre-shared EPR pairs and $m_2$ memory qubits. Codeword $\ket{\bar{\psi_2}}$ is then transmitted through a quantum channel that corrupts such quantum information. For this chapter, we will consider the depolarizing channel (Clifford twirl approximation) and, thus, with qubit-wise depolarizing probability $p$ an $n$-qubit Pauli error $\mathcal{E}_2\in\mathcal{G}_n$ will be inflicted to the codeword.

\begin{figure}[!h]
\centering
\includegraphics[width=\linewidth]{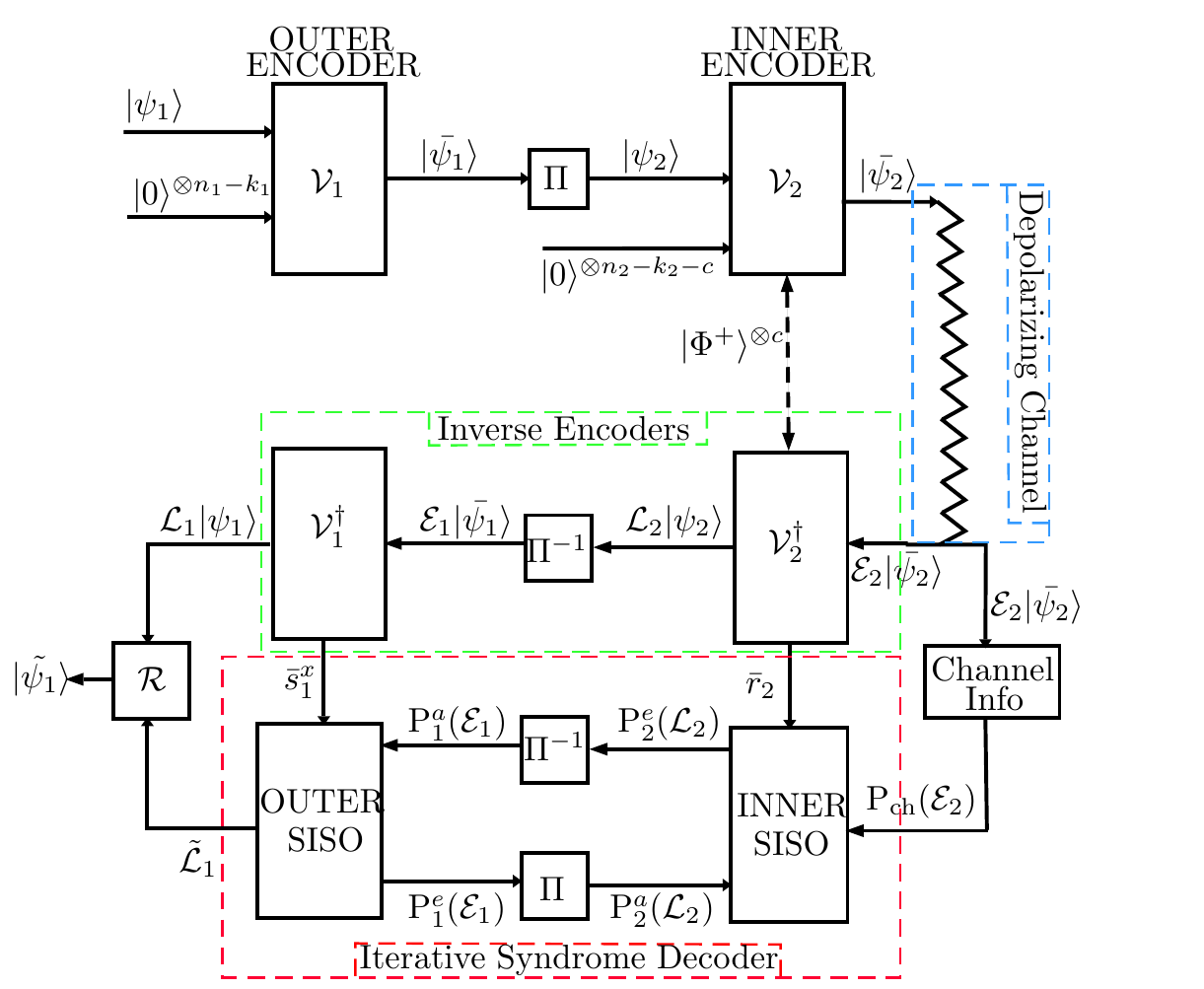}
\caption{Schematic of the Entanglement-Assisted Quantum Turbo Code (EAQTC). Note that the $c$ pre-shared EPR pairs $\ket{\Phi^+}$ are needed for the inner encoder to be both recursive and non-catastrophic \cite{EAQTC}. $\mathrm{P}_i^a(.)$ and $\mathrm{P}_i^e(.)$ denote the a-priori and extrinsic probabilities related to each of the SISO decoders used for turbo decoding.}
\label{fig:QTCsystem}
\end{figure}

At the output of the depolarizing channel, the state $\mathcal{E}_2\ket{\bar{\psi_2}}$ is fed to the inverse of the inner encoder $\mathcal{V}_2^\dagger$, which outputs the decoded state $\mathcal{L}_2\ket{\psi_2}$, where $\mathcal{L}_2\in\mathcal{G}_{n'}$ refers to the logical error suffered by the decoded state due to the operation of the channel; and the classical syndrome bits $\bar{r}_2=(\bar{s}^x_2,\bar{e}_2^{x,z})$ obtained from $\mathrm{Z}$ basis measurements on the ancilla qubits and Bell measurements on the pre-shared EPR pairs\footnote{Syndrome measurement of the EAQCC includes measurements of ancillas $\bar{s}^x$ and pre-shared ebits $\bar{e}^{x,z}$.}. The corrupted logical qubits are then passed through a de-interleaver $\Pi^{-1}$ resulting in the state $\mathcal{E}_1\ket{\bar{\psi_1}}$, which is supplied to the inverse of the outer encoder $\mathcal{V}_1^\dagger$. The resulting output is the state $\mathcal{L}_1\ket{\psi_1}$, which corresponds to the information quantum state corrupted by a logical error $\mathcal{L}_1\in\mathcal{G}_k$; and the classical syndrome bits $\bar{s}^x_1$ obtained after measuring the ancilla qubits on the $\mathrm{Z}$ basis. The classical syndromes $\bar{r}_2$ and $\bar{s}_1^x$, obtained in the inverse decoders $\mathcal{V}_2^\dagger$ and $\mathcal{V}_1^\dagger$, respectively, are then provided to the iterative syndrome decoder consisted of two serially concatenated SISO decoders, as shown in Figure \ref{fig:QTCsystem}. Based on $\bar{r}_2$ and $\bar{s}_1^x$, as well as the channel information $\mathrm{P}_\mathrm{ch}(\mathcal{E}_2)$, both SISO decoders engage in degenerate iterative decoding \cite{QTC,EAQTC} to estimate the most likely error coset $\tilde{\mathcal{L}}_1$ that has corrupted the information quantum state. Based on such an estimation, a recovery operation $\mathcal{R}$ is applied to the corrupted state $\mathcal{L}_1\ket{\psi_1}$, yielding the recovered output $\ket{\tilde{\psi_1}}$.

As already mentioned, the aim of this chapter is to show that interleaver design is beneficial to reduce the error floors of QTCs. To that end, the inner and outer QCCs we consider to construct the concatenated error correcting system presented in Figure \ref{fig:QTCsystem} are the EXIT-optimized codes introduced in \cite{EXITQTC}. The parameters of the inner and outer codes are presented in Table \ref{table:configurations}. Such configuration, utilizing a random interleaver for the serial concatenation of the constituent QCCs, was studied extensively in \cite{EXITQTC}. Therefore, it is an excellent benchmark for the interleaver constructions proposed in this chapter.

From the configuration parameters presented in Table \ref{table:configurations}, it can be seen that the QTCs used in this chapter will be rate $1/9$ turbo codes with an entanglement consumption rate of $6/9$. The noise limit $p^*$ can be found from the entanglement-assisted hashing bound, which for such parameters is $p^*=0.3779$ \cite{EAQTC}.

\begin{table}[H]
\centering
\begin{tabular}{|llllll|}
\hline
Config.             & Encoder & \multicolumn{1}{c}{$R_\mathrm{Q}$} & \multicolumn{1}{c}{E} & $m$ & \multicolumn{1}{c|}{Seed Transformation $\mathbf{\mathcal{U}}$} \\ \hline
\multirow{2}{*}{\begin{tabular}[c]{@{}l@{}}\\EXIT- \\ optimized\end{tabular}}  & Outer            &     $1/3$      &        \multicolumn{1}{c}{$0$}    &           $3$ &  \begin{tabular}[c]{@{}l@{}}$ \{1048, 3872, 3485, 2054, $\\ \enskip$983, 3164, 3145, 1824, $\\ \enskip$987, 3282, 2505, 1984\}_{10}$ \end{tabular}           \\ \cline{2-6}
                                 & Inner            &     $1/3$       &      $2/3$      &       $3$     &         \begin{tabular}[c]{@{}l@{}}$ \{4091, 3736, 2097, 1336, $\\ \enskip$1601, 279, 3093, 502, $\\ \enskip$1792, 3020, 226, 1100\}_{10}$ \end{tabular}       \\ \hline
\end{tabular}
\caption[Parameters of the Quantum Convolutional Codes (QCC) encoders.]{Parameters of the QCC encoders. $R_\mathrm{Q}$ and $\mathrm{E}$ refer to the coding rate and the entanglement consumption rate, respectively. $m$ refers to the memory qubits. The seed transformations $\mathcal{U}$ are represented using the decimal representation presented in \cite{EAQTC}.}
\label{table:configurations}
\end{table}

\section{Classical Interleavers for QTCs}\label{sub:srandom}
The concatenation between the outer and inner QCCs is done by the quantum interleaver $\Pi$ of size $N$, which is an $N$ to $N$-qubit symplectic transformation\footnote{The notation we utilize to describe the operation of a quantum interleaver is the so-called symplectic notation as used in \cite{QTC}.} composed by a permutation $\pi$ of the $N$ qubit registers and a tensor product of single-qubit symplectic transformations. It acts by multiplication on the right on $\bar{\epsilon}\in\Upsilon^{-1}([\mathcal{G}_N])$ as
\begin{equation}\label{eq:qinterleaver}
(\epsilon_1,\cdots,\epsilon_N)\rightarrow(\epsilon_{\pi(1)}K_1,\cdots,\epsilon_{\pi(N)}K_N),
\end{equation}
where $K_1,\cdots,K_N$ are some fixed symplectic matrices\footnote{Recall the $\Upsilon$ map in section \ref{ch:preliminary}.} acting on $\epsilon_{\pi(j)}\in\Upsilon^{-1}([\mathcal{G}_1])$, where $\mathcal{G}_N$ refers to the $N$-fold Pauli group \cite{QTC}.

In order to use classical interleavers in QTCs, they must properly match the definition of quantum interleavers in \eqref{qinterleaver}. Classical interleavers consist only of permutations of the information arrays with no such thing as individual symplectic transformations. Consequently, those operations will be taken as identity operators acting on the registers $K_i = I, \forall i$; and so the only non-trivial operation of such quantum interleavers will be the scrambling of the qubits in the quantum information stream\footnote{The random interleavers used in \cite{QTC, isomorphism, EAQECC, EAQTC, EXITQTC, QIrCC, EAQIRCC, QURC, MemQTC} also follow this approach.}. Following this approach, the quantum interleavers presented here will be $N$ to $N$ qubit symplectic transformations acting as
\begin{equation}\label{eq:qsrandom}
(\epsilon_1,\cdots,\epsilon_N)\rightarrow(\epsilon_{\pi(1)},\cdots,\epsilon_{\pi(N)})
\end{equation}
to the qubit data stream, and where the permutation pattern $\pi$ corresponds to the classical interleaver under consideration.

\subsection{Considered Classical Interleavers}\label{sub:consideredInterleavers}
All the QTC schemes in \cite{QTC, isomorphism, EAQECC, EAQTC, EXITQTC, QIrCC, EAQIRCC, QURC, MemQTC} are based on \textit{random interleavers}, that is, the interleaving patterns $\pi$ are selected at random. However, it is known from classical turbo codes that the usage of interleavers with some added structure is beneficial to reduce either the error floors or the memory requirements \cite{SrandomConstruction, Lazcano, twoStep, IntDesignforTC, Kovaci, Vafi, turboCoding, SrandomFlexible}. In order to show that interleaver design is also beneficial when implementing quantum turbo codes, we consider the following three types of classical interleavers. 

The first classical interleavers considered here are the \textit{$S$-random} interleavers \cite{SrandomConstruction}. They are randomly generated by imposing the following condition on the interleaving distance or spread ($S$): 
\begin{equation}\label{eq:srandom}
|\pi (i)-\pi (j)|>S \text{ for $i$ and $j$ with } |i-j|\leq S.
\end{equation}
 If the heuristic recommendation $S<\sqrt{\frac{N}{2}}$ is satisfied, where $N$ is the blocklength, $S$-random interleavers can usually be produced in reasonable time by repeatedly generating random integers until condition \eqref{srandom} is satisfied \cite{SrandomConstruction, Kovaci}. However, as indicated in \cite{SrandomFlexible}, reasonable values of $S$ are sometimes lower when $N$ is large.
We next consider the deterministic \textit{Welch-Costas} \cite{turboCoding} interleaver, which is defined by two parameters: $p = N + 1$, a prime number, and $\alpha$, a primitive element modulo $p$. Once  $\alpha$ is chosen for the selected blocklength, it is proven in \cite{turboCoding} that
\begin{equation}\label{eq:welchcostas}
\pi(i) = (\alpha^i \bmod{(N+1)}) - 1,\quad \forall i\in\{0,\cdots, N-1\}
\end{equation}
forms a permutation which is selected as the interleaving rule. Although the spreads obtained with Welch-Costas interleavers are low, their dispersion parameter, $\eta$, is maximized. The dispersion of an interleaver $\pi$ is a parameter that represents the randomness of the permutation, and it is calculated as \cite{turboCoding}
\begin{equation}
\eta = \frac{|D(\pi)|}{\binom{N}{2}},
\end{equation}
where $|\cdot|$ indicates cardinality and $D(\pi)$ is the set of \textit{displacement vectors} of $\pi$ defined as 
\begin{equation}
D(\pi) = \{ (j-i,\pi(j)-\pi(i))| 0\leq i < j < N\}.
\end{equation}

The final classical interleaver we will consider for quantum error correction is  the \textit{JPL} interleaver \cite{turboCoding}, recommended in the ``CCSDS Recommendation for
Telemetry Channel Coding'' standard \cite{ccsds} because it has good spreading and dispersion parameters. The permutations that define this family of interleavers are constructed by applying the following algorithm:

\begin{itemize}
\item Factorize the length of the interleaver $N=k_1k_2$, where $k_1=8$ usually.
\item For $s=1$ to $s=N$ do
\begin{enumerate}
\item $m = (s-1)\bmod{2}$
\item $i = \lfloor\frac{s-1}{2k_2}\rfloor$
\item $j = \lfloor\frac{s-1}{2}\rfloor-ik_2$
\item $t = (19i+1)\bmod{\frac{k_1}{2}}$
\item $q = t \bmod{8} + 1$
\item $c = (p_qj+21m)\bmod{k_2}$
\item $\pi(s)=2(t+c\frac{k_1}{2}+1)-m$
\end{enumerate}
\end{itemize}
where $p_q$ is defined as the primes $p_1=31,p_2=37,p_3=43,p_4=47,p_5=53,p_6=59,p_7=61,p_8=67$.

\section{Numerical Results}\label{sec:results}
In this section we present Monte Carlo simulations (see Appendix \ref{app:montecarlo} for details) to asses the performance of quantum turbo codes when using the interleavers proposed in Section \ref{sub:consideredInterleavers}. To that end, the EXIT chart optimized QCCs defined by the parameters in Table \ref{table:configurations} have been used. The interleaver length is $N=3000$, so that the input quantum information consists of $1000$ qubits. As in \cite{EXITQTC}, a maximum of $15$ iterations has been set for the turbo decoder. The operational figure of merit selected in order to evaluate the performance of these quantum error correction schemes is the Word Error Rate.

The following interleavers have been implemented: 
 
\begin{enumerate}
 \item $S$-random interleaver with parameter $S=25$, 
 \item Welch-Costas interleaver with parameter $\alpha=2987$ and 
 \item JPL interleaver.
\end{enumerate}
 
The spread ($S$) and dispersion ($\eta$) parameters of such designed interleavers are indicated in Table \ref{table:parameters}.

\begin{table}[!h]
\centering
\begin{tabular}{|lll|}
\hline
Name & Spread ($S$) & Dispersion ($\eta$) \\ \hline
Random \cite{turboCoding} & \multicolumn{1}{c}{$1$} & \multicolumn{1}{c|}{$\approx 0.81$} \\ 
\multicolumn{1}{|c}{$S$-random}      & \multicolumn{1}{c}{$25$}         & \multicolumn{1}{c|}{$0.8136$}            \\ 
Welch-Costas   & \multicolumn{1}{c}{$1$}          & \multicolumn{1}{c|}{$1$}                 \\ 
\multicolumn{1}{|c}{JPL}       & \multicolumn{1}{c}{$16$}         & \multicolumn{1}{c|}{$0.35$}              \\ \hline
\end{tabular}
\caption{Spread and dispersion parameters.}
\label{table:parameters}
\end{table}

Figure \ref{fig:results} shows the results of the Monte Carlo simulations for the different QTC schemes in the depolarizing channel. It can be seen that the performance of all QTCs is similar in the turbo-cliff region, achieving a gap\footnote{{Note that the gap is calculated in a similar fashion as in \cite{EXITQTC}, that is, the distance between the QTC convergence threshold $p=0.35$ and the EA 
Hashing limit is taken as $10\log_{10}(0.3779/0.35)=0.3$ dB.}} to the EA Hashing limit of $0.3$ dB, which is also the gap obtained by the random interleaver QTCs of \cite{EXITQTC}. However, different behaviors can be observed in the error floor region. As depicted in Figure \ref{fig:results}, both the JPL 
and the $S$-random interleavers present a much lower error floor than that of random interleaving. Specifically, the error floor of the original ramdomly interleaved QTCs is of the order of $10^{-2}$, while the error floor of the $S$-random interleaver is around $10^{-4}$ and that of the JPL is of order of $10^{-3}-10^{-4}$. Notice that the error floor of the JPL interleaver presents a larger slope than that of the other interleavers, achieving levels close to those of the $S$-random interleaver when the channel quality improves. This means that quantum error correction systems using $S$-random and JPL interleavers present error floors around two orders of magnitude better than those obtained with random interleavers. 

\begin{figure}[H]
\centering
\includegraphics[width=\linewidth]{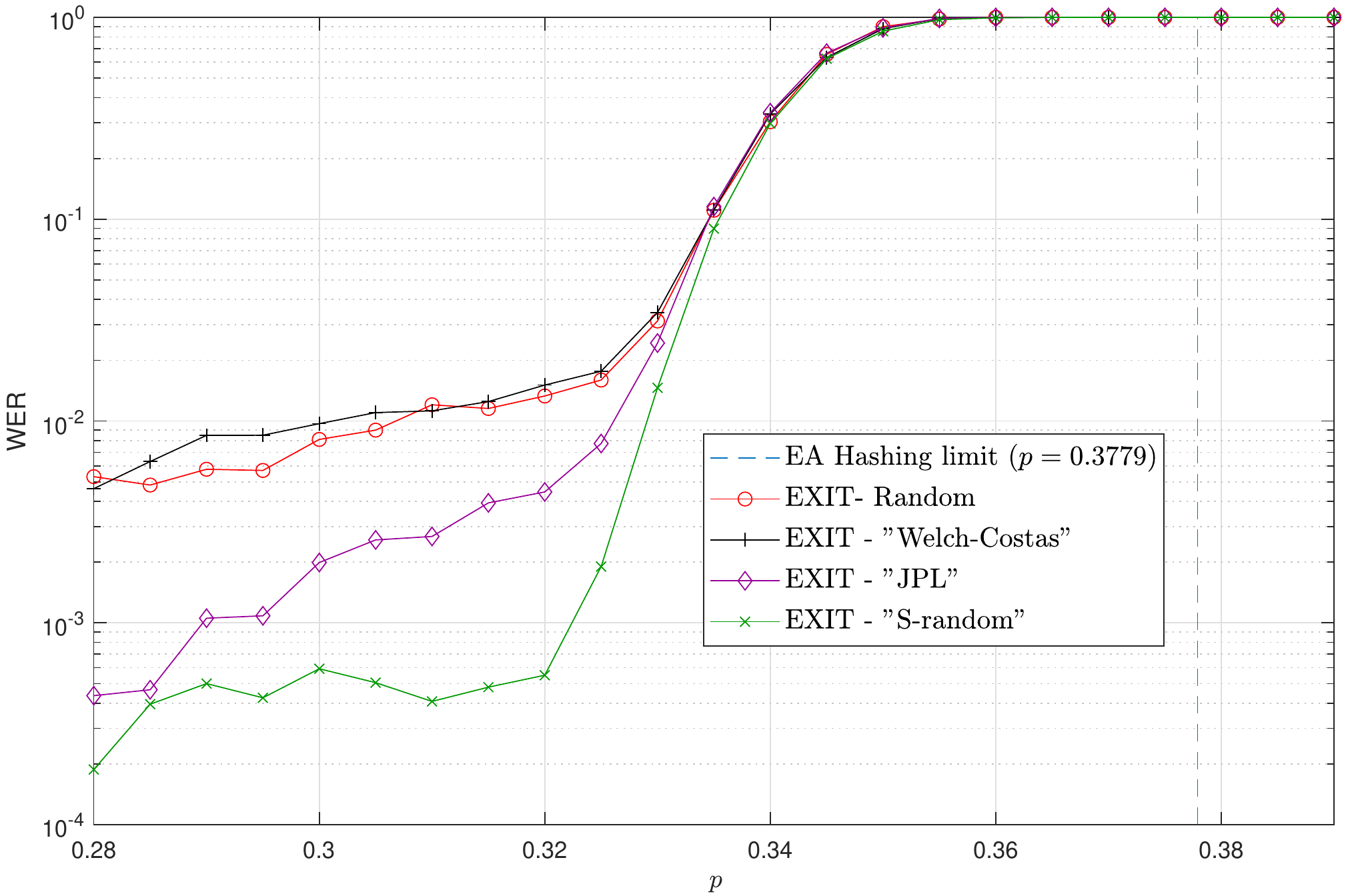}
\caption{Word error rate (WER) performance curves for the $1/9$-QTCs in Table \ref{table:configurations} when different interleaving methods are used.}
\label{fig:results}
\end{figure}

Figure \ref{fig:results} shows that the performance of the Welch-Costas interleaver is similar to that of the baseline random interleaver. However, the benefit of using the Welch-Costas interleaver is that its memory requirements are much smaller, as the permutation is defined just by parameters $N$ and $\alpha$. On the other hand, the random (and $S$-random) permutations must be stored completely. This means that using the Welch-Costas interleaver is a good option for systems with strict storage requirements. The memory requirements for the JPL interleaver are also lower than for the random and $S$-random ones, as the permutation can be defined by parameters $N$, $k_1$ and the primes $p_q$, while its performance is better than that of random interleavers but worse than that of $S$-random ones. Welch-Costas interleavers do have less memory requirements than JPL interleavers, but their performance is far worse in the error floor region. {Table \ref{table:memory} summarizes the memory requirements for each of the interleavers. Note that both random and $S$-random interleavers require a memory consumption that increases linearly with $N$, while the memory requirements for the JPL and Welch-Costas interleavers are constant.} Regarding the overall complexity of the decoding algorithm we have corroborated that it is practically the same for all interleavers.

\begin{table}[!h]
\centering
\begin{tabular}{|lll|}
\hline
Name & Parameters & Storage Requirements \\ \hline
Random \cite{turboCoding} & \multicolumn{1}{c}{$\pi$} & \multicolumn{1}{c|}{$N$} \\ 
\multicolumn{1}{|c}{$S$-random}      & \multicolumn{1}{c}{$\pi$}         & \multicolumn{1}{c|}{$N$}            \\ 
Welch-Costas   & \multicolumn{1}{c}{$N$,$\alpha$}          & \multicolumn{1}{c|}{$2$}                 \\ 
\multicolumn{1}{|c}{JPL}       & \multicolumn{1}{c}{$N$,$k_1$,$p_q$}         & \multicolumn{1}{c|}{$10$}              \\ \hline
\end{tabular}
\caption{{Memory requirements.}}
\label{table:memory}
\end{table}

{As explained in \cite{turboCoding}, the spread parameter is related to the number of distinct low weight error patterns that are possible to appear in the decoding, while the dispersion is related to the multiplicities of those existing low weight errors. Consequently, high spread is desired to avoid those harmful error sequences, while a dispersion parameter that is close to that of a random interleaver ($\eta \approx 0.81$) will mean that those harmful errors will happen less regularly. This is why}
in classical turbo coding it is known that to achieve good error correction performance in the error floor region, an interleaver must have a high spread parameter and a dispersion parameter similar to the dispersion of random permutations \cite{turboCoding}. The results presented in Figure \ref{fig:results} show that such rationale also applies to the QTC scenario.

\section{Discussion}\label{cp6_conclusion}

We have adapted three interleaving design methods used for classical turbo codes to the quantum domain, demonstrating the importance of interleaver construction in quantum turbo codes. By using Monte Carlo simulations, we have shown that these constructions do enhance the error correction capability of quantum codes in the error floor region, leading to error rate improvements of up to two orders of magnitude. We have also shown that memory requirements can be lowered by using specific interleaver designs, while the performance in the error floor region is comparable or even better than that of the original EXIT optimized QTCs.

\chapter{QTCs without channel state information} \label{cp7}
Decoding for QTCs typically assumes that perfect knowledge about the channel is available at the decoder, i.e. such a decoder has perfect channel state information (CSI). However, in realistic systems, such information must be estimated, and thus, there exists a mismatch between the true channel information and the estimated one. The effect of channel mismatch on QECCs has been studied for quantum low density parity check (QLDPC) codes in \cite{QLDPCmismatch,QLDPCMismatchMethods}, which showed that such codes are pretty insensitive to the errors introduced by the imperfect estimation of the depolarizing probability and proposed methods to improve the performance of such QECCs when the depolarizing probability is estimated. For classical parallel turbo codes, several studies showed their insensitivity to SNR mismatch \cite{ReedMismatch,SummersMismatch,MismatchNoEstimation,MismatchSOVA}. In such works, it was found that the performance loss of parallel turbo codes is small if the estimated SNR is above the actual value of the channel. However, QTCs can only be constructed as serially concatenated convolutional codes, and classical studies about channel information mismatch for such codes \cite{VarianceMismatchSCCC,SCCCmismatchTalwar} showed that such structures are more susceptible to channel identification mismatch, suffering a significant performance loss if the SNR is either overestimated or underestimated. This strongly suggests that (serially concatenated) QTCs will be more sensitive to CSI mismatch.

In this section, we first analyze the performance loss incurred by QTCs due to the channel's depolarizing probability mismatch. Then, we study how different off-line estimation protocols affect the overall QTCs' performance and provide some heuristic guidelines to selecting the number of quantum probes required to have an estimate of the depolarizing channel (Clifford twirl approximation) that keeps the performance degradation of such codes low. Moreover, we propose an on-line estimation procedure by utilizing a modified version of the turbo decoding algorithm that allows, at each decoding iteration, estimating the channel information to be fed to the inner SISO decoder. Finally, we propose an extension of the on-line estimation method presented for the depolarizing channel to cover the more general case of the Pauli channel with asymmetries (Pauli twirl approximation).

\section{Quantum Turbo Decoder Performance with Depolarizing Probability Mismatch}\label{sec:mismatch}
In this section, we study the performance sensitivity of quantum turbo codes when transmitting over the depolarizing channel, assuming that their turbo decoders have a mismatched value, $\hat{p}$, of the actual channel depolarizing probability, ${p}$.
To that end, Monte Carlo computer simulations of the entanglement-assisted QTC presented in section \ref{cp6} have been carried out.

Figure \ref{fig:mismatch} shows the WER versus the mismatched depolarizing probability, $\hat{p}$, fed to the turbo decoder. Six curves are plotted in the same figure, corresponding to six different depolarizing channels with probabilities $p\in\{0.35,0.34,0.33,0.32,0.31,0.30\}$. Figure \ref{fig:mismatch} shows the strong relationship between the mismatch sensitivity and the actual channel depolarizing probability, ${p}$. The smaller ${p}$ is, the wider the {flat}\footnote{We define the flat region as the zone in the sensitivity curve where the WER performance is not considerably degraded as a consequence of the mismatch.} region is.
 As expected, the ability of a QTC to correct errors improved as the depolarizing probability of the channel, ${p}$, decreased. 
Comparing these results to the mismatched sensitivity work presented in \cite{QLDPCmismatch,QLDPCMismatchMethods} for QLDPCs, we can conclude that QLDPCs were more robust against channel mismatch than quantum turbo codes, since QLDPCs had a behavior flatter than that of QTCs when the depolarizing probability was under-estimated. \mbox{This type} of behavior can be explained by comparing QTCs with classical turbo codes.
 In the realm of classical turbo codes, the same kind of\footnote{Note that in the classical realm, mismatch means that the noise variance of the channel has to be estimated.} behavior is encountered for parallel and serial turbo codes \cite{ReedMismatch,SummersMismatch,MismatchNoEstimation,MismatchSOVA}. That is, classical parallel turbo codes present a sensitive region flatter than serial ones \cite{VarianceMismatchSCCC,SCCCmismatchTalwar} when the variance of the AWGN channel is underestimated, which may be explained by the fact that in serial configurations, contrary to what happens in parallel turbo codes, the channel state information is only used in the inner decoding stage, and thus, errors due to mismatch only appear at the first decoding stage and propagate to the outer decoder in the turbo decoder.
 
\begin{figure}[H]
\centering
\includegraphics[width=\linewidth]{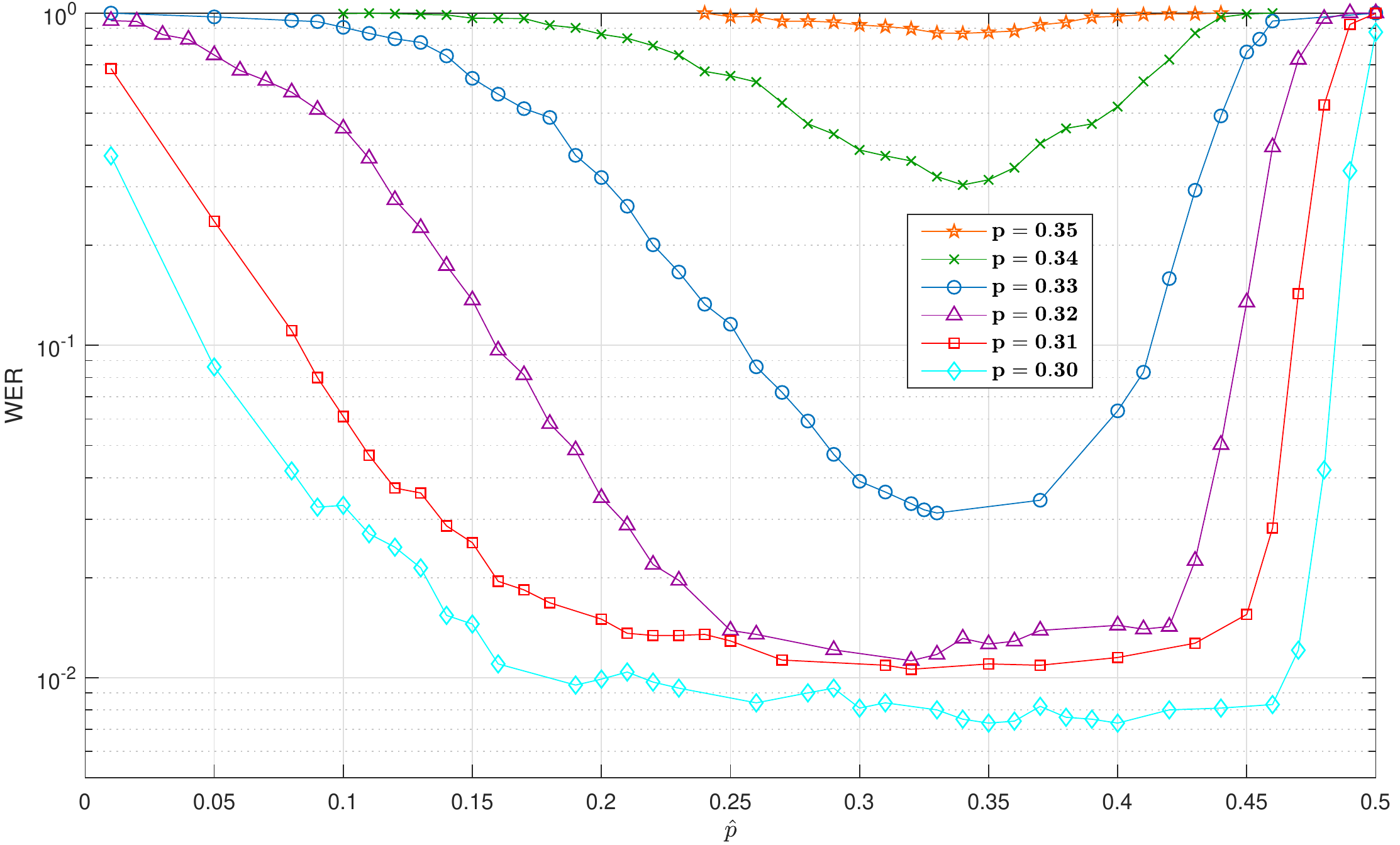}
\caption{WER variation for QTCs using an EXIT-random interleaver as a function of the mismatched depolarizing probability, $\hat{p}$. Each one of the curves corresponds to a value of the actual channel depolarizing probability. {The QTCs considered here have an entanglement consumption rate of $6/9$.}}
\label{fig:mismatch}
\end{figure}

\subsection{Interleaver Impact}\label{subsec:interleaving}
In Chapter \ref{cp6}, we have shown that the selection of the interleaver plays a central role in the design of QTCs due to their ability to reduce the error floors in such codes. This raises the question of whether different classes of interleavers may also influence the sensitivity behavior of the QTC under channel mismatch. To investigate this possibility, the previous simulation was also carried out for the following two interleavers, which according to Figure \ref{fig:results} were the ones with the lowest error floors:
\begin{itemize}
\item $S$-random interleaver with parameter $S=25$ and
\item JPL interleaver.
\end{itemize}

\begin{figure}[h!]
\centering
\includegraphics[width=\linewidth]{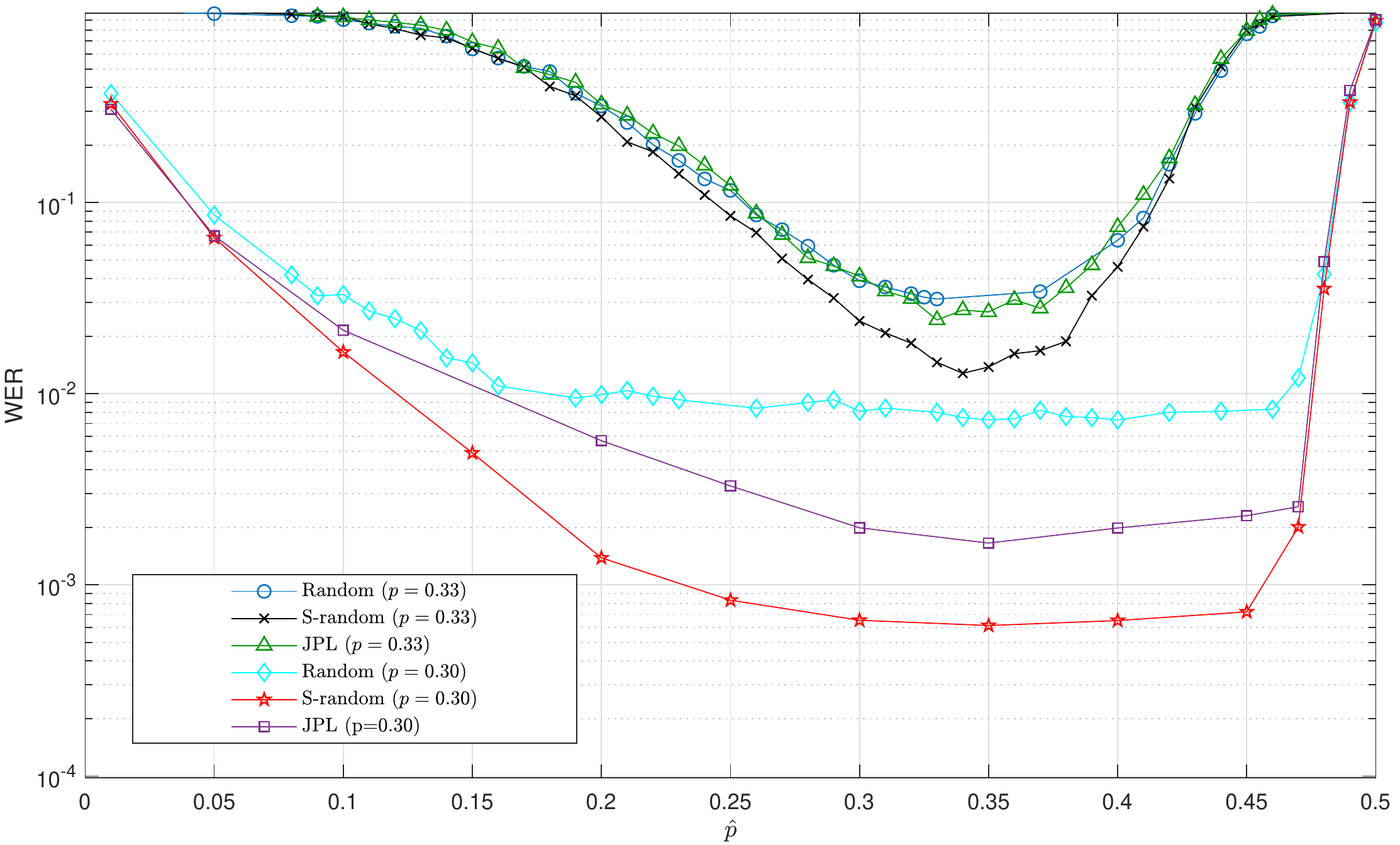}
\caption{WER variation for the QTCs with three different interleavers as a function of the mismatched depolarizing probability $\hat{p}$. The true depolarizing {probabilities selected are $p=0.33$ and $p=0.30$.}}
\label{fig:interleaversMism}
\end{figure}

The next set of simulations was performed by setting the depolarizing channel probability to $p=0.33$ {and $p=0.30$}, which by looking at the curves in Figure \ref{fig:results} were the values of $p$ where the WERs started to diverge {in the waterfall region and where the WERs lied in the error floor region, respectively}. The results are plotted in Figure \ref{fig:interleaversMism}. Note that in the same figure, we also included the curves corresponding to the EXIT-random interleaver previously plotted in Figure \ref{fig:mismatch}. Observe that the curves of WER versus $\hat{p}$ were very similar for the all three interleavers under consideration {as the flat region for each of them had approximately the same width. It is true that for $p=0.30$, both the JPL and the $S$-random interleavers presented a lower WER in the flat region, but that was due to the fact that their error floor was lower}. Therefore, we conclude that even though the selection of the interleaver had a significant impact on the resulting error floor of the QTC, the influence that interleavers had on the mismatch sensitivity is negligible.

The explanation for this invariance of the mismatch sensitivity to the interleaver choice is that the depolarizing probability was used only in the inner SISO decoder. This means that the extra errors that appeared in the decoding due to the mismatched depolarizing probability will be generated in the inner part of the turbo decoder, and they will propagate through the interleaver to the outer convolutional decoder. Consequently, even if using the $S$-random or the JPL interleaver means a higher error correction capability in the error floor region, the sensitivity of the overall decoder to channel mismatch, defined by the inner decoder, will not be affected by interleaver choice.

\section{Estimating the Depolarizing Probability}\label{sec:estimation}

The problem of quantum channel identification (quantum process tomography \cite{ProcessTomography}) is of fundamental relevance for several quantum information processing tasks. In the context of quantum error correction, it is often assumed that the decoder knows exactly the quantum channel. Unfortunately, such an assumption may not be valid as the channel is usually unknown. Consequently, the QTC decoder should be provided with an estimate of the channel. In this case, since the quantum channel acting on an arbitrary quantum state $\rho$ was the depolarizing channel, $\mathcal{N}_\mathrm{D}$, the decoder of the QTC should be fed with an estimate of the channel depolarizing probability, $p$. It will be assumed that $p$ remains constant for at least the duration of a QTC channel block.

\subsection{Off-Line Estimation Framework}\label{subsec:directestimation}

In this framework, the channel depolarizing probability is estimated before block transmission can begin. Note that if the channel remains constant during the transmission of all blocks (i.e., the quantum channel is unknown but constant) then the QTC rate loss will be asymptotically negligible as the number of transmitted blocks increases. On other hand, if the channel varies at every transmitted block (but remains constant during a block), i.e. the TVQC model discussed in Chapters \ref{cp4} and \ref{cp5}, this off-line estimation method substantially reduces the overall rate of the QTC.

\subsubsection{Quantum Channel Estimation}\label{subsec:estimation}
In general, quantum channel identification requires suitably prepared quantum states to be subjected to the operation of the quantum channel, $\Gamma(p)$, \cite{Estimation}. Figure \ref{fig:estimationScheme} shows the general estimation scheme used for channel identification. It used $N$ equal quantum probes, $\sigma$, that were then passed through unitary operators, $U_i$, while being submitted to the action of the target quantum channel $m$ times. The output quantum state, $\sigma_f(p)$, which depends on the parameter $p$, was then measured qubit-wise to obtain the classical information $(x_1,\cdots,x_N)$ that will be used to get the estimate $\hat{p}=p_{est}(x_1,\cdots,x_N)$.

\begin{figure}[H]
\centering
\includegraphics[width=\linewidth]{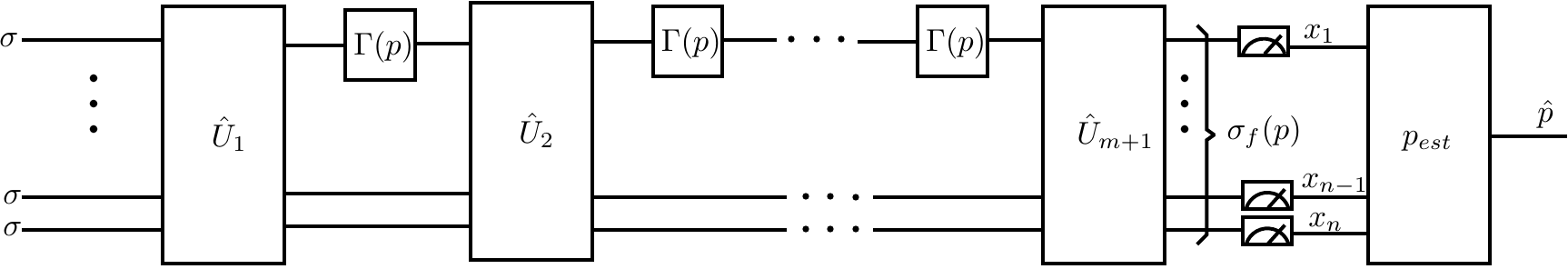}
\caption{General channel identification method for a $p$ parameter dependent quantum channel $\Gamma(p)$ \cite{Estimation}. For the scenario of this section, $\Gamma(p)=\mathcal{N}_\mathrm{D}$. This general scenario uses $n$ quantum probes $\sigma$ and $m$ channel invocations. $U_1,\cdots,U_{m+1}$ are parameter independent transformations that can be effectively selected to estimate $p$. $p_{est}$ is an estimator function that uses the classical information $(x_1,\cdots,x_N)$ obtained from measuring the output quantum state $\sigma_f(p)$ in order to get the estimate $\hat{p}$.}
\label{fig:estimationScheme}
\end{figure}

Since the estimation of ${p}$ depends on the measurements obtained from the state $\sigma_f(p)$ and such measurements in quantum physics are statistically distributed, the estimated depolarization probability, $\hat{p}$, will be a random variable. Consequently, the aim of channel identification is to chose a procedure that yields estimations of $p$ with minimal fluctuations around the actual parameter value. The way to quantify these fluctuations is by computing the variance of the estimation, $\mathrm{var}(\hat{p}) = \mathrm{E}\{(\hat{p}-p)^2\}$, where it is assumed that the estimator is unbiased, i.e., $\mathrm{E}\{\hat{p}\}=p$. Obviously, the variance will depend on the actual choice of the estimator used. For the purposes of this section, we will assume that the estimation scheme achieves the information-theoretical {optimal} performance, that is the {quantum Cram\'er-Rao bound}. In that sense, the results in this section will be bounds for off-line channel estimation.

The {quantum Cram\'er-Rao bound} states that the variance on any estimator is bounded below \mbox{as \cite{Estimation}:}
\begin{equation}\label{eq:CramerRao}
\mathrm{var}(\hat{p})\geq \frac{1}{J(p)},
\end{equation}
where $J(p)$ is the quantum Fisher information, which depends only on the output quantum state, $\sigma_f(p)$, but neither on the choice of measurement\footnote{Note that the classical Fisher information $F(p)$ depends on the choice of measurement, but as $F(p)\leq J(p)$, the Cram\'er-Rao bound can be rewritten as a function of the measurement-independent quantum Fisher information \cite{Estimation}.} $(x_1,\cdots,x_N)$, nor on the applied estimation scheme (refer to Figure \ref{fig:estimationScheme}). It is calculated as:
\begin{equation}\label{eq:Fisher}
J(p) = \mathrm{Tr}[\sigma_f(p)\hat{L}^2(p)],
\end{equation}
where $\hat{L}(p)$ is the symmetric logarithmic derivative (SLD) defined implicitly via:
\begin{equation}\label{eq:sld}
\frac{\partial \sigma_f(p)}{\partial p}=\frac{1}{2}[\hat{L}(p)\sigma_f(p)+\sigma_f(p)\hat{L}(p)].
\end{equation}

{Quite generally, there exists a measurement procedure that asymptotically attains the quantum Cram\'er-Rao bound \cite{FisherAchievable}.} From the fact that $\sigma_f(p)$ depends only on the resources used to test the channel, the same dependence will hold for $J(p)$, and in turn for the {quantum Cram\'er-Rao bound}. Note that there are many possible combinations of these resources, which will lead to different Fisher informations \cite{Estimation,Fisher}, but for the purposes of this study, we will only consider some of them.

\begin{enumerate}

\item Number of channel invocations: For our encoding-decoding system in Figure \ref{fig:QTCsystem}, sequential channel invocations are not considered. The reason is that once a quantum state goes through the operation of the depolarization channel, it cannot be sent again through the channel. Therefore, the number of channel invocations will be set to $m=1$.

\item Unitary $\hat{U}_1$ transformation: The goal of the unitary transformation, $\hat{U}_1$, applied to the $N$ input quantum probes, $\sigma$, is to introduce correlations among the quantum probes. In the particular case where the transformation is diagonal,
it results in independent instances of the quantum probes, i.e., in an {independent channel use protocol}. Figure \ref{fig:estimationSchemeReduced} shows such an estimation protocol. 

Furthermore, if $N$ is set to one, the above scheme reduces to the simplest estimation protocol, called {single-qubit, single-channel (SQSC)}. If we denote by $J_1(p)$ the Fisher information of this SQSC estimation protocol, then for any $N$ greater than one, it can be shown \cite{Estimation} that the corresponding overall Fisher information $J_N(p)$ is given by $J_N(p)=NJ_1(p)$. Therefore, for $N$ channel invocations, the quantum Cram\'er-Rao bound is:
\begin{equation}\label{eq:CramerRaoN}
\mathrm{var}(\hat{p})\geq \frac{1}{NJ_1(p)},
\end{equation}
so that the variance bound decreases linearly with $N$. In this section, we will use the estimation protocol shown in Figure \ref{fig:estimationSchemeReduced}.
 \begin{figure}[h!]
\centering
\includegraphics[width=\linewidth]{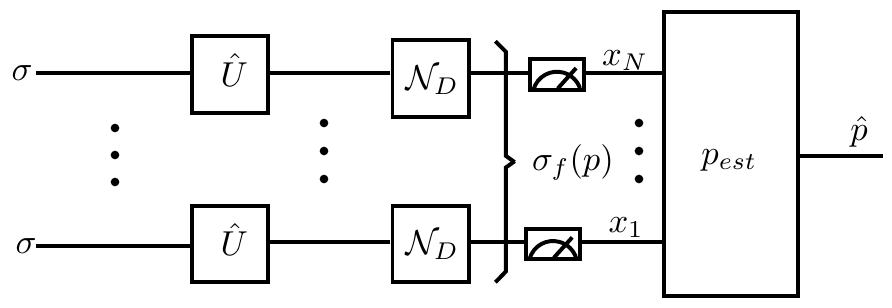}
\caption{Estimation protocol used in this section. $N$ independent channel invocations are performed in order to estimate the value of $p$.}
\label{fig:estimationSchemeReduced}
\end{figure}

\item Input probe $\sigma$: Two different state probes $\sigma$ are considered\footnote{We consider noiseless probes that are only affected by the depolarizing channel. Research about constructing robust quantum probe states to face such adverse noise has been addressed in \cite{RobustProbes1,RobustProbes2}.}.

\begin{itemize}
\item {Unentangled pure states:} The Fisher information for unentangled pure states as probes has been calculated in \cite{FishPureProbe} to be:
\begin{equation}\label{eq:fishpure}
J_1(p) = \frac{9}{8p(3-2p)}\;\; \mbox{and}\;\;J_N(p) = N\frac{9}{8p(3-2p)}.
\end{equation}
It should be mentioned that for SQSC protocols with no entanglement available, \eqref{fishpure} is the largest value that the Fisher information can attain. In other words, under pure quantum state probes, the scheme in Figure \ref{fig:estimationSchemeReduced} is optimal.
\item {Maximally entangled pure states:} When entanglement is available, maximally entangled pure states or EPR pairs, $|\Phi^+\rangle=(|00\rangle+|11\rangle)/\sqrt{2}$, can be used as probes for the depolarizing channel. It can be shown that if just one of the qubits of $|\Phi^+\rangle=(|00\rangle+|11\rangle)/\sqrt{2}$ is transformed by the depolarizing channel (i.e., the EPR pair goes through an extended channel $\mathcal{N}_\mathrm{D}\otimes \mathrm{id}$), then the corresponding Fisher information is \cite{Fisher}
:
\begin{equation}\label{eq:fishEPR}
J_1(p) = \frac{9}{16p(1 - p)}\;\; \mbox{and}\;\;J_N(p)=N\frac{9}{16p(1 - p)}.
\end{equation}

Note that this type of protocol requires that one of the entangled qubits is not affected by noise. This is not an issue in our scenario, since the codes we consider are entanglement assisted (there is pre-shared entanglement between the coder and the decoder), and thus, this protocol is suitable for the estimation of the depolarizing probability. It can be shown that the Fisher information value in \eqref{fishEPR}, higher than in \eqref{fishpure} due to the entanglement, is the largest value that can be achieved {by} SQSC estimation protocols {for the depolarizing channel \cite{Estimation}}.

\end{itemize}

Although the Fisher information given by the maximally entangled probes \eqref{fishEPR} is larger than the one for pure states \eqref{fishpure}, entanglement consumption is an expensive resource, and thus, a tradeoff between performance and implementation complexity exists for this kind of estimation protocol.
\end{enumerate}

\subsubsection{Computation of the Average Word Error Rate}

Based on the results of Section \ref{sec:mismatch} regarding the variation of the WER as a function of the mismatched depolarizing probability $\mathrm{WER}(\hat{p})$ (refer to Figure \ref{fig:mismatch}), one can obtain
the average $\tilde{\mathrm{WER}}(p)$ versus the actual depolarization probability $p$ of the channel as:

\begin{equation}\label{eq:calcWER}
\tilde{\mathrm{WER}}(p) = \int \mathrm{WER}(\hat{p})\mathrm{P}(\hat{p}) d\hat{p},
\end{equation}
where $\mathrm{P}(\hat{p})$ is the probability density function of the estimator being considered\footnote{Note that both $\mathrm{WER}(\hat{p})$ and $\mathrm{P}(\hat{p})$ are functions of the true depolarizing probability $p$.}. As in \cite{QLDPCmismatch}, this probability density will be assumed to be a truncated normal distribution between zero and one with the mean equal to the true depolarizing probability $p$ and variance $[J_N(p)]^{-1}$, that is, $\hat{p}\sim \mathcal{GN}_{[0,1]}(p,[J_N(p)]^{-1})$. {\mbox{The selection} of the inverse of the Fisher information as the variance of the distribution of the estimator, $\mathrm{P}(\hat{p})$, will give us a bound on the best performance that the system will asymptotically attain. \mbox{We followed} the reasoning of \cite{Estimation}, that is, as the Fisher information is asymptotically achievable, this will quantify the best possible accuracy of the quantum estimation process, giving us a bound on the performance.}

Figure \ref{fig:estimationDegradation} shows the performance results of the QTCs considered in this section versus $p$ as a function of the number of probes, $N$, used. In particular, Figure \ref{fig:estimationDegradation}a and \ref{fig:estimationDegradation}b show the results for pure state probes and for maximally entangled state probes, respectively. In both cases, the estimation protocol converged after $N\approx1000$. 
However, for $N=100$, the protocol based on EPR pairs converged faster than the pure protocol. This agrees with the fact that the variance of the entanglement probe estimator was smaller than the variance of the pure probe estimator (refer to Expressions \eqref{fishpure} and \eqref{fishEPR}). The drawback is that pre-shared entanglement was needed, which is an expensive resource.

\begin{figure}[H]
\centering
 \begin{subfigure}[b]{\textwidth}
 \includegraphics[width=\textwidth]{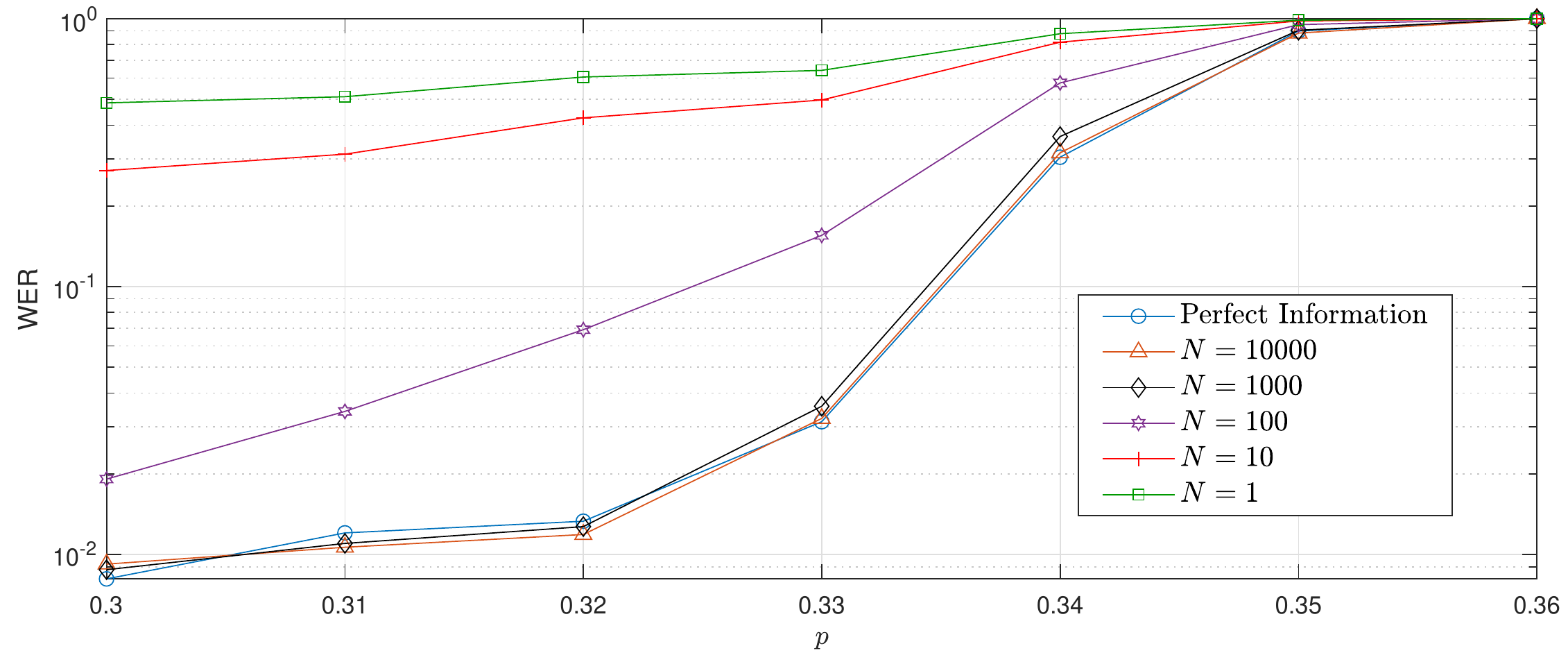}
 \caption{Pure quantum probes.}
 \label{fig:f1}
 \end{subfigure} \\
 \vspace{35pt}
 \begin{subfigure}[b]{\textwidth}
 \includegraphics[width=\textwidth]{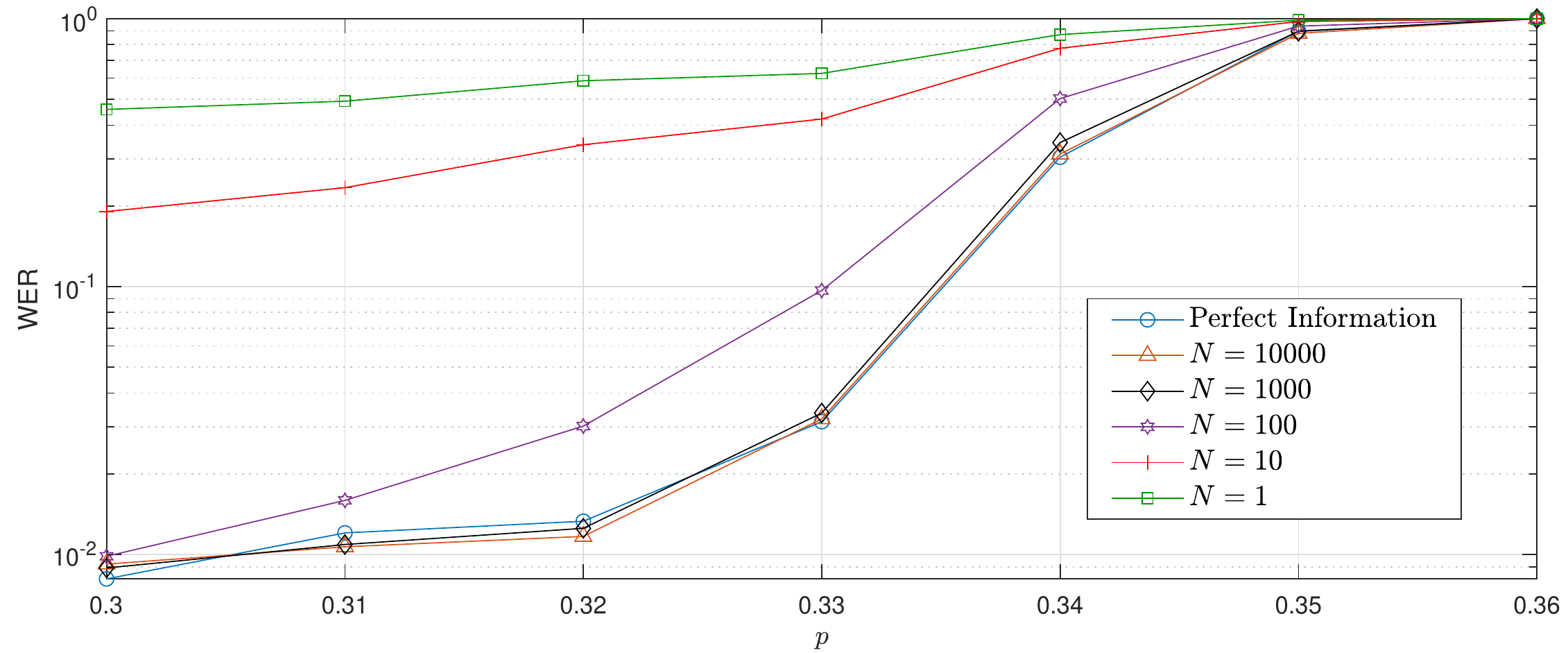}
 \caption{Maximally-entangled quantum probes (EPR).}
 \label{fig:f2}
 \end{subfigure}
 \vspace{13pt}
 \caption{WER degradation for the QTCs as a function of the number of probes, $N$, used for the initial depolarizing probability estimation. {Note that the number of probes used applies in the asymptotic setting.} (\textbf{a}) Estimation using pure quantum probes. (\textbf{b}) Estimation using maximally-entangled quantum probes (EPR pairs).}
\label{fig:estimationDegradation}
\end{figure}

\subsection{On-Line Estimation Framework}\label{subsec:onlineestimation}
In this section, we propose a modified version of the turbo decoding algorithm presented in Section \ref{sub:system} (described in detail in \cite{QTC}) that will allow estimating the depolarizing probability jointly with the decoding process. Contrary to what happened in the {off-line} scheme considered in \mbox{Section \ref{subsec:directestimation}}, this {on-line} estimator has the advantage of not requiring channel estimation before transmission begins. Thus, latency increase and rate reduction are avoided even when the quantum channel is block-to-block \mbox{time-varying}, such as the TVQC model discussed in Chapters \ref{cp4} and \ref{cp5}.

As mentioned in Section \ref{sub:system}, at each iteration $j$ of the turbo decoder, the inner decoder computes the probability distribution\footnote{Note that $\mathrm{P}^{e}_{2}(\mathcal{L}_2)$ also depends on iteration $j$, but such dependency is usually omitted for simplicity.} $\mathrm{P}^{e}_{2}(\mathcal{L}_2)$ of the logical errors $\mathcal{L}_2$ associated with the inner convolutional code and the probability distribution $\mathrm{P}^{(j)}(\mathcal{E}_2^{i}|\bar{r}_2)$ that the $i^{\text{th}}$ qubit of the transmitted block will endure a physical error conditioned on the measured syndrome.
 Note that although $\mathrm{P}_{2}^{e}(\mathcal{L}_2)$ is sent to the outer SISO decoder through the interleaver, $\mathrm{P}^{(j)}(\mathcal{E}_2^{i}|\bar{r}_2)$ is left unused. Such information may effectively be used to estimate the value of the depolarizing probability of the channel at decoding iteration $j$ as:
\begin{equation}\label{eq:online}
\hat{p}^{(j)} = 1 - \frac{1}{n}\sum_{i=1}^n\mathrm{P}^{(j)}(\mathcal{E}_2^i = \mathrm{I}|\bar{r}_2),
\end{equation}
where $n$ is the blocklength transmitted through the channel and $\mathrm{I}$ is the identity operator. 

Once $\hat{p}^{(j)}$ is obtained, it is fed to the inner SISO decoder. 

Under the low WER decoding operation, it is expected that the sequence of estimates $\{ \hat{p}^{(j)} \}$ will converge to the actual depolarization probability of the channel $p$. This procedure is schematically represented in Figure \ref{fig:onlineInnerSISO}.

\begin{figure}[H]
\centering
\includegraphics[width=0.7\linewidth]{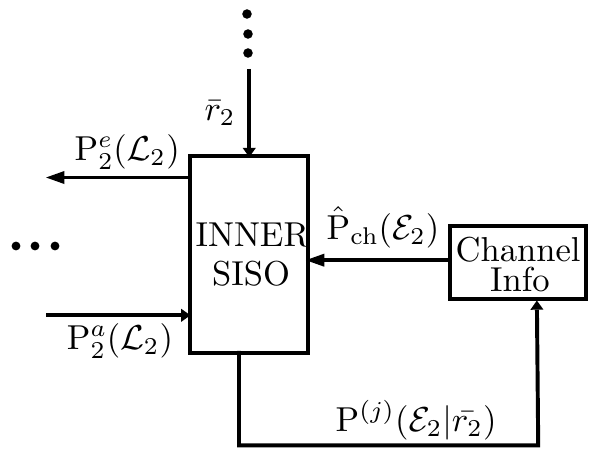}
\caption{Modified QTC decoder to perform on-line estimation of the depolarizing probability. The figure only presents the inner SISO part of the decoder, as the rest of the turbo decoder \mbox{remains unchanged}.}
\label{fig:onlineInnerSISO}
\end{figure}

A question that remains to be answered is what criterion one should follow to provide the initial value of the depolarization probability, $\hat{p}^{(1)}$, to the inner SISO decoder. The selection of the initial value may affect the convergence rate to the final value of $p$, and in turn the decoder operation. Figure \ref{fig:onlineVsEstimation} presents the results when the initial value, $\hat{p}^{(1)}$, is equal to the hashing limit of the QTC, i.e., $\hat{p}^{(1)} = p^*$. The rationale behind this selection is that no QECC operating over a channel with a hashing limit $p^*$ will be able to operate at low WER whenever the depolarization probability $p$ is larger than $p^*$, and thus, initialization with $\hat{p}^{(1)}>p^*$ makes no sense in principle. In our case, $p^*=0.3779$ (refer to Figure \ref{fig:results}).

Figure \ref{fig:onlineVsEstimation} shows the QTC performance in terms of WER versus $p$ when using the proposed {on-line} estimation method with initial value $\hat{p}^{(1)}=p^*=0.3779$. For comparison purposes, the figure also shows the WER achieved when perfect channel information is available to the inner SISO decoder. Notice that there was no WER degradation when using on-line estimation, since the turbo code with the {on-line} estimation of the depolarizing probability achieved the same error correction capability as when the channel information was perfectly available at the decoder. Importantly, this performance was achieved without requiring extra qubits and, therefore, without decreasing the rate of the QTC, since the information used for this {on-line} protocol was based on information that the SISO decoding algorithm already calculated.

\begin{figure}[H]
\centering
\includegraphics[width=\linewidth]{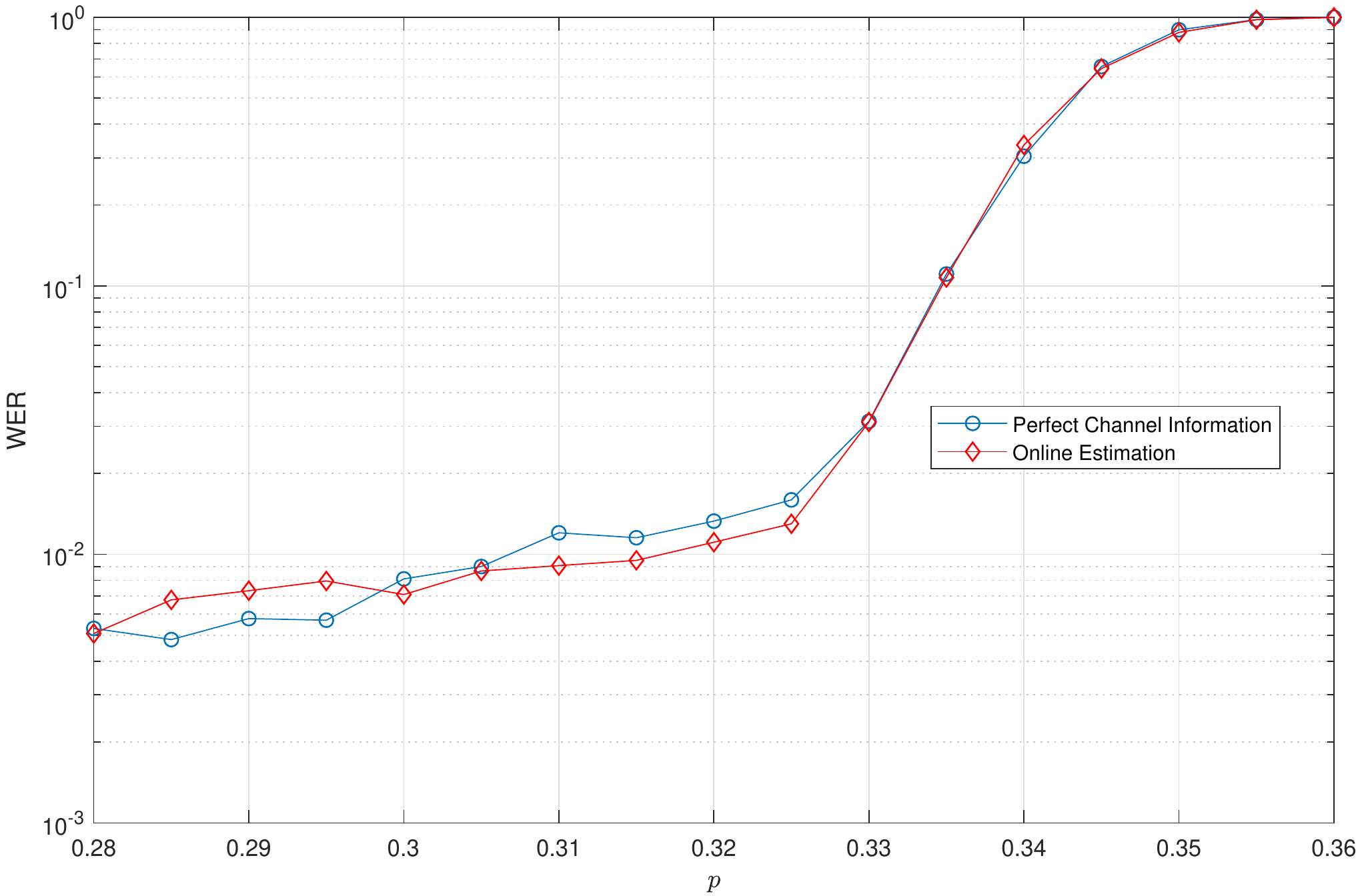}
\caption{WER performance for the QTCs when the proposed on-line estimation procedure for the depolarizing probability is utilized with initial value $\hat{p}^{(1)}=p^*=0.3779$. For comparison purposes, \mbox{the figure} also shows the WER achieved when perfect channel information is available.}
\label{fig:onlineVsEstimation}
\end{figure}

In order to assess the effect of the choice of the initial value of the depolarizing probability, \mbox{Figure \ref{fig:onlineMismatch}} depicts the WER performance when $\hat{p}^{(1)}$ is chosen as $\hat{p}$ in the $x$-axis. Each curve corresponds to a different actual value of the depolarizing probability, $p$. For comparison purposes, the curves in \mbox{Figure \ref{fig:mismatch}} are also depicted here (in that case, $\hat{p}$ is not estimated jointly with the decoding process). Notice the excellent performance achieved by the on-line estimation method independently of the initial value $\hat{p}^{(1)}$, which manifests by the much flatter nature of the curves obtained by the on-line estimation method when compared with those of Figure \ref{fig:mismatch}.

\begin{figure}[H]
\centering
\includegraphics[width=\linewidth]{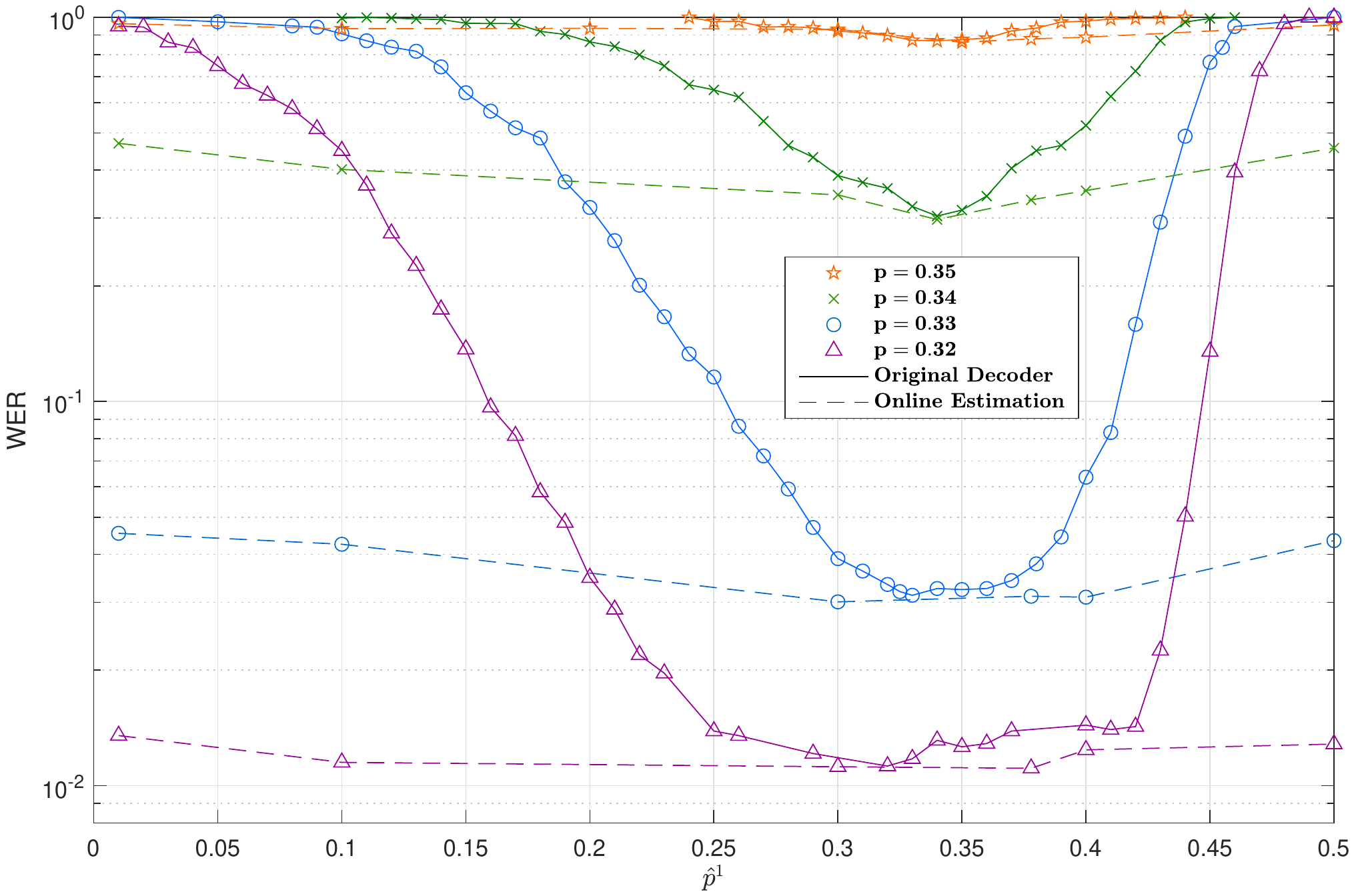}
\caption{WER variation for the QTCs as a function of the initial value of the depolarizing probability, $\hat{p}^{(1)}$. The continuous lines represent the sensitivity of the original turbo decoder \cite{QTC} (no estimation of the depolarizing probability is performed, and the decoder always utilizes $\hat{p}^{(1)}$), while the dashed lines represent the sensitivity of the modified decoder using the proposed on-line estimation method.}
\label{fig:onlineMismatch}
\end{figure}

\section{On-line estimation method for the Pauli channel}\label{sec:AsymOnline}
In this section, we describe a modified on-line estimation algorithm for the asymmetric Pauli channel. Furthermore, we present results that support the claim that the proposed method successfully achieves the task at hand.

\subsection{On-line method for the Pauli channel}\label{sub:online}
Here, we consider the PTA channel obatined by Pauli twirling the combined amplitude and phase damping channel (see Chapter \ref{ch:preliminary}). Thus, the on-line estimation method for such asymmetric Pauli channel should estimate the channel probabilities, $\hat{p}_g^{(j)}$, with $g\in\{\mathrm{X},\mathrm{Y},\mathrm{Z}\}$ at each iteration $j$.
Following the same reasoning as for the depolarizing channel, we propose the following estimator:
\begin{equation}\label{eq:asymmEst}
\hat{p}_g^{(j)} = \frac{1}{n}\sum_{i=1}^n \mathrm{P}^{j}(\mathcal{E}_2^i = g|\bar{r}_2),
\end{equation}
where $g\in\{\mathrm{X},\mathrm{Y},\mathrm{Z}\}$. Once such values are obtained, they are used to derive the \textit{a priori} channel probability distribution $\hat{\mathrm{P}}_{ch}(\mathcal{E}_2)$, which is fed back to the inner SISO decoder for the next iteration of the turbo decoding sum-product algorithm, similar to Figure \ref{fig:onlineInnerSISO}.

Note that in \eqref{asymmEst}, each of the probabilities associated with each of the non-identity Pauli matrices is calculated independently. One might wonder why not to use the fact that for the PTA channel being considered the probabilities for bit-flip and bit-and-phase-flip are equal, and thus to force the estimator to estimate $\hat{p}_\mathrm{x}^{(j)}=\hat{p}_\mathrm{y}^{(j)}$ by calculating the mean of both. In principle, this would lead to better performance, as information about the channel structure would be indirectly fed to the decoder. However, as it will be seen in the simulation results, the proposed decoder achieves the same performance as a decoder that has knowledge of the channel parameters. Consequently, the only benefit that might be obtained by applying such restriction in the estimation procedure would be to lower the decoding complexity. Reducing the complexity is desirable, but we consider that the method proposed in \eqref{asymmEst} is more interesting, as it will be valid for all types of Pauli channels and not just for the asymmetric model that the PTA channel represents.

The final matter that must be established in order to appropriately define the iterative on-line estimation algorithm is the selection of the values for the first iteration, $\hat{p}^{(1)}_g,g\in\{\mathrm{X},\mathrm{Y},\mathrm{Z}\}$. Following the reasoning done in the previous section for the depolarizing channel, it seems logical to select $\hat{p}_g^{(1)}=\alpha p^*_{\alpha} / (\alpha + 2)$ for $g=\mathrm{Z}$ and $\hat{p}_g^{(1)}=p^*_{\alpha} / (\alpha + 2)$ for $g=\{\mathrm{X},\mathrm{Y}\}$, where $p^*_{\alpha}$ refers to the noise limit for the code with channel asymmetry, $\alpha$. However, the parameter $\alpha$ is unknown by the receiver, and so we cannot set such a starting point for the algorithm. Consequently, we will select the value for the first iteration as a constant for each $g$. Observing Figure \ref{fig:onlineMismatch}, it can be inferred that the selection of the initial values is not crucial in terms of code performance if the number of iterations for the turbo decoding algorithm is sufficient. In light of these results, and of the fact that the noise limit, $p^*_\alpha$, increases\footnote{Since operating beyond the noise limit of a code does not make sense, it is logical to select the smallest noise limit as the starting point of all possible asymmetry coefficients $\alpha$. The smallest noise limit is given by the symmetric case $\alpha=1$ \cite{EAQIRCC}.} with the value of $\alpha$ \cite{EAQIRCC}, we will select the starting point of the algorithm as $\hat{p}_g^{(1)}= p^*_{\alpha=1} / 3, \forall g$. The estimated asymmetry parameter for each iteration can be obtained as

\begin{equation}
\hat{\alpha}^{(j)} = \frac{\hat{p}^{(j)}_{g=Z}}{\hat{p}^{(j)}_{g=X}}.
\end{equation}

The proposed estimation method is applicable for all possible values of $\alpha$.

\subsection{Simulation Results}
In this section we perform Monte Carlo simulations to asses the performance of Quantum Turbo Codes using the proposed on-line estimation method for asymmetric Pauli channels. For this purpose, the unassisted QTC scheme presented in Chapter \ref{cp5} has been used with a block length of $k=1000$ logical qubits. Note that this scheme is not the same as the one used for the depolarizing channel estimation in the previous section. The reason for doing so is that previous literature on QTCs operating over asymmetric Pauli channels was done for such scheme \cite{EAQIRCC}.

\begin{figure}[H]
\centering
\begin{subfigure}{.7\textwidth}
  \centering
  \includegraphics[width=\linewidth]{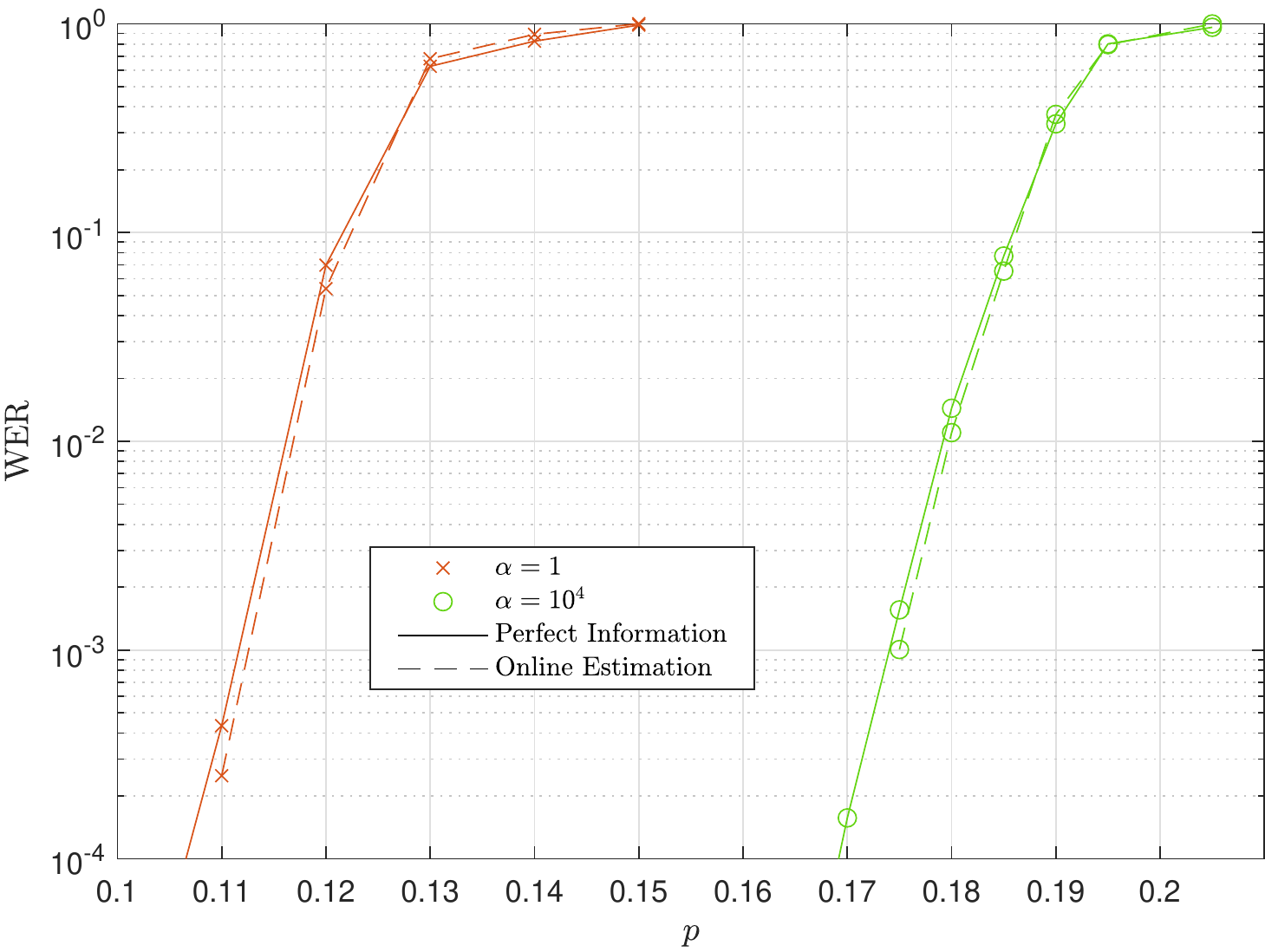}
  \caption{$\mathrm{WER}$ performance of the QTC for $\alpha = 1$ and $\alpha = 10^4$.}
  \label{fig:sub1}
\end{subfigure}\\ \vspace{35pt}
\begin{subfigure}{.7\textwidth}
  \centering
  \includegraphics[width=\linewidth]{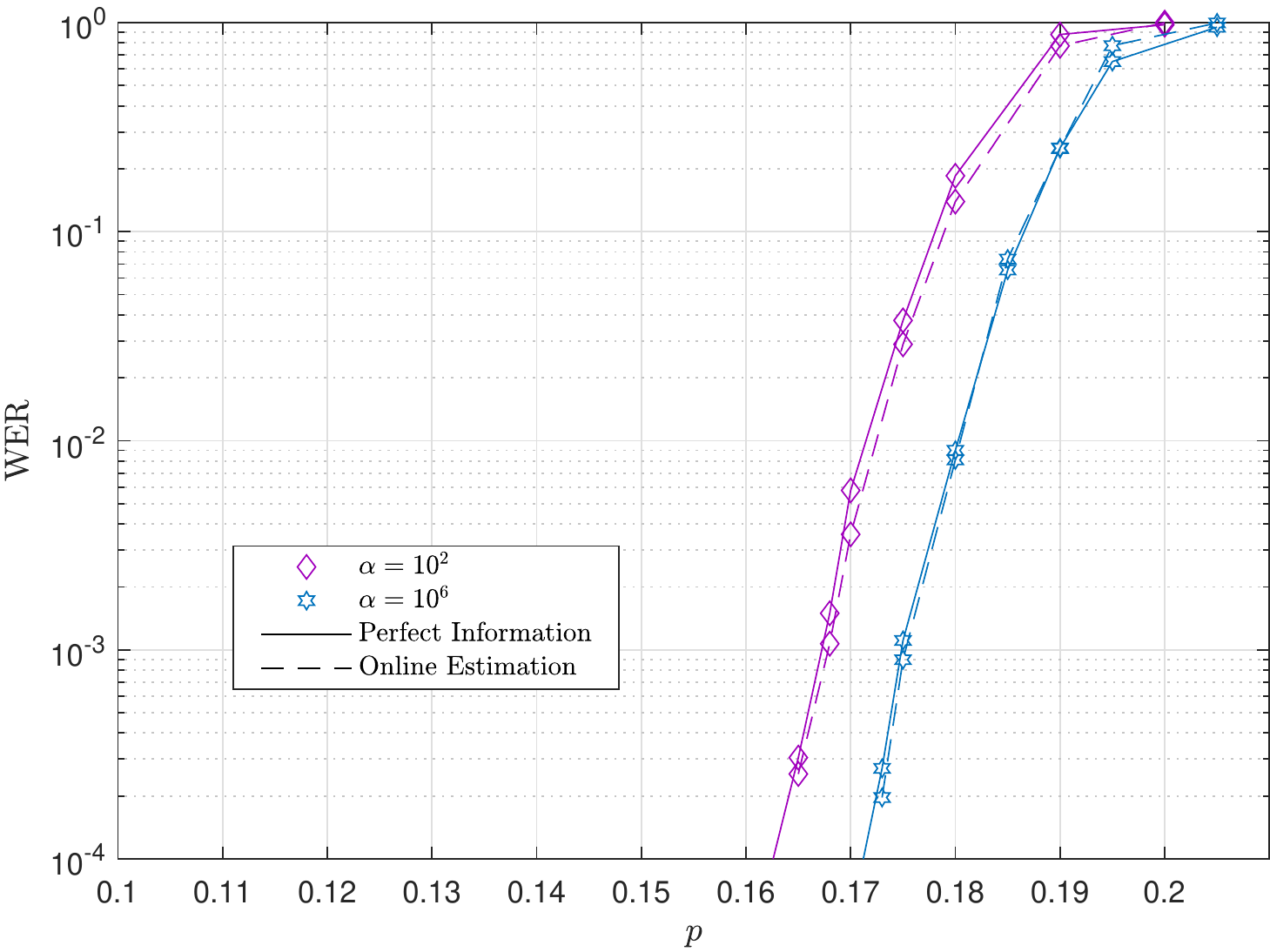}
  \caption{$\mathrm{WER}$ performance of the QTC for $\alpha = 10^2$ and $\alpha = 10^6$.}
  \label{fig:sub2}
\end{subfigure}
\vspace{13pt}
\caption{$\mathrm{WER}$ performance for the QTC under consideration when the proposed asymmetric on-line estimation procedure is utilized. The obtained performance curves are compared to the ones obtained in \cite{EAQIRCC} for the case where the QTC decoders have knowledge of the parameters $p$ and $\alpha$.}
\label{fig:resultsAsym}
\end{figure}

Figure \ref{fig:resultsAsym} shows the results obtained for the QTCs described in Chapter \ref{cp5} using the generalized on-line estimation method presented in section \ref{sub:online}. The waterfall shaped curves depicted in said figure compare the performance of the ``PTO1R-PTO1R'' configuration when it corrects errors of asymmetric nature and its decoder has access to perfect channel information, to the performance of the decoder that estimates those channel parameters while simultaneously decoding the information. It is obvious that both methods exhibit similar error correcting capabilities in terms of $\mathrm{WER}$ performance. Consequently, the generalized on-line estimation method allows the QTCs to attain excellent performance regardless of the nature of the channel. This outcome is of significant relevance, since it means that the QTC decoder can be applied to any Pauli channel. Therefore, no channel-specific QTC decoders that need to know the exact noise model beforehand are required.

To finalize the study of the QTCs that use the on-line estimator to tackle the lack of knowledge about the channel parameters, we discuss the number of iterations required by the turbo decoder algorithm to succesfully estimate the logical error coset. In \cite{EAQIRCC}, the authors claim that setting the maximum number of iterations beyond $I_{max}=4$ produced no benefits in the decoding process. However, the utilization of the on-line method introduced in our work requires the decoder to execute more iterations so that the recovery operation can be calculated correctly. This need arises because the decoder estimates the channel parameters during every decoding iteration, which inadvertently implies that subsequent estimates become increasingly accurate as more decoding rounds take place. Thus, the decoder requires sufficient iterations for the channel information to be accurate enough to produce satisfactory estimates of the logical error coset through additional turbo decoder rounds. For the codes considered in this section, we have observed that setting the maximum number of iterations to less than $I_{max}=32$ implies a loss in the $\mathrm{WER}$ performance, while increasing the parameter beyond this value yields no benefit for the QTCs under consideration. As a result, the worst case decoding scenario using our estimator takes $8$ times the number of decoding rounds required by a decoder that has knowledge of the channel parameters. However, the increase in complexity is justified by the flexibility provided by the on-line estimation method, since if we assume that the channel remains constant for the duration of the block, the QTC performance will not be compromised\footnote{Here we mean that the code will not be degraded due to the imperfect CSI. However, the performance of the code will suffer from the effects of time-variation as seen in Chapter \ref{cp5}.} if the nature of the channel changes from block to block, regardless of whether the change occurs in the gross error probability, $p$, or in the asymmetry level of the Pauli channel, $\alpha$. 
Therefore, QTCs will achieve excellent performance independently of the nature of the channel or its time evolution.

\begin{figure}[!h]
\centering
\includegraphics[width=\linewidth]{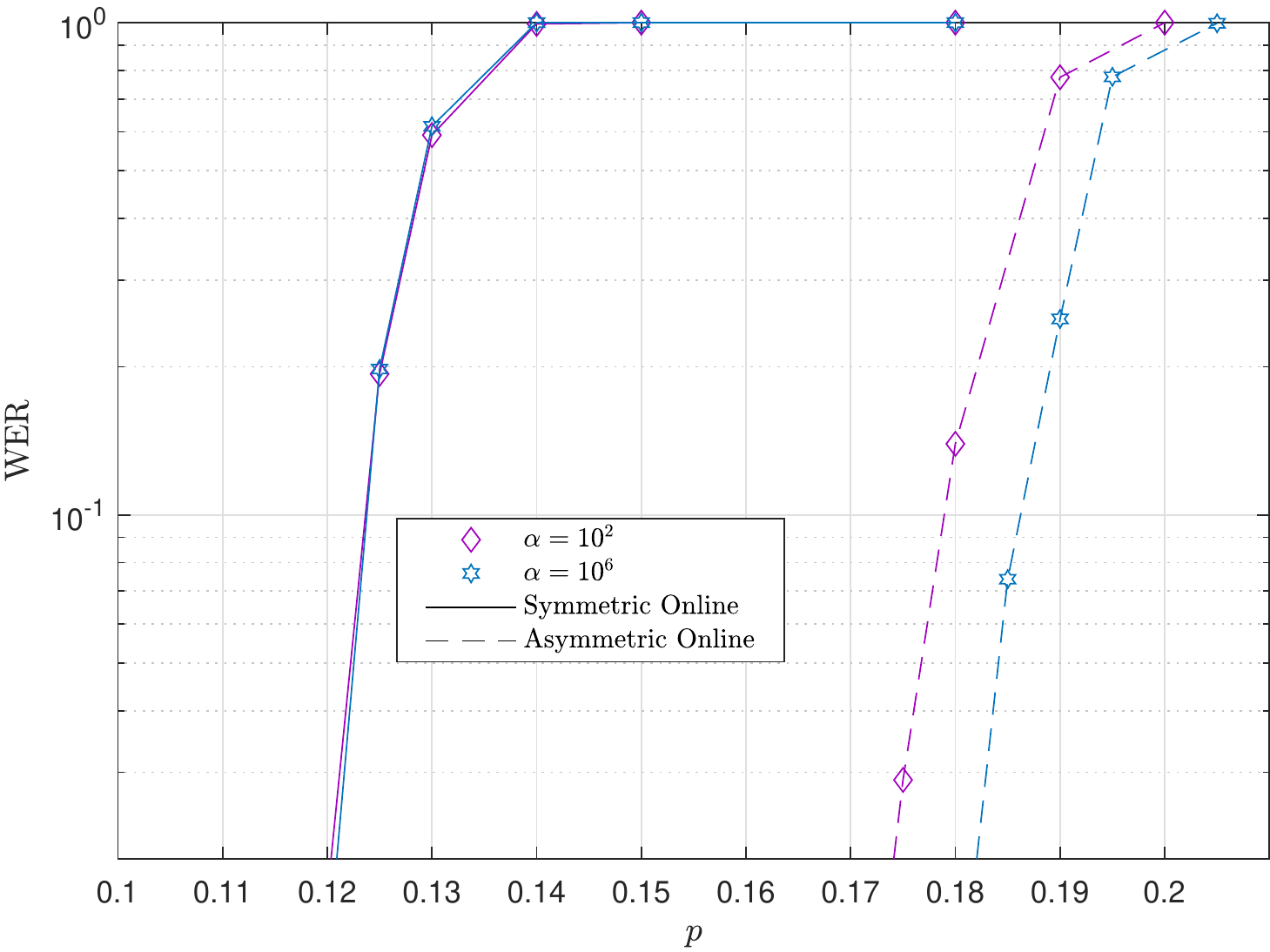}
\caption{$\mathrm{WER}$ performance of the QTCs using the symmetric on-line estimator against the performance of those QTCs using the on-line asymmetric estimator proposed. The considered asymmetry coefficients are $\alpha=10^2$ and $\alpha = 10^6$.}
\label{fig:asymmVSsymm}
\end{figure}

Another question that comes up when analyzing the performance of the proposed on-line asymmetry-considering estimator is how it measures up against the on-line estimator proposed in the previous section for the depolarizing channel.
 The answer to this question will give us information regarding the importance of asymmetry to decode QTCs. Figure \ref{fig:asymmVSsymm} shows the results obtained by running the turbo decoder with the estimator from section \ref{subsec:onlineestimation}, as well as the results obtained for the extended on-line estimator proposed in this section, when we consider the Pauli channels with asymmetry coefficients $\alpha=10^2$ and $\alpha = 10^6$. Notice that neglecting the possibility that the channel might be asymmetric and treating it like a depolarizing channel leads to huge performance losses. Upon further examination of Figure \ref{fig:asymmVSsymm}, it can be seen that using the estimator from section \ref{subsec:onlineestimation} for the Pauli channel always results, independently of the nature of $\alpha$, in the QTC performance corresponding to the symmetric case $\alpha=1$. As a consequence, if asymmetry is not considered in the error correction system, the performance loss experienced by the QTCs will correspond to the gap between the waterfall region presented by the decoder that has perfect knowledge of the channel parameters and the waterfall region of the decoder for the depolarizing channel, which increases with the growth of the asymmetry coefficient, $\alpha$. In practical situations, since the quantum channel might be asymmetric, using the extended on-line estimator proposed here is paramount if one wants to maintain excellent performance.

	\FloatBarrier
\section{Discussion}\label{cp7_sec: conclusion}
We have studied the performance sensitivity of QTCs applied over depolarizing channels when there exists a mismatch between the actual depolarizing probability and the probability fed to the decoders. To that end, we fed the turbo decoder with values different from the actual value of the depolarizing probability and observed the degradation in the WER due to such mismatch. Simulation results showed that the decoder is sensitive both to under- and over-estimation of the depolarizing probability, presenting a flat region around the actual value of the depolarizing probability. 

In order to overcome the lack of knowledge regarding the depolarization probability, we have proposed a novel on-line estimation procedure based on the information generated in the iterative turbo decoder. Such an on-line estimator outperformed existing off-line estimators in terms of overall coding rate and latency, while maintaining excellent WER performance. Simulation results have shown that the on-line estimation method is insensitive to the initial value of the depolarizing probability, and thus, the resulting performance experiences very little degradation in the presence of channel mismatch.
 
Moreover, we have presented an extension of the on-line estimation method for the depolarizing channel, making it capable of considering the asymmetries that are present in the decoherence models of realistic quantum devices. The noise in these devices is modeled using the well known Pauli twirl approximation channel, which has been shown to be biased towards $\mathrm{Z}$ type errors for some of the qubit construction technologies that are being used in this day and age. Monte Carlo simulations corroborate that the proposed on-line estimation method is successful in aiding the QTC to achieve the same performance as when it has access to the channel parameters. This positive outcome comes at the cost of a slight increase in the complexity of the decoding algorithm. However, the flexibility that the proposed method provides, allowing the QTCs to be able to operate close to their hashing bounds without requiring any knowledge of the Pauli channel parameters or of how they evolve in time, far outweighs this minor drawback.

\clearemptydoublepage
\chapter{Quantum Low-Density-Generator-Matrix Codes and Degeneracy} \label{cp8}
This chapter serves as a summary of the main results of other QEC-related research that the author has been involved in during this Ph.D. thesis. Given the varied nature of this work, the chapter is divided into five different sections, each one dedicated to a specific topic:

\begin{itemize}
\item Section \ref{sec:QLDGM1} studies the design of a class of non-Calderbank-Shor-Steane (CSS) Quantum Low-Density-Generator-Matrix (QLDGM) codes that outperforms other QLDPC schemes of the literature. It reviews the contents of \cite{patrick}.
\item Section \ref{sec:QLDGM2} discusses the scenario of non-CSS QLDGM codes when they operate without channel state information as it was done in Chapter \ref{cp7} for QTCs. The depolarizing channel is considered. It summarizes the research presented in \cite{patrick2}.
\item Section \ref{sec:QLDGM3} proposes the design of CSS QLDGM codes for the asymmetric Pauli channel. Our scheme outerperforms other QLDPC codes in the literature for some of the scenarios considered. This section presents the methods and results of \cite{patrick3}.
\item Section \ref{sec:degeneracy} presents a group theoretical approach to discuss the issue of degeneracy as it relates to sparse quantum codes. It summarizes the review article \cite{degen}.
\item Section \ref{sec:logical} discusses previously existing methods to compute the logical error rate and presents an efficient coset-based method inspired by classical coding strategies to estimate degenerate errors. It reviews the results and discussions of \cite{logicalRate}.
\end{itemize}

Before moving onward, it must be mentioned that the contents shown herein are meant only as a cursory overview. Readers should refer to the Ph.D. dissertation of Patricio Fuentes Ugartemendia (first-author of this research) or the journal articles themselves for a complete discussion regarding these findings.

\section{Approach for the construction of non-CSS low-density-generator-matrix-based quantum codes}\label{sec:QLDGM1}
Quantum LDPC codes are built by casting classical LDPC codes in the framework of stabilizer codes \cite{GottLDPC}, which enables the design of quantum codes from any arbitrary classical binary and quaternary codes. In \cite{qldpc15}, the authors document the design of QLDPC codes based on their classical counterparts, detailing numerous construction and decoding techniques along with their flaws and merits. Among the discussed methods, the construction of QLDPC codes based on LDGM codes is shown to yield performance and code construction improvements, albeit at an increase in decoding complexity. This method was originally proposed in \cite{jgfQLDGM1,jgfQLDGM2}, where Calderbank-Steane-Shor (CSS) quantum codes based on regular LDGM classical codes were shown to surpass the best quantum coding schemes of the time, and performance was significantly improved in \cite{jgfQLDGM2,jgfQLDGM3} by utilizing a parallel concatenation of two regular LDGM codes.

Quantum LDGM (QLDGM)-based quantum code implementations, as well as most QLDPC designs, are based on CSS constructions. CSS codes, simultaneously proposed by Calderbank, Shor, and Steane in \cite{CSS1,CSS2}, are a particular subset of the family of stabilizer codes. They provide a straightforward method to design quantum codes via existing classical codes. In general, decoding of quantum codes based on CSS designs is performed separately for bit- and phase-flip errors, which negatively impacts their performance. In fact, CSS constructions decoded separately are limited by an unsurpassable bound, referred to as the CSS lower bound \cite{CSSbound}. Joint decoding of bit- and phase-flip errors using modified CSS decoders capable of exploiting the correlation between the aforementioned errors has been considered in the literature \cite{modified1,modified2,modified3,modified4,modified5,modified6,modified7}, and performance improvements associated with these modified CSS decoders have been shown in these articles. Nonetheless, the improvements provided by these modified CSS decoders come at the expense of an increased decoding complexity, which, along with the performance limitations of conventional CSS decoding, inspires the search for non-CSS constructions, as they should theoretically be able to outperform CSS codes and avoid complex decoding strategies. Non-CSS QLDPC-based codes were proposed in \cite{noncssqldpc}. However, despite showing promise, they fail to outperform existing CSS QLDPC codes for comparable block lengths.

In such context, we introduced a technique to design non-CSS quantum codes based on the use of the generator and parity check matrices of LDGM codes. The proposed methods are based on modifying the upper layer of the decoding graph in CSS QLDGM constructions. An example can be seen in Figure \ref{fig:merge}. The feasibility of such modification of the upper layer of the decoding graph is based on the othogonality properties shown by the generator and parity-check-matrices of the classical LDGM codes used for constructing the code. The simplicity of the proposed scheme ensures that the high degree of flexibility in the choice of the quantum rate and the block length for the CSS code utilized as a starting point is translated to the non-CSS design. 

\begin{figure}[!ht]
\centering
\includegraphics[width=\linewidth]{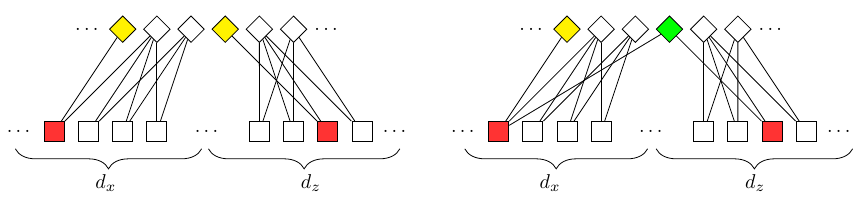}
\caption{Schematic representation of the combination of $\mathrm{x}$ and $\mathrm{z}$ nodes of the upper layer of the decoding graph.}
\label{fig:merge}
\end{figure}

Compared with quantum CSS codes based on the use of LDGM codes, the proposed non-CSS scheme is 0.2 dB closer to the hashing bound in the depolarizing channel, and outperforms all other existing quantum codes of comparable complexity. The results can be seen in Figure \ref{fig:nonCSSresults}.

\begin{figure}[!ht]
\centering
\includegraphics[width=\linewidth]{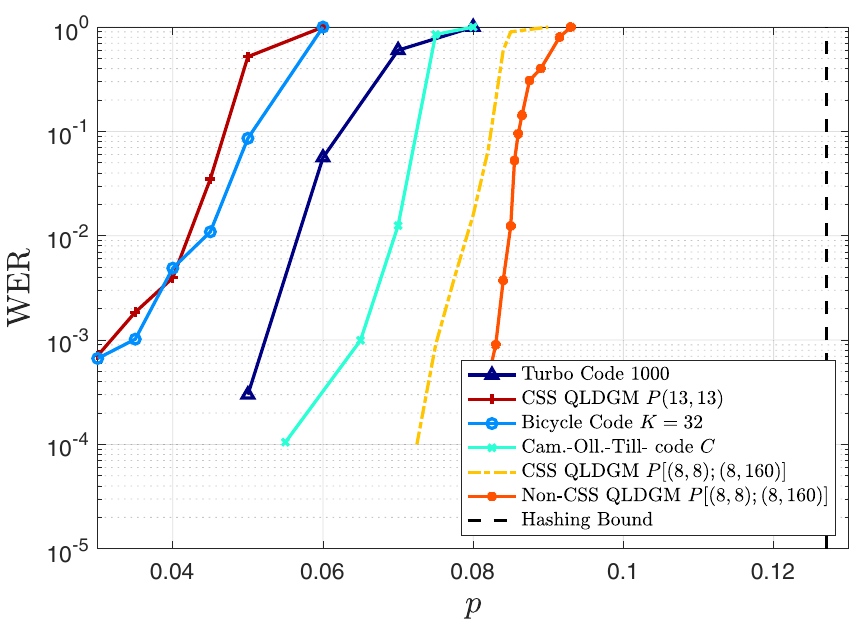}
\caption{Word error rate for various quantum codes when they are applied over the depolarizing channel. The hashing bound is also included. $p$ is the depolarizing probability.}
\label{fig:nonCSSresults}
\end{figure}

\section{Performance of non-CSS LDGM-based quantum codes over the misidentified depolarizing channel}\label{sec:QLDGM2}
Most of the research related to the performance of QLDPC codes has been conducted under the tacit premise that perfect knowledge of the quantum channel in question is available. In reality, such a scenario is highly unlikely, meaning that analyzing how the behaviour of these codes changes in terms of the existing information about the quantum channel is of significant relevance. Such a study was conducted for the quantum depolarizing channel in \cite{QLDPCmismatch}. In \cite{QLDPCMismatchMethods}, the same authors designed an improved decoding strategy for QLDPC codes when only an estimate of the channel depolarizing probability is available. The aforementioned method makes use of quantum channel identification, which requires the introduction of a probe (a known quantum state) into the quantum channel and the subsequent measurement of the channel output state to produce an accurate estimate of the depolarizing probability. This procedure typically makes use of additional qubits and results in a latency increase. Thus, the design of methodologies capable of minimizing this overhead while yielding performance similar to the perfect channel knowledge scenario is germane to this field of research.  \begin{figure}[!ht]
\centering
\includegraphics[width=\linewidth]{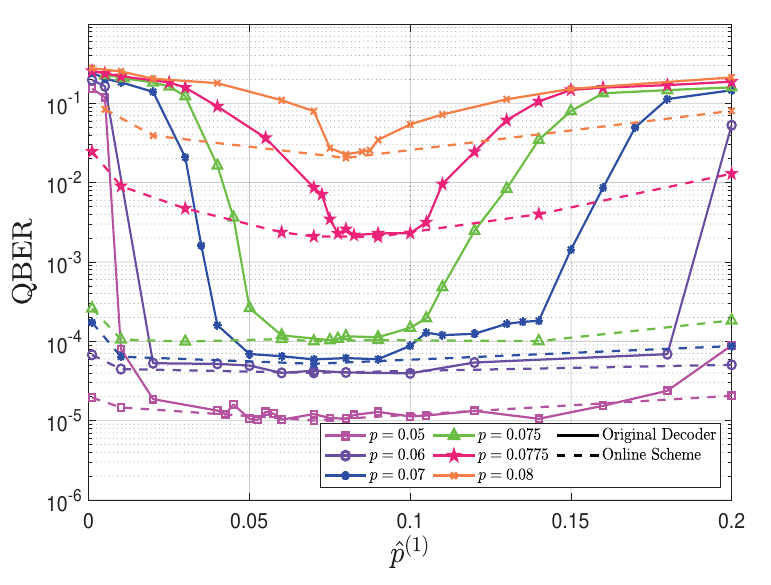}
\caption{Simulated QBER as a function of the initial estimate of the depolarizing probability of the channel $p^{(1)}$.}
\label{fig:mismatchnonCSSresults}
\end{figure}
In Chapter \ref{cp7} of this dissertation, published in \cite{josu2}, we have derived an on-line depolarizing probability estimation technique for Quantum Turbo Codes (QTCs). This method yields similar performance to that obtained when using the same QTCs with perfect channel information but without the need for additional resources. In light of this outcome, we propose a similar on-line estimation procedure for QLDPC codes.

Our analysis revealed an increasing impact of channel mismatch on decoder performance as the depolarizing probability of the channel grows. The mismatch effect is especially noticeable in the waterfall region of the code. As in the case of QTCs, the on-line estimation scheme outperforms off-line channel identification techniques in terms of overall coding rate, while maintaining excellent performance. The simulation results may be seen in Figure \ref{fig:mismatchnonCSSresults}. In contrast to what happens with QTCs, simulation results showed that the on-line estimation method is slightly dependent on the initial estimate of the depolarizing probability when operating in the waterfall region. Selecting the initial estimate of the depolarizing probability as the hashing limit of the code in question yielded good performance regardless of the actual value of the depolarizing probability.

\section{Design of low-density-generator-matrix-based quantum codes for asymmetric quantum channels}\label{sec:QLDGM3}
Most of the research related to the performance of quantum error correcting codes considers the symmetric Pauli channel, generally referred to as the depolarizing channel, which incurs bit flips, phase flips, and bit-and-phase flips with the same probability. However, realistic quantum devices, given the nature of the materials used to construct them, often exhibit
asymmetric behavior, where the probability of a phase flip taking place is orders of magnitude higher than the probability of a bit flip. As seen in Chapter \ref{ch:preliminary}, the behavior of these quantum devices is governed by the single-qubit relaxation time and the dephasing time of the device itself, the former sometimes being orders of magnitude larger than the latter. Generally, relaxation causes both bit-flips and phase-flips, while pure dephasing only leads to phase-flip errors. This difference in relaxation and dephasing times gives rise to the aforementioned asymmetric behavior, where bit-flip errors are much less likely to occur than phase flips, and it can be accurately modelled by the Pauli twirl approximation channel. Naturally, it stands to reason that the best QEC schemes for this asymmetric channel must somehow be able to exploit its asymmetry. In \cite{EAQIRCC}, the authors introduce an EXIT-chart based methodology to design quantum turbo codes (QTCs) specifically for the Pauli twirl approximation channel. We have extended such work in Chapter \ref{cp7}, published in \cite{josu3}, where an on-line estimation protocol to decode QTCs over Pauli twirl approximation channels has been proposed. These results speak to the merit of constructing a coding scheme tailored to the asymmetric characteristics of the quantum channel in question, since performance of the QTCs varies depending on the degree of asymmetry of the channel.

Consequently, we studied the performance of quantum CSS LDGM codes when they are applied over a Pauli channel. We showed how although they are not the best known codes for the depolarizing channel, their simplicity allows for them to be almost seamlessly adapted to the Pauli twirl approximation channel. Based on this result, we introduce a simple yet effective method to derive CSS QLDGM codes that perform well over channels with varying degrees of asymmetry. Such a strategy is necessary because CSS codes designed for the depolarizing channel perform poorly over its asymmetric counterpart. To the extent of our knowledge and at the time of writing, the research on designing quantum codes specifically for asymmetric quantum channels was quite limited \cite{twirl6,aymmetricCode}, especially when compared to results regarding the depolarizing channel. Thus, this work represents one of the first attempts at designing QLDPC codes specifically for asymmetric quantum channels.

\begin{figure}[!ht]
\centering
\includegraphics[width=\linewidth]{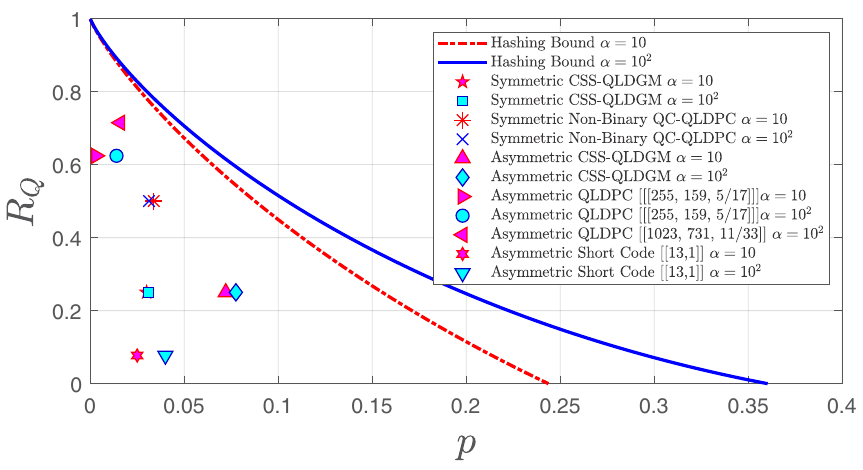}
\caption{Achievable coding rate at a $\mathrm{WER}$ of $10^{-3}$ for various QLDPC codes over Pauli channels with $\alpha = 10$ and $10^2$.}
\label{fig:asymmResults}
\end{figure}

The proposed methods were based on simple modifications to the upper layer of the decoding graph of a symmetric CSS QLDGM code designed for the depolarizing channel. Specifically, the factor graph is adapted to an asymmetric channel by increasing the number of syndrome nodes used to decode the $\mathrm{Z}$ operators and decreasing the number of syndrome nodes used to decode the $\mathrm{X}$ operators. This comes from the fact that $\mathrm{Z}$ errors are more likely to occur when the Pauli twirl approximation channel is considered (and $T_2<<T_1$), as it can be seen in section \ref{ch:preliminary} of this dissertation. Additionally, we showed how for larger block lengths, the proposed asymmetric CSS codes can be further optimized based on the asymmetry coefficient $\alpha$ by increasing the block length of the code and the number of $\mathrm{Z}$ syndrome nodes. Over Pauli channels with $\alpha = 10$ and $10^2$, the schemes proposed are closer to the theoretical limit than other existing asymmetric codes and the best codes designed for the depolarizing channel. The results may be seen in Figure \ref{fig:asymmResults}.

\section{Degeneracy and its impact on the decoding of sparse quantum codes}\label{sec:degeneracy}
The quantum phenomenon known as degeneracy should theoretically improve the performance of quantum codes if exploited appropriately. Unfortunately, degeneracy may end up sabotaging the performance of sparse quantum codes because its existence is neglected in the decoding process.
In this work we provided a broad overview of the role this phenomenon plays in the realm of quantum error correction and, more specifically, in the field of QLDPC codes. We started by introducing a group theoretical explanation of the most relevant concepts in the field of quantum error correction. This approach is helpful to lay out the necessary groundwork to later study the effects of degeneracy. Following this, we examined the differences between the classical decoding problem and its quantum counterpart. Despite the intricate similarities between classical and quantum decoders, the coset partitioning of the $N$-fold Pauli space (see Figure \ref{fig:cosets}) reveals the higher complexity of the quantum decoding task. Based on this discussion, we were able to show in a concise manner why applying the classical decoding algorithm for sparse codes
to the quantum problem is suboptimal. Then, we provided a detailed explanation regarding the origin of degeneracy and how it may have detrimental effects on traditional SPA-based decoders. Finally, we provided a simple example in order to facilitate the comprehension of the topics discussed in the work.

\begin{figure}[!ht]
\centering
\includegraphics[width=0.75\linewidth]{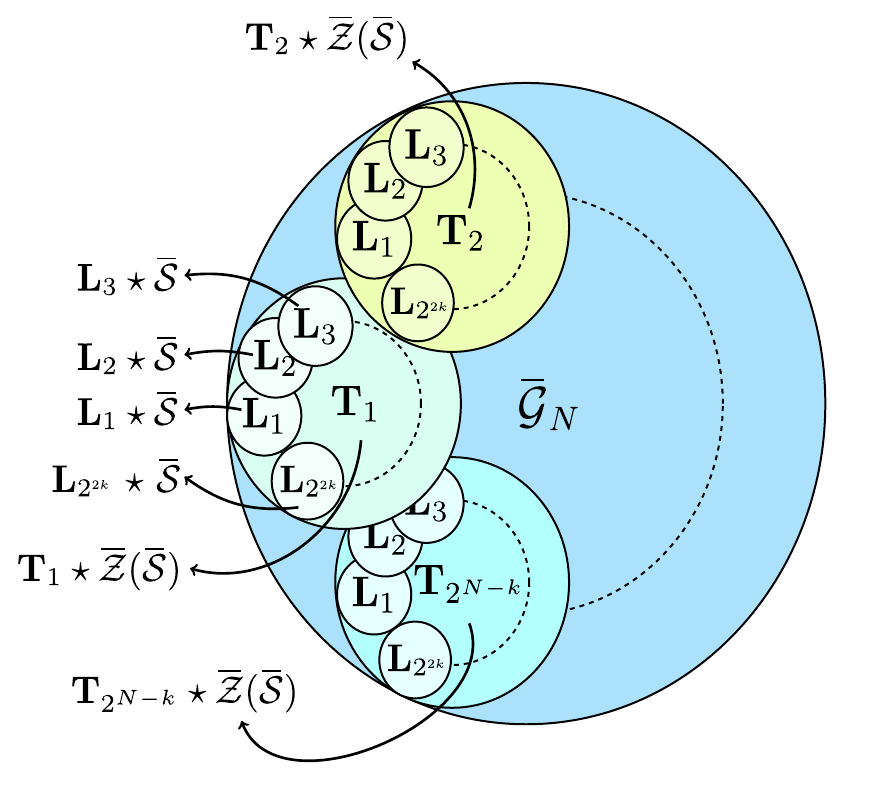}
\caption{Schematic representation of the coset partition of the effective Pauli group by the stabilizer.}
\label{fig:cosets}
\end{figure}

\section{On the logical error rate of sparse quantum codes}\label{sec:logical}
The manifestation of degeneracy in the design of sparse quantum codes and its effects on the decoding process has been studied extensively \cite{QSC,hardness,modified3,modified6,NPhardn,yoshida,viterbiDegen,poulinDegen}. Unfortunately, although degeneracy should theoretically improve performance, limited research exists on how to quantify the true impact that this phenomenon has on Quantum Low Density Parity Check (QLDPC) codes. This has resulted in the performance of sparse quantum codes being assessed differently throughout the literature; while some research considers the effects of degeneracy by computing the metric known as the logical error rate \cite{modified3,noncssqldpc,panteleev,landscape,kuo,kuo2,rigby}, other works employ the classical strategy of computing the physical error rate \cite{qldpc15,jgfQLDGM1,jgfQLDGM2,wang,qldpcBabar}, a metric which provides an upper bound on the performance of these codes since it ignores degeneracy. However, because these codes are highly degenerate \cite{QEClidar,hardness,modified3,NPhardn,vasic}, there is a significant gap between the upper bound provided by the physical error rate and their true performance. Consequently, we devised a group theoretic strategy to accurately assess the effects of degeneracy on sparse quantum codes and we explain another method that was used in \cite{panteleev,landscape,kuo,kuo2} to compute the logical error rate. Then, we used our methodology to show how sparse quantum codes should always be assessed using the logical error rate. \begin{figure}[!ht]
\centering
\includegraphics[width=0.8\linewidth]{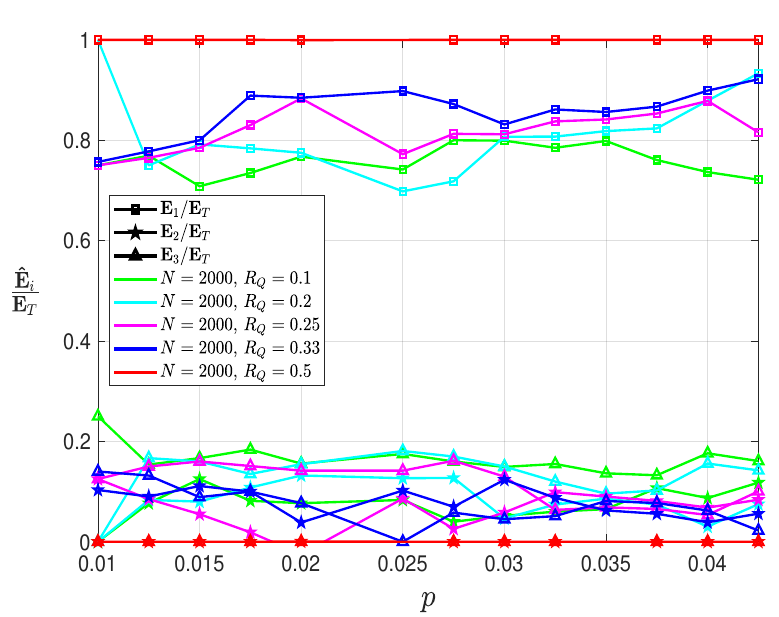}
\caption{Ratios of end-to-end different syndrome, identical syndrome and degenerate errors for the QLDPC code with blocklength $N=2000$ in consideration.}
\label{fig:logicalresults}
\end{figure}
The discrepancy between the logical error rate and the physical error rate is especially relevant to the field of sparse quantum codes because of their degenerate nature. This was reflected by the results (see Figure \ref{fig:logicalresults}) obtained for a specific family of QLDPC codes, whose performance can be up to $20\%$ better than would be expected based on previous results in the literature. In addition, these simulation outcomes serve to show how performance may be improved by constructing degenerate quantum codes, and they also speak towards the positive impact that modified decoding strategies can have on the performance of sparse quantum codes.

 \label{chapter8}
\clearemptydoublepage
\chapter{Conclusions and Future Work} 
This Ph.D. thesis set out with the following two objectives in the context of quantum  information and quantum error correction:
\begin{itemize}
\item The analysis of decoherence as the source of the errors that corrupt quantum information.
\item The optimization of the performance of quantum turbo codes when different scenarios are considered.
\end{itemize}

With these goals in mind, we commenced the dissertation by giving the basic tools regarding quantum computing, quantum information theory and quantum error correction needed in order to understand the research that is presented in the following parts of the thesis. In Chapter \ref{ch:preliminary}, we first introduced the postulates of quantum mechanics, both from the state vector and density matrix perspectives, which are the mathematical framework that describe the theory of Quantum Mechanics. We continued by describing decoherence as the source of the errors that make quantum computations unreliable, and provided the way it is mathematically modelled by the use of quantum channels. Moreover, we discussed the way in which such quantum channels can be approximated so that they can be efficiently simulated in classical computers via the quantum information technique named twirling. The concept of quantum capacity was also introduced. We finished this preliminary chapter by discussing the way in which quantum error correction codes are simulated in classical computers.

Following this, we got into Part \ref{part1} (consisted of Chapters \ref{cp3}, \ref{cp4} and \ref{cp5}) of the thesis which is focused on mathematically modelling decoherence and studying the asymptotical limits of error correction of quantum channels. In Chapter \ref{cp3} we proposed a time-varying quantum channel (TVQC) model. The motivation of such proposal was the recent experimental observation that the parameters of decoherence of superconduting qubits (relaxation time, $T_1$, and dephasing time, $T_2$) are not fixed, but fluctuate through time. Therefore, we studied the stochastic processes that describe the random temporal behaviour of such parameters and integrated such randomness into the context of quantum channels. Moreover, we compared the action of the static quantum channels with the proposed TVQCs by making use of the metric named diamond norm distance. We concluded that both models differ significantly in terms of diamond norm distance indicating that the performance of QECCs will be significantly affected as a result of the $T_1$ and $T_2$ flutctuations.

Chapter \ref{cp4} is the natural continuation to Chapter \ref{cp3}. Quantum channel capacity is the asymptotically achievable rate by QECCs when static quantum channels are considered. In this chapter, we studied the asymptotical limits for quantum error correction under the assumption that the realizations of $T_1$ and $T_2$ are the same for all the qubits of a block, which is reminiscent of the classical scenario of slow fading channels. For classical slow fading channels the operational concept of capacity fails to exist and has to be replaced by the so-called outage probability. Thus, we proposed the quantum outage probability as the asymptotically achievable error rate for QECCs when they operate over TVQCs. We provided closed-form expressions for the family of time-varying amplitude damping channels ($T_1$-limited qubits). Furthermore, we proposed the quantum hashing outage probability as an upper bound for the quantum outage probability when Pauli/Clifford twirl approximated channels are considered, and provided closed-form expressions for such bound. The results concluded that the quantum outage probability is a monotonically increasing function with the coefficient of variation of the relaxation time. Thus, we discussed that the asymptotically achievable performance of QECCs will be limited not only by the mean qubit relaxation time, but also by the relative deviation of such parameter respect to the mean.

We finished the Part \ref{part1} of the thesis by studying how the performance of QECCs is affected when they operate over TVQCs. In Chapter \ref{cp3} we discussed that, due to the deviation in terms of diamond norm distance between static and time-varying channels, the performance of QECCs may be compromised when time variations are considered. We proved this intuition by simulating the Kitaev toric code and quantum turbo codes operating over the TVQCs. We concluded that the performance curves of QECCs are flattened when the channels are time-varying. However, the flattening of the curve does not only depend on the nature of the channel, but also on the performance of the code when it operates over static channels. Thus, the degradation of the WER curve is a function of the coefficients of variation of the decoherence parameters and the steepness of WER curve of the code under static noise. This comes from the fact that when the performance of the QECCs is flat in static scenarios, the difference in terms of WER between ``good'' and ``bad'' channel realizations will not be very different and, therefore, the mean performance of the code in a TVQC will be similar to the static one. Finally, we used the quantum outage probability in order to benchmark the performance of the codes in consideration.

The main takeaways of Part \ref{part1} can be summarized as:
\begin{itemize}
\item Decoherence parameter fluctuations are a critical issue to consider when quantum error correction codes are constructed. Neglecting their time-varying nature would leave out an integral part of quantum noise, making the predictions of the performance of QECCs too optimistic in many scenarios.
\item The importance of superconducting qubit construction and cooldown, in the sense that if they are optimized correctly, the fluctuations relative to the mean will be milder. Traditionally, long mean decoherence times are seeked for the qubits to be able to handle long time duration algorithms, however, this should be done jointly with lowering the dispersion of such parameter.
\item Qubit benchmarking, at least for superconducting qubits at the moment, should include the variation parameters of their decoherence times. However, the current experimental research on the topic usually only focuses at the mean or best case scenarios, which is far from being accurate when describing the decoherence dynamics of the qubits.
\item Even if the current state-of-the-art experimentally implemented QECCs will not be significantly affected by the TVQC (as they are generally short codes such as the $3\times 3$ toric code discussed in Chapter \ref{cp5}), the proposed channel model will be essential when better QECCs are implemented, such as the $7\times 7$ and $9\times 9$ toric codes for the near term, or the more advanced QTCs or QLDPCs for the long term.
\end{itemize}

Part \ref{part2} of the thesis is focused on improving the error correction performance of QTCs. Taking this aim into acount, we attempted the optimization of the error floor performance of QTCs by implementing practical interleavers rather than the traditionally considered random ones (Chapter \ref{cp6}). Inspired by the classical-to-quantum isomorphism, we adapted certain interleaving patterns that proved to improve the performance of classical turbo codes to the paradigm of quantum error correction. In particular, we implemented $S$-random, Welch-Costas and JPL interleavers for scrambling the quantum information that comes out of the outer quantum convolutional encoder. Simulation results showed that the error floor performance is improved up to two orders of magnitude. We also showed that memory requirements can be lowered by using specific interleaver designs, while the performance in the error floor region is comparable or even better than that of the original QTCs. We discussed that this benefits come without any decoding complexity costs.

Chapter \ref{cp7} dealt with QTCs whose decoders operate without channel state information. We began studying the sensitivity that the turbo decoder has to a mismatch between the actual channel depolarizing probability and the probability fed to the decoder. Simulation results showed that the decoder is sensitive both to under- and over-estimation of the depolarizing probability, presenting a flat region around the actual value of the depolarizing probability. Therefore, we discussed that estimation techniques are necessary for the QTCs to operate reliably over depolarizing channels. With that in mind, we first considered existing off-line estimation techniques that estimate the noise level of the channel previous to the operation of the QTCs. This techniques proved to be succesful when aiding the turbo decoders to reach the same performance as when perfect CSI is available, but at the expenses of an increase of the qubit overhead (which in certain scenarios may lead to code rate reduction) and system latency.  To overcome such issues, we proposed an on-line estimation technique that is able to use the information being processed in the inner SISO decoder of the QTC for succesfully estimating the channel depolarizing probability. Such an on-line estimator outperformed existing off-line estimators in terms of overall coding rate and latency, while maintaining excellent WER performance and decoding complexity. Simulation results showed that the on-line estimation method is insensitive to the initial value of the depolarizing probability, and thus, the resulting performance experiences very little degradation in the presence of channel mismatch. We finished Chapter \ref{cp7} by extending the on-line estimation method to scenarios where the channel is asymmetric, i.e. the Pauli twirled approximation channel. Monte Carlo simulations corroborated that the extended on-line estimation method is successful in aiding the QTC to achieve the same performance as when it has access to the channel parameters. We discussed that, here, this positive outcome comes at the cost of a slight increase in the complexity of the decoding algorithm.

The main takeaways of Part \ref{part2} can be summarized as:
\begin{itemize}
\item The classical-to-quantum isomorphism allows to import the knowledge of classical coding theory to the context of quantum error correction. This way, the performance of QECCs can be improved by utilizing classical techniques that have been extensively developed in the literature if they are adapted correctly.
\item Interleavers play a central role in the design of QTCs due to their ability to reduce the error floors of such codes while maintaining the performance in the waterfall region. Some interleaver constructions do also reduce the memory requirements of QTCs.
\item Quantum turbo codes are susceptible to channel state information mismatch. Feeding the turbo decoder with channel estimates that are not accurate enough compromises the excellent performance of such family of QECCs.
\item Off-line estimation protocols are succesful in aiding QTCs to operate with the same performance as when the turbo decoder has access to perfect channel information. However, this comes at the expenses of increasing the qubit overhead and system latency.
\item The on-line estimation method proposed is succesful in the task of making QTCs that are blind to CSI to operate with a performance similar to the ones that have access to perfect CSI. This method is succesful without needing a higher number of qubits. The on-line estimation protocol requires lower system latency than the discussed off-line estimation protocols. This modified decoder is able to operate for depolarizing (CTA) and asymmetric (PTA) channels.
\end{itemize}

In the final chapter of the dissertation, Chapter \ref{cp8}, we provided a brief overview of other QEC-related works that the author has been involved in during his time as a Ph.D. student. Such work was related to the design and optimization of quantum low-density-parity-check-matrix codes and the study of the quantum-unique phenomenon named degeneracy.

\section{Future research lines}
Many of the presented discussions and results are encouraging and suggest further study and analysis. Here we present some future research lines that arise from the work done in this Ph.D. thesis.

\subsection*{Asymptotically achievable limits for the combined amplitude and phase damping channel}
As discussed in Chapter \ref{ch:preliminary}, the quantum capacity of the combined amplitude and phase damping channel is yet unknown. It was discussed that, in general, quantum noise for two-level coherent systems is described by such channel that includes both relaxation and pure dephasing effects. Therefore, the knowledge of the asymptotically achievable limits at the time of writing reduces to the set of $T_1$-limited qubits, given by the quantum capacity of the amplitude damping channel. This, however, neglects an integral part of decoherence since most of the run-of-the-mill experimental qubits do not saturate the Ramsey limit ($T_2\approx 2T_1$). With this in mind, it is critical to obtain an expression for the quantum capacity of the combined amplitude and phase damping channel.

Moreover, in Chapter \ref{cp4}, we proposed the concept of quantum outage probability as the asymptotically achievable error rate by QECCs when operating over TVQCs that consider the realizations of the decoherence parameters to be equal for all the qubits of a block. In such chapter, we derived closed-form expressions of $p_\mathrm{out}^Q$ for qubits that are $T_1$-limited. Similar studies for the quantum outage probability of the more general time-varying combined amplitude and phase damping channel will be considered in future work. The aforementioned quantum capacity for such channel under static noise is needed for achieving this last task. This will be critical in order to have a complete tableau of the information theoretical limits of error correction for superconducting qubits with pure dephasing channels.

\subsection*{Fast time-varying quantum channels}
In Chapters \ref{cp4} and \ref{cp5} we assumed that the realizations of the decoherence parameters were the same for each of the qubits in the block of an error correction round. Under such premise, the TVQC model resembled the classical scenario of slow fading channels. Having such similarity in mind, we derived the quantum outage probability as the asymptotically achievable error rate by QECCs and simulated Kitaev toric codes and QTCs over such channels. Nevertheless, such assumption may not be true in general, indicating that depending on the qubit-wise correlation presented by $T_1$ and $T_2$ the system will operate in a different way to what has been studied here. At the time of writing, little to none experimental results regarding the qubit-wise correlation of the fluctuations of $T_1$ and $T_2$ have been done with the sole exception of \cite{decoherenceBenchmarking}. In such study, the authors studied the correlation between the relaxation times of two superconducting qubits and concluded that they were completely uncorrelated. This has important implications on error correction. First, since the realizations of the channel will be different and independent for each of the qubits of the block, this scenario will be reminiscent of the classical fast fading channel. That is, each of the symbols of the block will be affected by an independent realization of the channel. Therefore, the classical concept of \textit{ergodic capacity} may be integrated in the context of quantum error correction as the asymptotically achievable limit. This would obviously have implications on the performance of QECCs. Furthermore, considering a case where there is partial qubit-wise correlation of $T_1$ and $T_2$ would generalize every possible scenario for the proposed TVQC model.

As a consequence, this topic requires to be tackled both from the experimental and theoretical perspectives. First of all, a more complete experimental benchmarking of the fluctuations of the decoherence parameters of qubits needs to be done so that they can be properly characterized. In this way, the correlations among the qubits of the system can be correctly addressed by theoretical models. This will probably depend on the experimental implementation of the system (qubit technology, architecture, ...) and, thus, it is probable that each of the possible time-varying scenarios will be valid for some of the experimental quantum hardwares. Therefore, the theoretical study of the TVQC scenarios with uncorrelated or partially correlated fluctuations will be very important for understanding how QECCs will operate when real hardware is considered.

\subsection*{Switching rate QEC protocols for TVQCs}
Part \ref{part1} of the dissertation described the theoretical modelling of the fluctuations of the decoherence parameters of superconducting qubits and how such fluctuations degrade the performance of QECCs. It is important to remark that the proposed TVQC model is relevant for protocols involving a large number of error correction cycles (e.g., a quantum memory or a long quantum algorithm, such as the one in ref. \cite{rsaRounds}). If a very short algorithm or QECC is run just once or a very few number of times, the parameters will not fluctuate and the effects of the proposed channel model will not be noticeable. Herein lies important future work on this topic, such as the construction of optimized error correction codes that are apt to address these dynamic scenarios. In this way, we consider relevant trying to revert the degradation suffered by QECCs as a consequence of the decoherence parameter fluctuations. Therefore, a possible solution might be based on a switching rate protocol that can adapt the rate of the QECCs depending on the noise realization. In this way, the protocol would switch to lower QECC rates when noise level realizations are high, and to higher rates when the realizations are low. Rate adaptive schemes have been extensively studied in the field of classical coding theory \cite{imagra1,imagra2}. By doing this, it is possible to obtain a protocol that assures that the WER of the code is lower than some threshold value, say $10^{-3}$, while trying to consume as little qubit resources as possible\footnote{Note that by designing the QECCs by considering that the realizations are always bad will imply that the WER target would be achieved, but the qubit consumption would be high. Qubits are an expensive resource and, thus, designing codes that are far from being optimized is not the best praxis.}. It is important to realize that for this protocol to work appropriately, we need to implement an estimation protocol so that, depending on the estimated noise level, the rate is adapted accordingly. The online estimation protocol proposed in Chapter \ref{cp7} of this thesis is a perfect match for doing so. Note that the time needed for an error correction round is much smaller than the stochastic process coherence time and, thus, several QECC blocks can be fitted in the same noise realization. In this way, the on-line estimation protocol can check when the noise level has changed and, therefore, switch the rate accordingly. Another thing to be considered is that the selected QECC can vary its rate in a seamless manner. Protocols such as this one may be needed for the QECCs to maintain excellent performance with low qubit overheads when time variations are important.

\subsection*{Optimizing interleaver constructions for QTCs}
In Chapter \ref{cp6} of this Ph.D. dissertation we demonstrated that implementing interleavers with some construction leads to an improvement of the performance of QTCs in the error floor region. For doing so, we imported some well known interleaving patterns from classical turbo coding theory to the context of QTCs. In this context, there is still a vast amount of knowledge regarding interleaver construction to be imported from the classical coding theory \cite{newInt1,newInt2,newInt3,newInt4,newInt5}. This further optimization of the scrambling patterns in between the convolutional encoders that form the turbo code is  based on the actual structure of those two constituent codes. Therefore, in order to import those interleaving methods to the context of quantum information, the structure of the QCCs must be studied appropriately. Furthermore, as it has been seen in Chapter \ref{cp6}, quantum interleavers have an additional degree of freedom since $1$-qubit symplectic transformations can be included while interleaving the quantum information stream. Nevertheless, at the time of writing, how those would impact the performance of turbo codes remains to be studied. Thus, we consider important to address such additional degree of freedom when optimizing the performance of QTCs. To sum up, we consider that the error floors of QTCs can still be improved significantly by just constructing interleavers that are a good fit for concatenating the constituent QCCs.

These are the main future research topics that follow the investigations conducted throguhout this thesis. However, quantum information and computing are rather young and, thus, a vast amount of possible paths to follow are open at the moment. Fault-tolerant implementations of NISQ-era algorithms, hardware specific implementations of QECCs or the construction of error correction codes that maximize degeneracy may serve as examples of such endeavours from the point of view of quantum error correction. Recently, claims for quantum advantage have been made by Google and China independently \cite{GoogleSup,ChinaSup} and the barrier of $100$ superconducting qubits has been surpassed by the new ibm\_washington processor \cite{ibmeagle}. Needless to say, quantum computing has a bright future that will mainly be developed in this decade and QEC is considered to be key in order to hit the jackpot. This will require an Herculean effort by the community of quantum computing and quantum information. Exciting times are ahead!
 \label{Conclusion}
\clearemptydoublepage

\def\chaptername{APPENDIX}
\appendix
\begin{appendices}
\chapter{Tensor Product}\label{app:tensor}
\textit{Tensor products} are mathematical operations that put vector spaces together in order to form larger vector spaces. The tensor product is a crucial part in order to understand multiparticle systems in quantum mechanics. 

Suppose that $V$ and $W$ are vector spaces of dimension $m$ and $n$ respectively. Then $V\otimes W$, or the tensor product between the vector spaces, is an $mn$-dimensional vector space. The elements of this composite vector space are linear combinations of tensor products $\ket{v}\otimes\ket{w}$ of the elements $\ket{v}$ of $V$ and $\ket{w}$ of $W$. In particular, if $\ket{i}$ and $\ket{j}$ are orthonormal bases for the spaces $V$ and $W$ respectively, then $\ket{i}\otimes\ket{j}$ is a basis for the composite vector space $V\otimes W$.

Tensor product in matrix vector spaces is performed by the so-called \textit{Kronecker product}.
\begin{definition}[Kronecker Product]
Suppose $A\in \mathbb{C}^{m\times n}$ and $B\in \mathbb{C}^{p\times q}$ are two arbitrary complex matrices. Then $A\otimes B\in \mathbb{C}^{mp\times nq}$ and is defined as
\begin{equation}\nonumber
A\otimes B \equiv \begin{pmatrix}
a_{11}B & a_{12}B & \cdots & a_{1n}B \\
a_{21}B & a_{22}B & \cdots
& a_{2n}B \\
\vdots & \vdots & \ddots & \vdots \\
a_{m1}B & a_{m2}B & \cdots & a_{mn}B
\end{pmatrix},
\end{equation}
\label{def:kron}
\end{definition}
where $a_{ij}$ refers to the element in row $i$ and column $j$ of matrix $A$.

The Kronecker product as defined in definition \ref{def:kron} will be the tensor product that will be considered in the thesis. This tensor product presents the next properties:
\begin{itemize}
\item \textbf{Non-commutativity}: $A\otimes B \neq B\otimes A$ in general.
\item \textbf{Bilineality and associativity}:
\begin{itemize}
\item $A \otimes (B+C) = A\otimes B + A\otimes C$, with $B$ and $C$ of same dimensions.
\item $(A+B)\otimes C =A\otimes C + B\otimes C$, with $A$ and $B$ of same dimensions.
\item $(kA)\otimes B= A\otimes (kB) = k(A\otimes B)$, where $k\in \mathbb{C}$.
\item $(A\otimes B)\otimes C=A\otimes (B\otimes C)$.
\end{itemize}
\item \textbf{Mixed-product}: $(A\otimes B)(C\otimes D) = AC \otimes BD$, where the dimensions must be appropriate for the ordinary matrix product. From here it is derived that:
\begin{itemize}
\item $(A\otimes B)^{-1}=A^{-1}\otimes B^{-1}$, from where it is clear that for $(A\otimes B)^{-1}$ to exist, both $A^{-1}$ and $B^{-1}$ must exist.
\item $(A\otimes B)^T=A^T\otimes B^T$.
\end{itemize}
\item \textbf{Trace}: $\mathrm{Tr}(A\otimes B)=\mathrm{Tr}(A)\mathrm{Tr}(B)$.
\item \textbf{Determinant}: $\mathrm{det}(A\otimes B) =\mathrm{det}(A)^q \mathrm{det}(B)^n$, where $A\in\mathbb{C}^{n\times n}$ and $B\in\mathbb{C}^{q\times q}$.
\item \textbf{Rank}: $\mathrm{rank}(A\otimes B)=\mathrm{rank}(A)\mathrm{rank}(B)$.
\end{itemize}

To finish with the tensor product, notation $A^{\otimes k}$ will mean $A^{\otimes k} = A\otimes A\otimes\cdots\otimes A$, that is $A$ matrix tensored by itself $k$ times. For example, $\ket{\psi}^{\otimes 3} = \ket{\psi}\otimes\ket{\psi}\otimes\ket{\psi}$.

\chapter{Decoding algorithms} \label{app:decoding}
In this appendix we will detail the decoding algorithms that have been used in this dissertation. We will discuss the decoding algorithm for the Quantum Turbo codes and the Minimum-Weight Perfect Matching (MWPM) decoder for Kitaev toric codes. We will not go into the subtleties of these decoding methods, but rather present their principles of operation. We provide appropriate references for the reader that would like to go deeper in understanding the decoding protocols for such quantum error correction codes.

\section{Quantum Turbo Code decoder}
Decoding of classical turbo codes rely on belief propagation methods such as the sum-product algorithm \cite{sumproduct} that ouput the most probable codeword transmitted based on the syndrome of the received codeword and the information possesed about the channel. In the quantum setting, decoding is not to find out the most probable quantum state transmitted\footnote{Note that quantum states are continuous, so this is impossible to do.}, but to find out the most probable error coset of effective Pauli errors that have affected the codeword, and that have been discretized after syndrome measurement. Consequently, sum-product like algorithms in the quantum decoding world estimate the most probable error coset and so afterwards such correction operation can be done to the received state.

For QTCs, a \textit{Soft-Input-Soft-Output} (SISO) decoding algorithm was proposed for turbo decoding in \cite{QTC} and it was later enhanced by means of the so called extrinsic information transfer in \cite{EAQTC}. Such decoding algorithm will be described next. The decoding scheme for QTCs is shown in Figure \ref{fig:decodingQTC}.

The SISO decoder for quantum convolutional codes is an adaptation of the usual classical algorithm \cite{QTC,sumproduct}. This algorithm relies on the convolutional structure of the convolutional encoder in order to perform a forward-backward algorithm that will output the estimated channel error. Two of these are used in order to decode the quantum turbo code, as the code is obtained as the interleaved concatenation of two quantum convolutional codes. The SISO decoding algorithm for QCCs is extensively described in \cite{QTC}.

\begin{figure}[h!]
\centering
\includegraphics[scale=1]{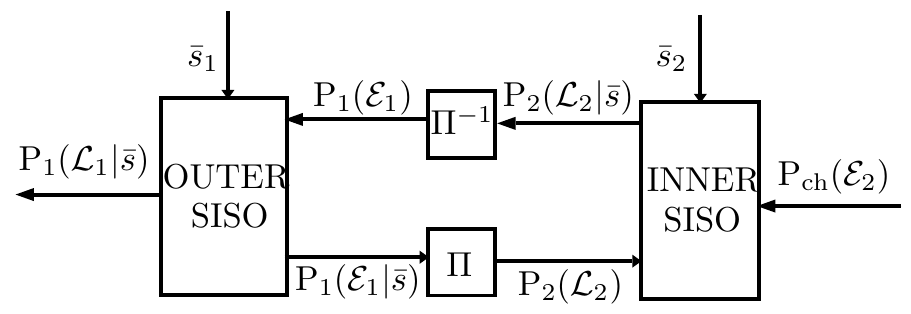}
\caption{Circuit for SISO decoding of quantum turbo codes.}
\label{fig:decodingQTC}
\end{figure}

From Figure \ref{fig:decodingQTC} it can be seen that the receiver first decodes the inner code based on the channel information, $\mathrm{P}_\mathrm{ch}(\mathcal{E}_2)$, and the measured syndrome, $\bar{s}_2$, for such inner code. The syndrome measurement is done as for block QECCs (see Chapter \ref{ch:preliminary}) and the estimation of the channel error is done by taking advantage of the structure of the inner convolutional decoding unitary. Afterwards, the obtained probabilities of the logical errors that may have happened during transmission conditioned to the syndrome read, $\mathrm{P}_2(\mathcal{L}_2|\bar{s})$, are passed thorugh the interleaver and given as a priori input information to the outer SISO decoder, $\mathrm{P}_1(\mathcal{E}_1)$. This element does the same operation as the inner SISO decoder with the new information, $\mathrm{P}_1(\mathcal{E}_1)$, and the syndrome read for the outer code, $\bar{s}_1$. The decoding for the quantum state is now done by taking advantage of the structure of the outer convolutional decoding unitary. After decoding, the information about the probability of the physical operator $\mathcal{E}_1$ conditioned with the syndrome read, $\mathrm{P}_1(\mathcal{E}_1|\bar{s})$, is deinterleaved and sent back to the the inner SISO decoder, $\mathrm{P}_2(\mathcal{L}_2)$. For the first iteration, such probability is taken as equiprobable. This loop is done several times until the estimates of the inner and outer decoders match or some number of limit iterations is achieved (convergence is the most probable thing to happen as stated in \cite{EAQTC}). The output of the decoder is the probability distribution $\mathrm{P}_1(\mathcal{L}_1|\bar{s})$ and, thus, the most probable error coset given the syndrome read can be estimated.

It has been stated that in \cite{EAQTC}, an improved version of the turbo decoding algorithm was presented, and we briefly present how such enchancement was done. The algorithm presented before is based on the exchange of \textit{a posteriori} information between the constituent quantum convolutional SISO decoders of the QTC. As the decoders pass along \textit{a posteriori} information, successive iterations of the constituent decoders are dependent on one another giving rise to a detrimental positive feedback effect that prevents the decoding algorithm to achieve the desired gains that iterative decoding usually presents.

For the sake of avoiding such effect, it is necessary that the \textit{a priori} information directly related to a given information qubit is not reused in the other decoder. Consequently, and similar to the approach employed in classical turbo decoding, this can be achieved by making one decoder to remove the \textit{a priori} information from the \textit{a posteriori} information before feeding it to the other decoder. More explicitly, the iterative decoding procedure should exchange only \textit{extrinsic} information that is new and unknown to the other decoder. To see how this information transfer works, consider a four-port\footnote{Note that in the scheme presented in Figure \ref{fig:decodingQTC}, the SISO decoders are not exactly like that, but this is a way to see how they work.} SISO decoder that generates soft output information pertaining to a logical error $\mathcal{L}$ and a physiscal error $\mathcal{E}$. Such circuit is presented in Figure \ref{fig:APP}.

\begin{figure}
\centering
\includegraphics[scale=1]{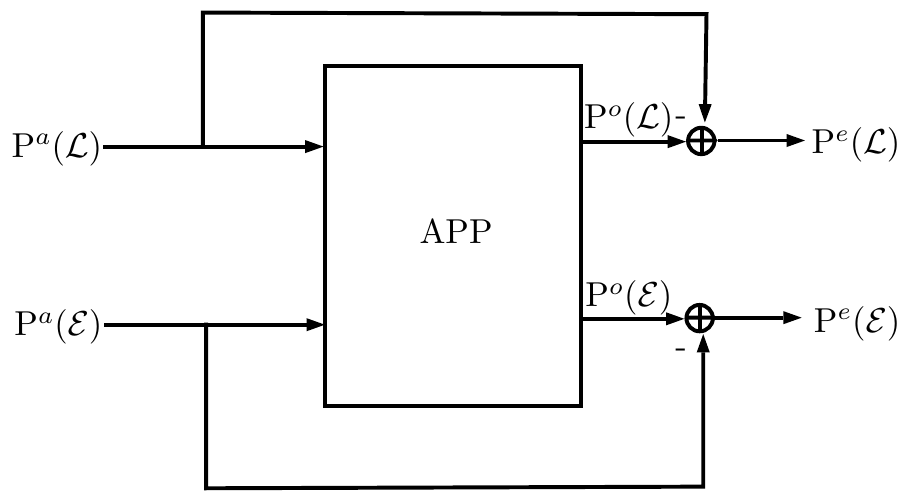}
\caption{A SISO decoder processes \textit{a priori} information and outputs \textit{a posteriori} information, from which the \textit{a priori} information is removed to obtain  \textit{extrinsic} information that is passed to the next decoder. Note that the probabilities here are in logarithmic scale so that they can be just substracted.}
\label{fig:APP}
\end{figure}

As it can be seen in Figure \ref{fig:APP}, a SISO decoder exploits an \textit{A Posteriori Probability} (APP) module that accepts the \textit{a priori} information $\mathrm{P}^a(\mathcal{L})$ and $\mathrm{P}^a(\mathcal{E})$ as input and outputs the \textit{a posteriori} information $\mathrm{P}^o(\mathcal{L})$ and $\mathrm{P}^o(\mathcal{E})$. The corresponding \textit{extrinsic} probabilities $\mathrm{P}^e(\mathcal{L}_i^j)$ and $\mathrm{P}^e(\mathcal{E}_i^j)$ for the $j^{th}$ qubit at time instant $i$ are then obtained by discarding the \textit{a priori} information from the \textit{a posteriori} information as follows
\begin{equation}\nonumber
\begin{split}
& \mathrm{P}^e(\mathcal{L}_i^j)=N_{\mathcal{L}^j}\frac{\mathrm{P}^o(\mathcal{L}_i^j)}{\mathrm{P}^a(\mathcal{L}_i^j)}, \\
& \mathrm{P}^e(\mathcal{E}_i^j)=N_{\mathcal{E}^j}\frac{\mathrm{P}^o(\mathcal{E}_i^j)}{\mathrm{P}^a(\mathcal{E}_i^j)},
\end{split}
\end{equation}
where $N_{\mathcal{L}^j}$ and $N_{\mathcal{E}^j}$ are normalization factors, which ensure that $\sum_{\mathcal{L}^j}\mathrm{P}^e(\mathcal{L}_i^j)=1$ and $\sum_{\mathcal{E}^j}\mathrm{P}^e(\mathcal{E}_i^j)=1$, respectively.

Furthermore, to avoid numerical instabilities and to reduce the computational complexity, log-domain arithmetics are conventionally employed, which convert the multiplicative operations to addition, as given below
\begin{equation}\nonumber
\begin{split}
& \ln[\mathrm{P}^e(\mathcal{L}_i^j)]=\ln[N_{\mathcal{L}^j}]+\ln[\mathrm{P}^o(\mathcal{L}_i^j)] - \ln[\mathrm{P}^a(\mathcal{L}_i^j)], \\
& \ln[\mathrm{P}^e(\mathcal{E}_i^j)]=\ln[N_{\mathcal{E}^j}]+\ln[\mathrm{P}^o(\mathcal{E}_i^j)] - \ln[\mathrm{P}^a(\mathcal{E}_i^j)].
\end{split}
\end{equation}

Therefore, the inputs and outputs of a SISO decoder are the logarithmic \textit{a priori} and extrinsic probabilities, repsectively.

In \cite{EAQTC} it was shown via simulation that this approach of extrinsic information transfer outerperforms the original turbo decoder proposed in \cite{QTC} that relied on the \textit{a posteriori} information.

The turbo decoding algorithm via SISO decoders for the convolutional codes has been presented here. Note that the turbo decoding algorithm is a sum-product like algorithm, similar to the ones used for decoding classical turbo codes with the slight difference that the effective Pauli errors are estimated instead of the codewords. Consequently, this algorithm is a \textit{classical} algorithm that will run on a \textit{classical computer}, and the quantum computer will be classically aided to solve the decoding problem and, thus, the complexity of the algorithm will be classical. Consequently, it is very important that the algorithm does not take more time than the decoherence time of the qubit, as if that was the case, then the qubits may suffer from additional errors that are not considered by the decoder. This implies that the study of quantum computing does also involve classical algorithms, and enhancing them will improve the quality of quantum error correction. Other approach to this problem would be to use the inherent parallel nature of quantum computers to develop quantum sum-product algorithms that would run in the quantum computers with the advantage to have a potentially better complexity than the classical ones and, therefore, improve the latency of the communications \cite{QBP1,QBP2}. Variations of those could also be useful to decode classical codes and reduce the latency of classical communications with hybrid classical-quantum computers \cite{EXITQTC}.

\section{Minimum-Weight Perfect Matching (MWPM) Decoder}
The Kitaev toric code is a surface code, which are codes defined on a 2D lattice of qubits \cite{toric}. Specifically, the toric codes have periodic boundaries in those latices, making them to have a torus-like shape. In surface codes, decoding a syndrome is equivalent to finding paths between the generators whose syndrome has been triggered. There will be multiple paths associated to the same syndrome and, consequently, the decoder must estimate the path that describes the most likely error. In general, the optimal decoder for toric codes has exponential complexity \cite{surface,DelftToric}, and, thus, it is impractical due to the fact that the classical algorithm must obtain an estimation of the error before the qubit suffers from additional errors.

A widely used decoding algorithm for the family of Kitaev toric codes is the so called \textit{Minimum-weight Perfect Matching} (MWPM) decoder \cite{surface,DelftToric}. This decoder is based on the MWPM problem of graph theory, where a matching\footnote{A set of edges without common vertices.} in which the sum of weights is minimized must be found. The term perfect refers to the fact that the matching matches all vertices of the graph. One can convert the lattice of the toric code with errors to a complete graph, where the generators with non-trivial syndrome are the nodes. The edges between the vetices have a weight equal to the minimum number of qubits between them. This way, solving the MWPM problem to such associated graph, the most likely error, i.e. the one with minimum weight, can be estimated \cite{surface,DelftToric}. The MWPM problem can be solved using Edmund's Blossom algorithm \cite{blossom} with some improvements in order to have a better computational complexity.

The toric codes of this dissertation are decoded via the MWPM decoder. We use the implementation of this decoder in the QECSIM tool \cite{surface} for the numerical simulations conducted.

\chapter{Monte Carlo numerical simulations}\label{app:montecarlo}
In this appendix we describe the numerical Monte Carlo simulations that we have conducted in this dissertation in order to asses the performance of the quantum error correction codes discussed. This way we are able to justify that the numerical simulations performed are accurate and, thus, the conclusions obtained from them valid.

The numerical simulations are based on the approximate error models that were explained in Chapter \ref{ch:preliminary}. Each round of the numerical simulation is performed by generating an $n$-qubit Pauli operator, calculating its associated syndrome and finally running the decoding algorithm using the syndrome as its input. The error operators are randomly generated using the distinct probability distributions for the Pauli channel that were discussed in this thesis. Once the logical error is estimated it is compared with the channel error\footnote{We observe if the logical error associated to the channel error and the estimated logical error belong to the same error coset.} in order to decide if the decoding round was succesful. The operational figure of merit selected in order to evaluate the performance of these quantum error correction schemes is the Word Error Rate ($\mathrm{WER}$), which is the probability that at least one qubit of the received block is incorrectly decoded.

Since the simulations we are conducting are based on the random sampling of the channel errors, the rules of Monte Carlo simulations are invoked in order to obtain results that are honest. Thus, for estimating the $\mathrm{WER}$ of the Kitaev toric codes and QTCs, the following Monte Carlo rule of thumb has been used in order to select the number of blocks to be transmitted, $N_{\mathrm{blocks}}$ \cite{montecarlo}:
\begin{equation}
N_{\mathrm{blocks}} = \frac{100}{\mathrm{WER}}.
\end{equation}
As explained in \cite{montecarlo}, under the assumption that the observed error events are independent, this results in a $95\%$ confidence interval of about $(0.8\hat{\mathrm{WER}},$ $  1.25\hat{\mathrm{WER}})$, where $\hat{\mathrm{WER}}$ refers to the empirically estimated value for the $\mathrm{WER}$. This way, we can justify that the results obtained via random sampling of the quantum channel are statistically representative and, thus, the conclusions obtained from those throughout the dissertation are valid.

\chapter{Diamond norm distance}\label{app:diamond}
In Chapter \ref{cp3}, the \textit{diamond norm distance} is used as the metric to compare the action of static quantum channels versus the proposed time-varying quantum channels. The objective of such analysis is to show that neglecting the fluctuating of $T_1$ and $T_2$ may result in an unrealistic model for quantum noise. Here we introduce the needed concepts regarding diamond norm distance in order to understand such comparison. We also provide a proof that shows that the diamond norm distances between the two ADPTA and two ADCTA channels obtained from the same two AD channels are the same. This applies for $T_1$-limited qubits.

\section{Diamond norm distance}\label{meth:diamondnorm}

The diamond norm distance \cite{diamondNat,FanoDiamond} between two quantum channels $\mathcal{N}_1$ and $\mathcal{N}_2$ is defined as
\begin{equation}\label{eq:diamond}
||\mathcal{N}_1 - \mathcal{N}_2||_\diamond = \sup_\rho ||\mathcal{N}_1\otimes \mathrm{I}(\rho) - \mathcal{N}_2\otimes \mathrm{I}(\rho)||_1,
\end{equation}
where $||\cdot||_1$ is the trace norm. The diamond norm considers states $\rho$ that might be entangled with some ancillary system that is not altered by the action of the channels. Consider two quantum channels $\mathcal{N}_1,\mathcal{N}_2$ and a single channel use. A quantum channel discrimination protocol aims to maximize the probability of correctly identifying which channel has acted on the quantum state if one of them is chosen uniformly at random. The diamond norm distance is related to the minimum probability of error, $p_\mathrm{e}$, that a discrimination protocol for $\mathcal{N}_1$ and $\mathcal{N}_2$ can achieve as
\begin{equation}\label{eq:chandisc}
p_\mathrm{e} = \frac{1}{2} - \frac{||\mathcal{N}_1 - \mathcal{N}_2||_\diamond}{4}.
\end{equation}
Therefore, the diamond norm distance serves as a measure of how differently two quantum channels affect input quantum states, and so it is a good metric to assess the dissimilarity between channels.

For Pauli channels $\mathcal{N}_\mathrm{P}$ defined as
\begin{equation}\label{eq:pauli}
\mathcal{N}_\mathrm{P}(\rho) = p_\mathrm{I}\rho + p_\mathrm{x}\mathrm{X}\rho \mathrm{X} + p_\mathrm{y} \mathrm{Y}\rho \mathrm{Y} + p_\mathrm{z} \mathrm{Z}\rho \mathrm{Z},
\end{equation}
with $p_\mathrm{I}=1-p_\mathrm{x}-p_\mathrm{y}-p_\mathrm{z}$, the diamond norm distance has the closed-form expression \cite{FanoDiamond}
\begin{equation}\label{eq:PauliDiamond}
||\mathcal{N}_\mathrm{P}^1 - \mathcal{N}_\mathrm{P}^2||_\diamond = \sum_{k} |p_k^1 - p_k^2|,
\end{equation}
where $k\in\{\mathrm{I,X,Y,Z}\}$.

Amplitude damping channels $\mathcal{N}_{\mathrm{AD}}$ also have a closed-form expression \cite{ADdiamond} for the diamond norm distance which is given by
\begin{equation}\label{eq:ADdiamond}
\begin{split}
  ||\mathcal{N}_{\mathrm{AD}}^1 -& \mathcal{N}_{\mathrm{AD}}^2||_\diamond  =  \\
     &\begin{cases}
       2|\gamma_1 - \gamma_2|  & \text{if } \sqrt{1-\gamma_1} + \sqrt{1-\gamma_2} >1 \\
       \\
       \frac{2|\sqrt{1-\gamma_1}-\sqrt{1-\gamma_2}|}{2 - (\sqrt{1-\gamma_1} + \sqrt{1-\gamma_2})} & \text{otherwise}
     \end{cases}
     \end{split} ,
\end{equation}
where $\gamma_1$ and $\gamma_2$ are the damping probailities that define the dynamics of each of the AD channels.

The action of phase damping or dephasing channels $\mathcal{N}_{\mathrm{PD}}$ can be easily rewritten as a function of the Pauli matrices as
\begin{equation}
\mathcal{N}_{\mathrm{PD}}(\rho) = \frac{1+\sqrt{1-\lambda}}{2}\mathrm{I}\rho \mathrm{I} + \frac{1-\sqrt{1-\lambda}}{2}\mathrm{Z}\rho \mathrm{Z},
\end{equation}
where $\lambda$ is interpreted as the scattering probability of a photon without loss of energy. Consequently, the PD channel is equivalent to a pure Pauli dephasing channel, that is, a Pauli channel where the only non-trivial errors are the phase-flip errors $\mathrm{Z}$. As a result, using \eqref{PauliDiamond}, the diamond norm distance between dephasing channels has the closed-form expression

\begin{equation}\label{eq:PDdiamond}
||\mathcal{N}_{\mathrm{PD}}^1 - \mathcal{N}_{\mathrm{PD}}^2||_\diamond = |\sqrt{1-\lambda_1} -\sqrt{1-\lambda_2}|.
\end{equation}

Finally, for quantum channels whose diamond norm distance has no closed-form expression, a method to efficiently compute it via semidefinite programming (SDP) was presented in \cite{SDPdiamond}. The QETLAB tool \cite{QETLAB} for MATLAB includes the calculation of diamond norm distances using those SDP methods. We use such tool in order to obtain the diamond norm distances between static and time-varying combined amplitude and phase damping channels.

\section{$||\mathcal{N}_{\mathrm{ADPTA}}(\gamma_1) - \mathcal{N}_{\mathrm{ADPTA}}(\gamma_2)||_\diamond = ||\mathcal{N}_{\mathrm{ADCTA}}(\gamma_1)-\mathcal{N}_{\mathrm{ADCTA}}(\gamma_2)||_\diamond$}
\begin{proposition}
The diamond norm distance between the Pauli twirl approximated channels (PTA) obtained from two amplitude damping channels (AD) of parameters $\gamma_1$ and $\gamma_2$ is the same as the diamond norm distance between the Clifford twirl approximated channels (CTA) obtained from the same AD channels:
\begin{equation}
||\mathcal{N}_{\mathrm{ADPTA}}(\gamma_1) - \mathcal{N}_{\mathrm{ADPTA}}(\gamma_2)||_\diamond = ||\mathcal{N}_{\mathrm{ADCTA}}(\gamma_1)-\mathcal{N}_{\mathrm{ADCTA}}(\gamma_2)||_\diamond.
\end{equation}
\end{proposition}
\begin{proof}
The PTA and CTA channels are Pauli channels \cite{twirl3,twirl6}. The PTA channel for an AD channel with damping probability $\gamma$ has the following probabilities for each of the Pauli matrices:
\begin{equation}\nonumber
p_\mathrm{I} = \frac{2-\gamma + 2 \sqrt{1-\gamma}}{4},
p_\mathrm{x} = p_\mathrm{y} = \frac{\gamma}{4}\text{ and }
 p_\mathrm{z}=\frac{2-\gamma - 2\sqrt{1-\gamma}}{4}.
\end{equation}
The CTA channel is a depolarizing channel (as a special case of the Pauli channels), with the following probabilities:
\begin{equation}\nonumber
 r_\mathrm{I} = \frac{2-\gamma+2\sqrt{1-\gamma}}{4}\text{ and }
r/3=r_\mathrm{x}=r_\mathrm{y}=r_\mathrm{z} = \frac{2+\gamma-2\sqrt{1-\gamma}}{12}.
\end{equation}
The diamond norm distance between Pauli channels has the closed-form expression \cite{FanoDiamond}
\begin{equation}
||\mathcal{N}_\mathrm{P}^1 - \mathcal{N}_\mathrm{P}^2||_\diamond = \sum_k |p_k^1 - p_k^2|,
\end{equation}
with $k\in\{\mathrm{I,X,Y,Z}\}$. Knowing that the PTAs and CTAs are Pauli channels, we can write the diamond norm distances between them obtained from approximating two AD channels:
\begin{enumerate}[(a)]
\item\label{PTA} $||\mathcal{N}_{\mathrm{ADPTA}}(\gamma_1)-\mathcal{N}_{\mathrm{ADPTA}}(\gamma_2)||_\diamond = |p_\mathrm{I}^1-p_\mathrm{I}^2|+|p_\mathrm{x}^1-p_\mathrm{x}^2|+|p_\mathrm{y}^1-p_\mathrm{y}^2|+|p_\mathrm{z}^1-p_\mathrm{z}^2|$
\item\label{CTA} $||\mathcal{N}_{\mathrm{ADCTA}}(\gamma_1)-\mathcal{N}_{\mathrm{ADCTA}}(\gamma_2)||_\diamond = |r_\mathrm{I}^1-r_\mathrm{I}^2|+\sum_{l=k}^3 |r_k^1-r_k^2| = |r_\mathrm{I}^1-r_\mathrm{I}^2|+3|r^1/3 - r^2/3|$
\end{enumerate}

Note that $p_\mathrm{I}=r_\mathrm{I},\forall\gamma$, so the contribution of the probabilities associated to non-identity components to the diamond norm distances must be compared. We start by expanding (\ref{CTA}):

\begin{equation}\label{eq:CTAStep1}
\begin{split}
3|r^1/3 - r^2/3| &= 3\left|\frac{\gamma_1 - \gamma_2 + 2(\sqrt{1-\gamma_2} - \sqrt{1-\gamma_1})}{12}\right| \\ & =  \frac{1}{4}|\gamma_1 - \gamma_2 + 2(\sqrt{1-\gamma_2} - \sqrt{1-\gamma_1})|.
\end{split}
\end{equation}

The results shown below for the terms $\gamma_1 - \gamma_2$  and $\sqrt{1-\gamma_2} - \sqrt{1-\gamma_1}$ that appear in equation \eqref{CTAStep1} will be useful in what follows:
\begin{itemize}
\item $\gamma_1 - \gamma_2 > 0 \rightarrow \gamma_1>\gamma_2 \rightarrow -\gamma_1<-\gamma_2 \rightarrow 1-\gamma_1<1-\gamma_2\rightarrow \sqrt{1-\gamma_1}<\sqrt{1-\gamma_2}\rightarrow \sqrt{1-\gamma_2} - \sqrt{1-\gamma_1}>0$.
\item $\gamma_1 - \gamma_2 < 0 \rightarrow \gamma_1<\gamma_2 \rightarrow -\gamma_1>-\gamma_2 \rightarrow 1-\gamma_1>1-\gamma_2\rightarrow \sqrt{1-\gamma_1}>\sqrt{1-\gamma_2}\rightarrow \sqrt{1-\gamma_2} - \sqrt{1-\gamma_1}<0$.
\item $\gamma_1-\gamma_2 = 0 \rightarrow \sqrt{1-\gamma_2} - \sqrt{1-\gamma_1}=0$.
\end{itemize}

From these observations, it is easy to see that $\gamma_1 - \gamma_2$  and $\sqrt{1-\gamma_2} - \sqrt{1-\gamma_1}$ share the same sign $\forall\gamma_1,\gamma_2$. Thus, we can split the sum inside the absolute value of \eqref{CTAStep1} into a sum of absolute values as
\begin{equation}\label{eq:CTAStep2}
\eqref{CTAStep1} = \frac{1}{4}|\gamma_1 - \gamma_2|+ \frac{1}{2}|\sqrt{1-\gamma_2} - \sqrt{1- \gamma_1}|.
\end{equation}

We continue by developing the non-identity part for the PTA as given in (\ref{PTA}):
\begin{equation}\label{eq:PTAStep1}
\begin{split}
|p_\mathrm{x}^1-p_\mathrm{x}^2|+|p_\mathrm{y}^1-p_\mathrm{y}^2|+|p_\mathrm{z}^1-p_\mathrm{z}^2| =& \frac{|\gamma_1-\gamma_2|}{4} + \frac{|\gamma_1-\gamma_2|}{4} \\& + \frac{|\gamma_2-\gamma_1 + 2(\sqrt{1-\gamma_2} - \sqrt{1-\gamma_1})|}{4}.
\end{split}
\end{equation}

We follow the same reasoning as for the CTA by making use of the following expressions that $\gamma_2 - \gamma_1$  and $\sqrt{1-\gamma_2} - \sqrt{1-\gamma_1}$ fulfil. They are later added within the absolute value of their corresponding terms in the expression:
\begin{itemize}
\item $\gamma_2 - \gamma_1 > 0 \rightarrow \gamma_2>\gamma_1 \rightarrow -\gamma_2<-\gamma_1 \rightarrow 1-\gamma_2<1-\gamma_1\rightarrow \sqrt{1-\gamma_2}<\sqrt{1-\gamma_1}\rightarrow \sqrt{1-\gamma_2} - \sqrt{1-\gamma_1}<0$.
\item $\gamma_2 - \gamma_1 < 0 \rightarrow \gamma_2<\gamma_1 \rightarrow -\gamma_2>-\gamma_1 \rightarrow 1-\gamma_2>1-\gamma_1\rightarrow \sqrt{1-\gamma_2}>\sqrt{1-\gamma_1}\rightarrow \sqrt{1-\gamma_2} - \sqrt{1-\gamma_1}>0$.
\item $\gamma_2-\gamma_1 = 0 \rightarrow \sqrt{1-\gamma_2} - \sqrt{1-\gamma_1}=0$.
\end{itemize}

From these outcomes, we can tell that $\gamma_2 - \gamma_1$  and $\sqrt{1-\gamma_2} - \sqrt{1-\gamma_1}$ have a different sign $\forall \gamma_1,\gamma_2$. In consequence we can write the sum inside the absolute value of the third term of \eqref{PTAStep1} as
\begin{equation}\label{eq:PTAStep2}
\eqref{PTAStep1} = \left|2|\sqrt{1-\gamma_2} - \sqrt{1-\gamma_1}| - |\gamma_2 - \gamma_1|\right|,
\end{equation}
that is to say, the absolute value of computing the substraction of each of the original absolute values. To further expand the expression in \eqref{PTAStep2}, we use
\begin{equation}\label{PTAmid1}
\begin{split}
& |(\sqrt{1-\gamma_2} - \sqrt{1-\gamma_1})(\sqrt{1-\gamma_2}+\sqrt{1-\gamma_1})|  = |1-\gamma_2 - 1 + \gamma_1| \\  &= |\gamma_2 - \gamma_1| \rightarrow 2 |\sqrt{1-\gamma_2} - \sqrt{1-\gamma_1}||\sqrt{1-\gamma_2}+ \sqrt{1-\gamma_1}| \\ &= 2|\gamma_2-\gamma_1| \rightarrow 2|\sqrt{1-\gamma_2}- \sqrt{1-\gamma_1}| = \frac{2|\gamma_2 - \gamma_1|}{|\sqrt{1-\gamma_2} + \sqrt{1-\gamma_1}|}\\ &\rightarrow 2|\sqrt{1-\gamma_2}- \sqrt{1-\gamma_1}| \geq |\gamma_2 - \gamma_1|,
\end{split}
\end{equation}
$\forall \gamma_1,\gamma_2$ where the last inequality comes from the fact that
\begin{itemize}
\item $|\sqrt{1-\gamma_2} - \sqrt{1-\gamma_1}|\in(0,2] \rightarrow 0< \frac{1}{2}|\sqrt{1-\gamma_2} - \sqrt{1-\gamma_1}|\leq 1 \rightarrow \frac{2}{|\sqrt{1-\gamma_2} - \sqrt{1-\gamma_1}|}\geq 1$
\item $|\sqrt{1-\gamma_2} - \sqrt{1-\gamma_1}| = 0 \rightarrow 2|\sqrt{1-\gamma_2} - \sqrt{1-\gamma_1}|=|\gamma_2-\gamma_1| = 0$.
\end{itemize}

Having shown this result, we can work on \eqref{PTAStep2} even further by removing the absolute value, since we have shown that the substraction will always be $\geq 0$:
\begin{equation}\label{eq:PTAmid2}
\eqref{PTAStep2} = 2|\sqrt{1-\gamma_2} - \sqrt{1-\gamma_1}| - |\gamma_2-\gamma_1|.
\end{equation}

We finish by introducing \eqref{PTAmid2} into expression \eqref{PTAStep1}:
\begin{equation}\label{eq:PTAStep3}
\begin{split}
\eqref{PTAStep1} =&\\& \frac{|\gamma_1-\gamma_2|}{4} + \frac{|\gamma_1-\gamma_2|}{4}   + \frac{1}{2}|\sqrt{1-\gamma_2}-\sqrt{1-\gamma_1}| - \frac{|\gamma_2-\gamma_1|}{4} \\ & = \frac{1}{4}|\gamma_1-\gamma_2| + \frac{1}{2}|\sqrt{1-\gamma_2}-\sqrt{1-\gamma_1}|
\end{split}
\end{equation}

If we now compare \eqref{CTAStep2} and \eqref{PTAStep3}, it is clear that both of them are the same, which implies that
\begin{equation}
||\mathcal{N}_{\mathrm{ADPTA}}(\gamma_1) - \mathcal{N}_{\mathrm{ADPTA}}(\gamma_2)||_\diamond = ||\mathcal{N}_{\mathrm{ADCTA}}(\gamma_1)-\mathcal{N}_{\mathrm{ADCTA}}(\gamma_2)||_\diamond,
\end{equation}
and so concludes the proof for the proposition.
\end{proof}

\chapter{Adjusted boxplots for skewed distributions}\label{app:boxSkew}
Boxplots are a heavily employed tool in statistics to see how specific statistical data is distributed. A boxplot is comprised of a box whose top edge refers to the third quartile ($\mathrm{Q3}$) and whose bottom edge refers to the first quartile ($\mathrm{Q1}$). A line inside this box represents the median and two whiskers at both ends indicate $\mathrm{Q3}+1.5(\mathrm{Q3}-\mathrm{Q1})$ for the top whisker and $\mathrm{Q1}-1.5(\mathrm{Q3}-\mathrm{Q1})$ for the bottom one. Data that is found outside whisker range are considered to be outliers in the data set and are individually plotted. Boxplots are based on the normal distribution, and so they are not very representative for heavily skewed distributions. Note that the length is the same for the top and the bottom whiskers. Consequently, the real skewness of heavily skewed data will not be represented in such boxplots, with several data being represented as outliers as a result of the long tails that such types of distributions have. In consequence, conventional boxplots do not reflect a valid range of outliers, since the long tail side requires that more data be accounted for as typical values.

Adjusted boxplots have been presented in the literature as a methodology to take the imbalanced weight of skewed data into account \cite{skewed}. The core idea of this adjusted boxplot is to use a robust measure of the skewness of the dataset to define the length of each of the whiskers. \cite{skewed} employs a measure known as the medcouple ($\mathrm{MC}$) \cite{medcouple} for this purpose. The whisker range $\Delta$ is then defined as
\begin{equation}\label{eq:adjBoxplot}
\begin{cases}
       [\mathrm{Q1}-1.5\mathrm{e}^{-4\mathrm{MC}}\mathrm{IQR}, \mathrm{Q3}+1.5\mathrm{e}^{3\mathrm{MC}}\mathrm{IQR}]  & \text{if }  \mathrm{MC}\geq 0 \\
       \\
       [\mathrm{Q1}-1.5\mathrm{e}^{-3\mathrm{MC}}\mathrm{IQR}, \mathrm{Q3}+1.5\mathrm{e}^{4\mathrm{MC}}\mathrm{IQR}] & \text{if } \mathrm{MC}<0
\end{cases},
\end{equation}
where $\mathrm{IQR} = \mathrm{Q3} - \mathrm{Q1}$ refers to the interquartile range. The whiskers of the boxplot will then take the skewness of the distribution into account, depending on whether the distribution is positively skewed ($\mathrm{MC}>0$), negatively skewed ($\mathrm{MC}<0$) or symmetric ($\mathrm{MC}=0$).
This way several data considered to be outliers in conventinal boxplots will not be considered as such by adjusted boxplots (skewness should be taken into account for outlier consideration).

The probability distribution of $||\mathcal{N}(\mu_{T_1},\mu_{T_2})-\mathcal{N}(\rho,\omega,t)||_\diamond$ discussed in Chapter \ref{cp3} appears to have a heavy positive skew, and thus the adjusted version of the boxplot presented in \cite{skewed} is necessary for a correct representation of such results.

\end{appendices} \label{Appendix}
\clearemptydoublepage

\def\bibname{References}
\addcontentsline{toc}{chapter}{\protect\numberline{}References}
\bibliographystyle{alpha}

\makeatletter
\renewenvironment{thebibliography}[1]
     {
      \list{\@biblabel{\@arabic\c@enumiv}}%
           {\settowidth\labelwidth{\@biblabel{#1}}%
            \leftmargin\labelwidth
            \advance\leftmargin\labelsep
            \@openbib@code
            \usecounter{enumiv}%
            \let\p@enumiv\@empty
            \renewcommand\theenumiv{\@arabic\c@enumiv}}%
      \sloppy
      \clubpenalty4000
      \@clubpenalty \clubpenalty
      \widowpenalty4000%
      \sfcode`\.\@m}
     {\def\@noitemerr
       {\@latex@warning{Empty `thebibliography' environment}}%
      \endlist}
\makeatother



\end{document}